\documentclass[10pt]{article}

\usepackage[table]{xcolor} %
\usepackage{vmargin}
\setmarginsrb{28mm}{25mm}{28mm}{25mm}{0pt}{0mm}{0pt}{0mm}
\usepackage{amssymb}
\usepackage{amsmath}
\usepackage{amsthm}
\usepackage{pgfplots} 
\pgfplotsset{compat=1.9} %
\usepackage{graphicx}
\graphicspath{{Figures/}}
\usepackage[utf8]{inputenc}
\usepackage{tikz}
\usepackage{bbm}
\usepackage{mathtools}
\usepackage{subcaption}
\usepackage[boxruled]{algorithm2e}
\usepackage{parskip}
\usepackage[title,titletoc]{appendix}
\usepackage{booktabs}
\usepackage[numbers,sort&compress]{natbib} %

\usepackage{enumerate}
\usepackage[]{hyperref}
\hypersetup{colorlinks,citecolor=blue,linkcolor=blue,urlcolor=blue}
\usepackage{mathrsfs} %
\usepackage{placeins} %


\newcommand{\F}{\mathcal{F}}
\newcommand{\norm}[1]{\left\lVert#1\right\rVert}

\renewcommand{\phi}{\varphi}

\newtheorem{theorem}{Theorem}[section]
\newtheorem{proposition}{Proposition}[section]
\newtheorem{corollary}{Corollary}[section]
\newtheorem{lemma}{Lemma}[section]
\newtheorem{definition}{Definition}[section]
\newtheorem{remark}{Remark}[section]
\newtheorem{example}{Example}[section]

\newcounter{auxlem}
\renewcommand{\theauxlem}{\Alph{auxlem}}
\newcommand{\auxlemma}[1]{\refstepcounter{auxlem}\noindent\textbf{Lemma \theauxlem\ (#1).}}

\newcommand{\R}{\mathbb{R}}
\newcommand{\E}{\mathbb{E}}

\newcommand{\Bn}{\widetilde{B}}                  %
\newcommand{\Bb}{\bar{B}}                        %
\newcommand{\Db}{\bar{D}}                        %
\newcommand{\Lm}{L_-}                            %
\newcommand{\Lmn}{\mathcal{L}_-}                 %
\DeclareMathOperator{\tr}{tr}
\DeclareMathOperator{\vol}{vol}
\newcommand{\volb}{\overline{\operatorname{vol}}}
\newcommand{\Ib}{\mathit{Ib}}
\newtheorem{alg}{Algorithm}[section]

\numberwithin{equation}{section}
\newlength{\perdsbar}\newlength{\perdsrowtop}

\title{A Bochner--Weitzenb\"ock Decomposition for Cellular Sheaf
Laplacians,\\ Curvature and Neural Sheaf Diffusion}

\author{%
  Francesco Maria Guadagnuolo$^{1,2}$\thanks{Corresponding author:
  \texttt{guadagnuolo.2315644@studenti.uniroma1.it}}
  \and Alessio Borgi$^{1,3}$
  \and Dorina Thanou$^{4}$
  \and Pietro Li\`o$^{3}$
}

\date{%
  \normalsize
  $^{1}$Sapienza University of Rome, Rome, Italy\\
  $^{2}$Italian Institute of Artificial Intelligence (AI4I), Turin, Italy\\
  $^{3}$University of Cambridge, Cambridge, United Kingdom\\
  $^{4}$\'Ecole Polytechnique F\'ed\'erale de Lausanne (EPFL), Lausanne,
  Switzerland
}

\begin{document}

\maketitle

\begin{abstract}
We extend the Bochner--Weitzenb\"ock decomposition to cellular sheaf
Laplacians, splitting it into a positive semidefinite term,
the Bochner sheaf Laplacian, and a block diagonal term, the Forman
curvature.
We show that the Bochner Laplacians form a class of cellular sheaves that
extends both connection and matrix-weighted Laplacians, and generalizes
signed graphs to $d$ dimensions.
We give a holonomy characterization of their kernel, prove the related
Cheeger inequalities, and attach to them a random walk that generalizes
the Gremban expansion.
For the Forman term, we prove a discrete Bochner theorem, which bounds the
first cohomology of a cellular sheaf with nonnegative curvature, and we read
the whole decomposition as an opinion-dynamics model.
We then apply this machinery to Sheaf Neural Networks, whose learnt
Laplacian is otherwise unobserved.
The Bochner interpretation yields a node-level conviction score that acts as
an operational importance measure under frozen-model interventions, while
the Forman profile provides a local descriptor of the learnt operator.
\end{abstract}

\section{Introduction}
\label{sec:introduction}

A sheaf prescribes how local data glue into global data, and turns the
obstruction to gluing into a cohomological
invariant~\cite{bredon1997sheaf,iversen1986cohomology,kashiwara1990sheaves,houzel1990short}.
On a graph or a simplicial complex, a \emph{cellular
sheaf}~\cite{curry2014sheavescosheavesapplications,rosiak2022sheaf} attaches
a vector space to each cell and a linear map to each incidence, so that the
theory becomes finite linear algebra.
In this form sheaves have been used for fusing sensor
readings~\cite{robinson2016sheavescanonicaldatastructuresensor}, for
distributed coordination and
optimization~\cite{hanks2025distributedmultiagentcoordinationcellular,distributed_optimization},
and in topological data
analysis~\cite{Curry2015,Edelsbrunner2002,edelsbrunner2010computational},
wherever a relation on a graph needs more than a scalar weight.
Their central operator is the sheaf Laplacian, which has the foundations of
a spectral theory~\cite{hansen2019toward} and a signal processing
theory~\cite{dacunto2026sheaftheoreticsignalprocessinggraphs}, and which
models opinion dynamics, the stalks carrying private opinions and the
restriction maps the discourse between agents~\cite{Discourse_sheaves}.

In machine learning, where topology has entered from geometric deep
learning~\cite{bronstein2021geometricdeeplearninggrids} to higher-order
domains such as simplicial complexes and
hypergraphs~\cite{pontryagin1952foundations,bretto2013hypergraph}, the most
notable application of sheaves is the \emph{Sheaf Neural
Network}~\cite{hansen2020sheaf,bodnar2022neural}.
A graph neural network diffuses node features along the edges, and fails on
heterophilic graphs, where neighbours tend to carry different labels.
A Sheaf Neural Network learns a sheaf over the graph and diffuses with its
sheaf Laplacian, so that each node transforms its message differently for
each neighbour, which restores separation power on heterophilic data.
The sheaf is then an object produced by training.
The limiting behaviour of sheaf diffusion is described by the Hodge
decomposition, which identifies the kernel of the sheaf Laplacian with sheaf
cohomology.
Three questions that are standard for the graph Laplacian, however, remain
open for a general sheaf Laplacian.
The first is \emph{spectral}: a general spectral theory of cellular sheaves
is still far from
complete~\cite{hansen2019toward,dacunto2026sheaftheoreticsignalprocessinggraphs},
and Cheeger-type bounds, classical for graphs~\cite{chung1997spectral}, are
available only in special cases, notably connection
Laplacians~\cite{bandeira2013cheegerinequalitygraphconnection,singer2011vectordiffusionmapsconnection}.
The second is \emph{probabilistic}: it is not clear how to attach to a sheaf
Laplacian a vector-valued random walk, a question left open
in~\cite{hansen2019toward}.
The third is \emph{combinatorial}: cellular sheaves lack an analogue of
Forman's Bochner decomposition~\cite{Forman2003Bochner}, a local
curvature, cheaper to compute than the spectrum and easier to interpret.

We address the three at once with a Bochner decomposition for cellular
sheaves.
Forman splits the Hodge Laplacian of a complex into a positive semidefinite
term and a diagonal curvature term, using the absolute value of the
off-diagonal entries; replacing it by the matrix absolute value
$|A|:=\sqrt{AA^\top}$ extends the construction to every symmetric block
matrix, hence to every sheaf Laplacian, on a graph or on a complex.
We call \emph{Bochner matrices} the class of Laplacians the first term
belongs to.
It extends the connection Laplacians, the matrix-weighted
Laplacians~\cite{trinh2026matrix,expansion_matrix_weighted_graphs} and the
weighted signed graphs~\cite{Tian_2024}: a Bochner matrix is an anisotropic
connection Laplacian, in which a positive semidefinite matrix replaces the
scalar weight of an edge.
We show that its kernel admits the same holonomy characterization known for
connection Laplacians.

We then prove \emph{Cheeger inequalities} for Bochner matrices, extending
those of Bandeira, Singer and Spielman for connection
Laplacians~\cite{bandeira2013cheegerinequalitygraphconnection}: they relate
the low eigenvalues of the normalized Bochner Laplacian to three frustration
constants, the residuals of synchronization over $\mathbb{S}^{d-1}$, partial
or full, and over $O(d)$.
We also attach to a Bochner matrix a \emph{random walk} on
$V\times\mathbb{S}^{d-1}$, whose jump law depends on the orientation the
walker carries.
In the language of signed graphs it is a \emph{generalized Gremban
expansion}, each node being lifted to a whole sphere instead of two copies,
and in the isotropic case it recovers the vector-valued walk of the
connection Laplacian~\cite{singer2011vectordiffusionmapsconnection}.
These two results answer the spectral and the probabilistic question.

The \emph{Forman term} answers the combinatorial one.
It is block diagonal, each block computed from the restriction maps at a
single node.
On $1$-cells we call it the \emph{sheaf Forman Ricci curvature}, and we prove
a Bochner theorem for cellular sheaves: under nonnegative curvature, the
first sheaf cohomology of a compact connected combinatorial $n$-manifold with
$d$-dimensional stalks has dimension at most $nd$ for $n\leq3$, and at most
$d\cdot D(M)$ in any dimension, with $D(M)$ the local homological dimension
of the complex.
Together, the two terms give the decomposition a dynamical reading that
generalizes~\cite{Nurisso2025}: sheaf diffusion can be seen as a
reaction-diffusion process, the Bochner term diffusive and the Forman term
reactive.
On a graph this is an opinion dynamics model, in which the Bochner term
drives the agents towards a common \emph{conviction}, the norm a stalk
carries, and the Forman block measures how loudly an agent speaks relative to
its neighbours, its \emph{relative volume}.

We then return to Sheaf Neural Networks.
The available theory describes what they can represent: convergence of sheaf
diffusion, and a universal approximation theorem by which a network can learn
any sheaf on a graph.
What is missing is a description of what the learning mechanism \emph{does}:
which region of the sheaf space a training pipeline reaches, and whether
different pipelines are biased towards different regions.
From the decomposition we derive two diagnostics of neural sheaf diffusion,
one per term: the \emph{conviction} and the \emph{Forman profile}.
The conviction is an operational score of node importance: on a frozen
trained network, removing the highest-conviction nodes is the most damaging
intervention at every budget and in every perturbation mode, while removing
the lowest-conviction ones costs less than removing arbitrary nodes of the
same graph.
The Forman profile tracks, layer by layer, how a learning mechanism navigates
the sheaf space, and each of the three mechanisms we test leaves its own
signature: normalized general sheaf learning converges to local alignment
with the normalized graph Laplacian, unnormalized learning aligns with
neither reference, and joint
diffusion~\cite{caralt2024jointdiffusionprocessesinductive} reverses the role
of normalization.
Evaluated without normalization, the profile of the same learnt maps stays
far from that of the identity sheaf, so we read the alignment as an inductive
bias towards one region of the sheaf space.

\subsection{Related work}
\label{sec:related_work}

Sheaves on directed and undirected graphs have recently been given a
differential calculus~\cite{fioresi2026sheavesgraphsdifferentialcalculi},
formalized for the Grothendieck topology and developed in the spirit of
noncommutative geometry, where a Bochner Laplacian is defined as
$\nabla^*\nabla$ for $\nabla$ the connection of a bundle over a graph.
We work instead with cellular sheaves, hence over the Alexandrov topology,
and our starting point is the decomposition of an existing sheaf Laplacian.
A recent work also brings Forman curvature and sheaves
together~\cite{wang2026curvatureguidedsheafdiffusionunsupervised}, building
restriction maps as a function of the Forman curvature of the edges and
using the resulting sheaf Laplacian for community detection; the direction
there is the opposite of ours: we define a curvature \emph{from} a cellular
sheaf, while they use a graph curvature to define a particular sheaf.
Finally, dynamical processes on connection Laplacians have received
considerable attention
lately~\cite{gallo2025global,gallo2026synchronization,gallo2026turing,lambiotte_matrix_weighted},
with a focus on Kuramoto models, Turing patterns and network coherence, and
our framework opens the possibility of carrying those results to the
anisotropic setting while still controlling the balance of the system
through holonomy.
We note that these works use the term ``matrix-weighted graph'' to mean a
connection Laplacian, whereas we reserve it for the sense
of~\cite{trinh2026matrix}.

Curvature has also been studied on connection Laplacians independently of
sheaves.
Liu, M\"unch and Peyerimhoff~\cite{Liu_2019_curvature_buser} define a
Bakry--\'Emery curvature for connection graphs and derive higher-order Buser
inequalities from it, with Cheeger constants mixing the expansion of the
underlying graph with the frustration index of the signature.
This notion of curvature was moreover reformulated
in~\cite{hu2026discretebakryemerycurvaturetensors}.
Such curvatures are available on a connection Laplacian because it generates
a random walk, with scalar transition rates.
For a general sheaf Laplacian it is not clear how to obtain one, and that
failure is exactly the Forman term, which measures how far the diagonal is
from the sum of the magnitudes of the off-diagonal blocks.
Sheaf Laplacians with vanishing Forman term are the Bochner matrices, and
these, as we show in this work, do carry a random walk, so we believe a Bakry--\'Emery curvature can
be built for them, and we leave this to future work.

Sheaf Neural Networks were introduced by Hansen and
Gebhart~\cite{hansen2020sheaf} on a sheaf fixed in advance; Bodnar et
al.~\cite{bodnar2022neural} learn the sheaf from the data, and prove that a
network can learn any sheaf on a graph.
Several works have since studied the parametrization of the restriction
maps, through connection Laplacians aligned with the data
manifold~\cite{pmlr-v196-barbero22a}, equivariant
transport~\cite{borgi2026equivariantsheafneuralnetworks}, or a diffusion
process on the restriction maps themselves, run jointly with the one on the
node features~\cite{caralt2024jointdiffusionprocessesinductive}.
Others have extended the architecture, to features on SPD
manifolds~\cite{peng2026sheafneuralnetworksspd} and to message diffusion over
flat vector bundles~\cite{bamberger2025bundleneuralnetworksmessage}.
The sheaf formalism has also been used to generalize attention in graph
neural networks~\cite{ali2026letheadstalkdiagonal}.
On the theoretical side, oversmoothing has been related to a degeneration of
the representation underlying the learnt
sheaf~\cite{donmez2026oversmoothingrepresentationdegeneracyneural}.
On the other hand, Hernandez Caralt et
al.~\cite{caralt2026necessitylearnablesheaflaplacians} showed empirically
that identity restriction maps perform similarly to learnt ones, which
questions to what extent learning the sheaf matters.
Finally, crucial features of graph neural networks such as over-squashing
have been related to the curvature of the graph~\cite{Topping2021UnderstandingOA}.
While we build a notion of Forman curvature for cellular sheaves, we do not
pursue this direction, which would involve the Forman term of the sheaf
Hodge Laplacian in degree $1$; in Section~\ref{sec:forman_profile} we
compute instead the Forman term of the graph sheaf Laplacian, in degree
$0$, which to avoid ambiguity we call Forman term and not curvature.
Studying over-squashing in Sheaf Neural Networks from this curvature
perspective is left for future work.
We refer to Sections~\ref{sec:snn} and~\ref{sec:forman_profile} for a further
discussion.

\subsection{Organization of the paper}
\label{sec:organization}

The paper is organized as follows.
Section~\ref{sec:background} fixes notation and recalls cellular sheaves,
signed graphs and Forman's decomposition.
Section~\ref{sec:bochner_decomposition} introduces the Bochner decomposition
for sheaf Laplacians, and Sections~\ref{sec:bochner_matrices}
and~\ref{sec:forman_term} study its two terms in turn: within the first,
Section~\ref{sec:cheeger} proves the Cheeger inequalities and
Section~\ref{sec:random_walk} attaches the random walk.
Section~\ref{sec:applications} develops the dynamical reading and the
applications to Sheaf Neural Networks.
We conclude in Section~\ref{sec:conclusion}.
Proofs deferred from the body are collected in Appendix~\ref{sec:proofs},
each statement being restated before its proof, further properties of
Bochner matrices are in Appendix~\ref{sec:further}, and the experimental
protocols with the additional material are in
Appendices~\ref{sec:supp_conv} and~\ref{sec:supp_forman}.

\section{Background}
\label{sec:background}

This section fixes notation for cellular sheaves and their Laplacians, and
recalls the facts on signed graphs and on Forman's Bochner decomposition that
the paper builds on.

\subsection{Cellular sheaves on simplicial complexes}
\label{sec:sheaves}

We follow~\cite{muhammad2006control_kyoto,grady2010discrete} for simplicial
cochain complexes and~\cite{rosiak2022sheaf,curry2014sheavescosheavesapplications}
for cellular sheaves.
Let $X$ be a finite simplicial complex, $X_k$ its set of $k$-simplices and
$n:=|X_0|$.
We write $\sigma\trianglelefteq\tau$ for inclusion of simplices, with no
restriction on the difference in dimension, so that the simplices of $X$ form
the face poset; the words face and coface are
reserved for the codimension-one case, the only one entering the formulas
below.
Fixing an orientation of each simplex, $\epsilon_{\sigma\tau}\in\{-1,0,+1\}$
is the relative orientation of a codimension-one pair, and vanishes when
$\sigma$ is not a face of $\tau$.

\begin{definition}[Cellular sheaf]
\label{def:cellular_sheaf}
A cellular sheaf $\F$ on $X$ assigns to every simplex $\sigma$ a
finite-dimensional real vector space $\F(\sigma)$, its stalk, and
to every pair $\sigma\trianglelefteq\tau$ a linear restriction map
$\F_{\sigma\trianglelefteq\tau}\colon\F(\sigma)\to\F(\tau)$, with
$\F_{\sigma\trianglelefteq\sigma}=\mathrm{id}_{\F(\sigma)}$ and
$\F_{\rho\trianglelefteq\tau}=\F_{\sigma\trianglelefteq\tau}\,\F_{\rho\trianglelefteq\sigma}$
whenever $\rho\trianglelefteq\sigma\trianglelefteq\tau$.
\end{definition}

Equivalently, $\F$ is a functor from the face poset to finite-dimensional
real vector spaces~\cite{curry2014sheavescosheavesapplications}, and it is
determined by the maps attached to faces.
Throughout, every stalk has dimension $d$.
On $C^k(X;\F)=\bigoplus_{\sigma\in X_k}\F(\sigma)$ the coboundary
$\delta_k\colon C^k(X;\F)\to C^{k+1}(X;\F)$ is:
\begin{equation}
  (\delta_kc)_\tau \;=\;
  \sum_{\substack{\sigma\trianglelefteq\tau\\ \dim\sigma=k}}
  \epsilon_{\sigma\tau}\,\F_{\sigma\trianglelefteq\tau}\,c_\sigma ,
\end{equation}
it satisfies $\delta_{k+1}\delta_k=0$, and
$\mathcal H^k(X;\F):=\ker\delta_k/\operatorname{im}\delta_{k-1}$ is the
sheaf cohomology of $\F$; the constant sheaf $\R$ recovers the
simplicial cochain complex of $X$.
With the standard inner product on each $C^k(X;\F)$ and $\delta_k^*$ the
adjoint of $\delta_k$, the Hodge sheaf Laplacian
$H_k^\F=\delta_{k-1}\delta_{k-1}^*+\delta_k^*\delta_k$
is a positive semidefinite matrix with $d\times d$ blocks indexed by $X_k$,
and $\ker H_k^\F\cong\mathcal H^k(X;\F)$.
We write $\partial:=\delta^*$ for the boundary operator, with the
degree left implicit, so that $\partial^*=\delta$.
At $k=0$ we write $L_\F:=H_0^\F=\delta_0^*\delta_0$ for the graph
sheaf Laplacian; on a graph $G=(V,E)$ its blocks are
$(L_\F)_{vv}=\sum_{e:\,v\trianglelefteq e}\F_{v\trianglelefteq e}^\top\F_{v\trianglelefteq e}$
and, on an edge $e=(v,u)$,
$(L_\F)_{vu}=-\F_{v\trianglelefteq e}^\top\F_{u\trianglelefteq e}$.

Two classes of graph sheaf Laplacians recur below.
Taking $\F_{v\trianglelefteq e}=\sqrt{w_{vu}}\,O_{v\trianglelefteq e}$ with
$O_{v\trianglelefteq e}$ orthogonal and $w_{vu}>0$ gives off-diagonal blocks
$-w_{vu}W_{vu}^\top$, with $W_{vu}^\top:=O_{v\trianglelefteq e}^\top
O_{u\trianglelefteq e}\in O(d)$: these are the \emph{connection
Laplacians}~\cite{singer2011vectordiffusionmapsconnection,bandeira2013cheegerinequalitygraphconnection},
called \emph{bundle sheaves} in Sheaf Neural
Networks~\cite{hansen2020sheaf,bodnar2022neural}.
Taking $\F_{v\trianglelefteq e}=M_{vu}^{1/2}$ with $M_{vu}\succeq0$ gives
off-diagonal blocks $-M_{vu}$, without rotation: these are the
\emph{matrix-weighted
Laplacians}~\cite{trinh2026matrix,expansion_matrix_weighted_graphs}.
With $D:=\operatorname{blockdiag}\bigl((L_\F)_{vv}\bigr)_{v\in V}$ the degree
matrix, the \emph{normalized sheaf Laplacian}
$\Delta_\F:=D^{-1/2}L_\F D^{-1/2}$ has spectrum in
$[0,2]$~\cite{hansen2019toward}.
We work with unweighted sheaves: every weighted sheaf Laplacian, and in
particular every normalized one, is the Laplacian of some unweighted
sheaf~\cite{hansen2019toward}, a fact Sections~\ref{sec:bochner_decomposition}
and~\ref{sec:forman_profile} use directly.

\subsection{Signed graphs, balance and holonomy}
\label{sec:signed_graphs}

A signed graph $\Gamma=(V,E,W)$ carries a symmetric weighted
adjacency matrix $W$ whose entries may be negative~\cite{Tian_2024}.
With degrees $d_v:=\sum_u|w_{vu}|$, $D:=\operatorname{diag}(d_v)$, and
$W=W^+-W^-$ split into its nonnegative parts, the signed Laplacian
is $L_\Gamma:=D-W=D-W^++W^-$, and its quadratic form is:
\begin{equation}
\label{eq:signed_energy}
  x^\top L_\Gamma x \;=\;
  \sum_{vu\in E}|w_{vu}|\,\bigl(x_v-\operatorname{sgn}(w_{vu})\,x_u\bigr)^2 ,
\end{equation}
so $L_\Gamma\succeq0$.
Unlike the unsigned case, the kernel may be trivial, and its dimension is a
structural invariant~\cite{ZASLAVSKY198247}, controlled by balance.
A cycle is positive if it carries an even number of negative edges.

\begin{definition}[Balance and antibalance]
\label{def:balance}
$\Gamma$ is balanced if every cycle is positive, and
antibalanced if the signed graph obtained by reversing every edge
sign is balanced.
\end{definition}

\begin{theorem}[Structure theorem for balance~\cite{Harary1953OnTN,structural_duality}]
\label{thm:structural_balance}
$\Gamma$ is balanced if and only if there is a partition $V=V_1\sqcup V_2$
such that every edge within $V_1$ or within $V_2$ is positive and every edge
between $V_1$ and $V_2$ is negative.
It is antibalanced if and only if the same holds with the two signs
exchanged.
\end{theorem}

The partition is visible in the
spectrum~\cite{Harary1953OnTN,gallier2016spectraltheoryunsignedsigned,Atay_2020_cheeger}.
Writing $\mathbbm{1}_{V_1}$ for the indicator vector of $V_1$:

\begin{theorem}[Kernel of a signed Laplacian]
\label{thm:kernel_signed}
Let $\Gamma$ be connected.
Then $\Gamma$ is balanced if and only if $\dim\ker L_\Gamma=1$, in which
case the kernel is spanned by $\mathbbm{1}_{V_1}-\mathbbm{1}_{V_2}$ for the
partition of Theorem~\ref{thm:structural_balance}.
\end{theorem}

At the other end of the spectrum, the normalized signed Laplacian
$I-D^{-1/2}WD^{-1/2}$ has spectrum in $[0,2]$ and attains $2$ when $\Gamma$ is
antibalanced~\cite{antibalance_signed_graph}; a connected non-bipartite
signed graph cannot be both balanced and antibalanced, since an odd cycle
would be positive and negative at once.
When $\Gamma$ is not balanced, its frustration index, the minimal
number of edges whose removal makes it balanced, measures how far it is from
balance~\cite{zaslavsky1987balanced,harary1959measurement}.
Other quantitative measures of balance have also been proposed, together
with related Cheeger inequalities for signed
graphs~\cite{Atay_2020_cheeger,Lange_2015_frustration}.
In the \emph{Gremban
expansion}~\cite{GrembanPHD,Numerical_methods_Gremban_expansion,Investigation_spectral_signed_matrix}
each vertex is split into a positive and a negative copy, a walker crossing a
negative edge switching copy, so that $L_\Gamma$ becomes an unsigned graph
Laplacian on $2n$ vertices~\cite{Tian_2024}.

On a connection Laplacian, balance becomes holonomy.
Every walk $\gamma:v_0\to\cdots\to v_k$ carries the \emph{transport map}
$T_\gamma:=W_{v_{k-1}v_k}\cdots W_{v_0v_1}\in O(d)$,
the connection has \emph{trivial holonomy} when $T_\gamma=I$ for every
closed walk, and the kernel of the connection Laplacian has dimension at most
$d$, with equality exactly in that
case~\cite{singer2011vectordiffusionmapsconnection,bandeira2013cheegerinequalitygraphconnection,bodnar2022neural}.
For $d=1$, $T_\gamma$ is the product of the signs along $\gamma$, trivial
holonomy is balance, and this is Theorem~\ref{thm:kernel_signed}.

\subsection{Forman's Bochner decomposition}
\label{sec:forman_background}

On a Riemannian manifold the Bochner--Weitzenb\"ock formula splits the Hodge
Laplacian on $1$-forms as
$\Delta=\nabla^*\nabla+\mathrm{Ric}$~\cite{Petersen2016RiemannianGeometry},
and the Bochner theorem follows from it: under nonnegative Ricci curvature
every harmonic $1$-form is parallel, so the first cohomology has dimension
at most $n$.
Forman gave a combinatorial counterpart of both~\cite{Forman2003Bochner},
resting on a splitting that applies to any symmetric matrix.

\begin{definition}[Bochner and curvature terms]
\label{def:scalar_bochner}
Let $A$ be a symmetric $n\times n$ matrix.
Define $B$ by $B_{ij}=A_{ij}$ for $i\neq j$ and $B_{ii}=\sum_{k\neq
i}|A_{ik}|$, and set $F:=A-B$, the diagonal matrix with
$F_{ii}=A_{ii}-\sum_{k\neq i}|A_{ik}|$.
\end{definition}

The matrix $B$ is diagonally dominant, hence positive semidefinite, and it is
a signed graph Laplacian; conversely a signed graph Laplacian is its own
Bochner term, with $F=0$.
In the scalar case the decomposition therefore reads: every symmetric matrix
is a signed graph Laplacian plus a diagonal.

When $A$ is a Hodge Laplacian, the diagonal remainder can be seen as a
curvature.
Forman works on weighted complexes, with a weight $w_\sigma>0$ on each
simplex.
Writing the Hodge Laplacian in the orthonormal basis of the weighted inner
product, applying Definition~\ref{def:scalar_bochner} and rescaling by
$w_\sigma$ gives Forman's curvature~\cite{Forman2003Bochner}:
\begin{equation}
  \begin{aligned}
  \mathscr{F}_k(\sigma) = w_\sigma\Biggl[
    &\sum_{\sigma \trianglelefteq\tau^{(k+1)}}\,\frac{w_\sigma}{w_\tau}
    + \sum_{\gamma^{(k-1)} \trianglelefteq \sigma} \frac{w_\gamma}{w_\sigma}\\
    &- \sum_{\sigma'\neq\sigma}
    \left|
      \sum_{\tau^{(k+1)}} \epsilon_{\sigma\tau}\epsilon_{\sigma'\tau}
        \frac{\sqrt{w_\sigma w_{\sigma'}}}{w_\tau}
      + \sum_{\gamma^{(k-1)}} \epsilon_{\gamma\sigma}\epsilon_{\gamma\sigma'}
        \frac{w_\gamma}{\sqrt{w_\sigma w_{\sigma'}}}
    \right|
  \Biggr],
  \end{aligned}
\end{equation}
the bracket being exactly the Forman term $(F_k)_{\sigma\sigma}$, so that
$\mathscr{F}_k(\sigma)=w_\sigma\cdot(F_k)_{\sigma\sigma}$ and the two
coincide when every weight is $1$.
A positive factor does not change the sign on which every result below
depends, and in the sheaf setting there are no weights, by the reduction of
Section~\ref{sec:sheaves}; therefore, we identify the curvature with the term
and write $F_k$ and $\mathrm{Ric}$.

With unit weights, the Forman term is a local count.
Two distinct $k$-simplices are upper adjacent,
$\sigma\frown\sigma'$, if they share a coface, and lower adjacent,
$\sigma\smile\sigma'$, if they share a face; they are \emph{parallel} if
exactly one of the two holds, and \emph{transverse} if both do, in which
case the contributions through the face and through the coface cancel.
With $\deg_u$ and $\deg_l$ the numbers of cofaces and of faces:
\begin{equation}
  F_k(\sigma) \;=\; \deg_u(\sigma) + \deg_l(\sigma)
  \;-\; \#\{\text{parallel neighbours of }\sigma\},
\end{equation}
which on an edge, where $\deg_l(e)=2$, is the \emph{Forman--Ricci
curvature}:
\begin{equation}
  \mathrm{Ric}(e) \;=\; \#\{\text{triangles containing } e\} + 2
  - \#\{\text{parallel neighbours of } e\}.
\end{equation}
At $k=0$ every neighbour of a vertex is parallel, so $F_0\equiv0$, and the
construction carries no information at the vertices of a graph.

Forman's theorem holds on quasiconvex complexes.
A CW complex is regular when the characteristic map of every cell
is a homeomorphism onto its image, so that the closure of each cell is a
closed ball; incidence numbers then lie in $\{0,\pm1\}$, and exactly two
$p$-cells lie between a $(p-1)$-cell and a $(p+1)$-cell.
A compact regular CW complex is quasiconvex when two
$(p+1)$-cells whose closures share a $p$-cell $\beta$ meet exactly in
$\bar\beta$.
Every simplicial and every polyhedral complex is quasiconvex; our results
hold at this generality, and we write ``simplicial complex'' for readability.
A combinatorial $n$-manifold is a complex in which the link of every
vertex $v$, the subcomplex $\{\sigma:v\notin\sigma,\ \sigma\cup\{v\}\in X\}$,
is a triangulated $(n-1)$-sphere.

\begin{theorem}[Forman's Bochner theorem]
\label{thm:forman_scalar}
Let $M$ be a compact, connected, quasiconvex combinatorial $n$-manifold with
$n\leq3$, and suppose $\mathrm{Ric}(e)\geq0$ for every edge $e$.
Then $\dim\mathcal H^1(M)\leq n$.
\end{theorem}

\section{The Bochner Decomposition of a Sheaf Laplacian}
\label{sec:bochner_decomposition}

For sheaf Laplacians with $d$-dimensional stalks the off-diagonal blocks
$H_{ij}$ are $d\times d$ matrices, so the scalar absolute value appearing in
Forman's construction must be replaced by the \emph{matrix absolute value}
$|A| := \sqrt{AA^\top}$.
All stalks are assumed to have the same dimension $d$.

\begin{definition}[Block Bochner decomposition]
\label{def:block_bochner}
Let $H$ be a symmetric block matrix in $\mathbb{R}^{nd,nd}$ with $d\times d$
blocks $H_{ij}$.
The \emph{block matrix Bochner decomposition} is:
\begin{equation}
  B_{ij} = \begin{cases} H_{ij} & i \neq j, \\ \sum_{j\neq i}|H_{ij}| & i=j, \end{cases}
  \qquad
  F_{ij} = \begin{cases} 0 & i \neq j, \\ H_{ij} - \sum_{k\neq i}|H_{ik}| & i=j. \end{cases}
\end{equation}
\end{definition}

For this to define a Laplacian, $B$ must be positive semidefinite.
This holds for a class of block matrices wider than the Bochner term itself.

\begin{definition}[Block diagonally dominant matrix]
\label{def:block_diag_dominant}
A block matrix $C \in \R^{nd,nd}$, with blocks $C_{ij}$ of dimension $d$, is
\emph{block diagonally dominant} if $C_{ii} \succeq \sum_{j\neq
i}|C_{ij}|$ for every $i$.
\end{definition}

\begin{proposition}
\label{prop:psd}
Every symmetric, block diagonally dominant matrix is positive semidefinite.
\end{proposition}
\begin{proof}
Fix $i \neq j$ and write $C_{ij} = U\Sigma V^\top$ via the SVD, so that
$|C_{ij}| = U\Sigma U^\top$ and $|C_{ij}^\top| = V\Sigma V^\top$.
The block matrix $\begin{psmallmatrix}|C_{ij}| & C_{ij}\\ C_{ij}^\top &
|C_{ij}^\top|\end{psmallmatrix} =
\begin{psmallmatrix}U\\V\end{psmallmatrix}\Sigma\begin{psmallmatrix}U^\top &
V^\top\end{psmallmatrix} \succeq 0$, and evaluating its quadratic form at
$(x_i^\top,x_j^\top)^\top$ gives:
\begin{equation}
  x_i^\top C_{ij}x_j + x_j^\top C_{ij}^\top x_i
  \;\geq\;
  -\bigl(x_i^\top|C_{ij}|x_i + x_j^\top|C_{ij}^\top|x_j\bigr).
\end{equation}
Summing this bound over every off-diagonal pair:
\begin{equation}
  x^\top Cx
  = \sum_i x_i^\top C_{ii}x_i + \sum_{i\neq j} x_i^\top C_{ij}x_j
  \;\geq\;
  \sum_i x_i^\top\Bigl(C_{ii} - \sum_{j\neq i}|C_{ij}|\Bigr)x_i \;\geq\; 0.
\end{equation}
\end{proof}

\begin{remark}
The step above, $\begin{psmallmatrix}|A| & A\\ A^\top &
|A^\top|\end{psmallmatrix}\succeq 0$ for every matrix $A$, is essentially
Corollary~1.3.7 of Bhatia~\cite{PSD_Bhatia2007}.
Feingold and Varga's block diagonal dominance
criterion~\cite{block_diagonally_dominant_matrices} uses operator norms in
place of the matrix absolute value, so Proposition~\ref{prop:psd} is a
different hypothesis.
\end{remark}

By the proposition above, $B \succeq 0$, hence $B = C^\top C$ for some $C$,
in analogy with $B_k = \nabla^*\nabla$ in the continuous case.

\noindent\textbf{Graph and Hodge sheaf Laplacians.}
Applying Definition~\ref{def:block_bochner} to the Hodge sheaf Laplacian
$H_k^\F = \delta_{k-1}\delta_{k-1}^* + \delta_k^*\delta_k$ gives
$(B_k)_{\sigma\sigma'} = (H_k)_{\sigma\sigma'}$ off the diagonal and
$(F_k)_{\sigma\sigma} = (H_k)_{\sigma\sigma} -
\sum_{\sigma'\neq\sigma}|(H_k)_{\sigma\sigma'}|$ on it, so that $H_k^\F =
B_k^\F + F_k^\F$, writing $L_\F = B_\F + F_\F$ for $k=0$, the graph case.
The Forman term is block diagonal, an endomorphism acting within each stalk,
not a scalar.
Following Forman's terminology we call $\mathrm{Ric}_\F := F_1^\F$ the
curvature on $1$-cells; on an edge $e$ its block is:
\begin{equation}
 \begin{aligned}
 \mathrm{Ric}_\F(e) = &\sum_{v:\,v\trianglelefteq e}\F_{v\trianglelefteq e}\F_{v\trianglelefteq e}^\top
  + \sum_{t:\,e\trianglelefteq t}\F_{e\trianglelefteq t}^\top\F_{e\trianglelefteq t}\\
  &- \sum_{\tilde e\neq e}\Bigl|
      \sum_v \epsilon_{ve}\epsilon_{v\tilde e}\F_{v\trianglelefteq e}\F_{v\trianglelefteq \tilde e}^\top
      + \sum_t \epsilon_{et}\epsilon_{\tilde e t}\F_{e\trianglelefteq t}^\top\F_{\tilde e\trianglelefteq t}
    \Bigr|,
 \end{aligned}
  \label{eq:hodge_sheaf_forman}
\end{equation}
$\epsilon_{\alpha\beta}\in\{-1,0,+1\}$ encoding relative orientation, and
the general $k$ case being identical with $k$-cells in place of edges.
Specializing to $k=0$, where cells have no faces, gives the vertex block:
\begin{equation}
  (F_\F)_{vv} = \sum_{e:\,v\trianglelefteq e}
    \Bigl(\F_{v\trianglelefteq e}^\top\F_{v\trianglelefteq e}
    - \sqrt{\F_{v\trianglelefteq e}^\top\F_{u\trianglelefteq e}\,\F_{u\trianglelefteq e}^\top\F_{v\trianglelefteq e}}\Bigr).
  \label{eq:sheaf_forman}
\end{equation}
with $u$ the other endpoint of $e$.
Unlike the classical graph Laplacian, whose Forman term vanishes at $k=0$,
here it need not, so the vertex-level case is worth analyzing in its own
right.
A second difference with the classical theory is that there is no weight
factor $w_\sigma$ out front: by~\cite{hansen2019toward}, every weighted
sheaf Laplacian is the Laplacian of some unweighted sheaf, so no generality
is lost by working unweighted, and the Forman term of the decomposition
is the Forman curvature itself, with no rescaling.
For a higher-order connection Laplacian, with $\F_{v\trianglelefteq
e}\F_{v\trianglelefteq\tilde e}^\top = \F_{e\trianglelefteq
t}^\top\F_{\tilde e\trianglelefteq t}$ whenever $e,\tilde e$ share both a
face and a coface, \eqref{eq:hodge_sheaf_forman} collapses to $(F_\F)_{ee} =
\chi(e)\cdot I$, with $\chi(e)$ exactly the scalar Forman curvature of the
unweighted simplicial complex.

\section{Bochner Matrices}
\label{sec:bochner_matrices}

The Bochner term of a sheaf Laplacian belongs to a class of block matrices
that is worth studying on its own, as we show it contains Connection
Laplacians, Matrix-Weighted Laplacians, and signed graphs.
This section develops that class: its kernel and the top of its spectrum
here, its Cheeger inequalities, and the random walk it carries.

\subsection{Polar factors, the kernel and holonomy}

We call a symmetric block matrix produced by
Definition~\ref{def:block_bochner} a Bochner matrix:
\begin{equation}
  B = \begin{pmatrix}
    \ddots & \vdots & & \vdots & \\
    \cdots & \sum_{w\sim v}|B_{vw}| & \cdots & -B_{vu} & \cdots \\
    & \vdots & \ddots & \vdots & \\
    \cdots & -B_{vu}^\top & \cdots & \sum_{w\sim u}|B_{uw}| & \cdots \\
    & \vdots & & \vdots & \ddots
  \end{pmatrix},
\end{equation}
supported exactly on the edges of $G$.
Throughout we assume (H1) \emph{block symmetry}, $B_{uv}=B_{vu}^\top$; (H2)
\emph{invertibility}, every $B_{vu}$ invertible; (H3) \emph{connectivity}:
the graph on $v\sim u\iff B_{vu}\neq0$ is connected.

Writing $B_{vu}=U_{vu}\Sigma_{vu}V_{vu}^\top$ via the SVD and setting
$\mathcal{B}_{v\trianglelefteq e}:=\Sigma_{vu}^{1/2}U_{vu}^\top$,
$\mathcal{B}_{u\trianglelefteq e}:=\Sigma_{vu}^{1/2}V_{vu}^\top$ on each
edge $e=(v,u)$ defines a sheaf $\mathcal{B}$ with graph Laplacian $B$.
So the class of Bochner matrices is included in the class of graph sheaf
Laplacians.
Specializing the Bochner decomposition to $k=0$ gives $L_\F =
L_{\mathcal{B}} + F_\F$: every sheaf Laplacian is a Bochner sheaf Laplacian
plus a block-diagonal perturbation, the Forman term.
More generally, considering the decomposition $H_k^\F = B_k^\F + F_k^\F$,
$B_k^{\F}$ lives instead on the order-$k$ adjacency graph of the
complex, whose nodes are the $k$-simplices, two being adjacent iff they are
lower or upper adjacent.

Writing $B_{vu}=U_{vu}\Sigma_{vu}V_{vu}^\top =
U_{vu}\Sigma_{vu}U_{vu}^\top\,U_{vu}V_{vu}^\top$ gives the polar
decomposition $B_{vu}=M_{vu}W_{vu}^\top$, with radial factor
$M_{vu}:=|B_{vu}|=U_{vu}\Sigma_{vu}U_{vu}^\top\succeq0$ and angular factor
$W_{vu}:=V_{vu}U_{vu}^\top\in O(d)$.
Since $B_{uv}=B_{vu}^\top$, we get $W_{uv}=W_{vu}^\top$ and
$W_{vu}M_{vu}W_{vu}^\top=M_{uv}$.
Write $\norm{z}_M^2:=z^\top Mz$ for $M\succeq0$.
Since $|B_{vu}|=M_{vu}$, the diagonal blocks of $B$ are
$D_v:=B_{vv}=\sum_{u\sim v}M_{vu}$, and we write
$D:=\operatorname{blockdiag}(D_1,\dots,D_n)$.
A Bochner matrix could thus be seen as an \emph{anisotropic bundle sheaf}:
a rotation $W_{vu}$, as in a connection Laplacian, together with a
connection strength $M_{vu}$ that varies not only across edges but across
directions of the stalk.
This is the feature separating Bochner matrices from connection Laplacians,
and what makes them a multidimensional generalization of weighted signed
graphs, the sign of an edge replaced by $W_{vu}$.

\begin{remark}[On the terminology]
Throughout this work we treat Bochner matrices as a class extending the
connection Laplacians of Singer and
coauthors~\cite{singer2011vectordiffusionmapsconnection,bandeira2013cheegerinequalitygraphconnection}.
In differential geometry the Bochner Laplacian and the connection Laplacian
are essentially the same operator.
In our discrete setting they refer to distinct families of Laplacians, but
this is simply because they come from two different constructions:
connection Laplacians were obtained by discretizing bundles
over manifolds, in the context of local PCA and vector diffusion maps,
whereas we arrive at Bochner matrices in purely algebraic terms, as the
first term of the Bochner--Weitzenb\"ock decomposition of
Section~\ref{sec:bochner_decomposition}.
The two nonetheless share the properties needed to model holonomy and
parallel transport on a graph.
\end{remark}

We now turn to the kernel characterization of the Bochner Laplacian.

\begin{lemma}[Energy identity]
\label{lem:energy_id}
For every $x\in\R^{nd}$, $x^\top Bx = \frac12\sum_{v\neq u}\norm{x_v -
W_{vu}^\top x_u}^2_{M_{vu}}\geq 0$, each unordered edge contributing equally
from either orientation.
\end{lemma}
\begin{proof}
Expand $\norm{x_v-W_{vu}^\top x_u}_{M_{vu}}^2 = x_v^\top M_{vu}x_v
-2x_v^\top M_{vu}W_{vu}^\top x_u + x_u^\top W_{vu}M_{vu}W_{vu}^\top x_u$.
The last term is $x_u^\top M_{uv}x_u$, so summing the first and last terms
over all ordered pairs and relabelling gives $\sum_v x_v^\top D_vx_v$; the
middle term is $-2x_v^\top B_{vu}x_u$, the off-diagonal contribution to
$x^\top Bx$.
Nonnegativity is immediate since each $M_{vu}\succeq 0$.
\end{proof}

When $M_{vu}=w_{vu}I$ this recovers the Dirichlet energy $\sum_{v\neq
u}w_{vu}\norm{x_v-W_{vu}^\top x_u}^2$ of the Connection
Laplacian~\cite{bandeira2013cheegerinequalitygraphconnection}.
Since $B\succeq0$, Lemma~\ref{lem:energy_id} shows $x\in\ker(B)$ iff every
summand vanishes:

\begin{proposition}[Kernel]
\label{prop:kernel}
Let $B$ satisfy (H1)--(H2).
Then $x\in\ker(B)$ iff $x_u = W_{vu}x_v$ for every edge $(v,u)$; in
particular $\norm{x_v}$ is constant on each connected component.
\end{proposition}
An immediate corollary: each summand vanishes iff $x_v=W_{vu}^\top x_u$, and
the orthogonal $W_{vu}$ then preserves the norm along every edge.

With the kernel in hand, the same evaluation-map argument used for the
Connection Laplacian~\cite{bandeira2013cheegerinequalitygraphconnection,
bodnar2022neural} applies verbatim (proof in Appendix~\ref{sec:proofs}):

\begin{proposition}[Kernel and holonomy]
\label{prop:holonomy}
Let $B$ satisfy (H1)--(H3), and for a walk $\gamma:v_0\to\cdots\to v_k$ let
$T_\gamma := W_{v_{k-1}v_k}\cdots W_{v_0v_1}\in O(d)$.
Then $\dim\ker(B)\leq d$, with equality iff $T_\gamma = I$ for every closed
walk $\gamma$.
\end{proposition}
While the technique is unchanged from the connection-Laplacian case, the
result is novel: despite carrying the extra radial data $M_{vu}$, Bochner
matrices admit exactly the same holonomy characterization of the kernel,
that depends only on angular factors.
The same two propositions also settle the opposite end of the spectrum.
By~\cite{hansen2019toward}, the normalized Laplacian of any weighted
cellular sheaf has spectrum in $[0,2]$, with the same degree normalization
used here; in particular $\operatorname{spec}(\Bn)\subseteq[0,2]$, writing
$\Bn:=D^{-1/2}BD^{-1/2}$.
The upper value is governed by the \emph{twist matrix} $N:=2D-B$: since
$N_{vu}=-B_{vu}=-M_{vu}W_{vu}^\top$, $N$ is again a Bochner matrix, with the
same radial factors as $B$ but angular factors $-W_{vu}$, whence $x^\top
Bx\leq 2x^\top Dx$.
Attaining the value $2$ is therefore the balance of $N$, read through its
twisted transport maps $\widehat
T_\gamma:=(-1)^{|\gamma|}T_\gamma$, and this is the multidimensional
counterpart of antibalance.

\begin{theorem}[Antibalance characterization]
\label{thm:antibalance}
Let $E_2(\Bn)$ be the eigenspace of $\Bn$ for eigenvalue $2$.
Then $E_2(\Bn) = D^{1/2}\ker(N)$, with $z\in\ker(N)$ iff $z_u=-W_{vu}z_v$ on
every edge; and $\dim E_2(\Bn)\leq d$, with equality iff $T_\gamma =
(-1)^{|\gamma|}I$ for every closed walk.
\end{theorem}
This is immediate from $2I-\Bn = D^{-1/2}ND^{-1/2}$ together with
Propositions~\ref{prop:kernel} and~\ref{prop:holonomy} applied to $N$.
When $G$ is not bipartite, every $x\in\ker(B)$ and every $z\in\ker(N)$,
the latter giving $E_2(\Bn)=D^{1/2}\ker(N)$, are moreover orthogonal at
every vertex, $\langle x_v,z_v\rangle=0$, and the two spaces can coexist,
which has no scalar analogue (Appendix~\ref{sec:further}).

\subsection{Cheeger inequalities and the synchronization problems}
\label{sec:cheeger}

Seeing Bochner matrices as a generalization of both matrix-weighted
Laplacians and connection Laplacians lets us extend techniques known for
these two classes.
Hansen's comparison between the block spectrum of a matrix-weighted
Laplacian and a scalar one~\cite{expansion_matrix_weighted_graphs} extends
to Bochner matrices (Appendix~\ref{sec:further}).
Here we follow the synchronization framework of Bandeira, Singer and
Spielman~\cite{bandeira2013cheegerinequalitygraphconnection,Gao2021},
extended to the anisotropic setting.
Each of the three theorems below comes with an algorithm attaining the
bound, since the nontrivial half of a Cheeger-type inequality is proved by
exhibiting a good combinatorial object.

\begin{definition}[Anisotropy constants and companion graph]
\label{def:anisotropy}
For each edge, $w^-_{vu}:=\lambda_{\min}(M_{vu})$,
$w^+_{vu}:=\lambda_{\max}(M_{vu})$, and we set:
\begin{equation}
  \kappa_E := \max_{e\in E}\frac{w^+_e}{w^-_e} \geq 1,
  \qquad
  \kappa_D := \max_{v\in V}\frac{\lambda_{\max}(D_v)}{\lambda_{\min}(D_v)}.
\end{equation}
For each vertex, $d^\pm_v := \sum_{u\sim v}w^\pm_{vu}$.
The companion graph $G_-$ is the scalar weighted graph on $V$ with
weights $w^-_{vu}$ and degree matrix $D^-_0:=\operatorname{diag}(d^-_v)$;
its Laplacian $\Lm$ has $(\Lm)_{vv}=d^-_v$ and $(\Lm)_{vu}=-w^-_{vu}$ for
$v\neq u$, and its normalized Laplacian is
$\Lmn:=(D^-_0)^{-1/2}\Lm(D^-_0)^{-1/2}$.
For $x\in\R^{nd}$, the local norm vector $n_x\in\R^n$ is
$(n_x)_v:=\norm{x_v}$.
\end{definition}
By (H2)--(H3), $G_-$ has the same edges as $G$ and $\lambda_2(\Lmn)>0$.
Connection Laplacians are exactly the case $\kappa_E=1$; matrix-weighted
graphs ($W\equiv I$) place no constraint on $\kappa_E$, so the two
degenerations are independent, and the results below reduce to the classical
ones in the first.

\begin{definition}[Frustration constants]
\label{def:frustration}
For $M\succeq0$ write $\norm{A}^2_{M,F}:=\tr(A^\top MA)$.
With $\eta(x):=x^\top Bx/x^\top Dx$, set:
\begin{equation}
  \eta_G := \min_{x:V\to\mathbb{S}^{d-1}}\eta(x),
  \qquad
  \eta^*_G := \min_{\substack{x:V\to\mathbb{S}^{d-1}\cup\{0\}\\x\neq0}}\eta(x),
\end{equation}
and, for a group potential $g:V\to O(d)$:
\begin{equation}
  \nu(g) := \frac{\frac12\sum_{v\neq u}\norm{g_v-W_{vu}^\top g_u}^2_{M_{vu},F}}{\vol(G)},
  \qquad
  \vol(G):=\sum_v\tr(D_v),
  \qquad
  \nu_G := \min_{g:V\to O(d)}\nu(g).
\end{equation}
\end{definition}
Each of the three measures how far $B$ is from having a kernel
configuration, in the sense of Proposition~\ref{prop:kernel}, just as the
frustration index of a signed graph (Section~\ref{sec:signed_graphs})
measures how far it is from being balanced.
The three constants correspond to three feasible sets: a unit vector at
\emph{some} vertices, at \emph{every} vertex, and a full orthonormal
\emph{frame} at every vertex, each measuring the residual of explaining the
measured rotations $W_{vu}$ as $g_ug_v^{-1}$.
Expanding $\nu(g)$ columnwise and applying Lemma~\ref{lem:energy_id} to each
column gives $\nu(g)=\sum_{k\leq d}(g^{(k)})^\top Bg^{(k)}\big/\sum_{k\leq
d}(g^{(k)})^\top Dg^{(k)}$, so $\nu$ could be seen as $\eta$ evaluated on
$d$ columns at once.
This is why $\sum_{k\leq d}\lambda_k(\Bn)$, not $\lambda_1(\Bn)$ alone, is
the quantity compared against $\nu_G$ below.
Figure~\ref{fig:synchronization} pictures the $\mathbb{S}^{d-1}$ and the
$O(d)$ problem.

\begin{figure}[htbp]
    \centering
    \includegraphics[width=\linewidth]{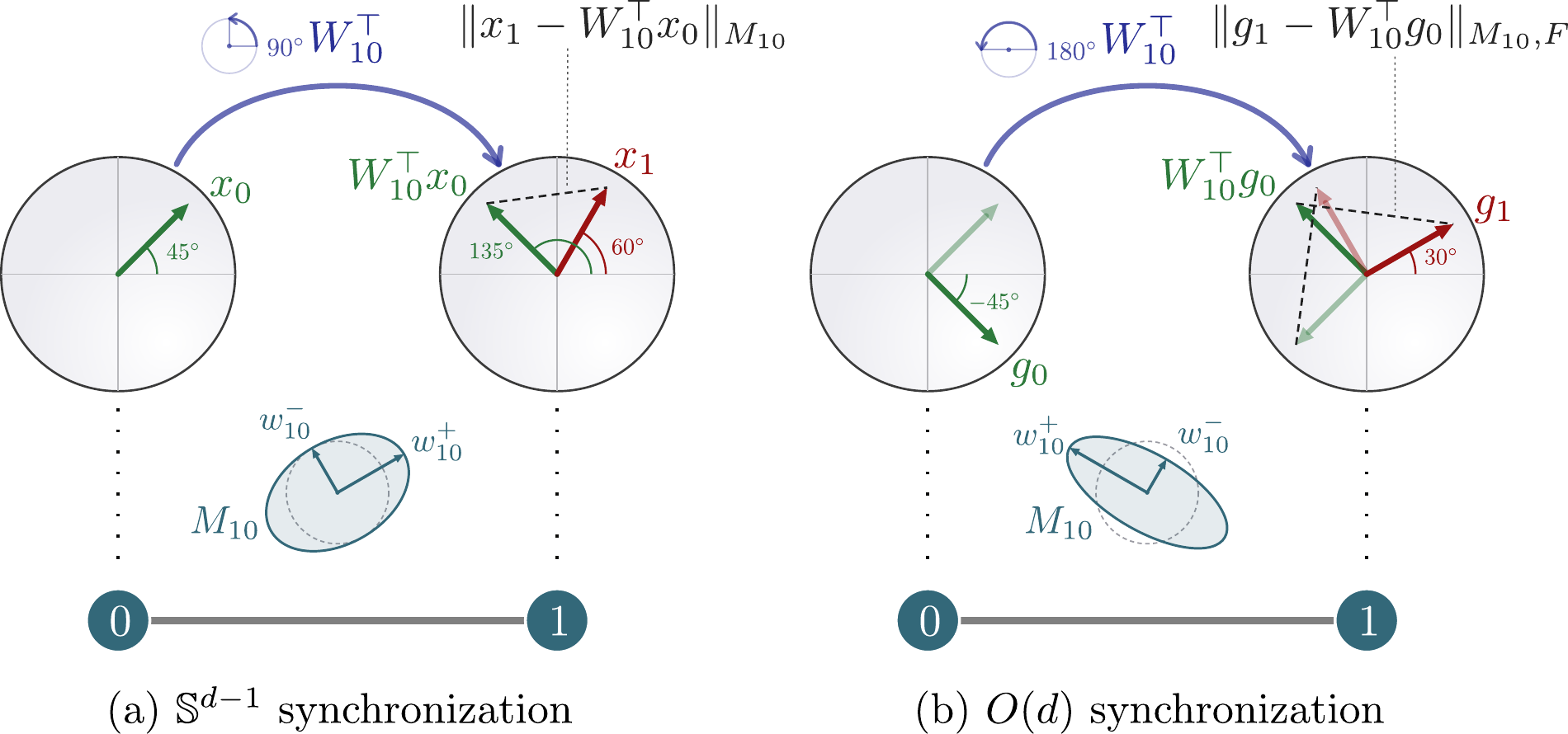}
    \caption{The $\mathbb{S}^{d-1}$ and $O(d)$ synchronization problems on a
single edge, $d=2$.
(a) Each node carries a unit vector. The angular factor $W_{10}^\top$, here
a rotation by $90^\circ$, transports $x_0$ to the stalk of node $1$, and the
frustration of the edge is the discrepancy between $x_1$ and $W_{10}^\top x_0$
(dashed), weighted by the radial factor $M_{10}$. The ellipse below the edge
draws $M_{10}$, with principal axes labelled by its largest and smallest
eigenvalues $w^+_{10}$ and $w^-_{10}$: a discrepancy along the first costs
more than one along the second, while the dashed circle is the isotropic
case of a connection Laplacian.
(b) Each node carries an orthogonal frame $g_v\in O(d)$, drawn as its two
columns, and the frustration is $\norm{g_1-W_{10}^\top g_0}_{M_{10},F}$.}
    \label{fig:synchronization}
\end{figure}

\begin{proposition}[Relaxation]
\label{prop:relaxation}
Under (H1)--(H3): (a) $\lambda_1(\Bn)\leq\eta^*_G\leq\eta_G$; (b)
$\lambda_1(\Bn)\leq\nu_G$; (c)
$\dfrac{1}{d\,\kappa_D}\sum_{k=1}^{d}\lambda_k(\Bn)\leq\nu_G$.
\end{proposition}
Proof in Appendix~\ref{sec:proofs}.
Parts (a) and (b) match the connection-Laplacian case; (c) does not, since
ordinary vertex degrees no longer make the columns of a group potential
$D$-orthogonal, and $\kappa_D$ is the penalty for that failure.

\begin{proposition}[Zero sets]
\label{prop:zero_sets}
Under (H1)--(H3), $\lambda_1(\Bn)=0 \iff \ker(B)\neq\{0\} \iff \eta_G=0 \iff
\eta^*_G=0$.
Moreover $\nu_G=0$ iff the transport maps $T_\gamma$ are path independent,
equivalently, by Proposition~\ref{prop:holonomy}, iff $\dim\ker(B)=d$.
\end{proposition}
Proof in Appendix~\ref{sec:proofs}.

\noindent\textbf{Partial synchronization.}
In the first problem some vertices are allowed to abstain.

\begin{alg}[Partial synchronization]
\label{alg:partial}
\emph{Input:} a Bochner matrix $B$ satisfying (H1)--(H3).
\begin{enumerate}
  \item Form $\Bn=D^{-1/2}BD^{-1/2}$ and compute a unit eigenvector $z$ for
$\lambda_1(\Bn)$.
  \item Set $x:=D^{-1/2}z$ and rescale it so that $\max_v\norm{x_v}=1$.
  \item For $t\geq0$ let $x^{>t}$ be the \emph{thresholded normalization}:
  \begin{equation}
    x^{>t}_v := \begin{cases}
      x_v/\norm{x_v}, & \norm{x_v}^2>t,\\
      0, & \text{otherwise,}
    \end{cases}
  \end{equation}
and form the candidates $x^{>t}$ at the $n$ thresholds $t=\norm{x_v}^2$,
$v\in V$.
  \item \emph{Output:} the candidate minimizing $\eta$.
\end{enumerate}
Steps 3 and 4 turn the eigenvector into an admissible assignment, the step
that Cheeger-type proofs call rounding.
Step 3 needs only finitely many thresholds: as $t$ varies, $x^{>t}$ changes
only when $t$ crosses one of the values $\norm{x_v}^2$, so trying these at
most $n$ values covers every possible candidate.
\end{alg}

\begin{theorem}[Partial synchronization]
\label{thm:partial}
Under (H1)--(H3):
\begin{equation}
  \lambda_1(\Bn) \;\leq\; \eta^*_G \;\leq\; \kappa_E^{3/2}\sqrt{10\,\lambda_1(\Bn)},
\end{equation}
and Algorithm~\ref{alg:partial} returns a partial assignment attaining the
upper bound.
\end{theorem}
Proof in Appendix~\ref{sec:proofs}.

\noindent\textbf{Full synchronization.}
Moving from partial to full synchronization additionally requires
controlling the connectivity of the companion graph, quantified by its
spectral gap $\lambda_2(\Lmn)$, exactly as when Bandeira, Singer and
Spielman move from partial to full
synchronization~\cite{bandeira2013cheegerinequalitygraphconnection}.

\begin{alg}[Full synchronization]
\label{alg:full}
\emph{Input:} a Bochner matrix $B$ satisfying (H1)--(H3).
\begin{enumerate}
  \item Form $\Bn=D^{-1/2}BD^{-1/2}$ and compute a unit eigenvector $z$ for
$\lambda_1(\Bn)$.
  \item Set $x:=D^{-1/2}z$.
  \item \emph{Output:} the normalized assignment $\tilde x$, given by
$\tilde x_v:=x_v/\norm{x_v}$ where $x_v\neq0$, and by an arbitrary unit
vector where $x_v=0$.
\end{enumerate}
\end{alg}

\begin{theorem}[Full synchronization]
\label{thm:full}
Under (H1)--(H3):
\begin{equation}
  \lambda_1(\Bn) \;\leq\; \eta_G \;\leq\; \frac{54\,\kappa_E^2}{\lambda_2(\Lmn)}\,\lambda_1(\Bn),
\end{equation}
and Algorithm~\ref{alg:full} returns a full assignment attaining the upper
bound.
\end{theorem}
Proof in Appendix~\ref{sec:proofs}.

\noindent\textbf{Synchronization in $O(d)$.}
The third problem asks for a full orthonormal frame at every vertex.
Its solution uses the polar projection: for nonsingular $X=P\Sigma
Q^\top$, $\mathrm{Pol}(X):=PQ^\top$ is the nearest orthogonal matrix in
Frobenius norm.
Alongside $D$ we need the isotropized degree:
\begin{equation}
  \bar d_v := \tfrac1d\tr(D_v), \quad
  \Db := \operatorname{blockdiag}(\bar d_vI_d), \quad
  \Bb := \Db^{-1/2}B\,\Db^{-1/2}, \quad
  \volb := \sum_v\bar d_v,
\end{equation}
whose relevant property is that, $\Db$ being a scalar multiple of the
identity in each block, $x^\top\Db x=\sum_v\bar
d_v\norm{x_v}^2=\norm{n_x}^2_{\bar D_0}$ exactly, for every $x$ (with $\bar
D_0:=\operatorname{diag}(\bar d_v)$).
The $O(d)$ argument needs the $d$ vectors at a vertex to be nearly
orthonormal; under $\Db$ this holds exactly, while under the ordinary degree
$D$ it holds only up to an additive error, where the argument needs a
multiplicative one. $\Db$ occurs only inside the construction: the statement of
Theorem~\ref{thm:od} involves $\Bn$, $D$ and $\lambda_2(\Lmn)$ only.

\begin{alg}[$O(d)$ synchronization]
\label{alg:od}
\emph{Input:} a Bochner matrix $B$ satisfying (H1)--(H3).
\begin{enumerate}
  \item Form $\bar d_v=\tfrac1d\tr(D_v)$ and $\Bb=\Db^{-1/2}B\Db^{-1/2}$.
  \item Compute orthonormal eigenvectors $z^1,\dots,z^d$ of $\Bb$ for its
$d$ smallest eigenvalues $\lambda_1(\Bb)\leq\cdots\leq\lambda_d(\Bb)$.
  \item Set $x^k:=\sqrt{\volb}\,\Db^{-1/2}z^k$ for $k=1,\dots,d$, and let
$X_v:=[x^1_v\ \cdots\ x^d_v]\in\R^{d\times d}$.
  \item \emph{Output:} the group potential $g_v:=\mathrm{Pol}(X_v)$,
computed from the SVD of $X_v$, setting $g_v:=I_d$ where $X_v$ is singular.
\end{enumerate}
\end{alg}

\begin{theorem}[$O(d)$ synchronization]
\label{thm:od}
Under (H1)--(H3):
\begin{equation}
  \frac{1}{d\,\kappa_D}\sum_{k=1}^{d}\lambda_k(\Bn)
  \;\leq\; \nu_G \;\leq\;
  \frac{2^{11}\,\kappa_E^3\,d^3}{\lambda_2(\Lmn)}\sum_{k=1}^{d}\lambda_k(\Bn),
\end{equation}
with $\kappa_D\leq\kappa_E$, and Algorithm~\ref{alg:od} returns a group
potential attaining the upper bound.
No hypothesis beyond (H1)--(H3) is required.
\end{theorem}
Proof in Appendix~\ref{sec:proofs}.

In the isotropic case every $M_{vu}=\omega_{vu}I_d$, so
$\kappa_E=\kappa_D=1$, $D_v=d_vI_d$, $\Db=D$, and the companion graph $G_-$
is the underlying weighted graph itself, so
$\lambda_2(\Lmn)=\lambda_2(\mathcal L_0)$, the second smallest eigenvalue of
its normalized Laplacian, simple since the graph is connected by (H3).
The three bounds then reduce to those of Bandeira, Singer and
Spielman~\cite{bandeira2013cheegerinequalitygraphconnection}, up to the
numerical constants: in place of our $54$ in Theorem~\ref{thm:full} and
$2^{11}$ in Theorem~\ref{thm:od}, they obtain $44$ and $2^{10}$, by sharper
bookkeeping at two steps we have not attempted to optimize.

The $\kappa_E^{3/2}$ of Theorem~\ref{thm:partial} is three half-powers of
the anisotropy: one from the estimate on a single edge, one from the
Cauchy--Schwarz comparison of $x^\top Bx$ against the companion graph, and
one from converting the resulting $\ell^1$ bound back into a bound on
$\eta$.
In Theorem~\ref{thm:full} the two comparisons appear outside any square root
and give $\kappa_E^2$, while Theorem~\ref{thm:od} carries one such
comparison together with two further powers from passing to the isotropized
degree, hence $\kappa_E^3$.
Our technique stems directly from the approach of Bandeira, Singer and
Spielman, with the anisotropy entering as a factor that rescales their
estimates; we leave open finding other techniques that could improve these
bounds.

\textbf{Related Cheeger inequalities.}
Liu, M\"unch and Peyerimhoff~\cite{Liu_2019_curvature_buser} prove
Buser-type inequalities for connection Laplacians with orthogonal or unitary
matrices, bounding eigenvalues from above under nonnegative Ricci curvature.
Their Cheeger constants are mixtures of the expansion rate of the underlying
graph and the frustration index of the signature.
One open question would be to generalize Ricci curvature for Connection
Laplacian to Bochner matrices, and, subsequently, apply analogous Cheeger
inequalities to Bochner matrices.

For general sheaves, First and
Kaufman~\cite{first2025cheegerinequalitycoboundaryexpansion} prove
Cheeger-type inequalities bounding the coboundary expansion of sheaves
obtained as quotients of a constant sheaf by a subsheaf.
Their technique applies only to certain graphs, which they characterize,
whereas the inequalities above hold for every Bochner matrix on an arbitrary
graph.
To our knowledge, the Bochner matrices defined in this work are therefore
the largest class of cellular sheaves for which a Cheeger inequality is
known on arbitrary graphs.

\subsection{The anisotropic random walk}
\label{sec:random_walk}

Attaching a random walk to a general sheaf Laplacian is nontrivial, and is
left open by Hansen and Ghrist~\cite{hansen2019toward}.
The obstruction is already visible for $d=1$: the diagonal entry at $v$ is
$\sum_e\F_{v\trianglelefteq e}^2$, the off-diagonal entry on $e=(v,u)$ is
the product $\F_{v\trianglelefteq e}\F_{u\trianglelefteq e}$, and nothing
forces the first to be the sum of the absolute values of the second, so no
transition probability falls out naturally.
It does for connection Laplacians, where $p_{vu}=w_{vu}/\sum_\rho w_{v\rho}$
induces the vector-valued walk of Bandeira and
Singer~\cite{bandeira2013cheegerinequalitygraphconnection,singer2011vectordiffusionmapsconnection};
Bochner matrices repair the general case, since by construction their
diagonal block is exactly the sum of the matrix absolute values of the
off-diagonal ones.

\begin{definition}[Anisotropic random walk]
\label{def:random_walk}
Let $B$ satisfy (H1)--(H3).
For each neighbour $u$ of $v$ set:
\begin{equation}
  P_{vu} := D_v^{-1/2}M_{vu}D_v^{-1/2} \succeq 0,
  \qquad
  \sum_{u\neq v}P_{vu} = D_v^{-1/2}\Bigl(\sum_{u\neq v}M_{vu}\Bigr)D_v^{-1/2} = I,
\end{equation}
a \emph{Positive Operator-Valued Measure} (POVM) at $v$.
Given a unit vector $\hat x_v\in\mathbb{S}^{d-1}$, the walker jumps to $u$
with probability $p_{vu}(\hat x_v) := \hat x_v^\top P_{vu}\hat x_v$ and
updates its vector deterministically as $\hat x\mapsto W_{vu}\hat x$, giving
a Markov chain on $\mathcal S = V\times\mathbb{S}^{d-1}$ that moves from
$(v,\hat x)$ to $(u,W_{vu}\hat x)$ with probability $p_{vu}(\hat x)$.
\end{definition}

The jump law is state-dependent: which neighbour the walker visits
next depends on the direction $\hat x$ it currently carries, large when
$\hat x$ aligns with the large eigenspace of $M_{vu}$ and small against its
small eigenspace.
This is the ingredient that the anisotropy adds to the walk of a
connection Laplacian.
When $M_{vu}=w_{vu}I$, $P_{vu}$ collapses to the scalar $w_{vu}/\sum_\rho
w_{v\rho}$ times the identity, and the walk reduces to the classical
vector-valued walk of the connection
Laplacian~\cite{bandeira2013cheegerinequalitygraphconnection,singer2011vectordiffusionmapsconnection}.
We also note that an effective conductance and resistance have recently been
developed out of that notion of random
walk~\cite{cloninger2023randomwalksconductanceresistance}.

For $d=1$, $O(1)=\{-1,+1\}$ and the walk is exactly the Gremban expansion of
the signed graphs of Section~\ref{sec:signed_graphs}: each node is lifted to
two copies, one per sign, and $P_{vu}$ collapses to the ordinary edge
weight.
The construction above lifts each node instead to a whole sphere
$\mathbb{S}^{d-1}$, one point per orientation available to a stalk vector,
with the POVM $\{P_{vu}\}$ playing the role the unsigned Gremban weights
play for $d=1$ (Figure~\ref{fig:gremban}).
We call it the \emph{generalized Gremban expansion}: it carries the scalar
construction from $\mathbb{S}^0=\{-1,+1\}$, where a step either flips a sign
or does not, to $O(d)$-valued transport, where a step rotates an entire
stalk vector.

\begin{figure}[htbp]
    \centering
    \includegraphics[width=\linewidth]{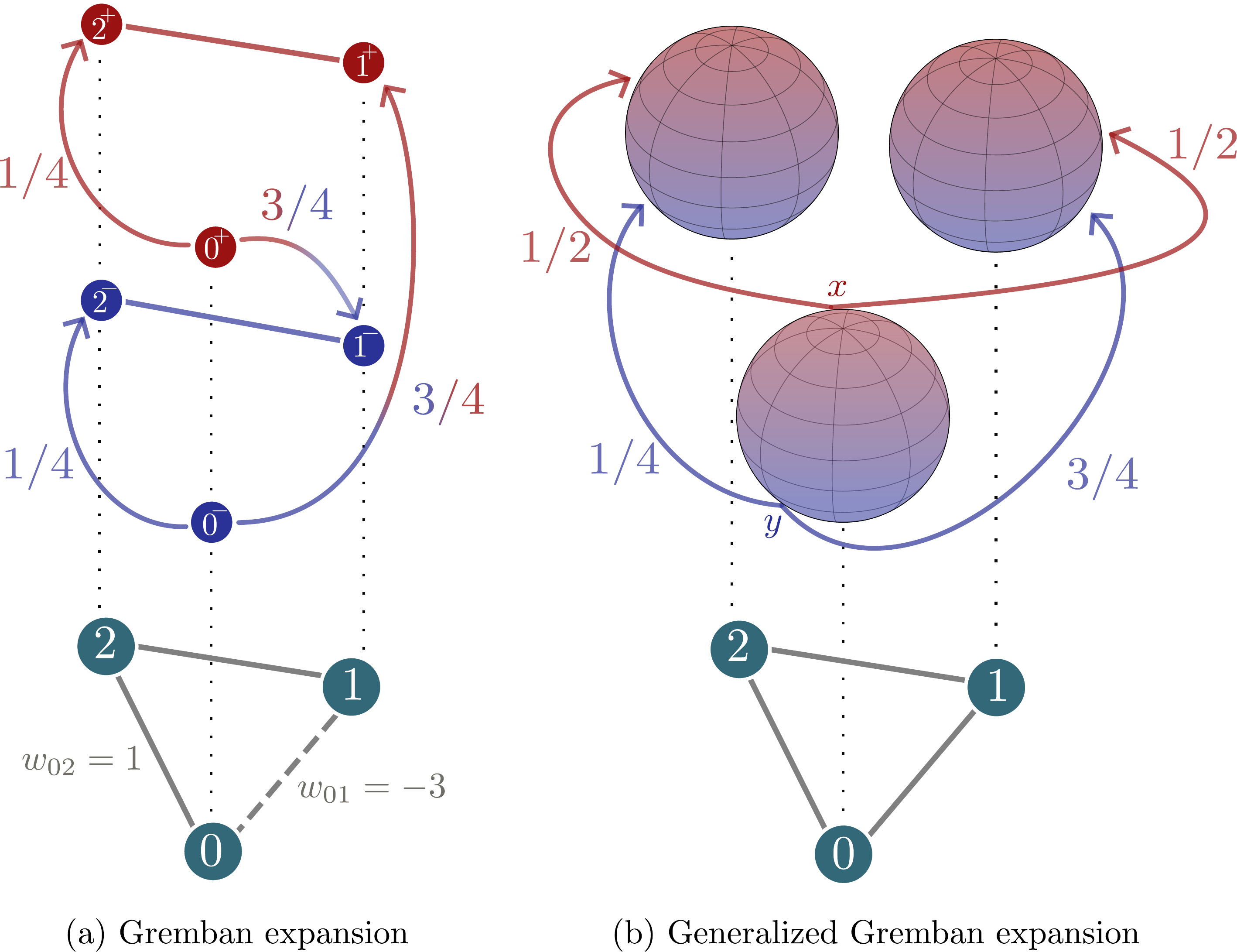}
    \caption{The random walk of a Bochner matrix as a generalized Gremban
expansion.
(a) For the signed triangle with $w_{01}=-3$ (dashed) and $w_{02}=1$, the
Gremban expansion doubles each node into $v^+$ and $v^-$. From node $0$ the
walker moves to node $1$ with probability $3/4$ and to node $2$ with
probability $1/4$, in proportion to $|w_{0u}|$; crossing the negative edge
flips its sign, $0^\pm\to1^\mp$, and crossing the positive one keeps it,
$0^\pm\to2^\pm$.
(b) For $d>1$ each node is lifted to a sphere $\mathbb{S}^{d-1}$, and the
jump probability $p_{0u}(\hat x)=\hat x^\top P_{0u}\hat x$ depends on the
direction the walker carries: in this example, from $y$ it reaches nodes $1$
and $2$ with probabilities $3/4$ and $1/4$, from $x$ with $1/2$ each
(Definition~\ref{def:random_walk}).}
    \label{fig:gremban}
\end{figure}

\noindent\textbf{Trajectories and conserved quantities.}
The fibre carries no randomness of its own.
Each step applies $W_{v_iv_{i+1}}$ to the current vector, so a walker
following $\gamma:v_0\to\cdots\to v_t$ ends in the state $(v_t,T_\gamma\hat
x_0)$: the path determines the fibre, the transport maps of
Proposition~\ref{prop:holonomy} reappear unchanged, and the randomness lies
entirely in which path is taken.
This makes the kernel visible along trajectories.
The observables $f_y(v,\hat x):=\langle\hat x,y_v\rangle$ that stay constant
along \emph{every} trajectory, and not merely in expectation, are exactly
those with $y\in\ker(B)$: such a $y$ is transported edge by edge by the same
orthogonal map as the walker's own vector, by Proposition~\ref{prop:kernel},
and orthogonal maps preserve inner products.
Both statements are proved in Appendix~\ref{sec:proofs}.
Kernel vectors therefore confine the walker to a subset of each fibre.
The walk is not reversible in general, while a variant with the same fibre
update and a modified jump law is; both facts are in
Appendix~\ref{sec:further}.

\section{Forman Sheaf Curvature}
\label{sec:forman_term}

So far we have presented the Bochner--Weitzenb\"ock decomposition for each
order $k$ of the Hodge Laplacian of a simplicial complex, highlighting that
even the node-graph case $k=0$ is non-trivial.
We then considered Bochner matrices, the first term arising from the Bochner
decomposition, and studied them on their own, setting aside the
decomposition they come from.
In this section we turn back to the simplicial complex setting, and to the
$k=1$ case.
Throughout this section we write $H_1^\F = B_1^\F + F_1^\F$ for the Bochner
decomposition of the $1$-dimensional \emph{Hodge} sheaf Laplacian, to keep
it apart from the graph sheaf Laplacian $L_\F$ of the previous sections.
We show that the sheaf Forman term, which we call \emph{sheaf Forman Ricci
curvature} $\mathrm{Ric}_\F := F_1^\F$, controls the first sheaf cohomology
dimension, $\dim\mathcal{H}^1(M;\F)$, when it is positive semidefinite
everywhere.
This is the analogue of the Bochner theorem for cellular sheaves, and we
obtain it by generalizing Forman's original techniques to our setting.
Throughout, $M$ is a compact, connected, quasiconvex regular CW complex and
$\F$ has invertible restriction maps, everywhere of dimension $d$.

\subsection{A Bochner theorem for cellular sheaves}

The two ingredients of the theorem are assembled here: nonnegative curvature
confines a harmonic cochain to the kernel of the Bochner term, and the
kernel characterization of Section~\ref{sec:bochner_matrices} then
propagates a vanishing from one edge to its neighbours.

\begin{corollary}
\label{cor:kernel_split_ricci}
Under $\mathrm{Ric}_\F\succeq0$,
$\ker(H_1^\F)=\ker(B_1^\F)\cap\ker(\mathrm{Ric}_\F)$.
\end{corollary}
\begin{proof}
$B_1^\F\succeq0$ always (Proposition~\ref{prop:psd}) and
$\mathrm{Ric}_\F\succeq0$ by hypothesis, so for $c\in\ker(H_1^\F)$,
$0=\langle H_1^\F c,c\rangle=\langle B_1^\F
c,c\rangle+\langle\mathrm{Ric}_\F c,c\rangle$ is a sum of two nonnegative
terms, forcing both to vanish; since each is a positive semidefinite
quadratic form, this forces $c\in\ker(B_1^\F)\cap\ker(\mathrm{Ric}_\F)$.
The reverse inclusion is immediate.
\end{proof}

\begin{lemma}
\label{lem:unique_cont_local}
Let $c=\sum_ec_ee\in\ker\partial^*\cap\ker B_1^\F$, let $f$ be a $2$-cell,
and let $e_1,e_2\trianglelefteq f$ be edges with $\bar e_1\cap\bar
e_2\neq\emptyset$ and $c_{e_1}=c_{e_2}=0$.
Then $c_e=0$ for every $e\trianglelefteq f$.
\end{lemma}
\begin{proof}
Let $e\neq e_1,e_2$ be the remaining edge of $f$.
If $e$ shares no vertex with $e_1$ or $e_2$, the off-diagonal block of
$B_1^\F$ between $e$ and $e_1$ reduces to its up contribution alone,
$\F_{e\trianglelefteq f}^\top\F_{e_1\trianglelefteq f}$, invertible as a
product of invertible maps; the Bochner kernel characterization
(Proposition~\ref{prop:kernel}) then gives $\|c_e\|=\|c_{e_1}\|=0$.

Otherwise $e$ shares a vertex with each of $e_1,e_2$, and since $M$ is a
quasiconvex regular CW complex these are the only three edges of $f$.
For $v_f\in\F(f)$, extended by $0$ elsewhere, $\partial v_f$ has one term
per edge of $f$, each the transpose restriction map applied to $v_f$; since
$c_{e_1}=c_{e_2}=0$, we get
$\langle\partial v_f,c\rangle=\pm\,v_f^\top\F_{e\trianglelefteq f}\,c_e$.
But $\langle\partial v_f,c\rangle=\langle v_f,\partial^*c\rangle=0$ by
hypothesis, for every $v_f$, so $\F_{e\trianglelefteq f}c_e=0$; since
$\F_{e\trianglelefteq f}$ is invertible, $c_e=0$.
\end{proof}
Invertibility is what makes the last step work: without it, the argument
only forces $c_e\in\ker(\F_{e\trianglelefteq f})$, not $c_e=0$.

\begin{theorem}[Unique continuation]
\label{thm:unique_cont}
Let $M$ be connected and $c\in\ker\partial^*\cap\ker B_1^\F$.
If $c_e=0$ for every edge $e$ incident to some fixed vertex, then $c=0$.
\end{theorem}
By induction on graph distance from that vertex: at distance $0$ the
hypothesis gives $c_e=0$ directly, and the inductive step propagates this to
every edge at distance $k+1$ via Lemma~\ref{lem:unique_cont_local}, using
either parallel propagation (Proposition~\ref{prop:kernel} applied to the
pair) or the vertex relation, depending on whether the two edges at that
step share a face.
Full proof in Appendix~\ref{sec:proofs}.

Propagation turns a local hypothesis into a global bound on $\mathcal H^1$:
it suffices to exhibit, at a single vertex, a small set of edges whose
vanishing forces the rest.

\begin{proposition}
\label{prop:ric_pos_vanishing}
Let $M$ be connected with $\mathrm{Ric}_\F\succeq0$, and suppose
$\mathrm{Ric}_\F(e)\succ0$ for every edge $e$ at some fixed vertex.
Then $\mathcal H^1(M;\F)=0$.
\end{proposition}
\begin{proof}
By Corollary~\ref{cor:kernel_split_ricci}, a harmonic $c$ lies in
$\ker(\mathrm{Ric}_\F)$, hence vanishes on every edge at that vertex by
strict positivity there, and Theorem~\ref{thm:unique_cont} gives $c=0$.
\end{proof}

\begin{theorem}[Bochner theorem for combinatorial manifolds]
\label{thm:bochner_manifold}
Let $M$ be a compact connected combinatorial $n$-manifold, $n\leq3$, with
$\mathrm{Ric}_\F\succeq0$.
Then $\dim\mathcal H^1(M;\F)\leq nd$.
\end{theorem}
We follow Forman's proof~\cite[Theorem 4.4]{Forman2003Bochner}.
By Corollary~\ref{cor:kernel_split_ricci}, a harmonic cochain $c$ lies in
$\ker B_1^\F$, and two facts about such a $c$ at a vertex $v$ carry the
argument.
The first is \emph{parallel propagation}: if two edges at $v$ are parallel
(Section~\ref{sec:forman_background}) and $c$ vanishes on one, it vanishes on the other
(Proposition~\ref{prop:kernel}).
The second is the \emph{vertex relation}: $\partial c=0$ at $v$, so if $c$
vanishes on every edge at $v$ but one, it vanishes on that one too, the
restriction maps being invertible.
Forman's proof chooses $n$ edges at $v$ such that, once $c$ vanishes on
them, these two facts force $c$ to vanish on every edge at $v$;
Theorem~\ref{thm:unique_cont} then gives $c=0$, so the evaluation
$c\mapsto(c_{e_1},\dots,c_{e_n})$ is injective on $\mathcal H^1(M;\F)$,
whose dimension is therefore at most $nd$.
For $n=1$ one edge suffices, since the two edges at a vertex are parallel.
For $n=2$ take two edges at $v$ sharing a face: every other edge at $v$ is
parallel to one of them, unless $v$ has only three edges, and then the
vertex relation kills the third.
For $n=3$ the choice of three edges uses the link of $v$, the cell
structure on $\mathbb S^2$ whose vertices are the edges at $v$, two being
adjacent when transverse, together with the fact that $K_5$ does not embed
in $\mathbb S^2$.
The choice depends only on the geometry of $M$; the sheaf enters only
through the two facts above, and Appendix~\ref{sec:proofs} checks them and
recalls the choice.

Beyond $n=3$ an explicit case analysis is out of reach, as Forman already
notes in his original work, but a coarser, purely local quantity still
bounds the cohomology.

\begin{definition}[Local homological dimension~\cite{Forman2003Bochner}]
\label{def:local_hom_dim}
At a vertex $v$, let $\sim$ be the equivalence relation on incident edges
generated by \emph{not} sharing a face, let $N(v)$ be its number of classes,
and set $D(v):=N(v)$ if every class has at least two elements,
$D(v):=N(v)-1$ otherwise.
The local homological dimension of $M$ is $D(M):=\inf_vD(v)$.
\end{definition}

\begin{theorem}
\label{thm:local_hom_dim}
Let $M$ be a compact connected combinatorial manifold with
$\mathrm{Ric}_\F\succeq0$.
Then $\dim\mathcal H^1(M;\F)\leq d\cdot D(M)$.
\end{theorem}
Full proof in Appendix~\ref{sec:proofs}.
Since $D(v)$ depends only on the neighbourhood of $v$, the bound is
computable from purely local data, which could make it usable in inductive
settings such as trajectory prediction, where only the edges near the region
of interest are known.

\subsection{Comparison with combinatorial curvature}
\label{sec:curvature_comparison}

The same complex also carries the scalar Forman curvature
$\mathrm{Ric}(e)=\deg_u(e)+\deg_l(e)-\#\{\text{parallel neighbours of }e\}$
of the unweighted case~\cite{Forman2003Bochner}.
The two curvatures are independent: neither $\mathrm{Ric}\geq0$ nor
$\mathrm{Ric}_\F\succeq0$ implies the other, even on manifolds, so neither
cohomology bound is a special case of the other.
Throughout this part $d=1$, already enough for every phenomenon involved.
The three examples below are drawn in Figure~\ref{fig:forman_examples}.

\begin{remark}[Curvature depends on the presentation]
\label{rem:gauge}
$\mathrm{Ric}_\F$ is not invariant under sheaf isomorphism: it depends on the chosen
restriction maps.
The hypothesis $\mathrm{Ric}_\F\succeq0$ therefore reads as in the classical
Bochner theorem, which does not ask every metric on a manifold to have
nonnegative Ricci curvature, but only that one such metric be exhibited.
\end{remark}

\begin{example}[Nonnegative $\mathrm{Ric}$ everywhere, negative $\mathrm{Ric}_\F$ on some edge]
\label{ex:triangle}
Let $M$ be the empty triangle, with vertices $0,1,2$ and no $2$-cell.
Every edge has two parallel neighbours, so $\mathrm{Ric}=0+2-2=0$.
Take all restriction maps equal to $1$ except
$\F_{0\trianglelefteq[0,2]}=\F_{1\trianglelefteq[1,2]}=2$.
The sheaf is invertible, and $\mathrm{Ric}_\F([0,1])=2-2-2=-2<0$.
\end{example}

\begin{example}[Negative $\mathrm{Ric}$ everywhere, presentations of both signs everywhere]
\label{ex:star}
Let $M$ be the star $S_4$, with center $0$ and leaves $1,\dots,4$.
Every spoke $e_i=[0,i]$ has three parallel neighbours, so
$\mathrm{Ric}(e_i)=0+2-3=-1$.
The sheaf with $\F_{0\trianglelefteq e_i}=1$ and $\F_{i\trianglelefteq e_i}=A$
has $\mathrm{Ric}_\F(e_i)=1+A^2-3\geq0$ for $A\geq\sqrt2$.
It is isomorphic, via the gauge $P_0=1$, $P_i=A$, to the constant sheaf
$\mathbb R$, whose curvature is $-1$ everywhere, so the same isomorphism
class admits presentations of both signs (Remark~\ref{rem:gauge}).
\end{example}

The following example instead considers a combinatorial $2$-manifold.

\begin{example}[Negative $\mathrm{Ric}$ on some edges, positive $\mathrm{Ric}_\F$ everywhere, on a manifold]
\label{ex:bipyramid}
We exhibit a combinatorial $2$-manifold $M$ and a cellular sheaf $\F$ with
invertible restriction maps such that $\mathrm{Ric}<0$ on some edges of $M$
while $\mathrm{Ric}_\F\succ0$ on \emph{every} edge of $M$.

Let $M$ be the bipyramid over a heptagon: equator $v_0,\dots,v_6$ a
$7$-cycle, apexes $A,B$, and the $14$ triangles $[v_i,v_{i+1},A]$,
$[v_i,v_{i+1},B]$ with indices mod $7$, so $9$ vertices, $21$ edges and $14$
triangles.
Every edge lies in exactly two triangles, so
$\mathrm{Ric}([u,v])=10-\deg(u)-\deg(v)$, which equator degrees $4$ and apex
degrees $7$ make $+2$ on the equator and $-1$ on each of the $14$ spokes.

For $\F$, fix a positive function $q$ on the vertices and set
$\F_{v\trianglelefteq e}:=q_w$ with $w$ the other endpoint of $e$, and
$\F_{e\trianglelefteq t}:=q_z$ with $z$ the vertex of $t$ opposite $e$; the
restriction maps are invertible.
On transverse neighbours the up and down contributions to $\mathrm{Ric}_\F$
cancel as in the unweighted case, leaving only parallel ones:
\begin{equation}
  \mathrm{Ric}_\F(e) = q_u^2+q_v^2+q_{z_1}^2+q_{z_2}^2
    - \sum_{\tilde e\,\text{parallel to}\,e}q_xq_y ,
\end{equation}
with $e=[w,x]$, each parallel neighbour $\tilde e=[w,y]$, and $z_1,z_2$ the
triangle apexes of $e$; so $q$ tunes the parallel neighbours against the
diagonal.
A spoke $[v_i,A]$ has parallel neighbours $[v_i,B]$ at $v_i$ and the four
spokes $[v_j,A]$, $j\neq i,i\pm1$, at $A$; an equator edge $[v_i,v_{i+1}]$
has $[v_{i-1},v_i]$ and $[v_{i+1},v_{i+2}]$.
Taking $(q_{v_0},\dots,q_{v_6})=(1,6,11,8,11,6,1)$ and $q_A=q_B=3$ gives:
\begin{center}
\renewcommand{\arraystretch}{1.25}
\begin{tabular}{lccc}
\hline
edges & $(H_1^\F)_{ee}$ & $\sum_{\tilde e\neq e}|(H_1^\F)_{e\tilde e}|$
      & $\mathrm{Ric}_\F$ \\
\hline
$[v_0,A],[v_0,B],[v_6,A],[v_6,B]$ & $47$  & $45$  & $+2$   \\
$[v_1,A],[v_1,B],[v_5,A],[v_5,B]$ & $167$ & $165$ & $+2$   \\
$[v_2,A],[v_2,B],[v_4,A],[v_4,B]$ & $230$ & $218$ & $+12$  \\
$[v_3,A],[v_3,B]$                 & $315$ & $121$ & $+194$ \\
$[v_6,v_0]$                       & $20$  & $12$  & $+8$   \\
$[v_0,v_1],[v_5,v_6]$             & $55$  & $17$  & $+38$  \\
$[v_2,v_3],[v_3,v_4]$             & $203$ & $169$ & $+34$  \\
$[v_1,v_2],[v_4,v_5]$             & $175$ & $59$  & $+116$ \\
\hline
\end{tabular}
\end{center}
Every entry is strictly positive, so $\mathrm{Ric}_\F\succ0$ on all $21$
edges, while the first four rows are exactly the $14$ spokes on which
$\mathrm{Ric}=-1$.
\end{example}

\begin{figure}[htbp]
    \centering
    \includegraphics[width=\linewidth]{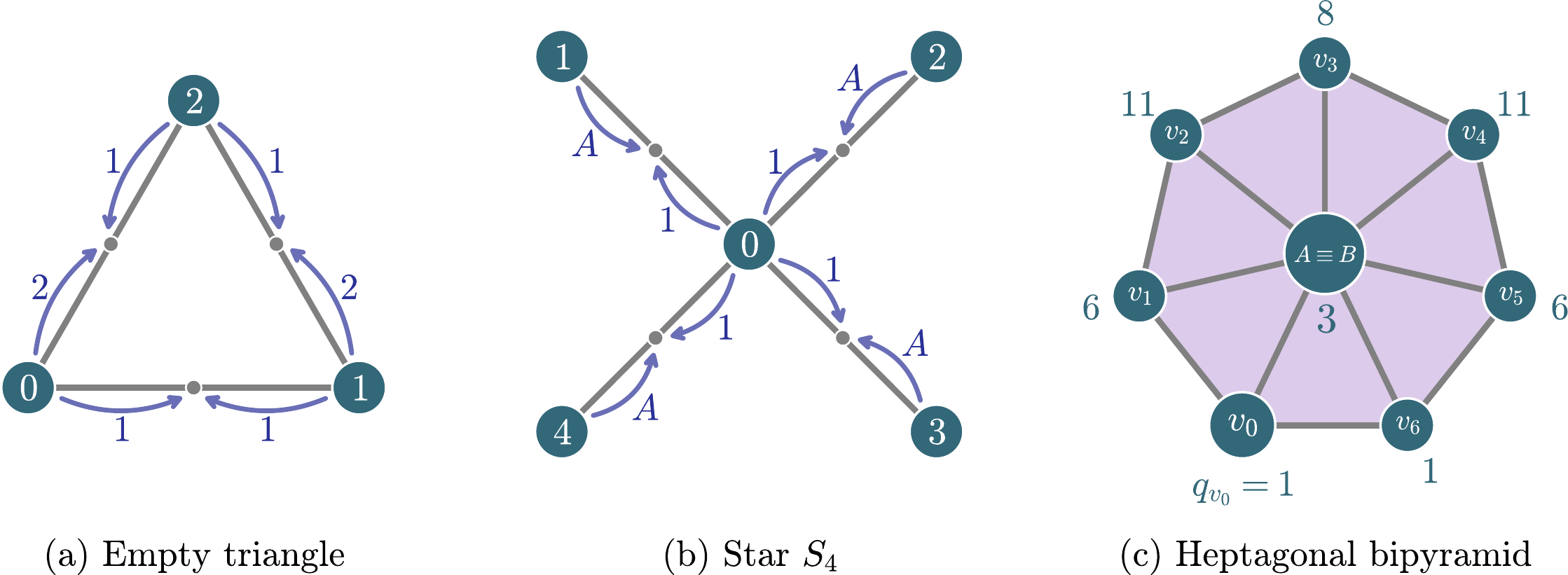}
    \caption{The three examples of this section, with $d=1$.
(a) The empty triangle of Example~\ref{ex:triangle}: each arrow carries the
restriction map from a vertex to an edge, all equal to $1$ except
$\F_{0\trianglelefteq[0,2]}=\F_{1\trianglelefteq[1,2]}=2$.
(b) The star $S_4$ of Example~\ref{ex:star}, with maps $1$ from the center
and $A$ from the leaves.
(c) The heptagonal bipyramid of Example~\ref{ex:bipyramid} seen from above:
the vertices $A$ and $B$, one on each side of the equator, project onto the
same point at the center. Next to each vertex is its value of $q$, which
sets the restriction maps of the construction.}
    \label{fig:forman_examples}
\end{figure}

\section{Dynamical Interpretation and Applications to Sheaf Neural Networks}
\label{sec:applications}

With both terms characterized, the decomposition admits a dynamical reading,
in which sheaf diffusion can be seen as a reaction-diffusion process, and
each term yields a computable quantity at every node: the conviction and the
relative volume.
We then apply this machinery to Sheaf Neural
Networks~\cite{hansen2020sheaf,bodnar2022neural}, which learn a sheaf
diffusion on a graph parametrically, and we characterize their behaviour
using these quantities.

\subsection{A dynamical reading}
\label{sec:dynamical_reading}

The decomposition $H_k^\F=B_k^\F+F_k^\F$ could be read as a
multi-dimensional reaction-diffusion system, following the analogous reading
of the Hodge Laplacian~\cite{Nurisso2025}: the Bochner term is diffusive,
driving the random walk of Section~\ref{sec:random_walk} on the
$\mathbb{S}^{d-1}$ fibers, while the Forman term is reactive, an
endomorphism forcing each stalk toward itself.

\begin{remark}
A random walk can also be built directly out of a Hodge Laplacian, without
passing through its Bochner term.
Schaub et al.~\cite{Schaub_2020} normalize the Hodge $1$-Laplacian and lift
the resulting operator, so that a walker carrying a quantity on a simplex
may hold it, pass it along, or reverse its orientation, the probability of
reversal being in general nonzero.
That construction does not extend to general sheaves: already for $d=1$ it
does not produce stochastic matrices.
Both dynamical readings are therefore available on a Hodge Laplacian, the
lifted random walk and the reaction-diffusion one, whereas on a sheaf Hodge
Laplacian only the reaction-diffusion interpretation carried by the Bochner
and Forman split is available up to now.
\end{remark}

\textbf{Opinion dynamics.}
At $k=0$, the decomposition $L_\F=B_\F+F_\F$ can also be read through
opinion dynamics~\cite{Discourse_sheaves}.
The Bochner term drives agents toward equilibria of equal stalk norm, the
shared value we call \textbf{conviction}: how strongly an opinion is held.
The Forman block $(F_\F)_{vv}$ measures how much more $v$'s opinion is
amplified on an edge than its neighbour's, the \textbf{relative volume}: how
loudly $v$ speaks relative to whom it speaks to.

If $F_\F\succ0$ everywhere, every agent is too confident for consensus to be
possible: $\ker(L_\F)=\{0\}$.
If $F_\F\succeq0$, consensus requires both equal conviction and compatible
volume: since $B_\F$ and $F_\F$ are both positive semidefinite,
$\ker(L_\F)=\ker(B_\F)\cap\ker(F_\F)$, as in
Corollary~\ref{cor:kernel_split_ricci}.
At equilibrium ($\F_{v\trianglelefteq e}x_v=\F_{u\trianglelefteq e}x_u$),
the two are complementary: a more convinced agent must express itself more
softly so as not to dominate a less convinced one, and vice versa.

Each term therefore yields a quantity measurable on a trained Sheaf Neural
Network.
Through the Bochner term we observe the \emph{signal}: how the conviction of
each node varies along the diffusion, and how far from equal it remains at
equilibrium.
Through the Forman term we observe the \emph{operator} that propagates it:
the reactive block at each node, its relative volume, read from the
restriction maps.

Sections~\ref{sec:conviction} and~\ref{sec:forman_profile} develop each in
turn on Sheaf Neural Networks, where the sheaf is learnt from data: the
conviction identifies the nodes a trained network relies on, and the Forman
profile tracks how the learning mechanism navigates the sheaf space.

\subsection{Sheaf Neural Networks}
\label{sec:snn}

A \emph{Sheaf Neural Network} learns the sheaf from data and unrolls a discretized sheaf diffusion as a network
layer~\cite{hansen2020sheaf,bodnar2022neural}.
Writing $X_t\in\R^{nd\times f_t}$ for the stacked node features and
$\Delta_\F = D_\F^{-1/2}L_\F D_\F^{-1/2}$ for the normalized sheaf Laplacian,
with $D_\F$ the block diagonal of $L_\F$, a layer takes the form:
\begin{equation}
  X_{t+1} = X_t - \varsigma\bigl(\Delta_\F(t)\,(I_n\otimes W_1)\,X_t\,W_2\bigr),
\end{equation}
with $W_1$ mixing the stalk dimensions, $W_2$ the feature channels, and
$\varsigma$ a pointwise nonlinearity.
The sheaf evolves with the features: the restriction maps at layer $t$ are
produced by a learnt map of the two endpoint representations,
$\F^{(t)}_{v\trianglelefteq e} = \varsigma(\Theta[x_v^{(t)}\|x_u^{(t)}])$,
with $\Theta$ a learnt weight matrix.

On heterophilic graphs, where adjacent nodes tend to carry different labels,
the harmonic space of the graph Laplacian can provably fail to separate the
classes in the infinite-depth limit; asymmetric and sign-changing
restriction maps restore separation power, and $d>1$ extends this to the
multi-class case~\cite{bodnar2022neural}.
The available theory consists of convergence results identifying the
equilibria of sheaf diffusion and of a universal approximation result, by
which a network can learn any sheaf on a graph.
Both describe what the model \emph{can} represent.
Neither describes what the learning mechanism \emph{does}: which region of
the sheaf space a training pipeline reaches, and whether different pipelines
are biased toward different regions.
The conviction and the Forman profile are computable on a trained network,
and Sections~\ref{sec:conviction} and~\ref{sec:forman_profile} test whether
what they measure is informative.

\textbf{Experimental setting.}
We train on ten node classification benchmarks: Amazon Ratings, Minesweeper,
Questions, Roman Empire and Tolokers~\cite{platonov2023critical}, Cora,
Citeseer and Pubmed~\cite{sen2008collective}, and Cornell and
Wisconsin~\cite{pei2020geomgcn,tang2009social}.
We report test accuracy, and ROC-AUC on the three binary tasks,
Minesweeper, Tolokers and Questions, under
$10$-fold cross-validation, with stalk dimension and depth chosen per
dataset.
Cornell and Wisconsin are left out of the aggregation of
Section~\ref{sec:conviction}, their training pools being too small for a
percentage budget, so the numbers there are over the remaining eight
datasets.
In the body we report the \emph{GeneralSheaf} learner, which produces
unconstrained $d\times d$ restriction maps, under the normalized Laplacian.
The remaining hyperparameters are in Appendix~\ref{sec:supp_forman_setup},
and the other learners and regimes, the full metric battery and the
per-dataset breakdowns in Appendices~\ref{sec:supp_forman}
and~\ref{sec:supp_conv}.

\subsection{Conviction analysis for SNNs}
\label{sec:conviction}

\begin{definition}[Conviction]
\label{def:conviction}
At diffusion layer $t$ the hidden state of node $u$ is a matrix
$X_t[u]\in\R^{d\times f_t}$, with $d$ the stalk dimension and $f_t$ the
number of channels.
Its \textbf{conviction} is $c_t(u):=\norm{X_t[u]}_F$, the Frobenius norm
of the stalk, that is, the Euclidean norm of its vectorization.
\end{definition}

The theory of Section~\ref{sec:bochner_matrices} is stated for a single
vector in $\R^d$, while a trained network carries $f_t$ channels at once, so
here it serves as a guide.
We use the conviction at the last diffusion layer, taken at the
best-validation checkpoint; a z-scored average across layers is reported in
Appendix~\ref{sec:supp_robust}.

Comparisons with other node statistics are in
Appendix~\ref{sec:supp_conv_observational}.
Here we test whether the conviction is an \emph{importance score}, that is,
whether the nodes it ranks highest are those a trained network relies on.
We train the network, record the conviction at the best-validation epoch,
freeze the weights, perturb a selected set of training nodes, and
re-evaluate.
Nothing is retrained, so any change is attributable to the perturbation and
not to a different optimization trajectory.
We use three perturbation modes: masking sets the input features of the
selected nodes to zero, gating multiplies their outgoing messages by zero and
leaves their self-term intact, and pruning deletes them with their incident
edges.

We compare nine selections of nodes: high and low conviction, high and low
degree, high and low confidence, two degree-matched controls, which are
random sets reproducing the degree distribution of the high- and
low-conviction groups, and a uniformly random set.
For each dataset we measure the drop in performance relative to the
baseline's headroom above chance (Appendix~\ref{sec:supp_impl}), and we then
average over datasets (Table~\ref{tab:conv_fwd_masking-zero_last}, for
masking).
The full protocol is in Appendix~\ref{sec:supp_impl}, and the per-dataset
results are in Appendix~\ref{sec:supp_perdataset}.

\begin{table}[htbp]
\definecolor{convworst}{RGB}{213,94,0}
\definecolor{convbest}{RGB}{0,114,178}
\caption{Relative drop in performance (\%) under zero masking, by selection
strategy and budget.
Mean$_{\pm\text{s.e.}}$ across 8 datasets, each averaged over its folds. In each column, the
\textcolor{convworst}{\textbf{most damaging}} and the
\textcolor{convbest}{\textbf{least damaging}} selection are highlighted.}
\label{tab:conv_fwd_masking-zero_last}
\centering
\small
\begin{tabular*}{\textwidth}{@{\extracolsep{\fill}}lccccc}
\toprule
\textbf{Zero masking} & \multicolumn{5}{c}{Budget (fraction of the training pool)} \\
\cmidrule(l){2-6}
Selection & 1\% & 5\% & 10\% & 20\% & 30\% \\
\midrule
High conviction & \textcolor{convworst}{$\mathbf{0.51}_{\pm0.17}$} & \textcolor{convworst}{$\mathbf{1.34}_{\pm0.39}$} & \textcolor{convworst}{$\mathbf{2.18}_{\pm0.70}$} & \textcolor{convworst}{$\mathbf{3.53}_{\pm1.26}$} & \textcolor{convworst}{$\mathbf{4.55}_{\pm1.83}$} \\
Low conviction & \textcolor{convbest}{$\mathbf{0.02}_{\pm0.01}$} & \textcolor{convbest}{$\mathbf{0.04}_{\pm0.04}$} & \textcolor{convbest}{$\mathbf{0.07}_{\pm0.07}$} & \textcolor{convbest}{$\mathbf{0.20}_{\pm0.14}$} & \textcolor{convbest}{$\mathbf{0.39}_{\pm0.32}$} \\
High degree & $0.37_{\pm0.16}$ & $0.75_{\pm0.28}$ & $1.06_{\pm0.40}$ & $1.76_{\pm0.68}$ & $2.46_{\pm0.98}$ \\
Low degree & $0.03_{\pm0.02}$ & $0.19_{\pm0.10}$ & $0.38_{\pm0.21}$ & $0.85_{\pm0.44}$ & $1.35_{\pm0.74}$ \\
Degree-matched (high) & $0.06_{\pm0.07}$ & $0.20_{\pm0.17}$ & $0.38_{\pm0.32}$ & $0.57_{\pm0.51}$ & $0.93_{\pm0.62}$ \\
Degree-matched (low) & $0.05_{\pm0.03}$ & $0.26_{\pm0.15}$ & $0.47_{\pm0.30}$ & $1.26_{\pm0.67}$ & $2.19_{\pm1.13}$ \\
High confidence & $0.15_{\pm0.07}$ & $0.55_{\pm0.16}$ & $0.92_{\pm0.28}$ & $1.66_{\pm0.54}$ & $2.50_{\pm0.89}$ \\
Low confidence & $0.04_{\pm0.03}$ & $0.17_{\pm0.12}$ & $0.35_{\pm0.21}$ & $0.80_{\pm0.43}$ & $1.32_{\pm0.69}$ \\
Random & $0.07_{\pm0.04}$ & $0.21_{\pm0.15}$ & $0.43_{\pm0.29}$ & $1.04_{\pm0.55}$ & $1.69_{\pm0.89}$ \\
\bottomrule
\end{tabular*}
\end{table}

\textbf{Conviction measures node importance with respect to performance.}
Removing the highest-conviction nodes is the most damaging intervention at
every budget, and removing the lowest-conviction ones the least.
At the $10\%$ budget of Table~\ref{tab:conv_fwd_masking-zero_last} the top
costs $2.18\pm0.70\%$, the bottom $0.07\pm0.07\%$, a random set
$0.43\pm0.29\%$, and a selection by degree $1.06\pm0.40\%$.
Low-conviction nodes are thus more dispensable than an arbitrary node of the
same graph, which is the comparison of interest for applications.
We obtain the same pattern for the other perturbation modes, as shown in
Appendix~\ref{sec:supp_forward}.

\textbf{The conviction is a graded score.}
In a decile experiment (Appendix~\ref{sec:supp_forward}) the training pool
is split into ten equal bins by conviction, and one bin is perturbed at a
time.
The damage stays close to zero over the bottom four bins and rises steadily
from the fifth onward.
From the bottom to the top decile, the damage goes from $0.07\%$ to $2.18\%$
for conviction, against $0.35\%$ to $0.92\%$ for confidence and $0.38\%$ to
$1.06\%$ for degree.
The conviction therefore grades how much each node matters, beyond
separating the nodes to remove from the others.
After this analysis of node importance, we believe the conviction could
serve future and more extensive work on Sheaf Neural Networks, on their
interpretability and explainability, and on their compression.

\subsection{The Forman profile}
\label{sec:forman_profile}

The Forman term is block diagonal (Section~\ref{sec:bochner_decomposition}),
each block computed from the restriction maps at one node, so its spectrum
describes the learnt operator node by node.
We use it to track how a Sheaf Neural Network moves through the space of
sheaves on a graph: any sheaf can be learnt~\cite{bodnar2022neural}, but
which sheaves training actually reaches, and whether it favours some regions
of the sheaf space, is open.
Moreover, a second motivation comes from the implementation: the sheaf is
re-estimated at every layer, so the network can be seen as running a
geometric flow~\cite{Hamilton1982ThreemanifoldsWP} on the sheaf, and the
descriptor should follow it across epochs and across layers.

\begin{definition}[Forman profile]
\label{def:forman_profile}
Let $\F$ have $d$-dimensional stalks on a graph with $n$ nodes, and for each
node $u$ let $\lambda_1(u)\leq\dots\leq\lambda_d(u)$ be the eigenvalues of
$(F_\F)_{uu}$.
The \emph{Forman profile} of $\F$ is the matrix
$f(\F)=\bigl(\lambda_k(u)\bigr)_{u\leq n,\,k\leq d}\in\R^{n\times d}$,
whose $k$-th column is the distribution of the $k$-th eigenvalue of the
Forman term across nodes.
\end{definition}

Compared with the spectrum of the whole Laplacian, the Forman profile is
local and cheap: it takes $n$ eigendecompositions of $d\times d$ matrices
instead of one of size $nd$.
It is also interpretable, as the reactive half of the decomposition
(Section~\ref{sec:dynamical_reading}) and as the term controlling the
harmonic space under nonnegativity (Section~\ref{sec:forman_term}).
We are interested in comparing the learnt profiles with two references.
The topological reference is the profile of the trivial sheaf under the
normalized graph Laplacian, and depends on the graph alone; the zero
reference is the profile of the unnormalized graph Laplacian, which vanishes
identically (Section~\ref{sec:bochner_decomposition}).
We call closeness to the topological reference \emph{local alignment}.

We report the $1$-Wasserstein distance $W_1$ between profile and reference,
averaged over the eigenvalues of the Forman term and over folds, and give
the other metrics in
Appendix~\ref{sec:supp_metrics}.
Spearman, Pearson and NRMSE agree with $W_1$, while RMSE does not, so the
alignment is in the shape of the distribution, up to an offset.
Each dataset is treated separately and then averaged, with layers binned by
relative depth as first, middle or last, since $L$ varies across datasets.
For readability, we report the spread across datasets: folds and layers
within one dataset are correlated re-measurements of a single quantity, and
pooling them would understate the uncertainty.
We also show individual results, for one representative dataset per regime,
in Tables~\ref{tab:perds_rm_norm}--\ref{tab:perds_questions_joint_unnorm}
(Appendix~\ref{sec:supp_forman_figures}).
The untrained checkpoint, where the sheaf is randomly initialized, is the
control.

\begin{table}[htbp]
\definecolor{fpa0}{RGB}{255,251,195}
\definecolor{fpa1}{RGB}{255,244,176}
\definecolor{fpa2}{RGB}{255,236,158}
\definecolor{fpa3}{RGB}{255,228,141}
\definecolor{fpa4}{RGB}{254,220,123}
\definecolor{fpa5}{RGB}{254,206,106}
\definecolor{fpa6}{RGB}{254,189,88}
\definecolor{fpa7}{RGB}{254,173,74}
\definecolor{fpa8}{RGB}{253,158,67}
\definecolor{fpa9}{RGB}{253,143,61}
\definecolor{fpa10}{RGB}{253,117,53}
\definecolor{fpa11}{RGB}{252,91,46}
\definecolor{fpa12}{RGB}{252,94,47}
\definecolor{fpa13}{RGB}{252,84,44}
\definecolor{fpa14}{RGB}{252,80,43}
\definecolor{fpa15}{RGB}{252,88,45}
\definecolor{fpa16}{RGB}{252,99,48}
\definecolor{fpa17}{RGB}{252,101,49}
\definecolor{fpa18}{RGB}{252,106,50}
\definecolor{fpa19}{RGB}{253,134,58}
\definecolor{fpa20}{RGB}{253,135,58}
\definecolor{fpa21}{RGB}{253,137,59}
\definecolor{fpa22}{RGB}{253,156,66}
\definecolor{fpa23}{RGB}{253,152,65}
\definecolor{fpa24}{RGB}{254,178,76}
\definecolor{fpa25}{RGB}{253,145,62}
\definecolor{fpa26}{RGB}{253,133,58}
\definecolor{fpa27}{RGB}{254,180,78}
\caption{$1$-Wasserstein distance between the Forman profile of the learnt
normalized GeneralSheaf and the topological reference, by training epoch
(rows; $0$ is the untrained sheaf) and relative depth of the layer
(columns). Darker cells are farther from the reference: the distance
shrinks with training and with depth. Mean$_{\pm\text{s.e.}}$ across 10
datasets, each averaged over its folds; 8 for the middle layer, which
networks with $L=2$ lack.}
\label{tab:global_general_norm}
\centering
\small
\renewcommand{\arraystretch}{1.2}%
\begin{tabular}{l>{\centering\arraybackslash}p{2.7cm}>{\centering\arraybackslash}p{2.7cm}>{\centering\arraybackslash}p{2.7cm}l}
\toprule
\multicolumn{5}{@{}l}{\textbf{GeneralSheaf, normalized}} \\
 & \multicolumn{3}{c}{Layer} & \\
\cmidrule(l){2-4}
Epoch & First & Middle & Last & \\
\midrule
0 & \cellcolor{fpa12}$0.408_{\pm0.047}$ & \cellcolor{fpa13}$0.422_{\pm0.057}$ & \cellcolor{fpa12}$0.408_{\pm0.047}$ & \smash{\begin{tikzpicture}[baseline={([yshift=-9pt]current bounding box.north)}]
\fill[fpa0] (0,0.0000) rectangle (0.28,0.2292);
\fill[fpa1] (0,0.2292) rectangle (0.28,0.4583);
\fill[fpa2] (0,0.4583) rectangle (0.28,0.6875);
\fill[fpa3] (0,0.6875) rectangle (0.28,0.9167);
\fill[fpa4] (0,0.9167) rectangle (0.28,1.1458);
\fill[fpa5] (0,1.1458) rectangle (0.28,1.3750);
\fill[fpa6] (0,1.3750) rectangle (0.28,1.6042);
\fill[fpa7] (0,1.6042) rectangle (0.28,1.8333);
\fill[fpa8] (0,1.8333) rectangle (0.28,2.0625);
\fill[fpa9] (0,2.0625) rectangle (0.28,2.2917);
\fill[fpa10] (0,2.2917) rectangle (0.28,2.5208);
\fill[fpa11] (0,2.5208) rectangle (0.28,2.7500);
\draw[line width=0.3pt] (0,0) rectangle (0.28,2.75);
\node[right,font=\scriptsize] at (0.28,2.75) {0.43};
\node[right,font=\scriptsize] at (0.28,1.375) {0.21};
\node[right,font=\scriptsize] at (0.28,0) {0.00};
\end{tikzpicture}} \\
1 & \cellcolor{fpa11}$0.412_{\pm0.046}$ & \cellcolor{fpa14}$0.427_{\pm0.058}$ & \cellcolor{fpa15}$0.417_{\pm0.046}$ &  \\
5 & \cellcolor{fpa16}$0.402_{\pm0.059}$ & \cellcolor{fpa17}$0.398_{\pm0.074}$ & \cellcolor{fpa18}$0.392_{\pm0.060}$ &  \\
15 & \cellcolor{fpa19}$0.353_{\pm0.061}$ & \cellcolor{fpa20}$0.352_{\pm0.075}$ & \cellcolor{fpa21}$0.350_{\pm0.060}$ &  \\
50 & \cellcolor{fpa22}$0.310_{\pm0.058}$ & \cellcolor{fpa23}$0.319_{\pm0.069}$ & \cellcolor{fpa24}$0.258_{\pm0.055}$ &  \\
200 & \cellcolor{fpa25}$0.334_{\pm0.058}$ & \cellcolor{fpa21}$0.350_{\pm0.070}$ & \cellcolor{fpa24}$0.259_{\pm0.055}$ &  \\
final & \cellcolor{fpa9}$0.339_{\pm0.059}$ & \cellcolor{fpa26}$0.355_{\pm0.071}$ & \cellcolor{fpa27}$0.253_{\pm0.056}$ &  \\
\bottomrule
\end{tabular}
\end{table}

\textbf{General (normalized) sheaf learning aligns with the normalized graph
Laplacian.} The learnt profile moves toward the topological reference along
both axes of Table~\ref{tab:global_general_norm}, with epochs and with
depth.
At initialization it sits at $W_1\approx0.41$ uniformly across depth; by the
final checkpoint the deepest layer is at $0.253\pm0.056$ against
$0.339\pm0.059$ for the first, with most of the movement in the first fifty
epochs.
The effect is absent at the untrained checkpoint, so normalization alone
does not produce it: it is an implicit bias of the \emph{learning}.
Node by node, the learnt sheaf and the normalized Laplacian have similar
Forman profiles, so at the end of training the relative volume of each node
resembles that of the normalized Laplacian.
This holds on most datasets, with the exceptions discussed below.

\textbf{Each learning mechanism has its own bias.}
We repeat the analysis for the unnormalized sheaf Laplacian and for Joint
Sheaf Diffusion~\cite{caralt2024jointdiffusionprocessesinductive},
normalized and unnormalized (Appendix~\ref{sec:supp_forman_tables}).
The unnormalized Laplacian matches neither the zero nor the topological
reference; unnormalized joint diffusion diverges in the first epochs and
middle layers, and recovers towards the end; normalized joint diffusion
diverges from both references as layers and epochs increase.
The Forman profile therefore distinguishes how each mechanism moves through
the sheaf space.

\textbf{The learnt maps stay away from the identity.}
To see whether the learnt maps are identity-like, we evaluate the Forman
profile of their unnormalized Laplacian, the maps having been trained with
the normalized one, and, conversely, the normalized profile of maps trained
without normalization (Appendix~\ref{sec:supp_other}).
In the first case the profile stays far from the zero reference
(Tables~\ref{tab:global_general_other_norm}
and~\ref{tab:global_general_other_unnorm} in
Appendix~\ref{sec:supp_forman_tables}).
Locally the learnt maps act as the normalized Laplacian does, yet they can
still differ from the identity sheaf, retaining features such as signs and
transport.

\textbf{Relation to representation degeneracy.}
The representation degeneracy reading of oversmoothing in Sheaf Neural
Networks~\cite{donmez2026oversmoothingrepresentationdegeneracyneural}
argues, from performance analysis, that the optimizer is drawn towards a
trivial subrepresentation of the learnt sheaf, without measuring the
operator itself.
Measuring it directly, we see local alignment with the topological reference
while the maps stay away from the identity, in networks with few layers, far
from oversmoothing; we interpret this as an inductive bias towards one
region of the sheaf space.
This may also explain why identity sheaves perform similarly to general
ones~\cite{caralt2026necessitylearnablesheaflaplacians}: general sheaves
align locally with the normalized Laplacian, and so resemble normalized
identity sheaves.
Global features of the maps, which our lens does not see, may be what
separates the two.

\textbf{Exceptions.}
Minesweeper, Tolokers and Questions do not follow this pattern.
None of the three loses task performance, with ROC-AUC $0.955\pm0.006$,
$0.854\pm0.007$ and $0.767\pm0.014$, so alignment is a tendency of the
learning dynamics and not a condition for the sheaf to be predictive.
Node homophily does not separate the two groups: Roman Empire and Amazon
Ratings are strongly heterophilic, and behave as the homophilic datasets do.
We hypothesize that the relevant factor is the feature-structure
relationship of these three datasets, the only ones with two classes of
nodes.

\textbf{Limits.}
Our aim was to test whether the Forman profile is a relevant
characteristic of sheaf diffusion in real systems, and Sheaf Neural
Networks are the setting in which we tested it.
The analysis above already relates the profile to the learning mechanism,
to normalization and to training; how the structure of the graph and that
of its features shape it, which the exceptions above point to, will require
a further analysis, for instance on synthetic benchmarks adapted to
sheaves, which we expect to flow into a machine-learning study of the
interpretability and explainability of Sheaf Neural Networks.
The descriptor is also local: it sees only the relative magnitude, or
volume, of the restriction maps, and misses their global features.
At $d=1$ it is insensitive to the sign of the maps, which is crucial for
non-aggregative behaviour on heterophilic graphs.
Comparing the distribution of the profile over the whole graph with a
reference partly compensates for this locality.
A spectral descriptor, such as the entropic
susceptibility~\cite{Villegas2023}, could complement it once spectral sheaf
theory is further developed.

\section{Conclusion and Open Problems}
\label{sec:conclusion}

In this work we introduced a Bochner decomposition for cellular sheaves,
splitting a sheaf Laplacian into a positive semidefinite term, the Bochner
sheaf Laplacian, and a block diagonal term, the Forman curvature, and we used
it to answer a spectral, a probabilistic and a combinatorial question.
The Bochner term belongs to the class of Bochner matrices, which contains
connection Laplacians, matrix-weighted Laplacians and signed graphs.
For this class we characterized the kernel by holonomy, proved Cheeger
inequalities extending those of Bandeira, Singer and Spielman, with
constants measuring the anisotropy, and attached a random walk on
$V\times\mathbb{S}^{d-1}$, the generalized Gremban expansion.
The Forman term gives a discrete Bochner theorem for cellular sheaves,
bounding the first sheaf cohomology under nonnegative curvature.
Together, the two terms allow sheaf diffusion to be seen as a
reaction-diffusion model in which, in an opinion dynamics interpretation,
the Bochner term drives the agents towards a common conviction and the
Forman term trades conviction off against volume.

Both terms also yield a quantity that can be measured on a trained model,
the conviction and the Forman profile, which we derived from the theory and
then tested on a real system, Sheaf Neural Networks.
The conviction is a score of node importance: removing high-conviction nodes
is the most damaging intervention in every perturbation mode, removing
low-conviction nodes costs less than removing arbitrary ones, a
degree-matched control rules out connectivity as the explanation, and the
damage grows steadily with the conviction rank.
The Forman profile tracks how a learning mechanism moves through the sheaf
space, and each mechanism leaves its own signature: normalized general sheaf
learning aligns locally with the normalized graph Laplacian while its maps
stay away from the identity, which we interpret as an inductive bias towards
one region of the sheaf space.
A systematic investigation of both, from the standpoint of machine-learning
interpretability and explainability, is left for future work: of the role
of the relationship between the structure of the graph and its features,
of whether these mechanisms vary across more recent sheaf
architectures~\cite{borgi2026polynomialneuralsheafdiffusion,
vanlooy2026oscillatoryneuraldynamicssheaves}, and of the extension of our
metrics, which assume a constant stalk dimension $d$, to heterogeneous
sheaves~\cite{braithwaite2026heterogeneoussheafneuralnetworks}.

Several questions remain open.
On the theoretical side, other techniques could sharpen the constants of
the Cheeger inequalities, and relating
the mixing time of the anisotropic random walk to the frustration constants
of Section~\ref{sec:cheeger} is left for future work.
Since Bochner matrices carry a random walk, we believe a Bakry--\'Emery
curvature can be built for them, as it has been for connection
Laplacians~\cite{Liu_2019_curvature_buser}.
Our framework also opens the possibility of carrying the dynamical processes
recently studied on connection Laplacians, such as synchronization and
Turing patterns~\cite{gallo2025global,gallo2026turing,gallo2026synchronization},
to the anisotropic setting, which we regard as one of the most promising
directions.

On the empirical side, we believe the conviction and the Forman profile
could become tools for the interpretability of Sheaf Neural Networks, and
for compressing or sparsifying the graphs they operate on.
The Forman profile could also be complemented by a global, spectral
descriptor, once spectral sheaf theory supplies one, such as the entropic
susceptibility~\cite{Villegas2023}, which detects the characteristic scales
of diffusion on a network.
Finally, the fact that every sheaf Laplacian is a Bochner Laplacian plus a
block diagonal term may itself be put to work, inside Sheaf Neural Networks
or to build new classes of sheaves.

\section*{Acknowledgements}
F.M.G.\ carried out this work as his Master's thesis at EPFL, while visiting
the University of Cambridge, and acknowledges the support of the ELSA
Mobility Program.
The authors are grateful to Kathryn Hess Bellwald for supervising F.M.G.'s Master's
thesis at EPFL, from which this work originates.
The authors also thank Ferran Hernandez Caralt for useful discussions and insights on
joint sheaf diffusion.

\section*{Code availability}
Our experiments build on the original implementation of Neural Sheaf
Diffusion by Bodnar et al.~\cite{bodnar2022neural},
\url{https://github.com/twitter-research/neural-sheaf-diffusion}, which we
modified to extract the conviction and the Forman profile.
The full code will be made available upon acceptance of the manuscript.

\clearpage
\bibliographystyle{sn-mathphys-num}
\bibliography{biblio}

@book{grady2010discrete,
  author    = {Grady, Leo J. and Polimeni, Jonathan R.},
  title     = {Discrete Calculus: Applied Analysis on Graphs for Computational Science},
  publisher = {Springer London},
  address   = {London},
  year      = {2010},
  edition   = {1},
  pages     = {XVI, 366},
  isbn      = {978-1-84996-289-6},
  doi       = {10.1007/978-1-84996-290-2},
}

@book{edelsbrunner2010computational,
  author    = {Edelsbrunner, Herbert and Harer, John L.},
  title     = {Computational Topology: An Introduction},
  publisher = {American Mathematical Society},
  address   = {Providence, RI},
  year      = {2010},
  series    = {Miscellaneous Books},
  volume    = {69},
  pages     = {xii+241},
  isbn      = {978-0-8218-4925-5},
  doi       = {10.1090/mbk/069},
}

@Article{Edelsbrunner2002,
author={{Edelsbrunner}
and {Letscher}
and {Zomorodian}},
title={Topological Persistence and Simplification},
journal={Discrete {\&} Computational Geometry},
year={2002},
month={Nov},
day={01},
volume={28},
number={4},
pages={511-533},
issn={1432-0444},
doi={10.1007/s00454-002-2885-2},
url={https://doi.org/10.1007/s00454-002-2885-2}
}

@misc{bronstein2021geometricdeeplearninggrids,
      title={Geometric Deep Learning: Grids, Groups, Graphs, Geodesics, and Gauges}, 
      author={Michael M. Bronstein and Joan Bruna and Taco Cohen and Petar Veličković},
      year={2021},
      eprint={2104.13478},
      archivePrefix={arXiv},
      primaryClass={cs.LG},
      url={https://arxiv.org/abs/2104.13478}, 
}

@book{pontryagin1952foundations,
  author      = {Pontryagin, L. S.},
  title       = {Foundations of Combinatorial Topology},
  publisher   = {Graylock Press},
  address     = {Rochester, NY},
  year        = {1952},
  pages       = {xii+99},
  note        = {Translated from the first (1947) Russian edition by F. Bagemihl, H. Komm, and W. Seidel},
  url         = {https://archive.org/details/pontryagin-foundations-of-combinatorial-topology}
}

@book{bretto2013hypergraph,
  author    = {Bretto, Alain},
  title     = {Hypergraph Theory: An Introduction},
  publisher = {Springer Cham},
  address   = {Cham},
  year      = {2013},
  series    = {Mathematical Engineering},
  edition   = {1},
  pages     = {XIII, 119},
  isbn      = {978-3-319-00079-4},
  issn      = {2192-4732},
  doi       = {10.1007/978-3-319-00080-0},
}

@book{chung1997spectral,
  author    = {Chung, Fan R. K.},
  title     = {Spectral Graph Theory},
  publisher = {American Mathematical Society},
  address   = {Providence, RI},
  year      = {1997},
  series    = {CBMS Regional Conference Series in Mathematics},
  number    = {92},
  pages     = {207},
  isbn      = {978-0-8218-0315-8},
  doi       = {10.1090/cbms/092},
}

@inproceedings{muhammad2006control_kyoto,
  author    = {Muhammad, Abubakr and Egerstedt, Magnus},
  title     = {Control Using Higher Order {L}aplacians in Network Topologies},
  booktitle = {Proceedings of the 17th International Symposium on Mathematical
               Theory of Networks and Systems (MTNS)},
  address   = {Kyoto, Japan},
  pages     = {1024--1038},
  year      = {2006}
}

@article{Forman2003Bochner,
  author    = {Forman, Robin},
  title     = {Bochner's Method for Cell Complexes and Combinatorial {Ricci} Curvature},
  journal   = {Discrete \& Computational Geometry},
  year      = {2003},
  volume    = {29},
  number    = {3},
  pages     = {323--374},
  doi       = {10.1007/s00454-002-0743-x},
  url       = {https://doi.org/10.1007/s00454-002-0743-x},
  issn      = {1432-0444}
}

@Article{Nurisso2025,
author={Nurisso, Marco
and Morandini, Marta
and Lucas, Maxime
and Vaccarino, Francesco
and Gili, Tommaso
and Petri, Giovanni},
title={Higher-order {Laplacian} renormalization},
journal={Nature Physics},
year={2025},
month={Apr},
day={01},
volume={21},
number={4},
pages={661-668},
issn={1745-2481},
doi={10.1038/s41567-025-02784-1},
url={https://doi.org/10.1038/s41567-025-02784-1}
}

@book{bredon1997sheaf,
  author    = {Bredon, Glen E.},
  title     = {Sheaf Theory},
  publisher = {Springer New York},
  address   = {New York, NY},
  year      = {1997},
  series    = {Graduate Texts in Mathematics},
  volume    = {170},
  edition   = {2},
  pages     = {XI, 504},
  isbn      = {978-0-387-94905-5},
  issn      = {0072-5285},
  doi       = {10.1007/978-1-4612-0647-7},
}

@book{rosiak2022sheaf,
  author    = {Rosiak, Daniel},
  title     = {Sheaf Theory through Examples},
  publisher = {MIT Press},
  address   = {Cambridge, MA},
  year      = {2022},
  pages     = {454},
  isbn      = {978-0-262-54215-9},
  url       = {https://direct.mit.edu/books/oa-monograph/5460/Sheaf-Theory-through-Examples}
}

@article{bandeira2013cheegerinequalitygraphconnection,
  author  = {Bandeira, Afonso S. and Singer, Amit and Spielman, Daniel A.},
  title   = {A {Cheeger} Inequality for the Graph Connection {Laplacian}},
  journal = {SIAM Journal on Matrix Analysis and Applications},
  volume  = {34},
  number  = {4},
  pages   = {1611--1630},
  year    = {2013},
  doi     = {10.1137/120875338}
}

@article{singer2011vectordiffusionmapsconnection,
  author  = {Singer, Amit and Wu, Hau-tieng},
  title   = {Vector Diffusion Maps and the Connection {Laplacian}},
  journal = {Communications on Pure and Applied Mathematics},
  volume  = {65},
  number  = {8},
  pages   = {1067--1144},
  year    = {2012},
  doi     = {10.1002/cpa.21395}
}

@misc{curry2014sheavescosheavesapplications,
      title={Sheaves, Cosheaves and Applications}, 
      author={Justin Curry},
      year={2014},
      eprint={1303.3255},
      archivePrefix={arXiv},
      primaryClass={math.AT},
      url={https://arxiv.org/abs/1303.3255}, 
}

@article{Discourse_sheaves,
author = {Hansen, Jakob and Ghrist, Robert},
title = {Opinion Dynamics on Discourse Sheaves},
journal = {SIAM Journal on Applied Mathematics},
volume = {81},
number = {5},
pages = {2033-2060},
year = {2021},
doi = {10.1137/20M1341088},
}

@article{hansen2019toward,
  author   = {Hansen, Jakob and Ghrist, Robert},
  title    = {Toward a Spectral Theory of Cellular Sheaves},
  journal  = {Journal of Applied and Computational Topology},
  volume   = {3},
  number   = {4},
  pages    = {315--358},
  year     = {2019},
  publisher = {Springer},
  issn     = {2367-1734},
  doi      = {10.1007/s41468-019-00038-7},
  url      = {https://doi.org/10.1007/s41468-019-00038-7}
}

@article{Tian_2024,
  author    = {Tian, Yu and Lambiotte, Renaud},
  title     = {Spreading and Structural Balance on Signed Networks},
  journal   = {SIAM Journal on Applied Dynamical Systems},
  volume    = {23},
  number    = {1},
  pages     = {50--80},
  year      = {2024},
  month     = jan,
  doi       = {10.1137/22M1542325},
  url       = {http://dx.doi.org/10.1137/22M1542325},
  issn      = {1536-0040}
}

@article{hansen2020sheaf,
  author  = {Hansen, Jakob and Gebhart, Thomas},
  title   = {Sheaf Neural Networks},
  journal = {NeurIPS Workshop: Topological Data Analysis and Beyond},
  year    = {2020}
}

@article{bodnar2022neural,
  author  = {Bodnar, Cristian and Di Giovanni, Francesco and Chamberlain, Benjamin Paul
             and Li\`{o}, Pietro and Bronstein, Michael M.},
  title   = {Neural Sheaf Diffusion: A Topological Perspective on Heterophily
             and Oversmoothing in {GNN}s},
  journal = {Advances in Neural Information Processing Systems (NeurIPS)},
  year    = {2022}
}

@InProceedings{caralt2024jointdiffusionprocessesinductive,
  title = 	 {Joint Diffusion Processes as an Inductive Bias in Sheaf Neural Networks},
  author =       {Hernandez Caralt, Ferran and Bern{\'a}rdez Gil, Guillermo and Duta, Iulia and Li{\`o}, Pietro and Alarc{\'o}n Cot, Eduard},
  booktitle = 	 {Proceedings of the Geometry-grounded Representation Learning and Generative Modeling Workshop (GRaM)},
  pages = 	 {249--263},
  year = 	 {2024},
  editor = 	 {Vadgama, Sharvaree and Bekkers, Erik and Pouplin, Alison and Kaba, Sekou-Oumar and Walters, Robin and Lawrence, Hannah and Emerson, Tegan and Kvinge, Henry and Tomczak, Jakub and Jegelka, Stephanie},
  volume = 	 {251},
  series = 	 {Proceedings of Machine Learning Research},
  month = 	 {29 Jul},
  url = 	 {https://proceedings.mlr.press/v251/hernandez-caralt24a.html}
}

@article{Hamilton1982ThreemanifoldsWP,
  title={Three-manifolds with positive {Ricci} curvature},
  author={Richard S. Hamilton},
  journal={Journal of Differential Geometry},
  year={1982},
  volume={17},
  pages={255-306},
  url={https://api.semanticscholar.org/CorpusID:117938273}
}

@article{Harary1953OnTN,
  title   = {On the notion of balance of a signed graph.},
  author  = {Frank Harary},
  journal = {Michigan Mathematical Journal},
  year    = {1953},
  volume  = {2},
  pages   = {143-146},
  doi     = {10.1307/mmj/1028989917},
  url     = {https://api.semanticscholar.org/CorpusID:120368252}
}

@article{structural_duality,
author = {Harary, Frank},
title = {Structural duality},
journal = {Behavioral Science},
volume = {2},
number = {4},
pages = {255-265},
doi = {https://doi.org/10.1002/bs.3830020403},
url = {https://onlinelibrary.wiley.com/doi/abs/10.1002/bs.3830020403},
eprint = {https://onlinelibrary.wiley.com/doi/pdf/10.1002/bs.3830020403},
year = {1957}
}

@article{ZASLAVSKY198247,
title = {Signed graphs},
journal = {Discrete Applied Mathematics},
volume = {4},
number = {1},
pages = {47-74},
year = {1982},
issn = {0166-218X},
doi = {https://doi.org/10.1016/0166-218X(82)90033-6},
url = {https://www.sciencedirect.com/science/article/pii/0166218X82900336},
author = {Thomas Zaslavsky}
}

@article{zaslavsky1987balanced,
  author  = {Zaslavsky, Thomas},
  title   = {Balanced Decompositions of a Signed Graph},
  journal = {Journal of Combinatorial Theory, Series B},
  volume  = {43},
  number  = {1},
  pages   = {1--13},
  year    = {1987}
}

@article{harary1959measurement,
  author  = {Harary, Frank},
  title   = {On the Measurement of Structural Balance},
  journal = {Behavioral Science},
  volume  = {4},
  number  = {4},
  pages   = {316--323},
  year    = {1959},
  doi     = {10.1002/bs.3830040405},
  url     = {https://doi.org/10.1002/bs.3830040405}
}

@INPROCEEDINGS{Investigation_spectral_signed_matrix,
  author={Fox, Alyson and Sanders, Geoffrey and Knyazev, Andrew},
  booktitle={2018 IEEE High Performance extreme Computing Conference (HPEC)}, 
  title={Investigation of Spectral Clustering for Signed Graph Matrix Representations}, 
  year={2018},
  volume={},
  number={},
  pages={1-7},
  doi={10.1109/HPEC.2018.8547575}}

@article{Numerical_methods_Gremban_expansion,
  title = {Numerical Methods for {Gremban}'s Expansion of Signed Graphs},
  author = {Alyson Fox and Thomas A. Manteuffel and Geoffrey Sanders},
  year = {2017},
  doi = {10.1137/16M1082433},
  url = {https://doi.org/10.1137/16M1082433},
  researchr = {https://researchr.org/publication/FoxMS17},
  cites = {0},
  citedby = {0},
  journal = {SIAM J. Scientific Computing},
  volume = {39},
  number = {5},
}

@phdthesis{GrembanPHD,
  author  = {Keith D. Gremban},
  title   = {Combinatorial Preconditioners for Sparse, Symmetric, Diagonally Dominant Linear Systems},
  school  = {Carnegie Mellon University},
  year    = {1996},
  address = {Pittsburgh},
  month   = {October},
  note    = {CMU CS Tech Report CMU-CS-96-123},
  url     = {https://www.cs.cmu.edu/%7Eglmiller/Publications/Papers/GrembanPHD.pdf}
}

@book{iversen1986cohomology,
  author    = {Iversen, Birger},
  title     = {Cohomology of Sheaves},
  publisher = {Springer},
  address   = {Berlin, Heidelberg},
  year      = {1986},
  series    = {Universitext},
  pages     = {XII, 464},
  isbn      = {978-3-540-16389-3},
  issn      = {0172-5939},
  doi       = {10.1007/978-3-642-82783-9},
}

@book{kashiwara1990sheaves,
  author    = {Kashiwara, Masaki and Schapira, Pierre},
  title     = {Sheaves on Manifolds},
  subtitle  = {With a Short History. «Les d{\'e}buts de la th{\'e}orie des faisceaux» by Christian Houzel},
  publisher = {Springer},
  address   = {Berlin, Heidelberg},
  year      = {1990},
  series    = {Grundlehren der mathematischen Wissenschaften},
  volume    = {292},
  pages     = {X, 512},
  isbn      = {978-3-540-51861-7},
  issn      = {0072-7830},
  doi       = {10.1007/978-3-662-02661-8},
}

@incollection{houzel1990short,
  author    = {Houzel, Christian},
  title     = {A Short History: Les D{\'e}buts de la Th{\'e}orie des Faisceaux},
  booktitle = {Sheaves on Manifolds},
  editor    = {Kashiwara, Masaki and Schapira, Pierre},
  publisher = {Springer},
  address   = {Berlin, Heidelberg},
  year      = {1990},
  pages     = {7--22},
  doi       = {10.1007/978-3-662-02661-8_2}
}

@article{expansion_matrix_weighted_graphs,
title = {Expansion in matrix-weighted graphs},
journal = {Linear Algebra and its Applications},
volume = {630},
pages = {252-273},
year = {2021},
issn = {0024-3795},
doi = {https://doi.org/10.1016/j.laa.2021.08.009},
url = {https://www.sciencedirect.com/science/article/pii/S0024379521003062},
author = {Jakob Hansen}
}

@book{trinh2026matrix,
  author    = {Trinh, Minh Hoang and Ahn, Hyo-Sung},
  title     = {Matrix-Weighted Graphs: Theory and Applications},
  publisher = {Springer Cham},
  address   = {Cham},
  year      = {2026},
  series    = {Lecture Notes in Control and Information Sciences},
  volume    = {498},
  edition   = {1},
  pages     = {XXIV, 361},
  isbn      = {978-3-032-03078-8},
  issn      = {0170-8643},
  doi       = {10.1007/978-3-032-03079-5},
}

@misc{fioresi2026sheavesgraphsdifferentialcalculi,
      title={Sheaves on Graphs and their Differential Calculi}, 
      author={Rita Fioresi and Angelica Simonetti and Ferdinando Zanchetta},
      year={2026},
      eprint={2602.21176},
      archivePrefix={arXiv},
      primaryClass={math.DG},
      url={https://arxiv.org/abs/2602.21176}, 
}

@misc{caralt2026necessitylearnablesheaflaplacians,
      title={On the Necessity of Learnable Sheaf {Laplacians}}, 
      author={{Hernandez Caralt}, Ferran and {Gonzàlez i Català}, Mar and Bazaga, Adrián and Liò, Pietro},
      year={2026},
      eprint={2603.05395},
      archivePrefix={arXiv},
      primaryClass={cs.LG},
      url={https://arxiv.org/abs/2603.05395}, 
}

@misc{donmez2026oversmoothingrepresentationdegeneracyneural,
      title={Oversmoothing as Representation Degeneracy in Neural Sheaf Diffusion}, 
      author={Arif Dönmez and Axel Mosig and Ellen Fritsche and Katharina Koch},
      year={2026},
      eprint={2605.11178},
      archivePrefix={arXiv},
      primaryClass={cs.LG},
      url={https://arxiv.org/abs/2605.11178}, 
}

@article{sen2008collective,
  author  = {Sen, Prithviraj and Namata, Galileo and Bilgic, Mustafa and Getoor, Lise and Galligher, Brian and Eliassi-Rad, Tina},
  title   = {Collective Classification in Network Data},
  journal = {AI Magazine},
  volume  = {29},
  number  = {3},
  pages   = {93--106},
  year    = {2008},
  doi     = {10.1609/aimag.v29i3.2157},
  url     = {https://doi.org/10.1609/aimag.v29i3.2157}
}

@inproceedings{pei2020geomgcn,
  author    = {Pei, Hongbin and Wei, Bingzhe and Chang, Kevin Chen-Chuan and Lei, Yu and Yang, Bo},
  title     = {Geom-GCN: Geometric Graph Convolutional Networks},
  booktitle = {International Conference on Learning Representations (ICLR)},
  year      = {2020},
  note      = {arXiv:2002.05287}
}

@inproceedings{tang2009social,
  author    = {Tang, Jie and Sun, Jimeng and Wang, Chi and Yang, Zi},
  title     = {Social Influence Analysis in Large-Scale Networks},
  booktitle = {Proceedings of the 15th ACM SIGKDD International Conference on Knowledge Discovery and Data Mining},
  pages     = {807--816},
  year      = {2009},
  publisher = {ACM},
  address   = {New York, NY, USA},
  doi       = {10.1145/1557019.1557108},
}

@inproceedings{platonov2023critical,
  author    = {Platonov, Oleg and Kuznedelev, Denis and Diskin, Michael and Babenko, Artem and Prokhorenkova, Liudmila},
  title     = {A Critical Look at the Evaluation of GNNs under Heterophily: Are We Really Making Progress?},
  booktitle = {The Eleventh International Conference on Learning Representations},
  year      = {2023}
}

@article{Li_perturbation_bounds,
author = {Li, Ren-Cang},
title = {New Perturbation Bounds for the Unitary Polar Factor},
journal = {SIAM Journal on Matrix Analysis and Applications},
volume = {16},
number = {1},
pages = {327-332},
year = {1995},
doi = {10.1137/S0895479893256359},
}

@article{gallo2025global,
  author  = {Gallo, Anna and Tian, Yu and Lambiotte, Renaud and Carletti, Timoteo},
  title   = {Global Synchronization in Matrix-Weighted Networks},
  journal = {Communications Physics},
  volume  = {8},
  number  = {1},
  pages   = {512},
  year    = {2025},
  doi     = {10.1038/s42005-025-02419-0},
  url     = {https://doi.org/10.1038/s42005-025-02419-0}
}

@article{gallo2026turing,
  author  = {Gallo, Anna and Segnou, William and Carletti, Timoteo},
  title   = {Turing Patterns in Matrix-Weighted Networks},
  journal = {arXiv preprint arXiv:2602.13080},
  year    = {2026}
}

@article{gallo2026synchronization,
  author  = {Gallo, Anna and Lambiotte, Renaud and Carletti, Timoteo},
  title   = {Synchronization of Higher-Dimensional {Kuramoto} Oscillators on Networks: From Scalar to Matrix-Weighted Couplings},
  journal = {Journal of Physics: Complexity},
  year    = {2026},
  doi     = {10.1088/2632-072X/ae76ad},
  url     = {https://doi.org/10.1088/2632-072X/ae76ad},
  note    = {arXiv:2603.08352}
}

@article{lambiotte_matrix_weighted,
  title = {Matrix-Weighted Networks for Modeling Multidimensional Dynamics: Theoretical Foundations and Applications to Network Coherence},
  author = {Tian, Yu and Kojaku, Sadamori and Sayama, Hiroki and Lambiotte, Renaud},
  journal = {Phys. Rev. Lett.},
  volume = {134},
  issue = {23},
  pages = {237401},
  numpages = {7},
  year = {2025},
  month = {Jun},
  publisher = {American Physical Society},
  doi = {10.1103/jw1k-6s7w},
  url = {https://link.aps.org/doi/10.1103/jw1k-6s7w}
}

@misc{dacunto2026sheaftheoreticsignalprocessinggraphs,
      title={Sheaf-theoretic Signal Processing on Graphs: Spectral Theory, Filtering, and Sampling}, 
      author={Gabriele D'Acunto and Leonardo Di Nino and Paolo Di Lorenzo and Sergio Barbarossa},
      year={2026},
      eprint={2608.01318},
      archivePrefix={arXiv},
      primaryClass={eess.SP},
      url={https://arxiv.org/abs/2608.01318}, 
}

@misc{wang2026curvatureguidedsheafdiffusionunsupervised,
      title={Curvature-Guided Sheaf Diffusion for Unsupervised Community Detection on Heterophilic Graphs}, 
      author={Feifan Wang},
      year={2026},
      eprint={2606.30249},
      archivePrefix={arXiv},
      primaryClass={cs.LG},
      url={https://arxiv.org/abs/2606.30249}, 
}

@article{Schaub_2020,
  author    = {Schaub, Michael T. and Benson, Austin R. and Horn, Paul
               and Lippner, Gabor and Jadbabaie, Ali},
  title     = {Random Walks on Simplicial Complexes and the Normalized
               {H}odge 1-{L}aplacian},
  journal   = {SIAM Review},
  volume    = {62},
  number    = {2},
  pages     = {353--391},
  year      = {2020},
  month     = jan,
  doi       = {10.1137/18M1201019},
  url       = {http://dx.doi.org/10.1137/18M1201019},
  issn      = {1095-7200}
}

@book{Petersen2016RiemannianGeometry,
  author    = {Petersen, Peter},
  title     = {Riemannian Geometry},
  edition   = {3rd},
  series    = {Graduate Texts in Mathematics},
  volume    = {171},
  year      = {2016},
  publisher = {Springer International Publishing},
  address   = {Cham},
  doi       = {10.1007/978-3-319-26654-1},
  isbn      = {978-3-319-26652-7}
}

@misc{gallier2016spectraltheoryunsignedsigned,
  author        = {Gallier, Jean},
  title         = {Spectral Theory of Unsigned and Signed Graphs.
                   Applications to Graph Clustering: a Survey},
  year          = {2016},
  eprint        = {1601.04692},
  archivePrefix = {arXiv},
  primaryClass  = {cs.LG},
  url           = {https://arxiv.org/abs/1601.04692}
}

@article{Atay_2020_cheeger,
   title={{Cheeger} constants, structural balance, and spectral clustering analysis for signed graphs},
   volume={343},
   ISSN={0012-365X},
   url={http://dx.doi.org/10.1016/j.disc.2019.111616},
   DOI={10.1016/j.disc.2019.111616},
   number={1},
   journal={Discrete Mathematics},
   publisher={Elsevier BV},
   author={Atay, Fatihcan M. and Liu, Shiping},
   year={2020},
   month=Jan, pages={111616} }

@article{Liu_2019_curvature_buser,
   title={Curvature and Higher Order {Buser} Inequalities for the Graph Connection {Laplacian}},
   volume={33},
   ISSN={1095-7146},
   url={http://dx.doi.org/10.1137/16M1056353},
   DOI={10.1137/16m1056353},
   number={1},
   journal={SIAM Journal on Discrete Mathematics},
   publisher={Society for Industrial & Applied Mathematics (SIAM)},
   author={Liu, Shiping and Münch, Florentin and Peyerimhoff, Norbert},
   year={2019},
   month=Jan, pages={257–305} }

@article{Lange_2015_frustration,
   title={Frustration index and {Cheeger} inequalities for discrete and continuous magnetic {Laplacians}},
   volume={54},
   ISSN={1432-0835},
   url={http://dx.doi.org/10.1007/s00526-015-0935-x},
   DOI={10.1007/s00526-015-0935-x},
   number={4},
   journal={Calculus of Variations and Partial Differential Equations},
   publisher={Springer Science and Business Media LLC},
   author={Lange, Carsten and Liu, Shiping and Peyerimhoff, Norbert and Post, Olaf},
   year={2015},
   month=Nov, pages={4165–4196} }

@article{cloninger2023randomwalksconductanceresistance,
  author  = {Cloninger, Alexander and Mishne, Gal and Oslandsbotn, Andreas and Robertson, Sawyer J. and Wan, Zhengchao and Wang, Yusu},
  title   = {Random Walks, Conductance, and Resistance for the Connection Graph {Laplacian}},
  journal = {SIAM Journal on Matrix Analysis and Applications},
  volume  = {45},
  number  = {3},
  pages   = {1541--1572},
  year    = {2024},
  doi     = {10.1137/23M1595400}
}

@article{kannan2022normalizedlaplaciansgaingraphs,
  author  = {Kannan, M. Rajesh and Kumar, Navish and Pragada, Shivaramakrishna},
  title   = {Normalized {Laplacians} for Gain Graphs},
  journal = {The American Journal of Combinatorics},
  volume  = {1},
  pages   = {20--39},
  year    = {2022},
  doi     = {10.63151/amjc.v1i.3}
}

@Article{Gao2021,
author={Gao, Tingran
and Brodzki, Jacek
and Mukherjee, Sayan},
title={The Geometry of Synchronization Problems and Learning Group Actions},
journal={Discrete {\&} Computational Geometry},
year={2021},
month={Jan},
day={01},
volume={65},
number={1},
pages={150-211},
issn={1432-0444},
doi={10.1007/s00454-019-00100-2},
url={https://doi.org/10.1007/s00454-019-00100-2}
}

@article{block_diagonally_dominant_matrices,
author = {Feingold, David and Varga, Richard},
year = {1962},
month = {12},
pages = {},
title = {Block diagonally dominant matrices and generalizations of the Gershgorin Theorem},
volume = {12},
journal = {Pac. J. Math.},
doi = {10.2140/pjm.1962.12.1241}
}

@Book{PSD_Bhatia2007,
author={Bhatia, Rajendra},
title={Positive Definite Matrices},
year={2007},
publisher={Princeton University Press},
address={Princeton, NJ},
isbn={9781400827787},
url={http://ieeexplore.ieee.org/document/9452503}
}

@misc{first2025cheegerinequalitycoboundaryexpansion,
      title={The {Cheeger} Inequality and Coboundary Expansion: Beyond Constant Coefficients}, 
      author={Uriya A. First and Tali Kaufman},
      year={2025},
      eprint={2208.01776},
      archivePrefix={arXiv},
      primaryClass={math.CO},
      url={https://arxiv.org/abs/2208.01776}, 
}

@inproceedings{Topping2021UnderstandingOA,
  author    = {Topping, Jake and Di Giovanni, Francesco and Chamberlain, Benjamin Paul and Dong, Xiaowen and Bronstein, Michael M.},
  title     = {Understanding over-squashing and bottlenecks on graphs via curvature},
  booktitle = {International Conference on Learning Representations},
  year      = {2022}
}

@article{antibalance_signed_graph,
author = {Li, Hong-Hai and Li, Jiong-Sheng},
year = {2009},
month = {06},
pages = {},
title = {Note on the normalized {Laplacian} eigenvalues of signed graphs},
volume = {44},
journal = {The Australasian Journal of Combinatorics [electronic only]}
}

@Article{Villegas2023,
author={Villegas, Pablo
and Gili, Tommaso
and Caldarelli, Guido
and Gabrielli, Andrea},
title={{Laplacian} renormalization group for heterogeneous networks},
journal={Nature Physics},
year={2023},
month={Mar},
day={01},
volume={19},
number={3},
pages={445-450},
issn={1745-2481},
doi={10.1038/s41567-022-01866-8},
url={https://doi.org/10.1038/s41567-022-01866-8}
}

@Article{Curry2015,
author={Curry, Justin Michael},
title={Topological data analysis and cosheaves},
journal={Japan Journal of Industrial and Applied Mathematics},
year={2015},
month={Jul},
day={01},
volume={32},
number={2},
pages={333-371},
issn={1868-937X},
doi={10.1007/s13160-015-0173-9},
url={https://doi.org/10.1007/s13160-015-0173-9}
}

@misc{hanks2025distributedmultiagentcoordinationcellular,
      title={Distributed Multi-agent Coordination over Cellular Sheaves}, 
      author={Tyler Hanks and Hans Riess and Samuel Cohen and Trevor Gross and Matthew Hale and James Fairbanks},
      year={2025},
      eprint={2504.02049},
      archivePrefix={arXiv},
      primaryClass={math.OC},
      url={https://arxiv.org/abs/2504.02049}, 
}

@misc{robinson2016sheavescanonicaldatastructuresensor,
      title={Sheaves are the canonical datastructure for sensor integration}, 
      author={Michael Robinson},
      year={2016},
      eprint={1603.01446},
      archivePrefix={arXiv},
      primaryClass={math.AT},
      url={https://arxiv.org/abs/1603.01446}, 
}

@INPROCEEDINGS{distributed_optimization,
  author={Hansen, Jakob and Ghrist, Robert},
  booktitle={2019 57th Annual Allerton Conference on Communication, Control, and Computing (Allerton)}, 
  title={Distributed Optimization with Sheaf Homological Constraints}, 
  year={2019},
  volume={},
  number={},
  pages={565-571},
  doi={10.1109/ALLERTON.2019.8919796}}

@misc{hu2026discretebakryemerycurvaturetensors,
      title={Discrete Bakry-\'Emery curvature tensors and matrices of connection graphs}, 
      author={Chunyang Hu and Shiping Liu},
      year={2026},
      eprint={2209.10762},
      archivePrefix={arXiv},
      primaryClass={math.CO},
      url={https://arxiv.org/abs/2209.10762}, 
}

@misc{borgi2026equivariantsheafneuralnetworks,
      title={Equivariant Sheaf Neural Networks: Learning Geometric Transport on Graphs}, 
      author={Alessio Borgi and Mario Severino and Fabrizio Silvestri and Pietro Liò},
      year={2026},
      eprint={2608.28853},
      archivePrefix={arXiv},
      primaryClass={cs.LG},
      url={https://arxiv.org/abs/2608.28853}, 
}

@InProceedings{pmlr-v196-barbero22a,
  title = 	 {Sheaf Neural Networks with Connection {Laplacians}},
  author =       {Barbero, Federico and Bodnar, Cristian and S\'aez de Oc\'ariz Borde, Haitz and Bronstein, Michael and Veli\v{c}kovi\'c, Petar and Li\`o, Pietro},
  booktitle = 	 {Proceedings of Topological, Algebraic, and Geometric Learning Workshops 2022},
  pages = 	 {28--36},
  year = 	 {2022},
  editor = 	 {Cloninger, Alexander and Doster, Timothy and Emerson, Tegan and Kaul, Manohar and Ktena, Ira and Kvinge, Henry and Miolane, Nina and Rieck, Bastian and Tymochko, Sarah and Wolf, Guy},
  volume = 	 {196},
  series = 	 {Proceedings of Machine Learning Research},
  month = 	 {25 Feb--22 Jul},
  url = 	 {https://proceedings.mlr.press/v196/barbero22a.html}
}

@InProceedings{peng2026sheafneuralnetworksspd,
  title = 	 {Sheaf Neural Networks on {SPD} Manifolds: Second-Order Geometric Representation Learning},
  author =       {Peng, Yuhan and Dong, Junwen and Zeng, Yuzhi and Li, Hao and Ju, Ce and Feng, Huitao and Taha, Diaaeldin and Wienhard, Anna and Xia, Kelin},
  booktitle = 	 {Proceedings of the 43rd International Conference on Machine Learning},
  pages = 	 {97981--98013},
  year = 	 {2026},
  editor = 	 {Zhang, Tong and Dudik, Miroslav and Jaggi, Martin and Agarwal, Alekh and Li, Sharon and Schuurmans, Dale and Zhu, Jerry and Berkenkamp, Felix and Dong, Hanze and Bietti, Alberto},
  volume = 	 {306},
  series = 	 {Proceedings of Machine Learning Research},
  month = 	 {06--11 Jul},
  url = 	 {https://proceedings.mlr.press/v306/peng26e.html}
}

@inproceedings{bamberger2025bundleneuralnetworksmessage,
  author    = {Bamberger, Jacob and Barbero, Federico and Dong, Xiaowen and Bronstein, Michael},
  title     = {Bundle Neural Network for Message Diffusion on Graphs},
  booktitle = {International Conference on Learning Representations},
  editor    = {Yue, Y. and Garg, A. and Peng, N. and Sha, F. and Yu, R.},
  pages     = {16788--16813},
  year      = {2025}
}

@misc{borgi2026polynomialneuralsheafdiffusion,
      title={Polynomial Neural Sheaf Diffusion: A Spectral Filtering Approach on Cellular Sheaves}, 
      author={Alessio Borgi and Fabrizio Silvestri and Pietro Liò},
      year={2026},
      eprint={2512.00242},
      archivePrefix={arXiv},
      primaryClass={cs.LG},
      url={https://arxiv.org/abs/2512.00242}, 
}

@misc{vanlooy2026oscillatoryneuraldynamicssheaves,
      title={Oscillatory Neural Dynamics over Sheaves}, 
      author={Jan-Willem Van Looy and Alessandro Trenta and Alessio Gravina and Alessio Borgi and Ferdinando Zanchetta and Pietro Liò and Davide Bacciu and Rita Fioresi},
      year={2026},
      eprint={2610.10018},
      archivePrefix={arXiv},
      primaryClass={cs.LG},
      url={https://arxiv.org/abs/2610.10018}, 
}

@misc{braithwaite2026heterogeneoussheafneuralnetworks,
      title={Heterogeneous Sheaf Neural Networks}, 
      author={Luke Braithwaite and Alessio Borgi and Gabriele Onorato and Kristjan Tarantelli and Francesco Restuccia and Fabrizio Silvestri and Pietro Liò},
      year={2026},
      eprint={2409.08036},
      archivePrefix={arXiv},
      primaryClass={cs.LG},
      url={https://arxiv.org/abs/2409.08036}, 
}

@misc{ali2026letheadstalkdiagonal,
      title={Let the Heads Talk: Beyond Diagonal Graph Attention}, 
      author={Riccardo Ali and Alessio Borgi and Mario Severino and Alessio Gravina and Davide Bacciu and Pietro Liò and Christopher Irwin},
      year={2026},
      eprint={2610.01494},
      archivePrefix={arXiv},
      primaryClass={cs.LG},
      url={https://arxiv.org/abs/2610.01494}, 
}

\clearpage
\appendix
\section{Additional Proofs}
\label{sec:proofs}

This appendix collects the proofs that do not sit inline with their
statement: either because they are longer, or because the technique carries
over unchanged from prior work and reads better collected in one place,
without interrupting the results above.
Each result is restated before its proof.

\medskip
\noindent\textbf{Proposition~\ref{prop:holonomy} (Kernel and holonomy).} Let
$B$ satisfy (H1)--(H3), and for a walk $\gamma:v_0\to\cdots\to v_k$ let
$T_\gamma := W_{v_{k-1}v_k}\cdots W_{v_0v_1}\in O(d)$.
Then $\dim\ker(B)\leq d$, with equality iff $T_\gamma = I$ for every closed
walk $\gamma$.

\begin{proof}
The evaluation map $\phi_v:\ker(B)\to\R^d$, $x\mapsto x_v$, is injective by
Proposition~\ref{prop:kernel} and connectedness (H3): $x_v=0$ forces $x_u=0$
at every neighbour, hence $x=0$ by induction along the graph.
So $\dim\ker(B)\leq d$.

If the transport maps are path independent, fix $v$ and $\alpha\in\R^d$;
setting $x_u:=T_\gamma\alpha$ for any path $\gamma$ from $v$ to $u$ is well
defined, and for an edge $(w_1,w_2)$, extending a path to $w_1$ by that edge
gives $x_{w_2}=W_{w_1w_2}x_{w_1}$, so $x\in\ker(B)$ by
Proposition~\ref{prop:kernel}; thus $\phi_v$ is surjective and $\dim\ker(B)=d$.
Conversely if $\phi_v$ is surjective, every $\alpha$ extends to some $x\in\ker(B)$
with $x_w=T_\gamma\alpha$ for \emph{every} path $\gamma$ from $v$ to $w$, by
Proposition~\ref{prop:kernel} again; two paths with the same endpoints then
give $T_\gamma\alpha=T_{\gamma'}\alpha$ for all $\alpha$, hence
$T_\gamma=T_{\gamma'}$, which is path independence, equivalent to
$T_\gamma=I$ on closed walks by concatenating a walk with the reverse of a
reference one.
\end{proof}

\medskip
\noindent The proofs of Section~\ref{sec:cheeger} use the following notation.
\begin{center}
\small
\renewcommand{\arraystretch}{1.25}
\begin{tabular}{@{}l>{\raggedright\arraybackslash}p{0.7\textwidth}@{}}
\toprule
Symbol & Meaning \\
\midrule
$D_v$, $D$ & $D_v=\sum_{u\sim v}M_{vu}$, the diagonal blocks of $B$, and $D=\operatorname{blockdiag}(D_v)$ \\
$\eta(x)$ & $x^\top Bx/x^\top Dx$ for $x\neq0$ \\
$\eta^*_G$ & $\min\eta(x)$ over nonzero $x:V\to\mathbb S^{d-1}\cup\{0\}$ \\
$n_x$ & the local norm vector, $(n_x)_v=\norm{x_v}$ \\
$w^\pm_{vu}$ & $\lambda_{\min}(M_{vu})$ and $\lambda_{\max}(M_{vu})$ \\
$\kappa_E$, $\kappa(M)$ & $\max_e w^+_e/w^-_e$, and $\lambda_{\max}(M)/\lambda_{\min}(M)$ for $M\succ0$ \\
$d^-_v$, $D^-_0$ & $\sum_{u\sim v}w^-_{vu}$, and $\operatorname{diag}(d^-_v)$ \\
$x^{>t}$ & $x^{>t}_v=x_v/\norm{x_v}$ if $\norm{x_v}^2>t$, and $0$ otherwise \\
$\Phi(y)$ & $\sum_{v\neq u}\sqrt{w^-_{vu}}\,\norm{y_v-W_{vu}^\top y_u}_{M_{vu}}\big/\sum_v d^-_v\norm{y_v}$ (Lemma~\ref{lem:global_l1}) \\
\bottomrule
\end{tabular}
\end{center}

\noindent\textbf{Auxiliary lemmas for Section~\ref{sec:cheeger}.} We record
three facts used throughout the Cheeger proofs, all resting on the
Courant--Fischer characterization:
\begin{equation}
\label{eq:rayleigh}
  \lambda_k(\Bn) \;=\; \min_{\dim S=k}\ \max_{0\neq x\in S}\ \eta(x),
  \qquad\text{in particular}\qquad
  \lambda_1(\Bn) = \min_{x\neq0}\eta(x),
\end{equation}
which follows from the substitution $y=D^{1/2}x$ turning $\eta$ into the
Rayleigh quotient of $\Bn$.
First, a \emph{compression bound}: for symmetric $A\in\R^{N\times N}$ and
$Q\in\R^{N\times k}$ with orthonormal columns,
$\lambda_i(Q^\top AQ)\geq\lambda_i(A)$ for $i\leq k$, hence
$\tr(Q^\top AQ)\geq\sum_{i\leq k}\lambda_i(A)$.
Indeed, by~\eqref{eq:rayleigh} $\lambda_i(Q^\top AQ)$ is a min-max over the
$i$-dimensional subspaces $S\subseteq\R^k$, and $Q$, being an isometry,
sends each of them to the $i$-dimensional subspace $QS\subseteq\R^N$ with
the same Rayleigh quotients; so the min-max for $\lambda_i(Q^\top AQ)$ runs
over a subfamily of the one for $\lambda_i(A)$, and its value can only be
larger.
Second, a \emph{comparison toolkit}.
We write $D^-_0:=\operatorname{diag}(d^-_v)$ for the degree matrix of the
companion graph $G_-$ of Definition~\ref{def:anisotropy} and
$\vol^-:=\sum_vd^-_v$ for its volume.
Then for all $v,u\in V$, $z\in\R^d$ and $x\in\R^{nd}$,
\begin{enumerate}
  \item[(i)] $d^-_v\leq\lambda_{\min}(D_v)\leq\lambda_{\max}(D_v)\leq d^+_v\leq\kappa_Ed^-_v$,
in particular $\kappa_D\leq\kappa_E$;
  \item[(ii)] $\norm{n_x}^2_{D^-_0}\leq x^\top Dx\leq\kappa_E\norm{n_x}^2_{D^-_0}$;
  \item[(iii)] $w^-_{vu}\norm z^2\leq\norm z^2_{M_{vu}}\leq w^+_{vu}\norm z^2$;
  \item[(iv)] $n_x^\top\Lm n_x\leq x^\top Bx$,
\end{enumerate}
where (i) follows from super/subadditivity of
$\lambda_{\min},\lambda_{\max}$ on $D_v=\sum_uM_{vu}$, (ii) and (iii) are
Rayleigh bounds, and (iv) follows from the reverse triangle inequality
applied to $\norm{x_v-W_{vu}^\top x_u}^2_{M_{vu}}\geq
w^-_{vu}(\norm{x_v}-\norm{x_u})^2$ together with Lemma~\ref{lem:energy_id}.
Item (iv) is the bridge from the block problem to the scalar companion
graph, and the source of most of the anisotropy constants.
Third, a \emph{columnwise dictionary}: for $X:V\to\R^{d\times d}$ and
$k=1,\dots,d$, let $x^k_v:=X_ve_k\in\R^d$ be the $k$-th column of $X_v$,
with $e_k$ the $k$-th standard basis vector of $\R^d$, and let
$x^k\in\R^{nd}$ collect these columns over all vertices:
\begin{equation}
  \tfrac12\sum_{v\neq u}\norm{X_v-W_{vu}^\top X_u}^2_{M_{vu},F}
  = \sum_{k\leq d}(x^k)^\top Bx^k,
\end{equation}
and if $X$ is orthogonal-valued, $\sum_k(x^k)^\top Dx^k=\vol(G)$; both
follow from $\norm A^2_{M,F}=\sum_k\norm{Ae_k}^2_M$ applied edgewise and
vertexwise, the second using $g_vg_v^\top=I$.

\medskip
\noindent\textbf{Proposition~\ref{prop:relaxation} (Relaxation).} (a)
$\lambda_1(\Bn)\leq\eta^*_G\leq\eta_G$; (b) $\lambda_1(\Bn)\leq\nu_G$; (c)
$\frac{1}{d\kappa_D}\sum_{k\leq d}\lambda_k(\Bn)\leq\nu_G$.

\begin{proof}
(a) By~\eqref{eq:rayleigh}, $\lambda_1(\Bn)=\min_{x\neq0}\eta(x)$, and both
frustration constants minimize $\eta$ over subsets of
$\R^{nd}\setminus\{0\}$; the second inequality holds since every $\mathbb
S^{d-1}$-valued assignment is admissible for $\eta^*_G$.

For (b) and (c), fix $g:V\to O(d)$, let $x^k$ be its columns, set
$y^k:=D^{1/2}x^k$ and $Y:=[y^1\cdots y^d]$.
Since $(x^k)^\top Bx^k=(y^k)^\top\Bn y^k$ and $(x^k)^\top
Dx^k=\norm{y^k}^2$, the columnwise dictionary gives $\nu(g)=\tr(Y^\top\Bn
Y)/\tr(Y^\top Y)$.

$Y$ has rank $d$: if $Yc=0$ then $\sum_kc_kx^k=0$, and evaluating at any
vertex gives $g_vc=0$, so $c=0$.
Let $Y=QR$ be a QR factorization, $R$ invertible, and set $S:=Q^\top\Bn
Q\succeq0$, $P:=RR^\top\succ0$; by cyclicity of the trace,
$\nu(g)=\tr(SP)/\tr(P)$.

(b) Diagonalizing $P=\sum_jp_jq_jq_j^\top$ with $p_j>0$ and $q_j$
orthonormal, the Rayleigh bound gives $q_j^\top Sq_j\geq\lambda_1(S)$, so
$\tr(SP)\geq\lambda_1(S)\tr(P)$; the compression bound gives
$\lambda_1(S)\geq\lambda_1(\Bn)$, and minimizing over $g$ proves (b).

(c) Since $R$ is square, $P=RR^\top$ and $R^\top R=Y^\top Y$ share
eigenvalues.
Moreover $Y^\top Y=\sum_vg_v^\top D_vg_v$, and conjugation by the
orthogonal $g_v$ preserves eigenvalues, so
$\lambda_{\min}(D_v)I\preceq g_v^\top D_vg_v\preceq\lambda_{\max}(D_v)I$;
summing over $v$, with $a:=\sum_v\lambda_{\min}(D_v)$ and
$b:=\sum_v\lambda_{\max}(D_v)$:
\begin{equation}
  aI\preceq Y^\top Y\preceq bI,
  \qquad\text{hence}\qquad
  \lambda_{\min}(P)\geq a,\quad \lambda_{\max}(P)\leq b .
\end{equation}
The trace of $P$ is the sum of its $d$ eigenvalues, each at most
$\lambda_{\max}(P)$, so $\tr(P)\leq d\lambda_{\max}(P)\leq db$; and
$\lambda_{\max}(D_v)\leq\kappa_D\lambda_{\min}(D_v)$ for every $v$ gives
$b\leq\kappa_Da$.
Together, $\lambda_{\min}(P)\geq a\geq b/\kappa_D\geq\tr(P)/(d\kappa_D)$.
Splitting $P=\lambda_{\min}(P)I+P'$ with $P'\succeq0$ and using
$\tr(SP')=\tr(S^{1/2}P'S^{1/2})\geq0$:
\begin{equation}
  \tr(SP) \geq \lambda_{\min}(P)\tr(S)
  \geq \lambda_{\min}(P)\sum_{k\leq d}\lambda_k(\Bn)
  \geq \frac{\tr(P)}{d\kappa_D}\sum_{k\leq d}\lambda_k(\Bn),
\end{equation}
using the compression bound for the middle step.
Divide by $\tr(P)$ and minimize over $g$.
\end{proof}

\medskip
\noindent\textbf{Proposition~\ref{prop:zero_sets} (Zero sets).}
$\lambda_1(\Bn)=0\iff\ker(B)\neq\{0\}\iff\eta_G=0\iff\eta^*_G=0$, and
$\nu_G=0$ iff the transport maps are path independent.

\begin{proof}
$\Bn=D^{-1/2}BD^{-1/2}$ with $D\succ0$, so $\Bn$ and $B$ have the same rank,
giving the first equivalence.
If $\ker(B)\neq\{0\}$, take $0\neq x\in\ker(B)$: by
Proposition~\ref{prop:kernel} the norm $\norm{x_v}$ is a constant $c>0$, so
$x/c$ maps $V$ into $\mathbb S^{d-1}$ with $\eta(x/c)=0$, giving $\eta_G=0$
and a fortiori $\eta^*_G=0$.
Conversely if $\eta^*_G=0$, witnessed by $x\neq0$ valued in $\mathbb
S^{d-1}\cup\{0\}$, then Lemma~\ref{lem:energy_id} and (H2) give
$x_v=W_{vu}^\top x_u$ on every edge; the support $S:=\{v:x_v\neq0\}$ is
nonempty, and no edge can leave it, since $v\in S$, $u\notin S$ would force
$x_v=W_{vu}^\top\cdot0=0$.
By (H3), $S=V$, so $x\in\ker(B)$ is nonzero.

For the last claim, the columnwise dictionary gives $\nu(g)=0$ iff every
column of $g$ lies in $\ker(B)$, that is, by Proposition~\ref{prop:kernel},
iff $g_u=W_{vu}g_v$ on every edge.
An orthogonal-valued $g$ with this property exists exactly when the
transport maps are path independent, which by
Proposition~\ref{prop:holonomy} is equivalent to $\dim\ker(B)=d$.
\end{proof}

\medskip
\noindent\textbf{Theorem~\ref{thm:partial} (Partial synchronization).}
$\lambda_1(\Bn)\leq\eta^*_G\leq\kappa_E^{3/2}\sqrt{10\lambda_1(\Bn)}$, and
Algorithm~\ref{alg:partial} attains the upper bound.

The lower bound is Proposition~\ref{prop:relaxation}(a).
For the upper bound we take an eigenvector for $\lambda_1(\Bn)$, turn it into a partial assignment by a random threshold, and control the
resulting quotient in four steps: a normalization inequality for a single $M$-norm
(Lemma~\ref{lem:weighted_rounding}), its consequence on one edge
(Lemma~\ref{lem:per_edge_rounding}), the sum over the graph in an $\ell^1$
form (Lemma~\ref{lem:global_l1}), and the passage back to the quotient
$\eta$ (Lemma~\ref{lem:l1_to_l2}).
Since $W_{vu}^\top$ is a Euclidean isometry, the threshold test is
unaffected by it, and $W_{vu}^\top(x^{>t}_u)=(W_{vu}^\top x_u)^{>t}$.

\medskip
\auxlemma{Weighted normalization inequality}\label{lem:weighted_rounding}
\emph{Let $M\succ0$ have condition number $\kappa(M):=\lambda_{\max}(M)/\lambda_{\min}(M)$.
For all Euclidean unit vectors $y,z\in\R^d$ and all $\alpha\geq1$:}
\begin{equation}
  \norm{z-y}_M + (\alpha^2-1)\norm{z}_M
  \;\leq\; \frac{\sqrt5}{2}\sqrt{\kappa(M)}\;\norm{\alpha z - y}_M\,(1+\alpha).
\end{equation}

\begin{proof}
The three inequalities below use, in turn, $\norm z=1$, the Euclidean
inequality
$\norm{z-y}+\alpha^2-1\leq\tfrac{\sqrt5}{2}\norm{\alpha z-y}(1+\alpha)$,
which is Proposition~A.1 of~\cite{bandeira2013cheegerinequalitygraphconnection},
and the lower Rayleigh bound
$\norm{\alpha z-y}\leq\lambda_{\min}(M)^{-1/2}\norm{\alpha z-y}_M$:
\begin{equation}
\begin{aligned}
  \norm{z-y}_M + (\alpha^2-1)\norm{z}_M
  &\leq \sqrt{\lambda_{\max}(M)}\,\bigl(\norm{z-y}+\alpha^2-1\bigr)\\
  &\leq \sqrt{\lambda_{\max}(M)}\;\frac{\sqrt5}{2}\norm{\alpha z-y}(1+\alpha)\\
  &\leq \frac{\sqrt5}{2}\sqrt{\kappa(M)}\;\norm{\alpha z-y}_M(1+\alpha).
\end{aligned}
\end{equation}
\end{proof}

The factor $\sqrt{\kappa(M)}$ cannot be dispensed with.
Take $M=\operatorname{diag}(1,\varepsilon)$, $\alpha=2$,
$z=(\tfrac12,\tfrac{\sqrt3}{2})$ and $y=(1,0)$.
Then $\alpha z-y=(0,\sqrt3)$, so $\norm{\alpha z-y}_M=\sqrt{3\varepsilon}$,
while $\norm{z-y}_M\geq\tfrac12$ and $\norm z_M\geq\tfrac12$ make the
left-hand side at least $2$.
Before normalization the two vectors differ only in the direction $M$ does
not see; normalizing both to unit length moves the discrepancy into the direction
$M$ does see, and this requires $\alpha\neq1$, so the effect comes from the
rescaling and not from the rotation.

\medskip
\auxlemma{Per-edge estimate}\label{lem:per_edge_rounding}
\emph{Let $x\in\R^{nd}$, with blocks $x_v\in\R^d$, satisfy
$\max_v\norm{x_v}=1$, and let $t$ be uniform on $[0,1]$.
Then for every ordered pair $(v,u)$:}
\begin{equation}
  \E_t\,\norm{x^{>t}_v - W_{vu}^\top x^{>t}_u}_{M_{vu}}
  \;\leq\; \frac{\sqrt5}{2}\sqrt{\kappa_E}\;
     \norm{x_v - W_{vu}^\top x_u}_{M_{vu}}
     \bigl(\norm{x_v}+\norm{x_u}\bigr).
\end{equation}

\begin{proof}
Write $M:=M_{vu}$, $a:=x_v$ and $b:=W_{vu}^\top x_u$, so that
$\norm b=\norm{x_u}$ and $b^{>t}=W_{vu}^\top x^{>t}_u$.
The claim is symmetric in $a$ and $b$, so we may assume $\norm b\leq\norm a$.
If $b=0$ it holds directly: both sides are multiples of $\norm a\norm a_M$,
the left with constant $1$ and the right with
$\tfrac{\sqrt5}{2}\sqrt{\kappa_E}\geq1$.
Otherwise $0<\norm b\leq\norm a\leq1$, and we set:
\begin{equation}
  z := \frac{a}{\norm a},\qquad y := \frac{b}{\norm b},\qquad
  \alpha := \frac{\norm a}{\norm b}\geq1 .
\end{equation}
As $t$ ranges over $[0,1]$ there are three regimes: with probability
$\norm b^2$ both vectors survive and the normalized difference is $z-y$; with
probability $\norm a^2-\norm b^2$ only $a$ survives, leaving $z$; otherwise
both vanish.
Hence:
\begin{equation}
  \E_t\norm{a^{>t}-b^{>t}}_M
  = \norm b^2\norm{z-y}_M + \bigl(\norm a^2-\norm b^2\bigr)\norm z_M
  = \norm b^2\Bigl[\norm{z-y}_M + (\alpha^2-1)\norm z_M\Bigr].
\end{equation}
Lemma~\ref{lem:weighted_rounding} bounds the bracket.
Since $\alpha z-y=(a-b)/\norm b$ and $1+\alpha=(\norm a+\norm b)/\norm b$,
the two factors of $\norm b$ cancel, and $\kappa(M_{vu})\leq\kappa_E$
concludes.
\end{proof}

\medskip
\auxlemma{Global $\ell^1$ estimate}\label{lem:global_l1}
\emph{Define:}
\begin{equation}
  \Phi(y) := \frac{\sum_{v\neq u}\sqrt{w^-_{vu}}\;
     \norm{y_v - W_{vu}^\top y_u}_{M_{vu}}}{\sum_v d^-_v\norm{y_v}} .
\end{equation}
\emph{If $\max_v\norm{x_v}=1$, there is $t\in[0,1)$ with $x^{>t}\neq0$ and
$\Phi(x^{>t})\leq\kappa_E\sqrt{10\,\eta(x)}$.}

\begin{proof}
Multiply Lemma~\ref{lem:per_edge_rounding} by $\sqrt{w^-_{vu}}$, sum over
ordered pairs, and apply Cauchy--Schwarz to the product of the two factors:
\begin{equation}
  (\mathrm A) := \norm{x_v - W_{vu}^\top x_u}_{M_{vu}}
  \qquad\text{and}\qquad
  (\mathrm B) := \sqrt{w^-_{vu}}\,\bigl(\norm{x_v}+\norm{x_u}\bigr).
\end{equation}
For the first factor, Lemma~\ref{lem:energy_id} and item (ii) of the
toolkit give:
\begin{equation}
  \sum_{v\neq u}(\mathrm A)^2 = 2\,x^\top Bx = 2\,\eta(x)\,x^\top Dx
  \leq 2\kappa_E\,\eta(x)\,\norm{n_x}^2_{D^-_0} .
\end{equation}
For the second, $(p+q)^2\leq2(p^2+q^2)$ and $w^-_{vu}=w^-_{uv}$ give:
\begin{equation}
  \sum_{v\neq u}(\mathrm B)^2
  \leq 2\sum_{v\neq u}w^-_{vu}\bigl(\norm{x_v}^2+\norm{x_u}^2\bigr)
  = 4\sum_v d^-_v\norm{x_v}^2 = 4\norm{n_x}^2_{D^-_0} .
\end{equation}
Multiplying the two square roots by the constant of
Lemma~\ref{lem:per_edge_rounding}:
\begin{equation}
\begin{aligned}
  \E_t[\text{numerator of }\Phi(x^{>t})]
  &\leq \frac{\sqrt5}{2}\sqrt{\kappa_E}\cdot\sqrt{2\kappa_E\,\eta(x)}\,
       \norm{n_x}_{D^-_0}\cdot 2\norm{n_x}_{D^-_0}\\
  &= \kappa_E\sqrt{10\,\eta(x)}\;\norm{n_x}^2_{D^-_0} .
\end{aligned}
\end{equation}
For the denominator, $\norm{x^{>t}_v}\in\{0,1\}$, since $x^{>t}_v$ is either
a unit vector or zero, and $\norm{x_v}^2\leq1$,
so $\E_t\norm{x^{>t}_v}=\Pr[t<\norm{x_v}^2]=\norm{x_v}^2$, and therefore the
expected denominator of $\Phi(x^{>t})$ is exactly
$\sum_vd^-_v\norm{x_v}^2=\norm{n_x}^2_{D^-_0}$.
With $c:=\kappa_E\sqrt{10\,\eta(x)}$ the two estimates give
$\E_t[\mathrm{num}(t)-c\,\mathrm{den}(t)]\leq0$, so some realization of $t$
has $\mathrm{num}(t)\leq c\,\mathrm{den}(t)$.
For every $t<1$ the vertex attaining $\norm{x_v}=1$ survives, so
$x^{>t}\neq0$ and the denominator is positive; dividing gives the claim.
\end{proof}

\medskip
\auxlemma{From $\ell^1$ to $\ell^2$}\label{lem:l1_to_l2}
\emph{For every nonzero $y:V\to\mathbb S^{d-1}\cup\{0\}$,
$\eta(y)\leq\sqrt{\kappa_E}\,\Phi(y)$.}

\begin{proof}
Each $y_v$ is a unit vector or zero, so $\norm{y_v}^2=\norm{y_v}$ and
$\norm{y_v-W_{vu}^\top y_u}\leq2$.
Writing the square as a product and bounding one factor by item (iii) of
the toolkit:
\begin{equation}
  \tfrac12\norm{y_v - W_{vu}^\top y_u}^2_{M_{vu}}
  \leq \tfrac12\sqrt{w^+_{vu}}\cdot 2\cdot
     \norm{y_v - W_{vu}^\top y_u}_{M_{vu}}
  \leq \sqrt{\kappa_E\,w^-_{vu}}\;\norm{y_v - W_{vu}^\top y_u}_{M_{vu}},
\end{equation}
using $w^+_{vu}\leq\kappa_Ew^-_{vu}$.
Summing over ordered pairs bounds the numerator of $\eta$ by
$\sqrt{\kappa_E}$ times that of $\Phi$.
For the denominator, item (i) of the toolkit gives:
\begin{equation}
  y^\top Dy = \sum_v y_v^\top D_vy_v
  \geq \sum_v\lambda_{\min}(D_v)\norm{y_v}^2
  \geq \sum_v d^-_v\norm{y_v},
\end{equation}
which is the denominator of $\Phi$.
\end{proof}

\begin{proof}[Proof of Theorem~\ref{thm:partial}]
Let $z$ be a unit eigenvector of $\Bn$ for $\lambda_1(\Bn)$ set
$x:=D^{-1/2}z$ and rescale it so that $\max_v\norm{x_v}=1$; since $\eta$ is
invariant under scaling, $\eta(x)=\lambda_1(\Bn)$.
By Lemma~\ref{lem:global_l1} there is $t$ with $x^{>t}\neq0$ and
$\Phi(x^{>t})\leq\kappa_E\sqrt{10\,\eta(x)}$.
The assignment $x^{>t}$ takes values in $\mathbb S^{d-1}\cup\{0\}$ and is
nonzero, so it is admissible in the definition of $\eta^*_G$, and
Lemma~\ref{lem:l1_to_l2} gives:
\begin{equation}
  \eta^*_G \leq \eta(x^{>t}) \leq \sqrt{\kappa_E}\,\Phi(x^{>t})
  \leq \kappa_E^{3/2}\sqrt{10\,\lambda_1(\Bn)} .
\end{equation}
Algorithm~\ref{alg:partial} enumerates all the breakpoints of
$t\mapsto x^{>t}$, hence finds a threshold at least as good.
\end{proof}

\medskip
\noindent\textbf{Theorem~\ref{thm:full} (Full synchronization).}
$\lambda_1(\Bn)\leq\eta_G\leq54\kappa_E^2\lambda_2(\Lmn)^{-1}\lambda_1(\Bn)$,
and Algorithm~\ref{alg:full} attains the upper bound.

The lower bound is Proposition~\ref{prop:relaxation}(a).
For the upper bound we again start from $x:=D^{-1/2}z$, with $z$ a unit
eigenvector of $\Bn$ for $\lambda_1(\Bn)$, so that $\eta(x)=\lambda_1(\Bn)$.
Its energy is already as small as it can be; the only obstruction to being
a full assignment is that the norms $\norm{x_v}$ are not all equal.
So we normalize it vertex by vertex, $\tilde x_v:=x_v/\norm{x_v}$, and show
that normalizing costs at most a bounded factor.
On an edge whose endpoints both have norm close to $1$, normalizing is
Lipschitz (Lemma~\ref{lem:normalize_two}).
The vertices of norm far from $1$ are few: the norm profile $n_x$ has small
energy on the companion graph, hence is nearly constant by its spectral gap
(Lemma~\ref{lem:norm_profile}), and a nearly constant profile leaves little
volume far from its mean (Lemma~\ref{lem:ill_balanced}).
On the edges touching those vertices we pay a crude bound, weighted by
their small volume, and adding the two contributions gives the constant
(Lemma~\ref{lem:cost_normalizing}).
Besides the notation fixed above, the proof uses the following.
\begin{center}
\small
\renewcommand{\arraystretch}{1.25}
\begin{tabular}{@{}l>{\raggedright\arraybackslash}p{0.7\textwidth}@{}}
\toprule
Symbol & Meaning \\
\midrule
$\eta_G$ & $\min\eta(x)$ over $x:V\to\mathbb S^{d-1}$ \\
$\tilde x$ & $\tilde x_v=x_v/\norm{x_v}$, and an arbitrary unit vector where $x_v=0$ \\
$\Lm$, $\Lmn$ & the Laplacian of the companion graph $G_-$, with $(\Lm)_{vv}=d^-_v$ and $(\Lm)_{vu}=-w^-_{vu}$, and $\Lmn=(D^-_0)^{-1/2}\Lm(D^-_0)^{-1/2}$ \\
$\vol^-$, $\vol^-(S)$ & $\sum_v d^-_v$, and $\sum_{v\in S}d^-_v$ for $S\subseteq V$ \\
$\Ib_x(\delta)$ & $\{v:|\,\norm{x_v}-1\,|\geq\delta\}$, the ill-balanced vertices \\
\bottomrule
\end{tabular}
\end{center}

\medskip
\auxlemma{Normalization of two vectors}\label{lem:normalize_two}
\emph{For nonzero $a,b\in\R^d$:}
\begin{equation}
  \Bigl\|\frac{a}{\norm a}-\frac{b}{\norm b}\Bigr\|
  \;\leq\; \frac{2\norm{a-b}}{\norm a+\norm b}
  \;\leq\; \frac{\norm{a-b}}{\min(\norm a,\norm b)} .
\end{equation}

\begin{proof}
Write $p:=\norm a$, $q:=\norm b$ and
$c:=\langle a/p,\,b/q\rangle\in[-1,1]$, so that:
\begin{equation}
  \Bigl\|\frac ap-\frac bq\Bigr\|^2 = 2-2c,
  \qquad
  \norm{a-b}^2 = p^2+q^2-2pqc .
\end{equation}
The first inequality, squared, reads $(2-2c)(p+q)^2\leq4(p^2+q^2-2pqc)$,
and indeed
$4(p^2+q^2-2pqc)-2(1-c)(p+q)^2=2(p-q)^2(1+c)\geq0$
since $c\geq-1$.
The second inequality follows from $p+q\geq2\min(p,q)$.
\end{proof}

\medskip
\auxlemma{Norm profile}\label{lem:norm_profile}
\emph{For every $x\neq0$ there is $\alpha_x\geq0$ such that
$u_x:=n_x-\alpha_x\mathbf 1$ satisfies $\mathbf 1^\top D^-_0u_x=0$ and the bound:}
\begin{equation}
  \norm{u_x}^2_{D^-_0} \;\leq\;
  \frac{\kappa_E\,\eta(x)}{\lambda_2(\Lmn)}\,\norm{n_x}^2_{D^-_0} .
\end{equation}

\begin{proof}
Let $\alpha_x:=(\mathbf 1^\top D^-_0n_x)/(\mathbf 1^\top D^-_0\mathbf 1)$,
the minimizer of $\alpha\mapsto\norm{n_x-\alpha\mathbf 1}_{D^-_0}$.
It is nonnegative since $n_x\geq0$ entrywise, and its stationarity
condition is exactly $\mathbf 1^\top D^-_0u_x=0$.
Since $\Lm\mathbf 1=0$, the variational characterization of the second
eigenvalue of the normalized Laplacian of $G_-$ reads:
\begin{equation}
  \lambda_2(\Lmn) = \min_{\mathbf 1^\top D^-_0u=0}
     \frac{u^\top\Lm u}{\norm u^2_{D^-_0}},
\end{equation}
which gives $u_x^\top\Lm u_x\geq\lambda_2(\Lmn)\norm{u_x}^2_{D^-_0}$.
On the other hand $\Lm\mathbf 1=0$ makes the quadratic form invariant under
adding multiples of $\mathbf 1$, so $u_x^\top\Lm u_x=n_x^\top\Lm n_x$, and
items (iv) and (ii) of the toolkit give:
\begin{equation}
  n_x^\top\Lm n_x \leq x^\top Bx = \eta(x)\,x^\top Dx
  \leq \kappa_E\,\eta(x)\,\norm{n_x}^2_{D^-_0} .
\end{equation}
Combining the two and dividing by $\lambda_2(\Lmn)>0$ gives the claim.
\end{proof}

\medskip
\auxlemma{Ill-balanced vertices are few}\label{lem:ill_balanced}
\emph{Normalize $x$ so that $\norm{n_x}^2_{D^-_0}=\vol^-$, and set
$\Ib_x(\delta):=\{v:|\,\norm{x_v}-1\,|\geq\delta\}$.
Then:}
\begin{equation}
  \frac{\vol^-\bigl(\Ib_x(\delta)\bigr)}{\vol^-}
  \;\leq\; \frac{4}{\delta^2}\cdot\frac{\kappa_E\,\eta(x)}{\lambda_2(\Lmn)} .
\end{equation}

\begin{proof}
Write $\varepsilon:=\kappa_E\eta(x)/\lambda_2(\Lmn)$, so that
Lemma~\ref{lem:norm_profile} reads $\norm{u_x}^2_{D^-_0}\leq\varepsilon\vol^-$.
First, $\alpha_x$ is close to $1$: since
$\norm{\mathbf 1}^2_{D^-_0}=\vol^-=\norm{n_x}^2_{D^-_0}$, the reverse
triangle inequality in $\norm\cdot_{D^-_0}$ gives:
\begin{equation}
  |1-\alpha_x|\sqrt{\vol^-}
  = \bigl|\,\norm{n_x}_{D^-_0} - \alpha_x\norm{\mathbf 1}_{D^-_0}\,\bigr|
  \leq \norm{u_x}_{D^-_0} \leq \sqrt{\varepsilon\vol^-},
\end{equation}
hence $|1-\alpha_x|\leq\sqrt\varepsilon$.
Second, ill-balanced vertices have a large fluctuation: for
$v\in\Ib_x(\delta)$:
\begin{equation}
  |(u_x)_v| = \bigl|\,\norm{x_v}-\alpha_x\,\bigr|
  \geq \bigl|\,\norm{x_v}-1\,\bigr| - |1-\alpha_x|
  \geq \delta-\sqrt\varepsilon .
\end{equation}
Suppose $\sqrt\varepsilon<\delta/2$, so that $\delta-\sqrt\varepsilon>\delta/2>0$.
Restricting the sum defining $\norm{u_x}^2_{D^-_0}$ to $\Ib_x(\delta)$, all
its terms being nonnegative:
\begin{equation}
  \varepsilon\vol^- \;\geq\; \norm{u_x}^2_{D^-_0}
  \;\geq\; \sum_{v\in\Ib_x(\delta)}d^-_v\,(u_x)_v^2
  \;\geq\; \Bigl(\frac{\delta}{2}\Bigr)^2\vol^-\bigl(\Ib_x(\delta)\bigr),
\end{equation}
which rearranges to the claim.
If instead $\sqrt\varepsilon\geq\delta/2$, then $4\varepsilon/\delta^2\geq1$,
while the left-hand side of the claim is at most $1$.
\end{proof}

\medskip
\auxlemma{Cost of normalizing}\label{lem:cost_normalizing}
\emph{For every $x\neq0$ and every $\gamma\in(0,1)$:}
\begin{equation}
  \eta(\tilde x) \;\leq\;
  \kappa_E^2\left[\frac{16}{\gamma^2\lambda_2(\Lmn)}
     + \frac{1}{(1-\gamma)^2}\right]\eta(x),
\end{equation}
\emph{and in particular, at $\gamma=2/3$,
$\eta(\tilde x)\leq54\,\kappa_E^2\,\lambda_2(\Lmn)^{-1}\eta(x)$.}

\begin{proof}
Both $\eta(x)$ and $\tilde x$ are invariant under $x\mapsto cx$ for $c>0$,
and $n_{cx}=c\,n_x$; so replacing $x$ by $cx$ with
$c:=\sqrt{\vol^-}/\norm{n_x}_{D^-_0}$ we may assume
$\norm{n_x}^2_{D^-_0}=\vol^-$; put $\Ib:=\Ib_x(\gamma)$.
By Lemma~\ref{lem:energy_id} the numerator of $\eta(\tilde x)$ is
$\frac12\sum_{v\neq u}\norm{\tilde x_v-W_{vu}^\top\tilde x_u}^2_{M_{vu}}$,
and we split it into the ordered pairs meeting $\Ib$ and those avoiding it.

\emph{Denominator.}
Since $\norm{\tilde x_v}=1$ for every $v$, item (i) of the toolkit gives:
\begin{equation}
  \tilde x^\top D\tilde x = \sum_v\tilde x_v^\top D_v\tilde x_v
  \geq \sum_v\lambda_{\min}(D_v) \geq \sum_v d^-_v = \vol^- .
\end{equation}

\emph{Pairs meeting $\Ib$.}
Here we use only $\norm{\tilde x_v-W_{vu}^\top\tilde x_u}\leq2$, so by item
(iii) of the toolkit each summand is at most $4w^+_{vu}$.
These pairs lie in $\{(v,u):v\in\Ib\}\cup\{(v,u):u\in\Ib\}$, and since
$w^+_{vu}=w^+_{uv}$ the two families contribute equally, so, using
$d^+_v\leq\kappa_Ed^-_v$ from item (i) and then
Lemma~\ref{lem:ill_balanced} with $\delta=\gamma$:
\begin{equation}
  \frac12\sum_{(v,u)\text{ meeting }\Ib}
     \norm{\tilde x_v-W_{vu}^\top\tilde x_u}^2_{M_{vu}}
  \;\leq\; 4\sum_{v\in\Ib}d^+_v
  \;\leq\; 4\kappa_E\vol^-(\Ib)
  \;\leq\; \frac{16\kappa_E^2}{\gamma^2}\cdot
     \frac{\eta(x)}{\lambda_2(\Lmn)}\,\vol^- .
\end{equation}

\emph{Pairs avoiding $\Ib$.}
For $v\notin\Ib$ we have $\norm{x_v}>1-\gamma>0$, so $x_v\neq0$; the
arbitrary choice in the definition of $\tilde x$ therefore occurs only
inside $\Ib$, which the previous case has absorbed.
Set $a:=x_v$ and $b:=W_{vu}^\top x_u$, both of norm greater than $1-\gamma$.
Since $W_{vu}^\top$ is an isometry, $b/\norm b=W_{vu}^\top\tilde x_u$, and
Lemma~\ref{lem:normalize_two} gives:
\begin{equation}
  \norm{\tilde x_v-W_{vu}^\top\tilde x_u}
  \leq \frac{\norm{x_v-W_{vu}^\top x_u}}{1-\gamma} .
\end{equation}
Passing to the edge metric on both sides by item (iii) of the toolkit costs
$w^+_{vu}/w^-_{vu}\leq\kappa_E$, and summing over all ordered pairs, every
term being nonnegative, then using item (ii):
\begin{equation}
\begin{aligned}
  \frac12\sum_{\text{avoiding }\Ib}
    \norm{\tilde x_v-W_{vu}^\top\tilde x_u}^2_{M_{vu}}
  &\leq \frac{\kappa_E}{(1-\gamma)^2}\cdot\frac12\sum_{v\neq u}
    \norm{x_v-W_{vu}^\top x_u}^2_{M_{vu}}\\
  &= \frac{\kappa_E}{(1-\gamma)^2}\,\eta(x)\,x^\top Dx
  \leq \frac{\kappa_E^2}{(1-\gamma)^2}\,\eta(x)\,\vol^- .
\end{aligned}
\end{equation}

\emph{Conclusion.}
Adding the two contributions and dividing by the denominator:
\begin{equation}
\begin{aligned}
  \eta(\tilde x)
  &\leq \frac{1}{\vol^-}\left[\frac{16\kappa_E^2}{\gamma^2}\cdot
     \frac{\eta(x)}{\lambda_2(\Lmn)}\,\vol^-
     + \frac{\kappa_E^2}{(1-\gamma)^2}\,\eta(x)\,\vol^-\right]\\
  &= \kappa_E^2\left[\frac{16}{\gamma^2\lambda_2(\Lmn)}
     + \frac{1}{(1-\gamma)^2}\right]\eta(x).
\end{aligned}
\end{equation}
For the constant, $\lambda_2(\Lmn)\leq2$ gives:
\begin{equation}
  \frac{16}{\gamma^2\lambda_2(\Lmn)} + \frac{1}{(1-\gamma)^2}
  = \frac{1}{\lambda_2(\Lmn)}\left[\frac{16}{\gamma^2}
     + \frac{\lambda_2(\Lmn)}{(1-\gamma)^2}\right]
  \leq \frac{1}{\lambda_2(\Lmn)}\left[\frac{16}{\gamma^2}
     + \frac{2}{(1-\gamma)^2}\right].
\end{equation}
The derivative of the bracket, $-32/\gamma^3+4/(1-\gamma)^3$, vanishes when
$\gamma^3=8(1-\gamma)^3$, that is $\gamma=2/3$, where the bracket equals
$36+18=54$.
\end{proof}

\begin{proof}[Proof of Theorem~\ref{thm:full}]
With $x:=D^{-1/2}z$ as above, $\eta(x)=\lambda_1(\Bn)$.
The normalized assignment $\tilde x$ maps $V$ into $\mathbb S^{d-1}$, so it
is admissible in the definition of $\eta_G$, and
Lemma~\ref{lem:cost_normalizing} gives:
\begin{equation}
  \eta_G \leq \eta(\tilde x)
  \leq 54\,\kappa_E^2\,\lambda_2(\Lmn)^{-1}\,\lambda_1(\Bn).
\end{equation}
This $\tilde x$ is exactly the output of Algorithm~\ref{alg:full}.
\end{proof}

\medskip
\noindent\textbf{Theorem~\ref{thm:od} ($O(d)$ synchronization).}
\begin{equation}
  \frac{1}{d\,\kappa_D}\sum_{k\leq d}\lambda_k(\Bn)
  \;\leq\; \nu_G \;\leq\;
  \frac{2^{11}\kappa_E^3d^3}{\lambda_2(\Lmn)}\sum_{k\leq d}\lambda_k(\Bn),
\end{equation}
and Algorithm~\ref{alg:od} attains the upper bound.

The lower bound is Proposition~\ref{prop:relaxation}(c).
For the upper bound we take the eigenvectors of the $d$ smallest eigenvalues,
one per column, and at each vertex we replace the $d\times d$ matrix they form by its nearest
orthogonal matrix, its polar factor.
This polar projection is stable where the matrix is already close to orthogonal,
that is, where its $d$ columns are nearly orthonormal; we show that this
fails only on a set of vertices of small volume, and pay a crude bound
there.

The lengths of the columns are controlled as in the full synchronization
proof, by the spectral gap of the companion graph.
Their inner products are not controlled directly, and we recover them from
lengths through the polarization identity: for
$u:=\tfrac{1}{\sqrt2}(x^k+x^m)$, expanding
$\norm{u_v}^2=\tfrac12\norm{x^k_v+x^m_v}^2$ gives:
\begin{equation}
  \langle x^k_v,\,x^m_v\rangle
  \;=\; \norm{u_v}^2 - \frac{\norm{x^k_v}^2+\norm{x^m_v}^2}{2},
\end{equation}
so that if the three lengths at $v$ lie in $[1-\delta,1+\delta]$, then
$|\langle x^k_v,x^m_v\rangle|\leq(1+\delta)^2-(1-\delta)^2=4\delta$.
This requires $u$ to carry the same normalization as $x^k$ and $x^m$, so
that the same length estimates apply to all three, and this is why the
argument runs in the isotropized normalization $\Db$ of
Algorithm~\ref{alg:od}.
Since $\Db$ is a scalar multiple of the identity in each block:
\begin{equation}
\label{eq:barD_exact}
  x^\top\Db x = \sum_v\bar d_v\norm{x_v}^2 = \norm{n_x}^2_{\bar D_0}
  \qquad\text{for every }x\in\R^{nd},
\end{equation}
with equality, and for $\Db$-orthogonal columns this makes the
normalization of $u$ exact (Lemma~\ref{lem:exact_polarization}).
If the algorithm used $D$ directly in place of $\Db$, the identity
\eqref{eq:barD_exact} would hold only up to item (ii) of the toolkit,
$\norm{n_x}^2_{D^-_0}\leq x^\top Dx\leq\kappa_E\norm{n_x}^2_{D^-_0}$.
Then $u$ would come out with squared norm $(1+\tau)\vol^-$, where $|\tau|$
is bounded by $\kappa_E-1$ and in general by nothing better, and the
polarization bound would become $4\delta+|\tau|(1+\delta)^2$.
The additional term $|\tau|(1+\delta)^2$ does not shrink with $\delta$;
since the inner products must be of order $1/d$, it could be controlled
only by assuming $\kappa_E-1$ of order $1/d$, that is, near isotropy.

The steps are then the following: convert between $\Db$ and $D$
(Lemma~\ref{lem:conversion}), control lengths in the isotropized
normalization (Lemma~\ref{lem:profile_isotropized}), make polarization exact
(Lemma~\ref{lem:exact_polarization}), bound the volume of the vertices where
the columns are not nearly orthonormal (Lemma~\ref{lem:bad_volume}), show
that the matrix is well conditioned everywhere else
(Lemma~\ref{lem:sigma_min}), show that the polar step is stable there
(Lemma~\ref{lem:polar_stability}), and assemble
(Lemma~\ref{lem:main_od}).
Besides the notation of the two previous cases, the proof uses the following.
\begin{center}
\small
\renewcommand{\arraystretch}{1.25}
\begin{tabular}{@{}l>{\raggedright\arraybackslash}p{0.7\textwidth}@{}}
\toprule
Symbol & Meaning \\
\midrule
$\norm{A}_{M,F}$ & $\tr(A^\top MA)^{1/2}$ for $A\in\R^{d\times d}$ and $M\succeq0$ \\
$\nu(g)$, $\nu_G$ & $\tfrac12\sum_{v\neq u}\norm{g_v-W_{vu}^\top g_u}^2_{M_{vu},F}/\vol(G)$, with $\vol(G)=\sum_v\tr(D_v)$, and its minimum over $g:V\to O(d)$ \\
$\kappa_D$ & $\max_v\lambda_{\max}(D_v)/\lambda_{\min}(D_v)$ \\
$\bar d_v$, $\Db$, $\bar D_0$ & $\tfrac1d\tr(D_v)$, $\operatorname{blockdiag}(\bar d_vI_d)$, and $\operatorname{diag}(\bar d_v)$ \\
$\Bb$, $\bar\eta(x)$ & $\Db^{-1/2}B\Db^{-1/2}$, and $x^\top Bx/x^\top\Db x$ \\
$\volb$, $\volb(S)$ & $\sum_v\bar d_v$, and $\sum_{v\in S}\bar d_v$ \\
$\bar\lambda_2$ & the second smallest eigenvalue of $\bar D_0^{-1/2}\Lm\bar D_0^{-1/2}$ \\
$\mathrm{Pol}(X)$ & $PQ^\top$, for an SVD $X=P\Sigma Q^\top$ of a nonsingular $X$ \\
$X_v$ & $[\,x^1_v\ \cdots\ x^d_v\,]\in\R^{d\times d}$ \\
\bottomrule
\end{tabular}
\end{center}

\medskip
\auxlemma{Conversion}\label{lem:conversion}
\emph{Let $\bar\lambda_2$ be the second smallest eigenvalue of
$\bar D_0^{-1/2}\Lm\bar D_0^{-1/2}$, the analogue of $\lambda_2(\Lmn)$ with
$\bar D_0$ in place of $D^-_0$.
Then}
\begin{enumerate}
  \item[(i)] \emph{$d^-_v\leq\lambda_{\min}(D_v)\leq\bar d_v\leq
  \lambda_{\max}(D_v)\leq d^+_v\leq\kappa_Ed^-_v$, hence
  $D^-_0\preceq\bar D_0\preceq\kappa_ED^-_0$ and
  $\kappa_D^{-1}\Db\preceq D\preceq\kappa_D\Db$;}
  \item[(ii)] \emph{$\bar\lambda_2\geq\lambda_2(\Lmn)/\kappa_E$;}
  \item[(iii)] \emph{$\bar\lambda_2\leq\frac{n}{n-1}\leq2$ for $n\geq2$;}
  \item[(iv)] \emph{$\lambda_k(\Bb)\leq\kappa_D\lambda_k(\Bn)$ for every
  $k$.}
\end{enumerate}

\begin{proof}
(i) Since $\bar d_v=\tfrac1d\tr(D_v)$ is the mean of the eigenvalues of
$D_v$, it lies between the extreme ones, and the outer bounds are item (i)
of the toolkit.
For the block form:
\begin{equation}
  D_v \preceq \lambda_{\max}(D_v)\,I
  \preceq \kappa_D\,\lambda_{\min}(D_v)\,I
  \preceq \kappa_D\,\bar d_v\,I,
\end{equation}
and symmetrically $D_v\succeq\kappa_D^{-1}\bar d_vI$.

(ii) For a positive diagonal $\Lambda$ and $\Lm\mathbf 1=0$, the second smallest
eigenvalue of $\Lambda^{-1/2}\Lm\Lambda^{-1/2}$ can be written as:
\begin{equation}
  \min_{u\text{ nonconstant}}\;
     \frac{u^\top\Lm u}{\min_\alpha\norm{u-\alpha\mathbf 1}^2_\Lambda},
\end{equation}
since subtracting from $u$ its $\Lambda$-projection on $\mathbf 1$ changes neither
$u^\top\Lm u$ nor the denominator, and makes $u$ orthogonal to $\mathbf 1$
in the $\Lambda$-inner product.
In this form the set over which one minimizes does not depend on $\Lambda$.
Let $\alpha^\ast$ minimize $\norm{u-\alpha\mathbf 1}_{D^-_0}$; then by (i):
\begin{equation}
  \min_\alpha\norm{u-\alpha\mathbf 1}^2_{\bar D_0}
  \leq \norm{u-\alpha^\ast\mathbf 1}^2_{\bar D_0}
  \leq \kappa_E\,\norm{u-\alpha^\ast\mathbf 1}^2_{D^-_0}
  = \kappa_E\min_\alpha\norm{u-\alpha\mathbf 1}^2_{D^-_0},
\end{equation}
so each quotient for $\bar D_0$ is at least $\kappa_E^{-1}$ times the
corresponding one for $D^-_0$, and minimizing over $u$ gives the claim.

(iii) The trace of $\bar D_0^{-1/2}\Lm\bar D_0^{-1/2}$ is
$\sum_vd^-_v/\bar d_v\leq n$ by (i).
Its smallest eigenvalue is $0$, so the second is at most the mean of the
remaining $n-1$, that is at most $n/(n-1)$.

(iv) From $D\preceq\kappa_D\Db$ and $B\succeq0$:
\begin{equation}
  \frac{x^\top Bx}{x^\top\Db x} \leq \kappa_D\,\frac{x^\top Bx}{x^\top Dx}
  \qquad\text{for every }x\neq0 .
\end{equation}
Both $\lambda_k(\Bb)$ and $\lambda_k(\Bn)$ are the min-max of these
quotients over the same collection of $k$-dimensional subspaces,
by~\eqref{eq:rayleigh} and its analogue for $\Db$, and a pointwise
inequality is preserved by taking maxima and then minima.
\end{proof}

\medskip
\auxlemma{Lengths in the isotropized normalization}\label{lem:profile_isotropized}
\emph{For $x\neq0$ there is $\alpha_x\geq0$ such that
$u_x:=n_x-\alpha_x\mathbf 1$ satisfies $\mathbf 1^\top\bar D_0u_x=0$ and the bound:}
\begin{equation}
  \norm{u_x}^2_{\bar D_0} \leq \frac{\bar\eta(x)}{\bar\lambda_2}\,
  \norm{n_x}^2_{\bar D_0} .
\end{equation}
\emph{Consequently, if $\norm{n_x}^2_{\bar D_0}=\volb$, then:}
\begin{equation}
  \frac{\volb\bigl(\Ib_x(\delta)\bigr)}{\volb}
  \;\leq\; \frac{4}{\delta^2}\cdot\frac{\bar\eta(x)}{\bar\lambda_2} .
\end{equation}

\begin{proof}
The proofs of Lemmas~\ref{lem:norm_profile} and~\ref{lem:ill_balanced}
carry over with $\bar D_0$ in place of $D^-_0$, $\volb$ in place of $\vol^-$
and $\alpha_x:=(\mathbf 1^\top\bar D_0n_x)/(\mathbf 1^\top\bar D_0\mathbf 1)$.
The only change is in the chain of Lemma~\ref{lem:norm_profile}, which now
reads:
\begin{equation}
  \bar\lambda_2\,\norm{u_x}^2_{\bar D_0} \leq u_x^\top\Lm u_x
  = n_x^\top\Lm n_x \leq x^\top Bx = \bar\eta(x)\,x^\top\Db x
  = \bar\eta(x)\,\norm{n_x}^2_{\bar D_0},
\end{equation}
the last step being~\eqref{eq:barD_exact}.
\end{proof}

In the remainder we fix $x^1,\dots,x^d$ pairwise $\Db$-orthogonal with
$\norm{n_{x^k}}^2_{\bar D_0}=\volb$ for each $k$, and we set
$X_v:=[\,x^1_v\ \cdots\ x^d_v\,]\in\R^{d\times d}$.

\medskip
\auxlemma{Exact polarization}\label{lem:exact_polarization}
\emph{For $k\neq m$ put $u:=\tfrac{1}{\sqrt2}(x^k+x^m)$.
Then $\norm{n_u}^2_{\bar D_0}=\volb$ exactly, and
$\bar\eta(u)\leq\bar\eta(x^k)+\bar\eta(x^m)$.}

\begin{proof}
By~\eqref{eq:barD_exact} and $\Db$-orthogonality:
\begin{equation}
  \norm{n_u}^2_{\bar D_0} = u^\top\Db u
  = \tfrac12\bigl[(x^k)^\top\Db x^k + (x^m)^\top\Db x^m\bigr]
    + (x^k)^\top\Db x^m
  = \tfrac12\,(2\volb) + 0 = \volb .
\end{equation}
Expanding both quadratic forms:
\begin{equation}
  u^\top Bu = \tfrac12(x^k+x^m)^\top B(x^k+x^m)
  = (x^k)^\top Bx^k + (x^m)^\top Bx^m - \tfrac12(x^k-x^m)^\top B(x^k-x^m)
  \leq (x^k)^\top Bx^k + (x^m)^\top Bx^m ,
\end{equation}
the last step since $B\succeq0$, and all three quotients $\bar\eta$ have
the same denominator $\volb$.
\end{proof}

\medskip
\auxlemma{Bad volume}\label{lem:bad_volume}
\emph{Put $\Ib_{xy}(\delta):=\{v:|\langle x_v,y_v\rangle|>\delta\}$ and define:}
\begin{equation}
  \mathcal G := V\setminus\Bigl[\,\bigcup_k\Ib_{x^k}\bigl(\tfrac{1}{8d}\bigr)
     \;\cup\; \bigcup_{k\neq m}\Ib_{x^kx^m}\bigl(\tfrac{1}{2d}\bigr)\Bigr].
\end{equation}
\emph{Then:}
\begin{equation}
  \frac{\volb(\mathcal G^c)}{\volb}
  \;\leq\; \frac{2^8d^3}{\bar\lambda_2}\sum_{k\leq d}\bar\eta(x^k) .
\end{equation}

\begin{proof}
Fix $k\neq m$, put $u^{km}:=\tfrac{1}{\sqrt2}(x^k+x^m)$, and suppose $v$
lies in none of $\Ib_{x^k}(\tfrac1{8d})$, $\Ib_{x^m}(\tfrac1{8d})$,
$\Ib_{u^{km}}(\tfrac1{8d})$.
The three norms at $v$ then lie in $[1-\tfrac1{8d},1+\tfrac1{8d}]$, and
since $u^{km}$ carries the same normalization by
Lemma~\ref{lem:exact_polarization}, the polarization identity gives:
\begin{equation}
  |\langle x^k_v,x^m_v\rangle|
  = \Bigl|\,\norm{u^{km}_v}^2 - \tfrac{\norm{x^k_v}^2+\norm{x^m_v}^2}{2}\,\Bigr|
  \leq \bigl(1+\tfrac1{8d}\bigr)^2 - \bigl(1-\tfrac1{8d}\bigr)^2
  = \tfrac{1}{2d},
\end{equation}
so $v\notin\Ib_{x^kx^m}(\tfrac1{2d})$.
Hence $\mathcal G^c$ is covered by the $d$ sets $\Ib_{x^k}(\tfrac1{8d})$
together with the $\binom d2$ sets $\Ib_{u^{km}}(\tfrac1{8d})$.
Lemma~\ref{lem:profile_isotropized} at $\delta=\tfrac1{8d}$, where
$4/\delta^2=2^8d^2$, applies to each of them, and with
$\bar\eta(u^{km})\leq\bar\eta(x^k)+\bar\eta(x^m)$:
\begin{equation}
  \frac{\volb(\mathcal G^c)}{\volb}
  \leq \frac{2^8d^2}{\bar\lambda_2}\Bigl[\sum_k\bar\eta(x^k)
     + \sum_{\{k,m\}}\bigl(\bar\eta(x^k)+\bar\eta(x^m)\bigr)\Bigr].
\end{equation}
Each index $k$ occurs in exactly $d-1$ unordered pairs, so the bracket is
$[1+(d-1)]\sum_k\bar\eta(x^k)=d\sum_k\bar\eta(x^k)$, which gives $2^8d^3$.
\end{proof}

\medskip
\auxlemma{Smallest singular value on $\mathcal G$}\label{lem:sigma_min}
\emph{For $v\in\mathcal G$, $\sigma_{\min}(X_v)\geq1/\sqrt2$; in particular
$X_v$ is nonsingular.}

\begin{proof}
Let $\Gamma:=X_v^\top X_v$, whose entries are
$\Gamma_{km}=\langle x^k_v,x^m_v\rangle$.
From an SVD $X_v=P\Sigma Q^\top$ we get $\Gamma=Q\Sigma^2Q^\top$, so
$\lambda_{\min}(\Gamma)=\sigma_{\min}(X_v)^2$.
On the diagonal, $v\notin\Ib_{x^k}(\tfrac1{8d})$ gives
$\Gamma_{kk}=\norm{x^k_v}^2\geq(1-\tfrac1{8d})^2\geq1-\tfrac1{4d}$; off the
diagonal, $v\notin\Ib_{x^kx^m}(\tfrac1{2d})$ gives
$|\Gamma_{km}|\leq\tfrac1{2d}$.

To bound $\lambda_{\min}(\Gamma)$ from below we use the Gershgorin circle
theorem: every eigenvalue $\lambda$ of $\Gamma$ satisfies
$|\lambda-\Gamma_{kk}|\leq\sum_{m\neq k}|\Gamma_{km}|$ for some $k$.
We recall the short argument: let $\Gamma c=\lambda c$ with $c\neq0$, and
choose $k$ with $|c_k|=\max_m|c_m|>0$.
The $k$-th row gives $(\lambda-\Gamma_{kk})c_k=\sum_{m\neq k}\Gamma_{km}c_m$,
whence:
\begin{equation}
  |\lambda-\Gamma_{kk}|\,|c_k|
  \leq \sum_{m\neq k}|\Gamma_{km}|\,|c_m|
  \leq |c_k|\sum_{m\neq k}|\Gamma_{km}|,
\end{equation}
so $\lambda\geq\Gamma_{kk}-\sum_{m\neq k}|\Gamma_{km}|$.
Each row has off-diagonal sum at most $(d-1)\tfrac1{2d}$, therefore:
\begin{equation}
  \lambda_{\min}(\Gamma) \;\geq\; 1-\frac1{4d}-\frac{d-1}{2d}
  \;=\; \frac12+\frac1{4d} \;\geq\; \frac12 .
\end{equation}
\end{proof}

\medskip
\auxlemma{Stability of the polar step}\label{lem:polar_stability}
\emph{For $v,u\in\mathcal G$ and $g_v:=\mathrm{Pol}(X_v)$:}
\begin{equation}
  \norm{g_v-W_{vu}^\top g_u}^2_{M_{vu},F}
  \;\leq\; 2\kappa_E\,\norm{X_v-W_{vu}^\top X_u}^2_{M_{vu},F} .
\end{equation}

\begin{proof}
First, the polar factor is equivariant: if $X=P\Sigma Q^\top$ is an SVD and
$W\in O(d)$, then $W^\top X=(W^\top P)\Sigma Q^\top$ is again one, so
$\mathrm{Pol}(W^\top X)=W^\top\mathrm{Pol}(X)$ and
$\sigma_{\min}(W^\top X)=\sigma_{\min}(X)$.
Put $A:=X_v$ and $A':=W_{vu}^\top X_u$.
Then $W_{vu}^\top g_u=\mathrm{Pol}(A')$, and both $\sigma_{\min}(A)$ and
$\sigma_{\min}(A')=\sigma_{\min}(X_u)$ are at least $1/\sqrt2$ by
Lemma~\ref{lem:sigma_min}.
The perturbation bound of~\cite{Li_perturbation_bounds}:
\begin{equation}
  \norm{\mathrm{Pol}(A)-\mathrm{Pol}(A')}_F
  \leq \frac{2}{\sigma_{\min}(A)+\sigma_{\min}(A')}\,\norm{A-A'}_F ,
\end{equation}
therefore has prefactor at most $2/\sqrt2=\sqrt2$, and squaring:
\begin{equation}
  \norm{g_v-W_{vu}^\top g_u}^2_F \leq 2\,\norm{X_v-W_{vu}^\top X_u}^2_F .
\end{equation}
Finally, item (iii) of the toolkit applied to each column gives
$w^-_{vu}\norm C^2_F\leq\norm C^2_{M_{vu},F}\leq w^+_{vu}\norm C^2_F$ for
every $C\in\R^{d\times d}$; converting the left-hand side with $w^+_{vu}$
and the right-hand side with $1/w^-_{vu}$ costs
$w^+_{vu}/w^-_{vu}\leq\kappa_E$.
\end{proof}

\medskip
\auxlemma{Main $O(d)$ estimate}\label{lem:main_od}
\emph{Let $g_v:=\mathrm{Pol}(X_v)$, and $g_v:=I_d$ where $X_v$ is singular.
Then:}
\begin{equation}
  \nu(g) \;\leq\; \frac{2^{11}\kappa_Ed^3}{\bar\lambda_2}
     \sum_{k\leq d}\bar\eta(x^k) .
\end{equation}

\begin{proof}
By Lemma~\ref{lem:sigma_min} the fallback $g_v=I_d$ occurs only on
$\mathcal G^c$, which the second case below covers.
We split the ordered pairs in the numerator of $\nu(g)$ accordingly.

\emph{Both endpoints in $\mathcal G$.}
By Lemma~\ref{lem:polar_stability}, extending the sum to all ordered pairs
since every term is nonnegative, and then applying the columnwise
dictionary to the matrix-valued map $X$:
\begin{equation}
  \frac12\sum_{v,u\in\mathcal G}
     \norm{g_v-W_{vu}^\top g_u}^2_{M_{vu},F}
  \leq 2\kappa_E\sum_k(x^k)^\top Bx^k
  = 2\kappa_E\,\volb\sum_k\bar\eta(x^k),
\end{equation}
using $(x^k)^\top Bx^k=\bar\eta(x^k)\,(x^k)^\top\Db x^k=\bar\eta(x^k)\volb$.

\emph{At least one endpoint in $\mathcal G^c$.}
Here we use only $g_v,g_u\in O(d)$: then $\norm{g_v}_F=\sqrt d$ and
$\norm{g_v-W_{vu}^\top g_u}_F^2\leq(\sqrt d+\sqrt d)^2=4d$, so by item
(iii) of the toolkit each summand is at most $4d\,w^+_{vu}$.
The two symmetric families of pairs contribute equally, and
$d^+_v\leq\kappa_E\bar d_v$ by Lemma~\ref{lem:conversion}(i), so with
Lemma~\ref{lem:bad_volume}:
\begin{equation}
  \frac12\sum_{(v,u)\text{ meeting }\mathcal G^c}
     \norm{g_v-W_{vu}^\top g_u}^2_{M_{vu},F}
  \;\leq\; 4d\sum_{v\in\mathcal G^c}d^+_v
  \;\leq\; 4d\,\kappa_E\,\volb(\mathcal G^c)
  \;\leq\; \frac{2^{10}\kappa_Ed^4}{\bar\lambda_2}\,\volb\sum_k\bar\eta(x^k).
\end{equation}

\emph{Conclusion.}
The denominator of $\nu$ is $\vol(G)=\sum_v\tr(D_v)=d\,\volb$, so:
\begin{equation}
  \nu(g) \;\leq\; \Bigl(\frac{2\kappa_E}{d}
     + \frac{2^{10}\kappa_Ed^3}{\bar\lambda_2}\Bigr)\sum_k\bar\eta(x^k).
\end{equation}
Since $\bar\lambda_2\leq2$ by Lemma~\ref{lem:conversion}(iii),
$2^{10}\kappa_Ed^3/\bar\lambda_2\geq2^9\kappa_Ed^3\geq2\kappa_E/d$, so the
first term is dominated by the second and the sum is at most twice the
second.
\end{proof}

\begin{proof}[Proof of Theorem~\ref{thm:od}]
Run Algorithm~\ref{alg:od}.
Since the $z^k$ are orthonormal, the vectors
$x^k=\sqrt{\volb}\,\Db^{-1/2}z^k$ satisfy
$(x^k)^\top\Db x^l=\volb\,(z^k)^\top z^l=\volb\,\delta_{kl}$,
so the family is $\Db$-orthogonal, $\norm{n_{x^k}}^2_{\bar D_0}=\volb$
by~\eqref{eq:barD_exact}, and
$\bar\eta(x^k)=(z^k)^\top\Bb z^k=\lambda_k(\Bb)$.
The output $g$ is one particular group potential, so $\nu_G\leq\nu(g)$, and
Lemma~\ref{lem:main_od} gives:
\begin{equation}
  \nu_G \;\leq\; \frac{2^{11}\kappa_Ed^3}{\bar\lambda_2}
     \sum_{k\leq d}\lambda_k(\Bb) .
\end{equation}
It remains to return to $\Bn$ and $\lambda_2(\Lmn)$.
By Lemma~\ref{lem:conversion}(ii),
$\bar\lambda_2^{-1}\leq\kappa_E\,\lambda_2(\Lmn)^{-1}$, and by
Lemma~\ref{lem:conversion}(iv) and item (i) of the toolkit,
$\lambda_k(\Bb)\leq\kappa_D\lambda_k(\Bn)\leq\kappa_E\lambda_k(\Bn)$.
These two factors, with the one already present, give:
\begin{equation}
  \nu_G \;\leq\; \frac{2^{11}\kappa_E^3d^3}{\lambda_2(\Lmn)}
     \sum_{k\leq d}\lambda_k(\Bn) .
\end{equation}
\end{proof}

\medskip
\noindent\textbf{Path determinacy (Section~\ref{sec:random_walk}).} If the
walker of Definition~\ref{def:random_walk} starts at $(v_0,\hat x_0)$ and
follows $\gamma:v_0\to\cdots\to v_t$, its state at time $t$ is
$(v_t,T_\gamma\hat x_0)$, with $T_\gamma$ the transport map of
Proposition~\ref{prop:holonomy}.

\begin{proof}
Each step applies $W_{v_iv_{i+1}}$ to the current vector, so composing along
$\gamma$ gives exactly $T_\gamma\hat x_0$ by definition of $T_\gamma$.
\end{proof}

\medskip
\noindent\textbf{Conserved observables (Section~\ref{sec:random_walk}).} For
$y\in\R^{nd}$ let $f_y(v,\hat x):=\langle\hat x,y_v\rangle$.
Then $f_y$ is constant along every trajectory, not merely in expectation, if
and only if $y\in\ker(B)$.

\begin{proof}
If $y\in\ker(B)$, Proposition~\ref{prop:kernel} gives $y_u=W_{vu}y_v$ on
every edge, so by path determinacy $\hat x_t$ and $y_{v_t}$ are the images
of $\hat x_0$ and $y_{v_0}$ under the same orthogonal $T_\gamma$, and
orthogonal maps preserve inner products.

Conversely, the increment of $f_y$ along a step $v\to u$ from state $\hat x$
is $\langle W_{vu}\hat x,y_u\rangle-\langle\hat x,y_v\rangle =\langle\hat
x,\,W_{vu}^\top y_u-y_v\rangle$.
By (H2) every edge is traversed with positive probability from every state,
so if $f_y$ is conserved this vanishes for every edge and every $\hat
x\in\mathbb S^{d-1}$, giving $y_u=W_{vu}y_v$ on every edge, that is
$y\in\ker(B)$ by Proposition~\ref{prop:kernel}.
\end{proof}

\medskip
\noindent\textbf{Theorem~\ref{thm:unique_cont} (Unique continuation).} Let
$M$ be connected and $c\in\ker\partial^*\cap\ker B_1^\F$.
If $c_e=0$ for every edge $e$ incident to some fixed vertex, then $c=0$.

\begin{proof}
Let $D(e)$ be the graph distance from $e$ to the fixed vertex's incident
edges, well defined and finite since $M$ is connected.
We show $c_e=0$ for every $e$ by induction on $D(e)$; the base case $D(e)=0$
is the hypothesis.

Assume $c_e=0$ for every $e$ with $D(e)\leq k$, and let $e$ have $D(e)=k+1$,
with neighbour $e_1$ at distance $k$.
If $e,e_1$ are parallel, Proposition~\ref{prop:kernel} gives $c_e=0$
directly.
Otherwise they share a face $f$; let $v_1\trianglelefteq
e,e_1\trianglelefteq f$ and let $e_2$ be the other edge of $f$ lower
adjacent to $e_1$, at the vertex $v_2\trianglelefteq e_1,e_2$ opposite $v_1$
(this edge exists since $M$ is regular).
It suffices to show $c_{e_2}=0$, since Lemma~\ref{lem:unique_cont_local}
then gives $c_e=0$.

If $k=0$, the fixed vertex is $v_1$ or $v_2$: the first would give $D(e)=0$,
a contradiction, so it is $v_2$, giving $D(e_2)=0$ directly.
If $k>0$, take $e_3$ with $D(e_3)=k-1$ sharing a vertex with $e_1$; if that
vertex were $v_1$ then $D(e)=k$, a contradiction, so it is $v_2$, giving
$D(e_2)=k$ and $c_{e_2}=0$ by the inductive hypothesis.
\end{proof}

\medskip
\noindent\textbf{Theorem~\ref{thm:bochner_manifold} (Bochner theorem for
combinatorial manifolds).} Let $M$ be a compact connected combinatorial
$n$-manifold, $n\leq3$, with $\mathrm{Ric}_\F\succeq0$.
Then $\dim\mathcal H^1(M;\F)\leq nd$.

\begin{proof}
In each case we exhibit $n$ edges at a vertex $v$ whose vanishing forces
every edge incident to $v$ to vanish; Theorem~\ref{thm:unique_cont} then
makes the evaluation map $\ker(H_1^\F)\to\F(e_1)\oplus\cdots\oplus\F(e_n)$
injective, giving $\dim\mathcal H^1(M;\F)=\dim\ker(H_1^\F)\leq nd$.

\emph{$n\leq1$.} For $n=0$, $M$ is a single vertex and the claim is trivial.
For $n=1$, $M$ is a circle: there are no $2$-cells, every off-diagonal block
of $B_1^\F$ is a single invertible restriction map, and
Proposition~\ref{prop:kernel} makes $\phi_e:\ker(H_1^\F)\to\F(e)$ injective
for any single edge $e$, so $n=1$ edge suffices.

\emph{$n=2$.} Let $e_1,e_2$ be transverse (sharing both a vertex $v$ and a
face $f$).
Order the edges at $v$ as $e_1,e_2,\dots,e_k$.
If $k\geq4$, every $2$-manifold forces each of $e_3,\dots,e_k$ to be
parallel to $e_1$ or $e_2$, so Proposition~\ref{prop:kernel} kills each of
them directly once $c_{e_1}=c_{e_2}=0$, and unique continuation gives $c=0$.
Otherwise $e_1,e_2,e_3$ are the only edges at $v$, pairwise sharing a
triangle; for $x_v\in\F(v)$ extended by $0$, harmonicity gives
$0=\langle\partial c,x_v\rangle=\pm\,c_{e_3}^\top\F_{v\trianglelefteq
e_3}x_v$ for every $x_v$, so $c_{e_3}=0$, and unique continuation concludes.

\emph{$n=3$.} We follow Forman's proof~\cite[Theorem 4.4]{Forman2003Bochner}.
As for $n\leq2$, only two of its steps involve the coefficients of $c$, and
both hold for sheaves with the same conclusion.
The first is \emph{parallel propagation}: if $e,\tilde e$ are parallel,
$c_e=0$ forces $c_{\tilde e}=0$ (Proposition~\ref{prop:kernel}).
The second is the \emph{vertex relation}: if $c$ vanishes on every edge at
$v$ except $e^*$, harmonicity gives $0=\langle\partial
c,x_v\rangle=\pm\,c_{e^*}^\top\F_{v\trianglelefteq e^*}x_v$ for every
$x_v\in\F(v)$; a bilinear form vanishing for every $x_v$ forces
$\F_{v\trianglelefteq e^*}^\top c_{e^*}=0$, hence $c_{e^*}=0$ by
invertibility, where Forman's scalar version reads $\pm\omega_vc_{e^*}=0$.
Everything else depends only on the geometry of $M$, not on the sheaf: the
choice of three edges at $v$ whose vanishing forces, through these two
steps, every edge at $v$ to vanish, after which
Theorem~\ref{thm:unique_cont} gives $c=0$.
We recall this choice briefly.
The edges at $v$ are the vertices of the link of $v$, a cell structure on
$\mathbb S^2$ in which two of them are adjacent exactly when transverse.
If $v$ has four edges, any three will do, and the vertex relation kills the
fourth.
Otherwise let $K$ be a maximal complete subgraph of the link; it has at
most four vertices, since $K_5$ does not embed in $\mathbb S^2$.
If $K$ has at most three vertices, take the corresponding edges: by
maximality, every other edge at $v$ is parallel to one of them.
If $K$ has four, with edges $e_1,\dots,e_4$, the link vertex of a fifth edge
$e_5$ lies in one of the four triangles cut out by $K$, say the one opposite
$e_4$, so $e_5$ is parallel to $e_4$.
Take $e_2,e_3,e_4$: propagation kills $e_5$, and then every further edge,
which cannot be transverse to both $e_4$ and $e_5$; the vertex relation
kills $e_1$.
\end{proof}

\medskip
\noindent\textbf{Theorem~\ref{thm:local_hom_dim}.} Let $M$ be a compact
connected combinatorial manifold with $\mathrm{Ric}_\F\succeq0$.
Then $\dim\mathcal H^1(M;\F)\leq d\cdot D(M)$.

\begin{proof}
Fix $v$; we show $\dim\ker(H_1^\F)\leq d\cdot D(v)$.
If every equivalence class (Definition~\ref{def:local_hom_dim}) has at least
two elements, $D(v)=N(v)$: choose one representative edge per class, and
note every other edge at $v$ shares its class with some representative, so
Proposition~\ref{prop:kernel} propagates a vanishing representative to the
whole class; once every representative vanishes, so does every edge at $v$,
and Theorem~\ref{thm:unique_cont} gives $c=0$.

If instead some class $\{e^*\}$ is a singleton, $D(v)=N(v)-1$: choose
representatives for the remaining classes, and as above every edge at $v$
other than $e^*$ vanishes once its representative does.
Harmonicity then gives $0=\langle\partial
c,x_v\rangle=c_{e^*}^\top\F_{v\trianglelefteq e^*}x_v$ for every $x_v$, so
$c_{e^*}=0$ as well, and unique continuation concludes.
\end{proof}

\section{Further Properties of Bochner Matrices}
\label{sec:further}

This appendix collects results on Bochner matrices that the main text does
not use: a comparison with a scalar Laplacian, the interplay of balance and
antibalance, and two further properties of the random walk.

\subsection{Trace comparison}

Following Hansen, we extend to Bochner matrices the comparison he develops
between the block spectrum and a scalar one for matrix-weighted
Laplacians~\cite{expansion_matrix_weighted_graphs}.

\begin{definition}[Trace shadow]
\label{def:trace_shadow}
Set $\hat w_{vu} := \tr(B_{vu})$ and define the trace shadow $\hat
L\in\R^{n\times n}$ by $\hat L_{vu} := -\hat w_{vu}$ ($v\neq u$), $\hat
L_{vv} := \tr(D_v)$.
\end{definition}

\begin{theorem}[Trace comparison]
\label{thm:trace_comparison}
$\hat L$ is symmetric and positive semidefinite.
For every $m$:
\begin{equation}
  \sum_{i\leq md}\lambda_i(B) \;\leq\; \sum_{j\leq m}\lambda_j(\hat L),
  \qquad
  \sum_{i>nd-md}\lambda_i(B) \;\geq\; \sum_{j>n-m}\lambda_j(\hat L),
\end{equation}
with equality at $m=n$.
\end{theorem}
\begin{proof}
\emph{Symmetry and the split-vector identity.} Symmetry is (H1), since $\hat
w_{uv}=\tr(B_{uv})=\tr(B_{vu}^\top)=\hat w_{vu}$.
For $\xi\in\R^n$ and the standard basis $e_1,\dots,e_d$, define $\xi\otimes
e_\alpha$ by $(\xi\otimes e_\alpha)_v:=\xi_ve_\alpha$.
By Lemma~\ref{lem:energy_id} and $\sum_\alpha(Ae_\alpha)^\top
M(Ae_\alpha)=\tr(A^\top MA)$:
\begin{equation}
  \sum_{\alpha=1}^d(\xi\otimes e_\alpha)^\top B(\xi\otimes e_\alpha)
  = \tfrac12\sum_{v\neq u}\tr\bigl[(\xi_vI-\xi_uW_{vu}^\top)^\top M_{vu}(\xi_vI-\xi_uW_{vu}^\top)\bigr],
\end{equation}
and expanding, using $\tr(WMW^\top)=\tr(M)$, this equals
$\sum_v\xi_v^2\tr(D_v)-\sum_{v\neq u}\xi_v\xi_u\hat w_{vu}=\xi^\top\hat
L\xi$.
Taking $\xi$ a unit vector shows $\hat L\succeq0$, since the left side is a
sum of nonnegative terms.

\emph{Conclusion.} Let $\xi^1,\dots,\xi^m$ be orthonormal eigenvectors of
$\hat L$ for its $m$ smallest eigenvalues and let $Q$ have the $md$ columns
$\xi^j\otimes e_\alpha$, orthonormal since $(\xi^j\otimes
e_\alpha)^\top(\xi^l\otimes
e_\beta)=((\xi^j)^\top\xi^l)\delta_{\alpha\beta}$.
The compression bound of Appendix~\ref{sec:proofs} applied to $A=B$ gives:
\begin{equation}
  \sum_{i\leq md}\lambda_i(B) \leq \tr(Q^\top BQ)
  = \sum_{j\leq m}\sum_\alpha(\xi^j\otimes e_\alpha)^\top B(\xi^j\otimes e_\alpha)
  = \sum_{j\leq m}\lambda_j(\hat L).
\end{equation}
The second inequality follows by the same argument applied to $-B$ and
$-\hat L$, whose eigenvalues are those of $B$ and $\hat L$ with sign and
order reversed.
At $m=n$, $Q$ is square orthogonal and $\tr(B)=\sum_v\tr(D_v)=\tr(\hat L)$,
giving equality.
\end{proof}

\subsection{Balance and antibalance}

\begin{proposition}[The two extremes]
\label{prop:extremes}
Let $x\in\ker(B)$, $z\in\ker(N)$, and $G$ connected.
If $G$ is \emph{not} bipartite, then $\langle x_v,z_v\rangle=0$ at every
vertex, hence $\dim\ker(B)+\dim E_2(\Bn)\leq d$.
If $G$ \emph{is} bipartite, then $\dim\ker(B)=\dim E_2(\Bn)$.
\end{proposition}
\begin{proof}
Put $f(v):=\langle x_v,z_v\rangle$.
Using $x_u=W_{vu}x_v$, $z_u=-W_{vu}z_v$ and orthogonality of $W_{vu}$,
$f(u)=-f(v)$ across every edge, so $f$ alternates sign along every walk.

If $G$ is not bipartite it has an odd cycle, and alternating sign around it
forces $f=0$ there, hence $f\equiv 0$ by connectedness: the evaluation
images $K=\{x_{v_0}\}$, $K'=\{z_{v_0}\}$ at a fixed vertex $v_0$ are then
orthogonal subspaces of $\R^d$, of dimension $\dim\ker(B)$ and $\dim
E_2(\Bn)$ by Proposition~\ref{prop:holonomy} and
Theorem~\ref{thm:antibalance}, so their dimensions sum to at most $d$.

If $G$ is bipartite with parts $V_+\sqcup V_-$, then $z_v:=x_v$ on $V_+$ and
$z_v:=-x_v$ on $V_-$ satisfies $z_u=-W_{vu}z_v$ on every edge (exactly one
endpoint changes sign), so $x\mapsto z$ is a linear isomorphism
$\ker(B)\to\ker(N)$, whence the two dimensions agree.
\end{proof}

Similarly to our treatment, balancedness, bipartiteness and their joint
effect on the spectrum of a normalized Laplacian have been investigated on
gain graphs by Kannan et
al.~\cite{kannan2022normalizedlaplaciansgaingraphs}.
A Bochner matrix is not a gain graph: its edge label
$B_{vu}=M_{vu}W_{vu}^\top$ carries the group element $W_{vu}\in O(d)$, but
the radial factor $M_{vu}$ obstructs the group structure.

The non-bipartite case has no scalar analogue: the two kernels are
orthogonal \emph{pointwise at every vertex}, yet can both be nonzero,
whereas for $d=1$ a connected non-bipartite graph cannot be simultaneously
balanced and antibalanced.
The triangle with $W_{12}=W_{23}=I$, $W_{31}=\operatorname{diag}(1,-1)$
attains this: $\dim\ker(B)=\dim E_2(\Bn)=1$, along the orthogonal directions
$e_1,e_2$.

\subsection{Two further properties of the random walk}

\noindent\textbf{Confinement.}
Since a nonzero kernel vector never vanishes, choosing
$x^1,\dots,x^k\in\ker(B)$, $k:=\dim\ker(B)$, with orthonormal values at a
fixed $v_0$ gives $k$ conserved directions at every node, and this confines
the walker.

\begin{corollary}[Confinement]
\label{cor:confinement}
With $x^1,\dots,x^k$ as above, write $c_i:=\langle\hat x_0,x^i_{v_0}\rangle$
and $r:=(1-\sum_ic_i^2)^{1/2}$.
Then at every time $t$, $\langle\hat x_t,x^i_{v_t}\rangle=c_i$ for every
$i$, and $\|\hat x_t-\sum_ic_ix^i_{v_t}\|=r$: the walker is confined to a
set of dimension $d-k-1$ inside the fibre.
In particular the fibre carries no freedom at all, $\hat x_t$ determined by
$v_t$ alone, when $k=d$ or when $\hat x_0=x_{v_0}$ for some unit
$x\in\ker(B)$.
\end{corollary}
\begin{proof}
By Proposition~\ref{prop:kernel} each $x^i$ transports as
$x^i_{v_t}=T_\gamma x^i_{v_0}$, and $T_\gamma$ orthogonal preserves inner
products, so the $x^i_v$ remain orthonormal at every node.
Conservation applied to each $x^i$ gives $\langle\hat
x_t,x^i_{v_t}\rangle=c_i$; since $\norm{\hat x_t}=1$, the component of $\hat
x_t$ orthogonal to $\operatorname{span}\{x^i_{v_t}\}$ has norm
$(1-\sum_ic_i^2)^{1/2}=r$.
If $k=d$ the $x^i_{v_t}$ form an orthonormal basis, so the $c_i$ determine
$\hat x_t$ and $r=0$; if $\hat x_0=x_{v_0}$ for a unit $x\in\ker(B)$, taking
$x^1:=x$ gives $c_1=1$ and again $r=0$.
\end{proof}
The bound can be strict: for $d=1$ on an unbalanced signed graph the
reachable set collapses to $\{-1,+1\}$, the ordinary Gremban walk; for $d=2$
on a cycle carrying an irrational rotation it is infinite, and dense in the circle.

\noindent\textbf{Reversibility.}
The walk of Definition~\ref{def:random_walk} is not reversible in general:
detailed balance would force $M_{vu}$, conjugated once by $D_v$ and once by
$D_u$, to give proportional quadratic forms, which fails unless every $D_v$
is a scalar multiple of the identity.

\begin{proposition}[Failure of reversibility]
\label{prop:not_reversible}
There exist Bochner matrices satisfying (H1)--(H3) for which the walk of
Definition~\ref{def:random_walk} admits no reversible measure.
\end{proposition}
\begin{proof}
Take $d=2$, $G$ the triangle on $\{1,2,3\}$, and every angular factor equal
to the identity, so the fibre never moves and $\{1,2,3\}\times\{e_1\}$ is a
closed class on which the walk is the three-state chain with
$q_{vu}=\omega_v^\top M_{vu}\omega_v$, $\omega_v$ the first column of
$D_v^{-1/2}$.
Take the edge weights:
\begin{equation}
  M_{12}=\begin{pmatrix}2&1\\1&2\end{pmatrix},\quad
  M_{13}=\begin{pmatrix}3&1\\1&3\end{pmatrix},\quad
  M_{23}=\begin{pmatrix}5/2&-1\\-1&5/2\end{pmatrix},
\end{equation}
all positive definite, so (H1)--(H3) hold.
Then $D_1=\begin{psmallmatrix}5&2\\2&5\end{psmallmatrix}$, $D_2=\tfrac92I$,
$D_3=\tfrac{11}2I$. $D_1$ has eigenvalues $7,3$ on
$f_\pm=\tfrac1{\sqrt2}(1,\pm1)$, so the entries of $D_1^{-1/2}$ are
$a=\tfrac{\sqrt7}{14}+\tfrac{\sqrt3}6$,
$b=\tfrac{\sqrt7}{14}-\tfrac{\sqrt3}6$, with $a^2+b^2=\tfrac5{21}$ and
$ab=-\tfrac1{21}$.
Since the quadratic form of $\begin{psmallmatrix}m&1\\1&m\end{psmallmatrix}$
at $(a,b)$ is $m(a^2+b^2)+2ab$, we get $q_{12}=\tfrac8{21}$,
$q_{13}=\tfrac{13}{21}$; at nodes $2,3$ the degree is scalar, so
$q_{21}=\tfrac49$, $q_{23}=\tfrac59$, $q_{31}=\tfrac6{11}$,
$q_{32}=\tfrac5{11}$.
Then $q_{12}q_{23}q_{31}=\tfrac{240}{2079}$ while
$q_{13}q_{32}q_{21}=\tfrac{260}{2079}$, so Kolmogorov's criterion fails and
the chain is not reversible; a reversible measure for the whole walk would
restrict to one on this closed class, so none exists.
\end{proof}

The fix is to normalize outside the quadratic form, dividing by
$\hat x^\top D_v\hat x$.
Since $M_{uv}=W_{vu}M_{vu}W_{vu}^\top$ by (H1), the scalar $\hat x^\top
M_{vu}\hat x$ is read identically from either end of the edge, which is what
reversibility needs.

\begin{proposition}[A reversible variant]
\label{prop:gauged_reversible}
Under (H1)--(H3), the walk with the same fibre update but jump law
\begin{equation}
  p^{\mathrm R}_{vu}(\hat x) := \frac{\hat x^\top M_{vu}\hat x}{\hat x^\top D_v\hat x}
\end{equation}
is reversible with respect to the measure $\mu$ assigning $\{v\}\times S$ the mass
$\frac{d}{\vol(G)}\int_S\hat x^\top D_v\hat x\,\sigma(d\hat x)$, for
$\sigma$ the uniform measure on $\mathbb{S}^{d-1}$, with marginal
$\pi_v=\tr(D_v)/\vol(G)$ on $V$.
\end{proposition}
\begin{proof}
That $p^{\mathrm R}_{vu}$ sums to $1$ over the neighbours of $v$ follows
from $\sum_{u\sim v}M_{vu}=D_v$, and $\hat x^\top D_v\hat x>0$ by (H2).
Since $\int_{\mathbb S^{d-1}}\hat x^\top D_v\hat x\,\sigma(d\hat
x)=\tr(D_v)/d$, $\mu$ has total mass $1$ and the stated marginal.
The transition from $(v,\hat x)$ along $(v,u)$ leads to $(u,W_{vu}\hat x)$,
and back along the same edge since $W_{uv}=W_{vu}^\top$.
The forward flow is:
\begin{equation}
  \frac{d}{\vol(G)}(\hat x^\top D_v\hat x)\frac{\hat x^\top M_{vu}\hat x}{\hat x^\top D_v\hat x}
  = \frac{d}{\vol(G)}\hat x^\top M_{vu}\hat x,
\end{equation}
and the backward flow, with $\hat y=W_{vu}\hat x$, is $\frac{d}{\vol(G)}\hat
y^\top M_{uv}\hat y$, equal to it since $W_{vu}M_{vu}W_{vu}^\top=M_{uv}$ by
(H1).
As $\hat x\mapsto W_{vu}\hat x$ preserves $\sigma$, this integrates to
detailed balance between measurable sets, and $\mu$ is stationary.
\end{proof}
The two jump laws coincide whenever $D_v$ is a scalar multiple of the
identity, in particular for connection Laplacians, and since both variants
share the fibre update, path determinacy, the conserved observables and
Corollary~\ref{cor:confinement} hold for either.

\section{Conviction: Protocol and Additional Material}
\label{sec:supp_conv}

\subsection{Protocol}
\label{sec:supp_impl}
\label{sec:supp_agg}

All interventions run on a model whose weights are loaded from the
best-validation-epoch checkpoint and never updated.
The candidate set is always a subset of the training nodes and the
evaluation set always the full test set, so the quantity compared is
identical across strategies, budgets and modes.

\noindent\textbf{Interventions.}
\emph{Masking} replaces the features of the selected nodes, by zero or by
the dataset feature mean, and leaves the graph untouched; the mean is an
unsupervised statistic computed over all nodes.
\emph{Gating} scales every cross-node entry of the sparse diffusion
Laplacian by the gate of the \emph{source} node, so a gated node keeps its
own state and only what it transmits is suppressed.
In the hard variant used in the body the gate is $0$ on the selected nodes
and $1$ elsewhere; in the soft variant of Appendix~\ref{sec:supp_forward}
it varies continuously:
\begin{equation}
    g(u) \;=\; g_{\min} + (1 - g_{\min})\,
    \frac{s(u) - \min_{v} s(v)}{\max_{v} s(v) - \min_{v} s(v)},
    \qquad g_{\min} = 0.1,
\end{equation}
with $s$ the ranking score.
\emph{Pruning} deletes the selected nodes together with their incident
edges and relabels the remaining ones; since the graph size is fixed at
model construction, a new model is built for the reduced graph and the same
frozen parameters are loaded into it, so no parameter is re-estimated.

\noindent\textbf{Selection strategies.}
At each budget $k\in\{1\%,5\%,10\%,20\%,30\%\}$ of the training pool we
compare nine strategies: the top and the bottom $k\%$ by conviction, by
degree and by prediction confidence; $k$ uniformly random nodes; and two
degree-matched controls, random sets reproducing the degree distribution of
the high- and of the low-conviction group, built by greedy nearest-degree
matching without replacement.
Several nodes share the same degree or the same confidence, so the ranking
by these criteria contains ties; they are broken at random with a seeded
key, which keeps the selection reproducible.

\noindent\textbf{Decile experiment.}
It sorts the training pool by one criterion, splits it into ten disjoint
bins of equal size, and perturbs exactly one bin at a time, each against the
same unperturbed baseline.
Instead, the budget sweep itself is cumulative: its selections are nested,
the $30\%$ one containing the $10\%$ one, since each budget takes the top
(or the bottom) $k\%$ of the pool.
The sweep therefore measures what a top-$k$ selection costs, and the deciles
how the damage is distributed along the score range.
The decile experiment is run for conviction, confidence and degree, so the
three orderings can be compared bin by bin.

\noindent\textbf{Aggregation.}
The relative drop is formed fold by fold, each fold against its \emph{own}
unperturbed baseline, and normalized by the headroom above chance:
\begin{equation}
\label{eq:drop}
    \mathrm{drop}\%(s) = 100 \cdot
    \frac{m_{\text{baseline}} - m_{s}}{m_{\text{baseline}} - m_{\text{chance}}},
    \qquad
    m_{\text{chance}} = \begin{cases} 0 & \text{accuracy},\\ 0.5 & \text{ROC-AUC},\end{cases}
\end{equation}
so a cell is the percentage of the learnt signal the intervention destroys,
and means the same on an accuracy dataset and on an AUC one.
Folds are then averaged within a dataset and the datasets averaged
unweighted, so every benchmark counts the same.
The quoted spread is the standard error across the per-dataset drops, and
shows the variation across datasets.

Cornell and Wisconsin are excluded from the aggregation, since their
training pools are too small for a percentage budget to be meaningful;
every number reported is therefore over the remaining eight datasets.

\definecolor{convworst}{RGB}{213,94,0}
\definecolor{convbest}{RGB}{0,114,178}
The tables of this appendix follow the conventions of
Table~\ref{tab:conv_fwd_masking-zero_last}.
Unless the caption says otherwise, a cell is the relative drop of
Eq.~\eqref{eq:drop}, as mean$_{\pm\text{s.e.}}$ across the eight datasets,
each averaged over its folds.
The bold label on top names the perturbation mode, and in the budget sweeps
the \textcolor{convworst}{\textbf{most damaging}} and the
\textcolor{convbest}{\textbf{least damaging}} selection of each column are
highlighted.

\subsection{What the conviction measures}
\label{sec:supp_conv_observational}

Section~\ref{sec:conviction} settles by intervention whether the conviction
identifies the nodes a trained network relies on.
Before intervening, we check what the conviction is by correlating it with node properties that are already interpretable: the
degree and the betweenness centrality, both topological; the model's softmax
confidence; and the input-gradient saliency, the norm of the gradient of the
loss with respect to the raw node features.
We also measure the Cohen's $d$ between the conviction of correctly and of
incorrectly classified nodes, which shows how far the conviction by itself
tells the two groups apart.
All six quantities are computed at the last diffusion layer, per fold,
averaged over the folds of a dataset and then over the eight datasets, at
the untrained and at the best-validation checkpoint.

The results are collected in Table~\ref{tab:conv_observational}.
The conviction is moderately correlated with degree, $0.33\pm0.07$, and with
betweenness, $0.28\pm0.06$, and both correlations are near zero at
initialization, so the ordering they express is produced by training.
The correlation with the model's confidence is higher, $0.66\pm0.05$, but
still moderate.
Its correlation with saliency is weak and negative, $-0.18\pm0.09$, the only
negative entry of the six.
The Cohen's $d$ is $0.54\pm0.09$ on the training split and $0.37\pm0.08$ on
the test split, a small to moderate effect: the conviction does not by
itself separate correct from incorrect nodes sharply.

The conviction is therefore not reducible to the topology or to the
confidence: the topological
correlations are moderate and appear only with training, and the agreement
with the model's confidence is partial.

\begin{table}[htbp]
\caption{Spearman correlation between the conviction and four node properties, and absolute Cohen's $d$ between the conviction of correctly and incorrectly classified nodes, at the last diffusion layer. \emph{Untrained} is the epoch-$0$ checkpoint, a control.}
\label{tab:conv_observational}
\centering
\small
\begin{tabular*}{\textwidth}{@{\extracolsep{\fill}}lcc}
\toprule
 & \multicolumn{2}{c}{Checkpoint} \\
\cmidrule(l){2-3}
Observable & Untrained & Best epoch \\
\midrule
Spearman $\rho$ with degree & $0.07_{\pm0.10}$ & $0.33_{\pm0.07}$ \\
Spearman $\rho$ with betweenness & $0.05_{\pm0.08}$ & $0.28_{\pm0.06}$ \\
Spearman $\rho$ with confidence & $0.17_{\pm0.07}$ & $0.66_{\pm0.05}$ \\
Spearman $\rho$ with saliency & $0.07_{\pm0.10}$ & $-0.18_{\pm0.09}$ \\
$|$Cohen's $d|$, train split & $0.21_{\pm0.05}$ & $0.54_{\pm0.09}$ \\
$|$Cohen's $d|$, test split & $0.24_{\pm0.05}$ & $0.37_{\pm0.08}$ \\
\bottomrule
\end{tabular*}
\end{table}

\FloatBarrier

\subsection{Robustness of the intervention}
\label{sec:supp_forward}
\label{sec:supp_robust}
\label{sec:supp_soft}

\noindent\textbf{The other perturbation modes.}
Tables~\ref{tab:conv_fwd_masking-mean_last},
\ref{tab:conv_fwd_gating-hard_last} and~\ref{tab:conv_fwd_pruning_last}
give the budget sweep for the three modes not shown in the body; zero
masking appears as Table~\ref{tab:conv_fwd_masking-zero_last}, and in all
four the colored entries mark the most and the least damaging selection at
each budget.
High conviction is the most damaging selection in all four tables and at all
five budgets, and low conviction the least; the only exception is the
$1\%$ gating budget, where low degree comes out marginally below it.

\begin{table}[htbp]
\definecolor{convworst}{RGB}{213,94,0}
\definecolor{convbest}{RGB}{0,114,178}
\caption{Budget sweep under mean masking, laid out as Table~\ref{tab:conv_fwd_masking-zero_last}. High conviction is the most damaging selection at every budget, low conviction the least.}
\label{tab:conv_fwd_masking-mean_last}
\centering
\small
\begin{tabular*}{\textwidth}{@{\extracolsep{\fill}}lccccc}
\toprule
\textbf{Mean masking} & \multicolumn{5}{c}{Budget (fraction of the training pool)} \\
\cmidrule(l){2-6}
Selection & 1\% & 5\% & 10\% & 20\% & 30\% \\
\midrule
High conviction & \textcolor{convworst}{$\mathbf{0.50}_{\pm0.16}$} & \textcolor{convworst}{$\mathbf{1.34}_{\pm0.39}$} & \textcolor{convworst}{$\mathbf{2.19}_{\pm0.72}$} & \textcolor{convworst}{$\mathbf{3.61}_{\pm1.30}$} & \textcolor{convworst}{$\mathbf{4.66}_{\pm1.85}$} \\
Low conviction & \textcolor{convbest}{$\mathbf{0.01}_{\pm0.02}$} & \textcolor{convbest}{$\mathbf{0.08}_{\pm0.05}$} & \textcolor{convbest}{$\mathbf{0.14}_{\pm0.10}$} & \textcolor{convbest}{$\mathbf{0.36}_{\pm0.17}$} & \textcolor{convbest}{$\mathbf{0.57}_{\pm0.38}$} \\
High degree & $0.38_{\pm0.16}$ & $0.94_{\pm0.38}$ & $1.35_{\pm0.59}$ & $2.13_{\pm0.94}$ & $2.91_{\pm1.17}$ \\
Low degree & $0.03_{\pm0.02}$ & $0.18_{\pm0.11}$ & $0.41_{\pm0.22}$ & $0.88_{\pm0.48}$ & $1.43_{\pm0.80}$ \\
Degree-matched (high) & $0.09_{\pm0.07}$ & $0.30_{\pm0.21}$ & $0.52_{\pm0.37}$ & $0.77_{\pm0.63}$ & $1.21_{\pm0.71}$ \\
Degree-matched (low) & $0.06_{\pm0.04}$ & $0.26_{\pm0.15}$ & $0.50_{\pm0.30}$ & $1.28_{\pm0.69}$ & $2.23_{\pm1.12}$ \\
High confidence & $0.15_{\pm0.07}$ & $0.51_{\pm0.16}$ & $0.86_{\pm0.24}$ & $1.57_{\pm0.48}$ & $2.25_{\pm0.78}$ \\
Low confidence & $0.06_{\pm0.04}$ & $0.30_{\pm0.15}$ & $0.61_{\pm0.28}$ & $1.29_{\pm0.58}$ & $1.99_{\pm0.94}$ \\
Random & $0.08_{\pm0.04}$ & $0.25_{\pm0.16}$ & $0.46_{\pm0.32}$ & $1.18_{\pm0.60}$ & $1.89_{\pm0.94}$ \\
\bottomrule
\end{tabular*}
\end{table}

\begin{table}[htbp]
\definecolor{convworst}{RGB}{213,94,0}
\definecolor{convbest}{RGB}{0,114,178}
\caption{As Table~\ref{tab:conv_fwd_masking-mean_last}, under hard gating. At the $1\%$ budget low degree is marginally the least damaging.}
\label{tab:conv_fwd_gating-hard_last}
\centering
\small
\begin{tabular*}{\textwidth}{@{\extracolsep{\fill}}lccccc}
\toprule
\textbf{Hard gating} & \multicolumn{5}{c}{Budget (fraction of the training pool)} \\
\cmidrule(l){2-6}
Selection & 1\% & 5\% & 10\% & 20\% & 30\% \\
\midrule
High conviction & \textcolor{convworst}{$\mathbf{0.81}_{\pm0.29}$} & \textcolor{convworst}{$\mathbf{1.80}_{\pm0.41}$} & \textcolor{convworst}{$\mathbf{2.77}_{\pm0.66}$} & \textcolor{convworst}{$\mathbf{4.32}_{\pm1.27}$} & \textcolor{convworst}{$\mathbf{5.67}_{\pm1.89}$} \\
Low conviction & $0.05_{\pm0.03}$ & \textcolor{convbest}{$\mathbf{0.11}_{\pm0.04}$} & \textcolor{convbest}{$\mathbf{0.18}_{\pm0.09}$} & \textcolor{convbest}{$\mathbf{0.43}_{\pm0.22}$} & \textcolor{convbest}{$\mathbf{0.71}_{\pm0.45}$} \\
High degree & $0.65_{\pm0.20}$ & $1.37_{\pm0.37}$ & $1.91_{\pm0.45}$ & $2.83_{\pm0.73}$ & $3.67_{\pm1.04}$ \\
Low degree & \textcolor{convbest}{$\mathbf{0.05}_{\pm0.02}$} & $0.26_{\pm0.12}$ & $0.58_{\pm0.27}$ & $1.20_{\pm0.59}$ & $1.91_{\pm0.98}$ \\
Degree-matched (high) & $0.12_{\pm0.05}$ & $0.37_{\pm0.19}$ & $0.68_{\pm0.35}$ & $1.02_{\pm0.60}$ & $1.44_{\pm0.86}$ \\
Degree-matched (low) & $0.06_{\pm0.03}$ & $0.31_{\pm0.18}$ & $0.62_{\pm0.38}$ & $1.53_{\pm0.84}$ & $2.71_{\pm1.36}$ \\
High confidence & $0.24_{\pm0.08}$ & $0.75_{\pm0.20}$ & $1.22_{\pm0.29}$ & $2.28_{\pm0.59}$ & $3.27_{\pm0.98}$ \\
Low confidence & $0.09_{\pm0.05}$ & $0.35_{\pm0.17}$ & $0.63_{\pm0.29}$ & $1.29_{\pm0.60}$ & $2.03_{\pm0.95}$ \\
Random & $0.09_{\pm0.04}$ & $0.34_{\pm0.17}$ & $0.63_{\pm0.34}$ & $1.46_{\pm0.69}$ & $2.21_{\pm1.03}$ \\
\bottomrule
\end{tabular*}
\end{table}

\begin{table}[htbp]
\definecolor{convworst}{RGB}{213,94,0}
\definecolor{convbest}{RGB}{0,114,178}
\caption{As Table~\ref{tab:conv_fwd_masking-mean_last}, under pruning.}
\label{tab:conv_fwd_pruning_last}
\centering
\small
\begin{tabular*}{\textwidth}{@{\extracolsep{\fill}}lccccc}
\toprule
\textbf{Pruning} & \multicolumn{5}{c}{Budget (fraction of the training pool)} \\
\cmidrule(l){2-6}
Selection & 1\% & 5\% & 10\% & 20\% & 30\% \\
\midrule
High conviction & \textcolor{convworst}{$\mathbf{1.65}_{\pm0.96}$} & \textcolor{convworst}{$\mathbf{2.95}_{\pm1.10}$} & \textcolor{convworst}{$\mathbf{4.02}_{\pm1.16}$} & \textcolor{convworst}{$\mathbf{5.64}_{\pm1.49}$} & \textcolor{convworst}{$\mathbf{7.08}_{\pm2.04}$} \\
Low conviction & \textcolor{convbest}{$\mathbf{0.01}_{\pm0.01}$} & \textcolor{convbest}{$\mathbf{0.07}_{\pm0.05}$} & \textcolor{convbest}{$\mathbf{0.16}_{\pm0.11}$} & \textcolor{convbest}{$\mathbf{0.42}_{\pm0.28}$} & \textcolor{convbest}{$\mathbf{0.81}_{\pm0.57}$} \\
High degree & $1.36_{\pm0.77}$ & $2.68_{\pm1.16}$ & $3.48_{\pm1.24}$ & $4.82_{\pm1.38}$ & $5.77_{\pm1.57}$ \\
Low degree & $0.04_{\pm0.03}$ & $0.24_{\pm0.14}$ & $0.54_{\pm0.30}$ & $1.13_{\pm0.67}$ & $1.88_{\pm1.10}$ \\
Degree-matched (high) & $0.20_{\pm0.05}$ & $0.52_{\pm0.21}$ & $0.89_{\pm0.41}$ & $1.29_{\pm0.76}$ & $1.73_{\pm1.09}$ \\
Degree-matched (low) & $0.07_{\pm0.04}$ & $0.37_{\pm0.20}$ & $0.67_{\pm0.41}$ & $1.72_{\pm0.96}$ & $2.92_{\pm1.52}$ \\
High confidence & $0.33_{\pm0.13}$ & $0.87_{\pm0.22}$ & $1.37_{\pm0.33}$ & $2.49_{\pm0.64}$ & $3.52_{\pm1.04}$ \\
Low confidence & $0.14_{\pm0.09}$ & $0.48_{\pm0.21}$ & $0.90_{\pm0.36}$ & $1.90_{\pm0.75}$ & $2.97_{\pm1.23}$ \\
Random & $0.11_{\pm0.04}$ & $0.43_{\pm0.20}$ & $0.78_{\pm0.38}$ & $1.69_{\pm0.83}$ & $2.79_{\pm1.24}$ \\
\bottomrule
\end{tabular*}
\end{table}

\noindent\textbf{The decile experiment.}
Table~\ref{tab:conv_dec_masking-zero_last} gives the deciles under zero
masking.
Of the three criteria, the conviction has the largest drop at bin $10$ and
usually the smallest at bin $1$; its lower bins are close to zero and most
of its growth occurs in the last two.
Confidence and degree also increase along the bins, but more flatly and more
irregularly.
The experiment was run in all four modes; here we show one of them
specifically.

\begin{table}[htbp]
\caption{Decile experiment under zero masking. The training pool is split into ten equal bins by each ranking criterion (columns), from the lowest-scoring bin $1$ to the highest-scoring bin $10$ (rows), and one bin is perturbed at a time. The conviction has the largest drop at bin $10$, and most of its growth occurs in the last two bins.}
\label{tab:conv_dec_masking-zero_last}
\centering
\small
\begin{tabular*}{\textwidth}{@{\extracolsep{\fill}}lccc}
\toprule
\textbf{Zero masking} & \multicolumn{3}{c}{Ranking criterion} \\
\cmidrule(l){2-4}
Bin & Conviction & Confidence & Degree \\
\midrule
1 & $0.07_{\pm0.07}$ & $0.35_{\pm0.21}$ & $0.38_{\pm0.21}$ \\
2 & $0.09_{\pm0.08}$ & $0.36_{\pm0.21}$ & $0.42_{\pm0.22}$ \\
3 & $0.06_{\pm0.20}$ & $0.39_{\pm0.26}$ & $0.41_{\pm0.24}$ \\
4 & $0.17_{\pm0.25}$ & $0.42_{\pm0.25}$ & $0.37_{\pm0.27}$ \\
5 & $0.37_{\pm0.26}$ & $0.48_{\pm0.38}$ & $0.45_{\pm0.29}$ \\
6 & $0.41_{\pm0.30}$ & $0.58_{\pm0.41}$ & $0.46_{\pm0.31}$ \\
7 & $0.55_{\pm0.33}$ & $0.61_{\pm0.37}$ & $0.55_{\pm0.29}$ \\
8 & $0.63_{\pm0.41}$ & $0.58_{\pm0.37}$ & $0.44_{\pm0.31}$ \\
9 & $0.93_{\pm0.48}$ & $0.61_{\pm0.29}$ & $0.56_{\pm0.34}$ \\
10 & $2.18_{\pm0.70}$ & $0.92_{\pm0.28}$ & $1.06_{\pm0.40}$ \\
\bottomrule
\end{tabular*}
\end{table}

\noindent\textbf{A different conviction score.}
The body ranks nodes by \texttt{last}, the conviction at the final diffusion
layer, since that is the representation the classifier reads, but this rests
on a single layer whose scale is arbitrary.
We therefore repeat the analysis with \texttt{avg\_z}, in which the
conviction of each post-diffusion layer is z-scored separately and the
results averaged, so no single layer dominates the ranking
(Table~\ref{tab:conv_fwd_masking-zero_avg_z}).
Every conclusion is unchanged and the effect is slightly sharper: high
conviction remains the most damaging selection and low conviction the least
at every budget.

\begin{table}[htbp]
\definecolor{convworst}{RGB}{213,94,0}
\definecolor{convbest}{RGB}{0,114,178}
\caption{As Table~\ref{tab:conv_fwd_masking-zero_last}, with nodes ranked by the layer-averaged score \texttt{avg\_z}. The effect is unchanged, and slightly sharper.}
\label{tab:conv_fwd_masking-zero_avg_z}
\centering
\small
\begin{tabular*}{\textwidth}{@{\extracolsep{\fill}}lccccc}
\toprule
\textbf{Zero masking, \texttt{avg\_z}} & \multicolumn{5}{c}{Budget (fraction of the training pool)} \\
\cmidrule(l){2-6}
Selection & 1\% & 5\% & 10\% & 20\% & 30\% \\
\midrule
High conviction & \textcolor{convworst}{$\mathbf{0.56}_{\pm0.17}$} & \textcolor{convworst}{$\mathbf{1.68}_{\pm0.57}$} & \textcolor{convworst}{$\mathbf{2.65}_{\pm0.98}$} & \textcolor{convworst}{$\mathbf{4.10}_{\pm1.61}$} & \textcolor{convworst}{$\mathbf{5.18}_{\pm2.21}$} \\
Low conviction & \textcolor{convbest}{$\mathbf{0.01}_{\pm0.01}$} & \textcolor{convbest}{$\mathbf{0.02}_{\pm0.03}$} & \textcolor{convbest}{$\mathbf{0.02}_{\pm0.06}$} & \textcolor{convbest}{$\mathbf{0.13}_{\pm0.18}$} & \textcolor{convbest}{$\mathbf{0.34}_{\pm0.29}$} \\
High degree & $0.37_{\pm0.16}$ & $0.75_{\pm0.28}$ & $1.06_{\pm0.40}$ & $1.76_{\pm0.69}$ & $2.46_{\pm0.98}$ \\
Low degree & $0.03_{\pm0.02}$ & $0.19_{\pm0.10}$ & $0.38_{\pm0.21}$ & $0.85_{\pm0.44}$ & $1.35_{\pm0.74}$ \\
Degree-matched (high) & $0.07_{\pm0.04}$ & $0.19_{\pm0.15}$ & $0.27_{\pm0.26}$ & $0.52_{\pm0.49}$ & $0.77_{\pm0.55}$ \\
Degree-matched (low) & $0.06_{\pm0.03}$ & $0.26_{\pm0.12}$ & $0.49_{\pm0.33}$ & $1.17_{\pm0.66}$ & $2.22_{\pm1.23}$ \\
High confidence & $0.15_{\pm0.07}$ & $0.55_{\pm0.16}$ & $0.92_{\pm0.28}$ & $1.66_{\pm0.54}$ & $2.50_{\pm0.89}$ \\
Low confidence & $0.04_{\pm0.03}$ & $0.17_{\pm0.12}$ & $0.35_{\pm0.21}$ & $0.80_{\pm0.43}$ & $1.32_{\pm0.69}$ \\
Random & $0.07_{\pm0.04}$ & $0.21_{\pm0.15}$ & $0.43_{\pm0.29}$ & $1.04_{\pm0.55}$ & $1.69_{\pm0.89}$ \\
\bottomrule
\end{tabular*}
\end{table}

\noindent\textbf{Soft gating.}
The soft variant attenuates the whole training pool at once: each node's
outgoing message is attenuated in proportion to its own score, so there is
neither a budget nor a selection axis.
Attenuating by conviction costs $3.09\pm0.95\%$, by degree $4.19\pm1.95\%$
and by confidence $0.67\pm0.37\%$ (Table~\ref{tab:conv_soft_last}).
Unlike in the budget sweeps, conviction is therefore not the most damaging
score here.
A plausible reading is that the soft variant suppresses the low-conviction
nodes and softens the high-conviction ones at the same time, so the two
effects the budget sweep isolates are superimposed and partly cancel.

\begin{table}[htbp]
\caption{Relative drop under soft gating, which attenuates the outgoing message of every training node in proportion to its score (gate floor $g_{\min}=0.1$); with no budget, each ranking score gives one number.}
\label{tab:conv_soft_last}
\centering
\small
\begin{tabular*}{\textwidth}{@{\extracolsep{\fill}}lc}
\toprule
\textbf{Soft gating} & \\
Ranking score & Relative drop (\%) \\
\midrule
Confidence & $0.67_{\pm0.37}$ \\
Conviction & $3.09_{\pm0.95}$ \\
Degree & $4.19_{\pm1.95}$ \\
\bottomrule
\end{tabular*}
\end{table}

\FloatBarrier

\subsection{Per-dataset results}
\label{sec:supp_perdataset}

\noindent\textbf{Where the aggregate comes from.}
Tables~\ref{tab:conv_perds_high} and~\ref{tab:conv_perds_low} give the sweep
dataset by dataset for the high and the low selections, under zero masking
and conviction score \texttt{last}.
Here a cell is the mean over the ten folds of that dataset, and the quoted
spread is the standard error across folds, which shows how stable the effect
is within a dataset; the aggregated tables instead show how it varies across
the datasets.
The high-conviction effect is large on Minesweeper and Roman Empire, which
at the $10\%$ budget lose $5.40\pm0.61\%$ and $4.02\pm0.17\%$ of the learnt
signal, moderate on Amazon Ratings, Questions and Cora, and almost absent on
Pubmed, Citeseer and Tolokers.
That dispersion is what produces the wide between-dataset error bars of the
aggregated tables.
The ordering between strategies survives the disaggregation wherever there
is an effect at all: at the $10\%$ budget high conviction is more damaging
than high degree on seven of the eight datasets; the exception is
Tolokers, where neither strategy has an appreciable effect.
Low conviction is the least damaging of the four strategies on the two
datasets where the effect is large: it costs $0.00\pm0.02\%$ on Minesweeper
and $0.33\pm0.05\%$ on Roman Empire, against the $5.40\%$ and $4.02\%$ above.
A few cells are slightly negative, on Citeseer and Pubmed in particular,
which means that the perturbed model scores marginally above its own
baseline; these
are the datasets on which no selection has an appreciable effect, so we read
those values as noise around zero.

\begin{table}[htbp]
\caption{Budget sweep under zero masking, dataset by dataset, for the high-degree and high-conviction selections. Unlike the other tables, the spread is the standard error across the ten folds of each dataset.}
\label{tab:conv_perds_high}
\centering
\footnotesize
\begin{tabular*}{\textwidth}{@{\extracolsep{\fill}}llccccc}
\toprule
\multicolumn{2}{@{}l}{\textbf{Zero masking, high}} & \multicolumn{5}{c}{Budget (fraction of the training pool)} \\
\cmidrule(l){3-7}
Dataset & Selection & 1\% & 5\% & 10\% & 20\% & 30\% \\
\midrule
Amazon Ratings & Degree & $1.25_{\pm0.17}$ & $1.75_{\pm0.28}$ & $2.14_{\pm0.32}$ & $3.14_{\pm0.32}$ & $3.90_{\pm0.27}$ \\
 & Conviction & $1.16_{\pm0.17}$ & $2.31_{\pm0.26}$ & $2.93_{\pm0.29}$ & $3.81_{\pm0.28}$ & $4.48_{\pm0.28}$ \\
\addlinespace
Citeseer & Degree & $-0.04_{\pm0.16}$ & $-0.39_{\pm0.21}$ & $-0.65_{\pm0.35}$ & $-1.25_{\pm0.82}$ & $-1.74_{\pm1.05}$ \\
 & Conviction & $0.06_{\pm0.17}$ & $0.41_{\pm0.21}$ & $0.62_{\pm0.22}$ & $1.20_{\pm0.31}$ & $0.67_{\pm0.46}$ \\
\addlinespace
Cora & Degree & $0.46_{\pm0.17}$ & $0.87_{\pm0.20}$ & $0.99_{\pm0.37}$ & $1.53_{\pm0.44}$ & $2.21_{\pm0.45}$ \\
 & Conviction & $0.46_{\pm0.16}$ & $0.87_{\pm0.28}$ & $1.23_{\pm0.32}$ & $1.78_{\pm0.37}$ & $2.29_{\pm0.49}$ \\
\addlinespace
Minesweeper & Degree & $0.14_{\pm0.03}$ & $0.67_{\pm0.09}$ & $1.46_{\pm0.10}$ & $3.26_{\pm0.24}$ & $5.31_{\pm0.29}$ \\
 & Conviction & $0.54_{\pm0.12}$ & $2.59_{\pm0.39}$ & $5.40_{\pm0.61}$ & $10.27_{\pm0.74}$ & $14.70_{\pm0.77}$ \\
\addlinespace
Pubmed & Degree & $-0.00_{\pm0.02}$ & $0.03_{\pm0.05}$ & $0.01_{\pm0.08}$ & $0.17_{\pm0.09}$ & $0.29_{\pm0.12}$ \\
 & Conviction & $0.03_{\pm0.02}$ & $0.09_{\pm0.07}$ & $0.11_{\pm0.10}$ & $0.18_{\pm0.08}$ & $0.27_{\pm0.11}$ \\
\addlinespace
Questions & Degree & $0.79_{\pm0.34}$ & $1.64_{\pm0.36}$ & $1.85_{\pm0.42}$ & $1.98_{\pm0.52}$ & $2.12_{\pm0.55}$ \\
 & Conviction & $1.26_{\pm0.34}$ & $2.19_{\pm0.31}$ & $3.10_{\pm0.41}$ & $3.40_{\pm0.34}$ & $3.32_{\pm0.32}$ \\
\addlinespace
Roman Empire & Degree & $0.37_{\pm0.04}$ & $1.37_{\pm0.08}$ & $2.48_{\pm0.11}$ & $4.74_{\pm0.13}$ & $6.73_{\pm0.11}$ \\
 & Conviction & $0.56_{\pm0.04}$ & $2.30_{\pm0.11}$ & $4.02_{\pm0.17}$ & $7.22_{\pm0.22}$ & $9.95_{\pm0.20}$ \\
\addlinespace
Tolokers & Degree & $-0.00_{\pm0.09}$ & $0.10_{\pm0.15}$ & $0.18_{\pm0.19}$ & $0.52_{\pm0.31}$ & $0.83_{\pm0.38}$ \\
 & Conviction & $0.01_{\pm0.09}$ & $-0.05_{\pm0.12}$ & $0.02_{\pm0.13}$ & $0.34_{\pm0.18}$ & $0.72_{\pm0.30}$ \\
\bottomrule
\end{tabular*}
\end{table}

\begin{table}[htbp]
\caption{As Table~\ref{tab:conv_perds_high}, for the low-degree and low-conviction selections.}
\label{tab:conv_perds_low}
\centering
\footnotesize
\begin{tabular*}{\textwidth}{@{\extracolsep{\fill}}llccccc}
\toprule
\multicolumn{2}{@{}l}{\textbf{Zero masking, low}} & \multicolumn{5}{c}{Budget (fraction of the training pool)} \\
\cmidrule(l){3-7}
Dataset & Selection & 1\% & 5\% & 10\% & 20\% & 30\% \\
\midrule
Amazon Ratings & Degree & $0.02_{\pm0.05}$ & $-0.02_{\pm0.08}$ & $0.05_{\pm0.12}$ & $0.48_{\pm0.18}$ & $0.73_{\pm0.20}$ \\
 & Conviction & $-0.02_{\pm0.03}$ & $-0.02_{\pm0.08}$ & $0.07_{\pm0.09}$ & $0.13_{\pm0.10}$ & $0.56_{\pm0.22}$ \\
\addlinespace
Citeseer & Degree & $-0.04_{\pm0.04}$ & $-0.05_{\pm0.11}$ & $-0.12_{\pm0.16}$ & $-0.15_{\pm0.22}$ & $-0.29_{\pm0.35}$ \\
 & Conviction & $0.08_{\pm0.09}$ & $-0.12_{\pm0.18}$ & $-0.30_{\pm0.20}$ & $-0.62_{\pm0.27}$ & $-1.60_{\pm0.67}$ \\
\addlinespace
Cora & Degree & $0.00_{\pm0.00}$ & $0.09_{\pm0.06}$ & $0.09_{\pm0.14}$ & $0.11_{\pm0.19}$ & $0.14_{\pm0.18}$ \\
 & Conviction & $-0.02_{\pm0.02}$ & $0.05_{\pm0.06}$ & $0.04_{\pm0.12}$ & $0.29_{\pm0.15}$ & $0.50_{\pm0.17}$ \\
\addlinespace
Minesweeper & Degree & $0.10_{\pm0.03}$ & $0.62_{\pm0.07}$ & $1.36_{\pm0.10}$ & $2.98_{\pm0.16}$ & $5.21_{\pm0.26}$ \\
 & Conviction & $0.00_{\pm0.00}$ & $-0.01_{\pm0.01}$ & $0.00_{\pm0.02}$ & $0.29_{\pm0.11}$ & $0.86_{\pm0.16}$ \\
\addlinespace
Pubmed & Degree & $-0.01_{\pm0.01}$ & $-0.02_{\pm0.02}$ & $-0.03_{\pm0.03}$ & $-0.05_{\pm0.04}$ & $-0.07_{\pm0.05}$ \\
 & Conviction & $-0.00_{\pm0.00}$ & $0.01_{\pm0.02}$ & $0.01_{\pm0.03}$ & $-0.05_{\pm0.03}$ & $-0.01_{\pm0.05}$ \\
\addlinespace
Questions & Degree & $0.03_{\pm0.02}$ & $0.10_{\pm0.06}$ & $0.21_{\pm0.08}$ & $0.45_{\pm0.14}$ & $0.57_{\pm0.14}$ \\
 & Conviction & $0.08_{\pm0.05}$ & $0.22_{\pm0.06}$ & $0.26_{\pm0.13}$ & $0.33_{\pm0.20}$ & $0.50_{\pm0.25}$ \\
\addlinespace
Roman Empire & Degree & $0.14_{\pm0.03}$ & $0.63_{\pm0.04}$ & $1.27_{\pm0.08}$ & $2.71_{\pm0.14}$ & $4.16_{\pm0.17}$ \\
 & Conviction & $0.02_{\pm0.01}$ & $0.14_{\pm0.03}$ & $0.33_{\pm0.05}$ & $0.74_{\pm0.08}$ & $1.39_{\pm0.16}$ \\
\addlinespace
Tolokers & Degree & $0.02_{\pm0.03}$ & $0.16_{\pm0.05}$ & $0.23_{\pm0.06}$ & $0.26_{\pm0.10}$ & $0.36_{\pm0.10}$ \\
 & Conviction & $0.03_{\pm0.01}$ & $0.09_{\pm0.05}$ & $0.15_{\pm0.09}$ & $0.49_{\pm0.17}$ & $0.95_{\pm0.29}$ \\
\bottomrule
\end{tabular*}
\end{table}

\subsection{The cost of pruning}
\label{sec:supp_prune_diag}

Pruning is the only intervention that changes the graph, so for every
(strategy, budget) pair we record the fraction of nodes and edges remaining,
the number of connected components and the size of the giant component.
These diagnostics show that under pruning degree and conviction are not on
the same footing.
Removing the top $k\%$ by degree deletes by construction the largest
possible number of edges, $51.6\pm4.3\%$ of the graph at the $30\%$ budget
against $39.4\pm2.8\%$ for the top $k\%$ by conviction
(Table~\ref{tab:conv_prune_edges_last}).
Even so, its average cost at that budget, $5.77\pm1.57\%$, is not higher than
that of conviction, $7.08\pm2.04\%$ (Table~\ref{tab:conv_fwd_pruning_last}).

The same diagnostics qualify the degree-matched control.
At the $1\%$ budget it removes $5.0\pm1.3\%$ of the edges against
$7.7\pm3.1\%$ for the high-conviction group, five times what a random
selection of the same size removes, and it costs $0.20\pm0.05\%$ against
$1.65\pm0.96\%$.
At the largest budgets the matching is diluted, since the pool of unused
high-degree candidates is exhausted, and by the $30\%$ budget the control
deletes $28.3\pm1.7\%$ of the edges, indistinguishable from a random
selection.
The control is therefore most informative at the small budgets, and should
be read with this in mind at the large ones.

\begin{table}[htbp]
\caption{Percentage of the graph's edges deleted under pruning, by selection strategy and budget; the performance cost of the same selections is in Table~\ref{tab:conv_fwd_pruning_last}. Selecting by degree deletes the most edges.}
\label{tab:conv_prune_edges_last}
\centering
\small
\begin{tabular*}{\textwidth}{@{\extracolsep{\fill}}lccccc}
\toprule
\textbf{Pruning} & \multicolumn{5}{c}{Budget (fraction of the training pool)} \\
\cmidrule(l){2-6}
Selection & 1\% & 5\% & 10\% & 20\% & 30\% \\
\midrule
High conviction & $7.7_{\pm3.1}$ & $15.9_{\pm3.5}$ & $22.0_{\pm3.4}$ & $31.3_{\pm3.1}$ & $39.4_{\pm2.8}$ \\
Low conviction & $0.5_{\pm0.1}$ & $2.7_{\pm0.5}$ & $5.6_{\pm1.0}$ & $11.6_{\pm1.8}$ & $18.0_{\pm2.5}$ \\
High degree & $10.2_{\pm3.4}$ & $23.5_{\pm5.1}$ & $32.7_{\pm5.5}$ & $44.2_{\pm5.1}$ & $51.6_{\pm4.3}$ \\
Low degree & $0.4_{\pm0.1}$ & $1.9_{\pm0.6}$ & $4.2_{\pm1.2}$ & $8.7_{\pm2.5}$ & $13.3_{\pm3.7}$ \\
Degree-matched (high) & $5.0_{\pm1.3}$ & $11.5_{\pm1.4}$ & $16.6_{\pm1.3}$ & $23.6_{\pm1.1}$ & $28.3_{\pm1.7}$ \\
Degree-matched (low) & $0.5_{\pm0.1}$ & $2.8_{\pm0.6}$ & $5.9_{\pm1.1}$ & $12.4_{\pm1.9}$ & $19.3_{\pm2.6}$ \\
High confidence & $2.0_{\pm0.4}$ & $6.9_{\pm1.1}$ & $11.7_{\pm1.7}$ & $20.3_{\pm2.6}$ & $28.1_{\pm3.1}$ \\
Low confidence & $1.7_{\pm0.9}$ & $7.1_{\pm2.6}$ & $12.2_{\pm3.5}$ & $20.7_{\pm4.2}$ & $28.2_{\pm4.5}$ \\
Random & $1.0_{\pm0.1}$ & $5.0_{\pm0.2}$ & $10.0_{\pm0.4}$ & $19.5_{\pm0.7}$ & $28.3_{\pm1.0}$ \\
\bottomrule
\end{tabular*}
\end{table}

\FloatBarrier

\section{The Forman Profile: Protocol and Additional Material}
\label{sec:supp_forman}

\subsection{Protocol}
\label{sec:supp_forman_setup}
\label{sec:supp_metrics}
\label{sec:supp_forman_agg}

\noindent\textbf{Models and training.}
Every configuration follows the same loop: the network is trained on the
node classification task, and at selected epochs the Forman profile of the
sheaf Laplacian currently produced by the restriction maps is extracted,
layer by layer, and compared to the reference profiles.
To extract it, we take the sheaf Laplacian produced at layer $\ell$, apply
Definition~\ref{def:block_bochner} to isolate the Forman term, and
diagonalize each of its $n$ diagonal blocks.

Two model families are studied.
In \emph{GeneralSheaf}, the learner reported in the body, a neural map
sends each edge to an unconstrained $d\times d$ restriction map,
$\F^{(\ell)}_{v\trianglelefteq e} =
\varsigma(\Theta[x_v^{(\ell)}\|x_u^{(\ell)}])$.
In \emph{JointSheafParams}~\cite{caralt2024jointdiffusionprocessesinductive}
the node signals and the restriction maps are diffused jointly at each
layer $\ell$:
\begin{equation}
\begin{cases}
X^{(\ell+1)} = X^{(\ell)} - \varsigma\bigl((\mathbf{I}_{nd} - \alpha\,\Delta_{\F^{(\ell)}})(I_n \otimes \Omega_1^{(\ell)})X^{(\ell)}\Omega_2^{(\ell)}\bigr),\\[2pt]
\F^{*(\ell+1)} = \F^{*(\ell)} - \varsigma\bigl((\mathbf{I}_{2md} - \beta\,\Delta_{X^{(\ell)}})(I_{2m} \otimes \Omega_1^{*(\ell)})\F^{*(\ell)}\Omega_2^{*(\ell)}\bigr),
\end{cases}
\end{equation}
where $\F^*\in\R^{2md\times d}$ stacks the $2m$ transposed restriction maps of
the $m$ edges, $\Delta_{X^{(\ell)}}$ is the sheaf Laplacian built with the node
signals in place of the restriction maps, $\Omega_1,\Omega_2$ and their starred
counterparts are learnt weight matrices, and $\alpha,\beta$ are fixed
coefficients weighting the two diffusions.
The maps are initialized at the identity sheaf, $\F^{(0)}=I_d$, so that for
this family the profile measures drift from the trivial sheaf.

Each model is trained independently under both normalization regimes, with $\Delta_\F = D_\F^{-1/2}L_\F D_\F^{-1/2}$ and with the raw $L_\F$.
Training uses $10$-fold cross-validation, $500$ epochs and early stopping
with patience $200$.
Stalk dimension and depth are chosen per dataset, $d\in\{2,3,5\}$ and
$L\in\{2,\dots,6\}$: Cornell and Wisconsin use $d=2$ and $L=2$, the larger
benchmarks up to $d=5$ and $L=6$.
Profiles are saved at epochs $\{0,1,5,15,50,200\}$ and at the last epoch,
for every fold and every layer.
Epoch $0$ is the control: the sheaf is randomly initialized for
GeneralSheaf, and is the identity sheaf for JointSheafParams.

\noindent\textbf{Alignment metrics.}
All metrics are computed per eigenvalue of the Forman term, per layer, per
checkpoint and per fold, then averaged over eigenvalues and folds, so a configuration
yields one scalar per (checkpoint, layer) cell, with checkpoints on the rows and
layers on the columns.

For normalized runs each eigenvalue channel $f_k$ is compared to the
topological reference $f^{\mathrm{top}}$ by Spearman $\rho$ and Pearson $r$,
by the $1$-Wasserstein distance $W_1$, by the RMSE between raw values, and
by the NRMSE, the same $\ell_2$ distance after z-scoring both profiles.
RMSE and NRMSE differ only by the z-scoring, so the first sees location and
scale and the second does not, and when they disagree, the profiles match in shape and differ in mean
or scale.
On the normalized runs the two do disagree, which is why the alignment of
Section~\ref{sec:forman_profile} is agreement in shape up to an offset, and
why $W_1$ is the statistic reported in the body.

For unnormalized runs the baseline is ambiguous, since the unnormalized
graph Laplacian has an identically zero profile.
We therefore use both references:
$W_1(f_k,f_k^{\mathrm{top}})$ and $W_1(f_k,\mathbf 0)$ for distributional
distance to the topological and the flat reference, Pearson $r$ and NRMSE
against $f_k^{\mathrm{top}}$ for shape alignment independent of scale, and
RMSE against $\mathbf 0$, the root-mean-square magnitude of the learnt
profile itself.

\noindent\textbf{Aggregation.}
The global tables report $W_1$ as a single statistic per checkpoint: for
each dataset we average over folds and eigenvalues, then across datasets.
Since $L$ differs across datasets, layers are binned by relative depth: first, middle or last.
When a bin spans more than one layer, as the middle bin of a six-layer
network does, those layers are averaged together with the folds into one
number per dataset, exactly as the eigenvalue channels already are; only
then do we average across datasets.
The quoted spread is therefore across datasets, and shows how much a
pattern varies from one dataset to another.
The global tables of Appendix~\ref{sec:supp_forman_tables} report $W_1$
only; the per-dataset tables of Appendix~\ref{sec:supp_forman_figures} report
all the metrics above.

\subsection{Separating normalization from training}
\label{sec:supp_other}

To separate the effect of the learning from that of the normalization we
recompute, for GeneralSheaf and at every checkpoint, a second Laplacian
built from the \emph{same learnt restriction maps} under the opposite
normalization, the \emph{other Laplacian}.
This gives four profiles per dataset.

\begin{center}
\renewcommand{\arraystretch}{1.3}
\small
\setlength{\tabcolsep}{4pt}%
\begin{tabular}{@{}lcc@{}}
\toprule
 & \textbf{Evaluated normalized} $\Delta_{\F}$
 & \textbf{Evaluated unnormalized} $L_{\F}$ \\
\midrule
Maps trained \texttt{norm=True}  & primary run     & other Laplacian \\
Maps trained \texttt{norm=False} & other Laplacian & primary run     \\
\bottomrule
\end{tabular}
\end{center}

Each column fixes the normalization used at evaluation and varies only the
regime the maps were trained under, so comparing the two entries of one
column isolates the role of the training regime.
The \emph{ablation} compares each other Laplacian against the reference
profiles, with the metrics above.
The \textbf{cross-comparison} instead compares the two Laplacians of one
column with each other: the Forman profile of the Laplacian built from the
maps trained with \texttt{norm=True} against that of the one built from the
maps trained with \texttt{norm=False}, both evaluated with the same
normalization, so that no reference profile is involved.
It is done for GeneralSheaf, dataset by dataset, at every checkpoint and
layer, and aggregated as in Appendix~\ref{sec:supp_forman_setup}; it
measures how far apart the profiles produced by the two training regimes
are.

\FloatBarrier

\subsection{The Forman fingerprint of each learning mechanism}
\label{sec:supp_forman_tables}

Section~\ref{sec:forman_profile} reports the normalized GeneralSheaf runs,
which converge to local alignment with the normalized graph Laplacian.
The remaining mechanisms each leave a different signature: we read the
difference in Forman-profile behaviour between these sheaf learning
mechanisms as a fingerprint of the learning process.

All tables of this section share the layout of Table~\ref{tab:global_general_norm}.
Each cell is the $1$-Wasserstein distance between two Forman profiles, with
training epochs on the rows, epoch $0$ being the initial sheaf, and the
relative depth of the layer on the columns.
Values are mean$_{\pm\text{s.e.}}$ across 10 datasets, each averaged over its
folds, 8 for the middle layer, and darker cells are farther apart, on the
scale of the colorbar beside each panel.
The bold label names the mechanism; the caption gives the reference and the
pattern.

\noindent\textbf{Unnormalized general sheaf learning aligns with neither
reference.}
The metrics against the topological reference stay almost constant in layer
and epoch, while the distance from the zero reference is large and generally
grows: on the large heterophilic datasets it peaks at intermediate depth and
epoch, then recedes slightly but stays substantial
(Table~\ref{tab:global_general_unnorm}).
We have no explanation for the non-monotonicity and read no mechanism into
it.

\begin{table}[htbp]
\definecolor{fpb0}{RGB}{255,251,195}
\definecolor{fpb1}{RGB}{255,244,176}
\definecolor{fpb2}{RGB}{255,236,158}
\definecolor{fpb3}{RGB}{255,228,141}
\definecolor{fpb4}{RGB}{254,220,123}
\definecolor{fpb5}{RGB}{254,206,106}
\definecolor{fpb6}{RGB}{254,189,88}
\definecolor{fpb7}{RGB}{254,173,74}
\definecolor{fpb8}{RGB}{253,158,67}
\definecolor{fpb9}{RGB}{253,143,61}
\definecolor{fpb10}{RGB}{253,117,53}
\definecolor{fpb11}{RGB}{252,91,46}
\definecolor{fpb12}{RGB}{253,145,62}
\definecolor{fpb13}{RGB}{253,128,56}
\definecolor{fpb14}{RGB}{253,133,58}
\definecolor{fpb15}{RGB}{252,109,51}
\definecolor{fpb16}{RGB}{253,138,59}
\definecolor{fpb17}{RGB}{252,81,43}
\definecolor{fpb18}{RGB}{252,93,46}
\definecolor{fpb19}{RGB}{253,127,56}
\definecolor{fpb20}{RGB}{252,78,42}
\definecolor{fpb21}{RGB}{252,97,47}
\definecolor{fpb22}{RGB}{253,132,57}
\definecolor{fpb23}{RGB}{252,94,46}
\definecolor{fpb24}{RGB}{253,111,51}
\definecolor{fpb25}{RGB}{253,134,58}
\definecolor{fpb26}{RGB}{252,108,51}
\definecolor{fpb27}{RGB}{253,129,57}
\definecolor{fpb28}{RGB}{255,254,201}
\definecolor{fpb29}{RGB}{254,161,69}
\definecolor{fpb30}{RGB}{255,253,200}
\definecolor{fpb31}{RGB}{253,155,66}
\definecolor{fpb32}{RGB}{253,157,67}
\definecolor{fpb33}{RGB}{255,252,197}
\definecolor{fpb34}{RGB}{253,156,66}
\definecolor{fpb35}{RGB}{254,165,71}
\definecolor{fpb36}{RGB}{255,251,193}
\definecolor{fpb37}{RGB}{253,150,64}
\definecolor{fpb38}{RGB}{254,183,81}
\definecolor{fpb39}{RGB}{255,247,185}
\definecolor{fpb40}{RGB}{254,198,98}
\definecolor{fpb41}{RGB}{255,246,181}
\definecolor{fpb42}{RGB}{252,98,48}
\definecolor{fpb43}{RGB}{254,204,104}
\definecolor{fpb44}{RGB}{255,245,180}
\caption{$W_1$ distance of the unnormalized GeneralSheaf from the topological reference (top) and the zero reference (bottom). The profile aligns with neither.}
\label{tab:global_general_unnorm}
\centering
\small
\renewcommand{\arraystretch}{1.2}%
\begin{tabular}{l>{\centering\arraybackslash}p{2.7cm}>{\centering\arraybackslash}p{2.7cm}>{\centering\arraybackslash}p{2.7cm}l}
\toprule
\multicolumn{5}{@{}l}{\textbf{GeneralSheaf, unnormalized}} \\
 & \multicolumn{3}{c}{Layer} & \\
\cmidrule(l){2-4}
Epoch & First & Middle & Last & \\
\midrule
\multicolumn{5}{@{}l}{\emph{against the topological reference}} \\
0 & \cellcolor{fpb12}$0.437_{\pm0.041}$ & \cellcolor{fpb13}$0.472_{\pm0.047}$ & \cellcolor{fpb14}$0.463_{\pm0.040}$ & \smash{\begin{tikzpicture}[baseline={([yshift=-9pt]current bounding box.north)}]
\fill[fpb0] (0,0.0000) rectangle (0.28,0.2292);
\fill[fpb1] (0,0.2292) rectangle (0.28,0.4583);
\fill[fpb2] (0,0.4583) rectangle (0.28,0.6875);
\fill[fpb3] (0,0.6875) rectangle (0.28,0.9167);
\fill[fpb4] (0,0.9167) rectangle (0.28,1.1458);
\fill[fpb5] (0,1.1458) rectangle (0.28,1.3750);
\fill[fpb6] (0,1.3750) rectangle (0.28,1.6042);
\fill[fpb7] (0,1.6042) rectangle (0.28,1.8333);
\fill[fpb8] (0,1.8333) rectangle (0.28,2.0625);
\fill[fpb9] (0,2.0625) rectangle (0.28,2.2917);
\fill[fpb10] (0,2.2917) rectangle (0.28,2.5208);
\fill[fpb11] (0,2.5208) rectangle (0.28,2.7500);
\draw[line width=0.3pt] (0,0) rectangle (0.28,2.75);
\node[right,font=\scriptsize] at (0.28,2.75) {0.56};
\node[right,font=\scriptsize] at (0.28,1.375) {0.28};
\node[right,font=\scriptsize] at (0.28,0) {0.00};
\end{tikzpicture}} \\
1 & \cellcolor{fpb9}$0.442_{\pm0.042}$ & \cellcolor{fpb15}$0.505_{\pm0.047}$ & \cellcolor{fpb10}$0.491_{\pm0.045}$ &  \\
5 & \cellcolor{fpb16}$0.453_{\pm0.040}$ & \cellcolor{fpb17}$0.554_{\pm0.051}$ & \cellcolor{fpb18}$0.534_{\pm0.045}$ &  \\
15 & \cellcolor{fpb19}$0.473_{\pm0.040}$ & \cellcolor{fpb20}$0.560_{\pm0.037}$ & \cellcolor{fpb21}$0.526_{\pm0.037}$ &  \\
50 & \cellcolor{fpb22}$0.464_{\pm0.038}$ & \cellcolor{fpb23}$0.532_{\pm0.048}$ & \cellcolor{fpb24}$0.502_{\pm0.040}$ &  \\
200 & \cellcolor{fpb25}$0.460_{\pm0.039}$ & \cellcolor{fpb26}$0.506_{\pm0.047}$ & \cellcolor{fpb27}$0.469_{\pm0.032}$ &  \\
final & \cellcolor{fpb25}$0.460_{\pm0.039}$ & \cellcolor{fpb15}$0.505_{\pm0.047}$ & \cellcolor{fpb27}$0.469_{\pm0.030}$ &  \\
\midrule
\multicolumn{5}{@{}l}{\emph{against the zero reference}} \\
0 & \cellcolor{fpb28}$0.172_{\pm0.143}$ & \cellcolor{fpb7}$7.619_{\pm6.955}$ & \cellcolor{fpb29}$8.361_{\pm7.008}$ & \smash{\begin{tikzpicture}[baseline={([yshift=-9pt]current bounding box.north)}]
\fill[fpb0] (0,0.0000) rectangle (0.28,0.2292);
\fill[fpb1] (0,0.2292) rectangle (0.28,0.4583);
\fill[fpb2] (0,0.4583) rectangle (0.28,0.6875);
\fill[fpb3] (0,0.6875) rectangle (0.28,0.9167);
\fill[fpb4] (0,0.9167) rectangle (0.28,1.1458);
\fill[fpb5] (0,1.1458) rectangle (0.28,1.3750);
\fill[fpb6] (0,1.3750) rectangle (0.28,1.6042);
\fill[fpb7] (0,1.6042) rectangle (0.28,1.8333);
\fill[fpb8] (0,1.8333) rectangle (0.28,2.0625);
\fill[fpb9] (0,2.0625) rectangle (0.28,2.2917);
\fill[fpb10] (0,2.2917) rectangle (0.28,2.5208);
\fill[fpb11] (0,2.5208) rectangle (0.28,2.7500);
\draw[line width=0.3pt] (0,0) rectangle (0.28,2.75);
\node[right,font=\scriptsize] at (0.28,2.75) {12.12};
\node[right,font=\scriptsize] at (0.28,1.375) {6.06};
\node[right,font=\scriptsize] at (0.28,0) {0.00};
\end{tikzpicture}} \\
1 & \cellcolor{fpb30}$0.229_{\pm0.180}$ & \cellcolor{fpb31}$8.775_{\pm7.724}$ & \cellcolor{fpb32}$8.671_{\pm6.717}$ &  \\
5 & \cellcolor{fpb33}$0.399_{\pm0.287}$ & \cellcolor{fpb34}$8.718_{\pm7.220}$ & \cellcolor{fpb35}$8.101_{\pm6.057}$ &  \\
15 & \cellcolor{fpb36}$0.603_{\pm0.409}$ & \cellcolor{fpb37}$9.075_{\pm6.932}$ & \cellcolor{fpb38}$6.955_{\pm4.448}$ &  \\
50 & \cellcolor{fpb39}$1.064_{\pm0.573}$ & \cellcolor{fpb19}$10.240_{\pm6.554}$ & \cellcolor{fpb40}$6.004_{\pm3.448}$ &  \\
200 & \cellcolor{fpb41}$1.277_{\pm0.530}$ & \cellcolor{fpb42}$11.338_{\pm6.261}$ & \cellcolor{fpb43}$5.680_{\pm2.724}$ &  \\
final & \cellcolor{fpb44}$1.347_{\pm0.591}$ & \cellcolor{fpb20}$12.125_{\pm7.013}$ & \cellcolor{fpb5}$5.536_{\pm2.636}$ &  \\
\bottomrule
\end{tabular}
\end{table}

\noindent\textbf{The alignment comes from normalization, and the maps stay away
from the identity.}
Maps trained normalized but evaluated unnormalized lose their alignment to
both references.
Were they identity-like, evaluating them unnormalized would reproduce the
unnormalized graph Laplacian, whose profile vanishes identically, and the
profile would sit at the zero reference.
In the other direction, maps trained unnormalized but evaluated normalized
do align with the topological reference, in $W_1$ as in the other metrics:
they
are not the identity, yet they build a Laplacian which, once normalized,
sits close to the reference.
The effect is absent at initialization
(Tables~\ref{tab:global_general_other_norm}
and~\ref{tab:global_general_other_unnorm}).

\begin{table}[htbp]
\definecolor{fpc0}{RGB}{255,251,195}
\definecolor{fpc1}{RGB}{255,244,176}
\definecolor{fpc2}{RGB}{255,236,158}
\definecolor{fpc3}{RGB}{255,228,141}
\definecolor{fpc4}{RGB}{254,220,123}
\definecolor{fpc5}{RGB}{254,206,106}
\definecolor{fpc6}{RGB}{254,189,88}
\definecolor{fpc7}{RGB}{254,173,74}
\definecolor{fpc8}{RGB}{253,158,67}
\definecolor{fpc9}{RGB}{253,143,61}
\definecolor{fpc10}{RGB}{253,117,53}
\definecolor{fpc11}{RGB}{252,91,46}
\definecolor{fpc12}{RGB}{253,136,59}
\definecolor{fpc13}{RGB}{253,129,57}
\definecolor{fpc14}{RGB}{253,128,56}
\definecolor{fpc15}{RGB}{253,131,57}
\definecolor{fpc16}{RGB}{252,108,51}
\definecolor{fpc17}{RGB}{253,125,56}
\definecolor{fpc18}{RGB}{252,110,51}
\definecolor{fpc19}{RGB}{252,79,42}
\definecolor{fpc20}{RGB}{253,111,51}
\definecolor{fpc21}{RGB}{253,126,56}
\definecolor{fpc22}{RGB}{253,142,60}
\definecolor{fpc23}{RGB}{253,159,68}
\definecolor{fpc24}{RGB}{254,160,68}
\definecolor{fpc25}{RGB}{254,183,81}
\definecolor{fpc26}{RGB}{254,214,115}
\definecolor{fpc27}{RGB}{254,216,117}
\definecolor{fpc28}{RGB}{254,182,80}
\definecolor{fpc29}{RGB}{254,215,116}
\definecolor{fpc30}{RGB}{254,213,114}
\caption{$W_1$ distance of the other Laplacian, with maps trained unnormalized and evaluated normalized (Appendix~\ref{sec:supp_other}), from the topological reference. Once normalized, the maps align with it as training proceeds.}
\label{tab:global_general_other_norm}
\centering
\small
\renewcommand{\arraystretch}{1.2}%
\begin{tabular}{l>{\centering\arraybackslash}p{2.7cm}>{\centering\arraybackslash}p{2.7cm}>{\centering\arraybackslash}p{2.7cm}l}
\toprule
\multicolumn{5}{@{}l}{\textbf{Other Laplacian: trained unnormalized, evaluated normalized}} \\
 & \multicolumn{3}{c}{Layer} & \\
\cmidrule(l){2-4}
Epoch & First & Middle & Last & \\
\midrule
0 & \cellcolor{fpc12}$0.408_{\pm0.047}$ & \cellcolor{fpc13}$0.419_{\pm0.062}$ & \cellcolor{fpc14}$0.420_{\pm0.052}$ & \smash{\begin{tikzpicture}[baseline={([yshift=-9pt]current bounding box.north)}]
\fill[fpc0] (0,0.0000) rectangle (0.28,0.2292);
\fill[fpc1] (0,0.2292) rectangle (0.28,0.4583);
\fill[fpc2] (0,0.4583) rectangle (0.28,0.6875);
\fill[fpc3] (0,0.6875) rectangle (0.28,0.9167);
\fill[fpc4] (0,0.9167) rectangle (0.28,1.1458);
\fill[fpc5] (0,1.1458) rectangle (0.28,1.3750);
\fill[fpc6] (0,1.3750) rectangle (0.28,1.6042);
\fill[fpc7] (0,1.6042) rectangle (0.28,1.8333);
\fill[fpc8] (0,1.8333) rectangle (0.28,2.0625);
\fill[fpc9] (0,2.0625) rectangle (0.28,2.2917);
\fill[fpc10] (0,2.2917) rectangle (0.28,2.5208);
\fill[fpc11] (0,2.5208) rectangle (0.28,2.7500);
\draw[line width=0.3pt] (0,0) rectangle (0.28,2.75);
\node[right,font=\scriptsize] at (0.28,2.75) {0.50};
\node[right,font=\scriptsize] at (0.28,1.375) {0.25};
\node[right,font=\scriptsize] at (0.28,0) {0.00};
\end{tikzpicture}} \\
1 & \cellcolor{fpc15}$0.416_{\pm0.051}$ & \cellcolor{fpc16}$0.452_{\pm0.066}$ & \cellcolor{fpc17}$0.425_{\pm0.057}$ &  \\
5 & \cellcolor{fpc18}$0.450_{\pm0.050}$ & \cellcolor{fpc19}$0.499_{\pm0.068}$ & \cellcolor{fpc20}$0.448_{\pm0.063}$ &  \\
15 & \cellcolor{fpc21}$0.424_{\pm0.053}$ & \cellcolor{fpc14}$0.420_{\pm0.070}$ & \cellcolor{fpc22}$0.398_{\pm0.075}$ &  \\
50 & \cellcolor{fpc23}$0.352_{\pm0.060}$ & \cellcolor{fpc24}$0.348_{\pm0.098}$ & \cellcolor{fpc23}$0.351_{\pm0.073}$ &  \\
200 & \cellcolor{fpc25}$0.287_{\pm0.058}$ & \cellcolor{fpc26}$0.207_{\pm0.070}$ & \cellcolor{fpc27}$0.203_{\pm0.034}$ &  \\
final & \cellcolor{fpc28}$0.290_{\pm0.059}$ & \cellcolor{fpc29}$0.205_{\pm0.069}$ & \cellcolor{fpc30}$0.209_{\pm0.032}$ &  \\
\bottomrule
\end{tabular}
\end{table}

\begin{table}[htbp]
\definecolor{fpd0}{RGB}{255,251,195}
\definecolor{fpd1}{RGB}{255,244,176}
\definecolor{fpd2}{RGB}{255,236,158}
\definecolor{fpd3}{RGB}{255,228,141}
\definecolor{fpd4}{RGB}{254,220,123}
\definecolor{fpd5}{RGB}{254,206,106}
\definecolor{fpd6}{RGB}{254,189,88}
\definecolor{fpd7}{RGB}{254,173,74}
\definecolor{fpd8}{RGB}{253,158,67}
\definecolor{fpd9}{RGB}{253,143,61}
\definecolor{fpd10}{RGB}{253,117,53}
\definecolor{fpd11}{RGB}{252,91,46}
\definecolor{fpd12}{RGB}{253,128,56}
\definecolor{fpd13}{RGB}{252,110,51}
\definecolor{fpd14}{RGB}{253,123,55}
\definecolor{fpd15}{RGB}{253,125,55}
\definecolor{fpd16}{RGB}{252,107,50}
\definecolor{fpd17}{RGB}{253,121,54}
\definecolor{fpd18}{RGB}{252,100,48}
\definecolor{fpd19}{RGB}{252,109,51}
\definecolor{fpd20}{RGB}{253,120,54}
\definecolor{fpd21}{RGB}{252,81,43}
\definecolor{fpd22}{RGB}{252,96,47}
\definecolor{fpd23}{RGB}{253,119,54}
\definecolor{fpd24}{RGB}{252,90,45}
\definecolor{fpd25}{RGB}{252,108,51}
\definecolor{fpd26}{RGB}{253,122,55}
\definecolor{fpd27}{RGB}{252,86,44}
\definecolor{fpd28}{RGB}{253,111,51}
\definecolor{fpd29}{RGB}{252,87,45}
\definecolor{fpd30}{RGB}{252,105,50}
\definecolor{fpd31}{RGB}{255,241,170}
\definecolor{fpd32}{RGB}{255,228,142}
\definecolor{fpd33}{RGB}{255,251,194}
\definecolor{fpd34}{RGB}{255,239,165}
\definecolor{fpd35}{RGB}{255,227,139}
\definecolor{fpd36}{RGB}{254,220,125}
\definecolor{fpd37}{RGB}{254,212,112}
\definecolor{fpd38}{RGB}{255,235,155}
\definecolor{fpd39}{RGB}{254,180,79}
\definecolor{fpd40}{RGB}{253,159,68}
\definecolor{fpd41}{RGB}{255,232,150}
\definecolor{fpd42}{RGB}{254,170,72}
\definecolor{fpd43}{RGB}{252,97,48}
\definecolor{fpd44}{RGB}{254,176,75}
\definecolor{fpd45}{RGB}{252,78,42}
\caption{$W_1$ distance of the other Laplacian, with maps trained normalized and evaluated unnormalized (Appendix~\ref{sec:supp_other}), from the topological reference (top) and the zero reference (bottom). The maps lose their alignment and stay far from the zero reference.}
\label{tab:global_general_other_unnorm}
\centering
\small
\renewcommand{\arraystretch}{1.2}%
\begin{tabular}{l>{\centering\arraybackslash}p{2.7cm}>{\centering\arraybackslash}p{2.7cm}>{\centering\arraybackslash}p{2.7cm}l}
\toprule
\multicolumn{5}{@{}l}{\textbf{Other Laplacian: trained normalized, evaluated unnormalized}} \\
 & \multicolumn{3}{c}{Layer} & \\
\cmidrule(l){2-4}
Epoch & First & Middle & Last & \\
\midrule
\multicolumn{5}{@{}l}{\emph{against the topological reference}} \\
0 & \cellcolor{fpd12}$0.437_{\pm0.041}$ & \cellcolor{fpd13}$0.468_{\pm0.049}$ & \cellcolor{fpd14}$0.446_{\pm0.042}$ & \smash{\begin{tikzpicture}[baseline={([yshift=-9pt]current bounding box.north)}]
\fill[fpd0] (0,0.0000) rectangle (0.28,0.2292);
\fill[fpd1] (0,0.2292) rectangle (0.28,0.4583);
\fill[fpd2] (0,0.4583) rectangle (0.28,0.6875);
\fill[fpd3] (0,0.6875) rectangle (0.28,0.9167);
\fill[fpd4] (0,0.9167) rectangle (0.28,1.1458);
\fill[fpd5] (0,1.1458) rectangle (0.28,1.3750);
\fill[fpd6] (0,1.3750) rectangle (0.28,1.6042);
\fill[fpd7] (0,1.6042) rectangle (0.28,1.8333);
\fill[fpd8] (0,1.8333) rectangle (0.28,2.0625);
\fill[fpd9] (0,2.0625) rectangle (0.28,2.2917);
\fill[fpd10] (0,2.2917) rectangle (0.28,2.5208);
\fill[fpd11] (0,2.5208) rectangle (0.28,2.7500);
\draw[line width=0.3pt] (0,0) rectangle (0.28,2.75);
\node[right,font=\scriptsize] at (0.28,2.75) {0.52};
\node[right,font=\scriptsize] at (0.28,1.375) {0.26};
\node[right,font=\scriptsize] at (0.28,0) {0.00};
\end{tikzpicture}} \\
1 & \cellcolor{fpd15}$0.443_{\pm0.039}$ & \cellcolor{fpd16}$0.472_{\pm0.048}$ & \cellcolor{fpd10}$0.455_{\pm0.039}$ &  \\
5 & \cellcolor{fpd17}$0.449_{\pm0.037}$ & \cellcolor{fpd18}$0.483_{\pm0.043}$ & \cellcolor{fpd19}$0.469_{\pm0.034}$ &  \\
15 & \cellcolor{fpd20}$0.451_{\pm0.038}$ & \cellcolor{fpd21}$0.515_{\pm0.043}$ & \cellcolor{fpd22}$0.491_{\pm0.038}$ &  \\
50 & \cellcolor{fpd23}$0.452_{\pm0.038}$ & \cellcolor{fpd24}$0.500_{\pm0.042}$ & \cellcolor{fpd25}$0.470_{\pm0.040}$ &  \\
200 & \cellcolor{fpd26}$0.447_{\pm0.034}$ & \cellcolor{fpd27}$0.506_{\pm0.042}$ & \cellcolor{fpd28}$0.466_{\pm0.042}$ &  \\
final & \cellcolor{fpd20}$0.451_{\pm0.035}$ & \cellcolor{fpd29}$0.505_{\pm0.042}$ & \cellcolor{fpd30}$0.476_{\pm0.043}$ &  \\
\midrule
\multicolumn{5}{@{}l}{\emph{against the zero reference}} \\
0 & \cellcolor{fpd0}$0.172_{\pm0.143}$ & \cellcolor{fpd31}$0.650_{\pm0.597}$ & \cellcolor{fpd32}$1.203_{\pm1.128}$ & \smash{\begin{tikzpicture}[baseline={([yshift=-9pt]current bounding box.north)}]
\fill[fpd0] (0,0.0000) rectangle (0.28,0.2292);
\fill[fpd1] (0,0.2292) rectangle (0.28,0.4583);
\fill[fpd2] (0,0.4583) rectangle (0.28,0.6875);
\fill[fpd3] (0,0.6875) rectangle (0.28,0.9167);
\fill[fpd4] (0,0.9167) rectangle (0.28,1.1458);
\fill[fpd5] (0,1.1458) rectangle (0.28,1.3750);
\fill[fpd6] (0,1.3750) rectangle (0.28,1.6042);
\fill[fpd7] (0,1.6042) rectangle (0.28,1.8333);
\fill[fpd8] (0,1.8333) rectangle (0.28,2.0625);
\fill[fpd9] (0,2.0625) rectangle (0.28,2.2917);
\fill[fpd10] (0,2.2917) rectangle (0.28,2.5208);
\fill[fpd11] (0,2.5208) rectangle (0.28,2.7500);
\draw[line width=0.3pt] (0,0) rectangle (0.28,2.75);
\node[right,font=\scriptsize] at (0.28,2.75) {4.18};
\node[right,font=\scriptsize] at (0.28,1.375) {2.09};
\node[right,font=\scriptsize] at (0.28,0) {0.00};
\end{tikzpicture}} \\
1 & \cellcolor{fpd33}$0.188_{\pm0.145}$ & \cellcolor{fpd34}$0.735_{\pm0.648}$ & \cellcolor{fpd35}$1.249_{\pm1.128}$ &  \\
5 & \cellcolor{fpd1}$0.530_{\pm0.331}$ & \cellcolor{fpd36}$1.537_{\pm0.908}$ & \cellcolor{fpd37}$1.785_{\pm1.072}$ &  \\
15 & \cellcolor{fpd38}$0.933_{\pm0.540}$ & \cellcolor{fpd39}$2.458_{\pm1.259}$ & \cellcolor{fpd40}$2.943_{\pm1.615}$ &  \\
50 & \cellcolor{fpd41}$1.027_{\pm0.658}$ & \cellcolor{fpd42}$2.697_{\pm1.388}$ & \cellcolor{fpd43}$3.922_{\pm1.849}$ &  \\
200 & \cellcolor{fpd2}$0.875_{\pm0.672}$ & \cellcolor{fpd44}$2.548_{\pm1.593}$ & \cellcolor{fpd45}$4.179_{\pm2.132}$ &  \\
final & \cellcolor{fpd38}$0.938_{\pm0.705}$ & \cellcolor{fpd8}$2.968_{\pm1.895}$ & \cellcolor{fpd11}$4.006_{\pm1.898}$ &  \\
\bottomrule
\end{tabular}
\end{table}

\noindent\textbf{Cross-comparison: the two normalization regimes.}
For GeneralSheaf, dataset by dataset, we compare the Forman profile of the
maps trained normalized with that of the maps trained unnormalized, both
evaluated normalized (Table~\ref{tab:global_cross_a}) or both evaluated
unnormalized (Table~\ref{tab:global_cross_b}); no reference profile is
involved (Appendix~\ref{sec:supp_other}).
In both columns $W_1$ stays sustained across checkpoints, neither vanishing
nor blowing up, with values of the same size in the two (up to about $0.27$
and $0.23$).
The Forman profiles produced by the two training regimes therefore neither
converge to one another nor diverge, even where each aligns locally with the
same reference.

\begin{table}[htbp]
\definecolor{fpe0}{RGB}{255,251,195}
\definecolor{fpe1}{RGB}{255,244,176}
\definecolor{fpe2}{RGB}{255,236,158}
\definecolor{fpe3}{RGB}{255,228,141}
\definecolor{fpe4}{RGB}{254,220,123}
\definecolor{fpe5}{RGB}{254,206,106}
\definecolor{fpe6}{RGB}{254,189,88}
\definecolor{fpe7}{RGB}{254,173,74}
\definecolor{fpe8}{RGB}{253,158,67}
\definecolor{fpe9}{RGB}{253,143,61}
\definecolor{fpe10}{RGB}{253,117,53}
\definecolor{fpe11}{RGB}{252,91,46}
\definecolor{fpe12}{RGB}{254,225,135}
\definecolor{fpe13}{RGB}{254,221,127}
\definecolor{fpe14}{RGB}{254,178,76}
\definecolor{fpe15}{RGB}{254,210,110}
\definecolor{fpe16}{RGB}{254,185,84}
\definecolor{fpe17}{RGB}{253,146,62}
\definecolor{fpe18}{RGB}{254,195,94}
\definecolor{fpe19}{RGB}{253,138,59}
\definecolor{fpe20}{RGB}{253,150,64}
\definecolor{fpe21}{RGB}{253,155,66}
\definecolor{fpe22}{RGB}{253,142,61}
\definecolor{fpe23}{RGB}{252,96,47}
\definecolor{fpe24}{RGB}{254,181,79}
\definecolor{fpe25}{RGB}{253,127,56}
\definecolor{fpe26}{RGB}{252,80,43}
\definecolor{fpe27}{RGB}{253,152,65}
\definecolor{fpe28}{RGB}{253,113,52}
\definecolor{fpe29}{RGB}{253,149,64}
\definecolor{fpe30}{RGB}{253,132,57}
\caption{$W_1$ distance between the two Laplacians evaluated normalized, built from the maps trained normalized and unnormalized (cross-comparison, Appendix~\ref{sec:supp_other}). They stay apart and neither converge nor diverge.}
\label{tab:global_cross_a}
\centering
\small
\renewcommand{\arraystretch}{1.2}%
\begin{tabular}{l>{\centering\arraybackslash}p{2.7cm}>{\centering\arraybackslash}p{2.7cm}>{\centering\arraybackslash}p{2.7cm}l}
\toprule
\multicolumn{5}{@{}l}{\textbf{Cross-comparison, evaluated normalized}} \\
 & \multicolumn{3}{c}{Layer} & \\
\cmidrule(l){2-4}
Epoch & First & Middle & Last & \\
\midrule
0 & \cellcolor{fpe12}$0.086_{\pm0.029}$ & \cellcolor{fpe13}$0.096_{\pm0.014}$ & \cellcolor{fpe14}$0.162_{\pm0.037}$ & \smash{\begin{tikzpicture}[baseline={([yshift=-9pt]current bounding box.north)}]
\fill[fpe0] (0,0.0000) rectangle (0.28,0.2292);
\fill[fpe1] (0,0.2292) rectangle (0.28,0.4583);
\fill[fpe2] (0,0.4583) rectangle (0.28,0.6875);
\fill[fpe3] (0,0.6875) rectangle (0.28,0.9167);
\fill[fpe4] (0,0.9167) rectangle (0.28,1.1458);
\fill[fpe5] (0,1.1458) rectangle (0.28,1.3750);
\fill[fpe6] (0,1.3750) rectangle (0.28,1.6042);
\fill[fpe7] (0,1.6042) rectangle (0.28,1.8333);
\fill[fpe8] (0,1.8333) rectangle (0.28,2.0625);
\fill[fpe9] (0,2.0625) rectangle (0.28,2.2917);
\fill[fpe10] (0,2.2917) rectangle (0.28,2.5208);
\fill[fpe11] (0,2.5208) rectangle (0.28,2.7500);
\draw[line width=0.3pt] (0,0) rectangle (0.28,2.75);
\node[right,font=\scriptsize] at (0.28,2.75) {0.27};
\node[right,font=\scriptsize] at (0.28,1.375) {0.14};
\node[right,font=\scriptsize] at (0.28,0) {0.00};
\end{tikzpicture}} \\
1 & \cellcolor{fpe15}$0.118_{\pm0.025}$ & \cellcolor{fpe16}$0.152_{\pm0.028}$ & \cellcolor{fpe17}$0.209_{\pm0.052}$ &  \\
5 & \cellcolor{fpe18}$0.139_{\pm0.027}$ & \cellcolor{fpe19}$0.219_{\pm0.060}$ & \cellcolor{fpe20}$0.203_{\pm0.050}$ &  \\
15 & \cellcolor{fpe21}$0.195_{\pm0.027}$ & \cellcolor{fpe22}$0.214_{\pm0.030}$ & \cellcolor{fpe23}$0.255_{\pm0.038}$ &  \\
50 & \cellcolor{fpe24}$0.158_{\pm0.023}$ & \cellcolor{fpe25}$0.228_{\pm0.046}$ & \cellcolor{fpe26}$0.268_{\pm0.047}$ &  \\
200 & \cellcolor{fpe27}$0.200_{\pm0.028}$ & \cellcolor{fpe28}$0.240_{\pm0.038}$ & \cellcolor{fpe17}$0.208_{\pm0.035}$ &  \\
final & \cellcolor{fpe29}$0.204_{\pm0.028}$ & \cellcolor{fpe28}$0.240_{\pm0.038}$ & \cellcolor{fpe30}$0.224_{\pm0.043}$ &  \\
\bottomrule
\end{tabular}
\end{table}

\begin{table}[htbp]
\definecolor{fpf0}{RGB}{255,251,195}
\definecolor{fpf1}{RGB}{255,244,176}
\definecolor{fpf2}{RGB}{255,236,158}
\definecolor{fpf3}{RGB}{255,228,141}
\definecolor{fpf4}{RGB}{254,220,123}
\definecolor{fpf5}{RGB}{254,206,106}
\definecolor{fpf6}{RGB}{254,189,88}
\definecolor{fpf7}{RGB}{254,173,74}
\definecolor{fpf8}{RGB}{253,158,67}
\definecolor{fpf9}{RGB}{253,143,61}
\definecolor{fpf10}{RGB}{253,117,53}
\definecolor{fpf11}{RGB}{252,91,46}
\definecolor{fpf12}{RGB}{254,221,127}
\definecolor{fpf13}{RGB}{254,214,115}
\definecolor{fpf14}{RGB}{254,183,81}
\definecolor{fpf15}{RGB}{254,192,92}
\definecolor{fpf16}{RGB}{254,174,74}
\definecolor{fpf17}{RGB}{253,149,63}
\definecolor{fpf18}{RGB}{254,166,71}
\definecolor{fpf19}{RGB}{252,105,50}
\definecolor{fpf20}{RGB}{252,78,42}
\definecolor{fpf21}{RGB}{252,89,45}
\definecolor{fpf22}{RGB}{252,82,43}
\definecolor{fpf23}{RGB}{253,147,63}
\definecolor{fpf24}{RGB}{253,142,60}
\definecolor{fpf25}{RGB}{253,118,53}
\definecolor{fpf26}{RGB}{253,144,61}
\definecolor{fpf27}{RGB}{253,156,67}
\definecolor{fpf28}{RGB}{252,92,46}
\definecolor{fpf29}{RGB}{253,150,64}
\definecolor{fpf30}{RGB}{254,160,68}
\caption{As Table~\ref{tab:global_cross_a}, with both Laplacians evaluated unnormalized. The distances are of the same size.}
\label{tab:global_cross_b}
\centering
\small
\renewcommand{\arraystretch}{1.2}%
\begin{tabular}{l>{\centering\arraybackslash}p{2.7cm}>{\centering\arraybackslash}p{2.7cm}>{\centering\arraybackslash}p{2.7cm}l}
\toprule
\multicolumn{5}{@{}l}{\textbf{Cross-comparison, evaluated unnormalized}} \\
 & \multicolumn{3}{c}{Layer} & \\
\cmidrule(l){2-4}
Epoch & First & Middle & Last & \\
\midrule
0 & \cellcolor{fpf12}$0.082_{\pm0.030}$ & \cellcolor{fpf13}$0.095_{\pm0.015}$ & \cellcolor{fpf14}$0.132_{\pm0.025}$ & \smash{\begin{tikzpicture}[baseline={([yshift=-9pt]current bounding box.north)}]
\fill[fpf0] (0,0.0000) rectangle (0.28,0.2292);
\fill[fpf1] (0,0.2292) rectangle (0.28,0.4583);
\fill[fpf2] (0,0.4583) rectangle (0.28,0.6875);
\fill[fpf3] (0,0.6875) rectangle (0.28,0.9167);
\fill[fpf4] (0,0.9167) rectangle (0.28,1.1458);
\fill[fpf5] (0,1.1458) rectangle (0.28,1.3750);
\fill[fpf6] (0,1.3750) rectangle (0.28,1.6042);
\fill[fpf7] (0,1.6042) rectangle (0.28,1.8333);
\fill[fpf8] (0,1.8333) rectangle (0.28,2.0625);
\fill[fpf9] (0,2.0625) rectangle (0.28,2.2917);
\fill[fpf10] (0,2.2917) rectangle (0.28,2.5208);
\fill[fpf11] (0,2.5208) rectangle (0.28,2.7500);
\draw[line width=0.3pt] (0,0) rectangle (0.28,2.75);
\node[right,font=\scriptsize] at (0.28,2.75) {0.23};
\node[right,font=\scriptsize] at (0.28,1.375) {0.12};
\node[right,font=\scriptsize] at (0.28,0) {0.00};
\end{tikzpicture}} \\
1 & \cellcolor{fpf15}$0.121_{\pm0.026}$ & \cellcolor{fpf16}$0.143_{\pm0.035}$ & \cellcolor{fpf17}$0.174_{\pm0.033}$ &  \\
5 & \cellcolor{fpf18}$0.153_{\pm0.033}$ & \cellcolor{fpf19}$0.210_{\pm0.055}$ & \cellcolor{fpf20}$0.232_{\pm0.051}$ &  \\
15 & \cellcolor{fpf9}$0.181_{\pm0.038}$ & \cellcolor{fpf21}$0.222_{\pm0.036}$ & \cellcolor{fpf22}$0.227_{\pm0.021}$ &  \\
50 & \cellcolor{fpf23}$0.176_{\pm0.027}$ & \cellcolor{fpf24}$0.183_{\pm0.022}$ & \cellcolor{fpf25}$0.201_{\pm0.017}$ &  \\
200 & \cellcolor{fpf26}$0.180_{\pm0.030}$ & \cellcolor{fpf27}$0.165_{\pm0.014}$ & \cellcolor{fpf28}$0.220_{\pm0.029}$ &  \\
final & \cellcolor{fpf29}$0.173_{\pm0.025}$ & \cellcolor{fpf30}$0.160_{\pm0.012}$ & \cellcolor{fpf20}$0.232_{\pm0.037}$ &  \\
\bottomrule
\end{tabular}
\end{table}

\noindent\textbf{Joint diffusion reverses the role of normalization.}
Both joint regimes start at the identity sheaf, so their distances measure
deviation from the identity.
Unnormalized, the standard regime for this architecture, the small datasets
stay very close to the identity sheaf while on the larger ones the maps
deviate through the layers and are drawn back as training proceeds
(Table~\ref{tab:global_joint_unnorm}).
Normalized, the regime is less stable and needs a lower stalk dimension and
fewer layers; there the maps deviate largely from the identity as training
proceeds, which is the reverse of the unnormalized pattern
(Table~\ref{tab:global_joint_norm}).

\begin{table}[htbp]
\definecolor{fpg0}{RGB}{255,251,195}
\definecolor{fpg1}{RGB}{255,244,176}
\definecolor{fpg2}{RGB}{255,236,158}
\definecolor{fpg3}{RGB}{255,228,141}
\definecolor{fpg4}{RGB}{254,220,123}
\definecolor{fpg5}{RGB}{254,206,106}
\definecolor{fpg6}{RGB}{254,189,88}
\definecolor{fpg7}{RGB}{254,173,74}
\definecolor{fpg8}{RGB}{253,158,67}
\definecolor{fpg9}{RGB}{253,143,61}
\definecolor{fpg10}{RGB}{253,117,53}
\definecolor{fpg11}{RGB}{252,91,46}
\definecolor{fpg12}{RGB}{252,78,42}
\definecolor{fpg13}{RGB}{253,147,63}
\definecolor{fpg14}{RGB}{253,146,62}
\definecolor{fpg15}{RGB}{253,148,63}
\definecolor{fpg16}{RGB}{253,145,62}
\definecolor{fpg17}{RGB}{253,147,62}
\definecolor{fpg18}{RGB}{253,139,59}
\definecolor{fpg19}{RGB}{253,140,60}
\definecolor{fpg20}{RGB}{253,135,58}
\definecolor{fpg21}{RGB}{253,136,59}
\definecolor{fpg22}{RGB}{255,255,204}
\definecolor{fpg23}{RGB}{254,162,69}
\definecolor{fpg24}{RGB}{252,108,50}
\definecolor{fpg25}{RGB}{254,163,69}
\definecolor{fpg26}{RGB}{252,102,49}
\definecolor{fpg27}{RGB}{254,164,70}
\definecolor{fpg28}{RGB}{254,167,71}
\definecolor{fpg29}{RGB}{252,81,43}
\definecolor{fpg30}{RGB}{254,203,103}
\definecolor{fpg31}{RGB}{254,172,73}
\definecolor{fpg32}{RGB}{255,243,175}
\definecolor{fpg33}{RGB}{255,242,173}
\definecolor{fpg34}{RGB}{255,243,174}
\definecolor{fpg35}{RGB}{255,242,172}
\caption{$W_1$ distance of the unnormalized JointSheafParams, started at the identity sheaf, from the topological reference (top) and the zero reference (bottom), the latter measuring the deviation from the identity. The maps deviate through the layers and are drawn back with training.}
\label{tab:global_joint_unnorm}
\centering
\small
\renewcommand{\arraystretch}{1.2}%
\begin{tabular}{l>{\centering\arraybackslash}p{2.7cm}>{\centering\arraybackslash}p{2.7cm}>{\centering\arraybackslash}p{2.7cm}l}
\toprule
\multicolumn{5}{@{}l}{\textbf{JointSheafParams, unnormalized}} \\
 & \multicolumn{3}{c}{Layer} & \\
\cmidrule(l){2-4}
Epoch & First & Middle & Last & \\
\midrule
\multicolumn{5}{@{}l}{\emph{against the topological reference}} \\
0 & \cellcolor{fpg12}$0.581_{\pm0.046}$ & \cellcolor{fpg13}$0.444_{\pm0.058}$ & \cellcolor{fpg14}$0.448_{\pm0.053}$ & \smash{\begin{tikzpicture}[baseline={([yshift=-9pt]current bounding box.north)}]
\fill[fpg0] (0,0.0000) rectangle (0.28,0.2292);
\fill[fpg1] (0,0.2292) rectangle (0.28,0.4583);
\fill[fpg2] (0,0.4583) rectangle (0.28,0.6875);
\fill[fpg3] (0,0.6875) rectangle (0.28,0.9167);
\fill[fpg4] (0,0.9167) rectangle (0.28,1.1458);
\fill[fpg5] (0,1.1458) rectangle (0.28,1.3750);
\fill[fpg6] (0,1.3750) rectangle (0.28,1.6042);
\fill[fpg7] (0,1.6042) rectangle (0.28,1.8333);
\fill[fpg8] (0,1.8333) rectangle (0.28,2.0625);
\fill[fpg9] (0,2.0625) rectangle (0.28,2.2917);
\fill[fpg10] (0,2.2917) rectangle (0.28,2.5208);
\fill[fpg11] (0,2.5208) rectangle (0.28,2.7500);
\draw[line width=0.3pt] (0,0) rectangle (0.28,2.75);
\node[right,font=\scriptsize] at (0.28,2.75) {0.58};
\node[right,font=\scriptsize] at (0.28,1.375) {0.29};
\node[right,font=\scriptsize] at (0.28,0) {0.00};
\end{tikzpicture}} \\
1 & \cellcolor{fpg12}$0.581_{\pm0.046}$ & \cellcolor{fpg15}$0.443_{\pm0.056}$ & \cellcolor{fpg15}$0.442_{\pm0.053}$ &  \\
5 & \cellcolor{fpg12}$0.581_{\pm0.046}$ & \cellcolor{fpg14}$0.448_{\pm0.052}$ & \cellcolor{fpg16}$0.451_{\pm0.048}$ &  \\
15 & \cellcolor{fpg12}$0.581_{\pm0.046}$ & \cellcolor{fpg17}$0.446_{\pm0.052}$ & \cellcolor{fpg16}$0.452_{\pm0.042}$ &  \\
50 & \cellcolor{fpg12}$0.581_{\pm0.046}$ & \cellcolor{fpg16}$0.453_{\pm0.048}$ & \cellcolor{fpg13}$0.444_{\pm0.033}$ &  \\
200 & \cellcolor{fpg12}$0.581_{\pm0.046}$ & \cellcolor{fpg18}$0.468_{\pm0.037}$ & \cellcolor{fpg19}$0.465_{\pm0.027}$ &  \\
final & \cellcolor{fpg12}$0.581_{\pm0.046}$ & \cellcolor{fpg20}$0.475_{\pm0.031}$ & \cellcolor{fpg21}$0.473_{\pm0.028}$ &  \\
\midrule
\multicolumn{5}{@{}l}{\emph{against the zero reference}} \\
0 & \cellcolor{fpg22}$0.000_{\pm0.000}$ & \cellcolor{fpg23}$2.760_{\pm1.584}$ & \cellcolor{fpg24}$3.632_{\pm1.811}$ & \smash{\begin{tikzpicture}[baseline={([yshift=-9pt]current bounding box.north)}]
\fill[fpg0] (0,0.0000) rectangle (0.28,0.2292);
\fill[fpg1] (0,0.2292) rectangle (0.28,0.4583);
\fill[fpg2] (0,0.4583) rectangle (0.28,0.6875);
\fill[fpg3] (0,0.6875) rectangle (0.28,0.9167);
\fill[fpg4] (0,0.9167) rectangle (0.28,1.1458);
\fill[fpg5] (0,1.1458) rectangle (0.28,1.3750);
\fill[fpg6] (0,1.3750) rectangle (0.28,1.6042);
\fill[fpg7] (0,1.6042) rectangle (0.28,1.8333);
\fill[fpg8] (0,1.8333) rectangle (0.28,2.0625);
\fill[fpg9] (0,2.0625) rectangle (0.28,2.2917);
\fill[fpg10] (0,2.2917) rectangle (0.28,2.5208);
\fill[fpg11] (0,2.5208) rectangle (0.28,2.7500);
\draw[line width=0.3pt] (0,0) rectangle (0.28,2.75);
\node[right,font=\scriptsize] at (0.28,2.75) {4.01};
\node[right,font=\scriptsize] at (0.28,1.375) {2.00};
\node[right,font=\scriptsize] at (0.28,0) {0.00};
\end{tikzpicture}} \\
1 & \cellcolor{fpg22}$0.000_{\pm0.000}$ & \cellcolor{fpg25}$2.737_{\pm1.553}$ & \cellcolor{fpg26}$3.700_{\pm1.808}$ &  \\
5 & \cellcolor{fpg22}$0.000_{\pm0.000}$ & \cellcolor{fpg27}$2.718_{\pm1.518}$ & \cellcolor{fpg12}$4.011_{\pm1.973}$ &  \\
15 & \cellcolor{fpg22}$0.000_{\pm0.000}$ & \cellcolor{fpg28}$2.655_{\pm1.480}$ & \cellcolor{fpg29}$3.967_{\pm2.068}$ &  \\
50 & \cellcolor{fpg22}$0.000_{\pm0.000}$ & \cellcolor{fpg30}$1.892_{\pm1.027}$ & \cellcolor{fpg31}$2.545_{\pm1.468}$ &  \\
200 & \cellcolor{fpg22}$0.000_{\pm0.000}$ & \cellcolor{fpg32}$0.523_{\pm0.211}$ & \cellcolor{fpg33}$0.574_{\pm0.252}$ &  \\
final & \cellcolor{fpg22}$0.000_{\pm0.000}$ & \cellcolor{fpg34}$0.550_{\pm0.284}$ & \cellcolor{fpg35}$0.585_{\pm0.317}$ &  \\
\bottomrule
\end{tabular}
\end{table}

\begin{table}[htbp]
\definecolor{fph0}{RGB}{255,251,195}
\definecolor{fph1}{RGB}{255,244,176}
\definecolor{fph2}{RGB}{255,236,158}
\definecolor{fph3}{RGB}{255,228,141}
\definecolor{fph4}{RGB}{254,220,123}
\definecolor{fph5}{RGB}{254,206,106}
\definecolor{fph6}{RGB}{254,189,88}
\definecolor{fph7}{RGB}{254,173,74}
\definecolor{fph8}{RGB}{253,158,67}
\definecolor{fph9}{RGB}{253,143,61}
\definecolor{fph10}{RGB}{253,117,53}
\definecolor{fph11}{RGB}{252,91,46}
\definecolor{fph12}{RGB}{255,255,204}
\definecolor{fph13}{RGB}{255,227,140}
\definecolor{fph14}{RGB}{254,202,102}
\definecolor{fph15}{RGB}{254,226,137}
\definecolor{fph16}{RGB}{254,197,96}
\definecolor{fph17}{RGB}{254,222,129}
\definecolor{fph18}{RGB}{254,190,89}
\definecolor{fph19}{RGB}{254,211,111}
\definecolor{fph20}{RGB}{254,171,73}
\definecolor{fph21}{RGB}{254,201,100}
\definecolor{fph22}{RGB}{253,134,58}
\definecolor{fph23}{RGB}{254,199,99}
\definecolor{fph24}{RGB}{252,78,42}
\definecolor{fph25}{RGB}{254,187,86}
\caption{$W_1$ distance of the normalized JointSheafParams, started at the identity sheaf, from the topological reference, which here measures the deviation from the identity. The deviation grows with training. Tolokers is excluded since its training is unstable ($9$ datasets, $7$ for the middle layer). On Cora, PubMed and Questions training was interrupted, and these datasets enter with the folds that completed (between two and six of ten).}
\label{tab:global_joint_norm}
\centering
\small
\renewcommand{\arraystretch}{1.2}%
\begin{tabular}{l>{\centering\arraybackslash}p{2.7cm}>{\centering\arraybackslash}p{2.7cm}>{\centering\arraybackslash}p{2.7cm}l}
\toprule
\multicolumn{5}{@{}l}{\textbf{JointSheafParams, normalized}} \\
 & \multicolumn{3}{c}{Layer} & \\
\cmidrule(l){2-4}
Epoch & First & Middle & Last & \\
\midrule
0 & \cellcolor{fph12}$0.000_{\pm0.000}$ & \cellcolor{fph13}$0.124_{\pm0.098}$ & \cellcolor{fph14}$0.201_{\pm0.078}$ & \smash{\begin{tikzpicture}[baseline={([yshift=-9pt]current bounding box.north)}]
\fill[fph0] (0,0.0000) rectangle (0.28,0.2292);
\fill[fph1] (0,0.2292) rectangle (0.28,0.4583);
\fill[fph2] (0,0.4583) rectangle (0.28,0.6875);
\fill[fph3] (0,0.6875) rectangle (0.28,0.9167);
\fill[fph4] (0,0.9167) rectangle (0.28,1.1458);
\fill[fph5] (0,1.1458) rectangle (0.28,1.3750);
\fill[fph6] (0,1.3750) rectangle (0.28,1.6042);
\fill[fph7] (0,1.6042) rectangle (0.28,1.8333);
\fill[fph8] (0,1.8333) rectangle (0.28,2.0625);
\fill[fph9] (0,2.0625) rectangle (0.28,2.2917);
\fill[fph10] (0,2.2917) rectangle (0.28,2.5208);
\fill[fph11] (0,2.5208) rectangle (0.28,2.7500);
\draw[line width=0.3pt] (0,0) rectangle (0.28,2.75);
\node[right,font=\scriptsize] at (0.28,2.75) {0.42};
\node[right,font=\scriptsize] at (0.28,1.375) {0.21};
\node[right,font=\scriptsize] at (0.28,0) {0.00};
\end{tikzpicture}} \\
1 & \cellcolor{fph12}$0.000_{\pm0.000}$ & \cellcolor{fph15}$0.130_{\pm0.097}$ & \cellcolor{fph16}$0.211_{\pm0.077}$ &  \\
5 & \cellcolor{fph12}$0.000_{\pm0.000}$ & \cellcolor{fph17}$0.146_{\pm0.094}$ & \cellcolor{fph18}$0.226_{\pm0.075}$ &  \\
15 & \cellcolor{fph12}$0.000_{\pm0.000}$ & \cellcolor{fph19}$0.181_{\pm0.090}$ & \cellcolor{fph20}$0.267_{\pm0.068}$ &  \\
50 & \cellcolor{fph12}$0.000_{\pm0.000}$ & \cellcolor{fph21}$0.203_{\pm0.084}$ & \cellcolor{fph22}$0.346_{\pm0.052}$ &  \\
200 & \cellcolor{fph12}$0.000_{\pm0.000}$ & \cellcolor{fph23}$0.207_{\pm0.076}$ & \cellcolor{fph24}$0.421_{\pm0.039}$ &  \\
final & \cellcolor{fph12}$0.000_{\pm0.000}$ & \cellcolor{fph25}$0.232_{\pm0.076}$ & \cellcolor{fph24}$0.422_{\pm0.033}$ &  \\
\bottomrule
\end{tabular}
\end{table}

\FloatBarrier

\subsection{Per-dataset tables}
\label{sec:supp_forman_figures}

The tables above average across datasets and report $W_1$ only.
Tables~\ref{tab:perds_rm_norm}--\ref{tab:perds_questions_joint_unnorm} give the
per-dataset view, one representative dataset per regime, under the metrics
of Appendix~\ref{sec:supp_metrics}.
Each table has one block per metric, with training epochs on the rows and
layers on the columns; each cell is averaged over the eigenvalues of the
Forman term and given as mean$_{\pm\text{s.e.}}$ across the ten folds of that dataset.
Correlations are on a blue--red scale centred at zero, distances on a
yellow--red scale, and a blank cell marks a correlation that is undefined
because the profile is constant.
The remaining datasets show similar patterns, with the exceptions discussed
in Section~\ref{sec:forman_profile}.

Table~\ref{tab:perds_rm_norm} is the per-dataset counterpart of
Table~\ref{tab:global_general_norm}: on Roman Empire, alignment with the
topological reference strengthens with epochs and, at fixed epoch, with
depth.
Table~\ref{tab:perds_citeseer_unnorm} is the counterpart of
Table~\ref{tab:global_general_unnorm}: on Citeseer, the unnormalized profile
matches neither reference.
Table~\ref{tab:perds_questions_joint_unnorm} is the counterpart of
Table~\ref{tab:global_joint_unnorm}: on Questions, the maps deviate from the
identity through the layers, and the deviation is attenuated as training
proceeds.

\begin{table}[p]
\definecolor{pda0}{RGB}{59,136,190}
\definecolor{pda1}{RGB}{95,165,205}
\definecolor{pda2}{RGB}{141,194,220}
\definecolor{pda3}{RGB}{179,214,232}
\definecolor{pda4}{RGB}{213,231,241}
\definecolor{pda5}{RGB}{236,242,245}
\definecolor{pda6}{RGB}{249,239,233}
\definecolor{pda7}{RGB}{252,222,205}
\definecolor{pda8}{RGB}{249,194,167}
\definecolor{pda9}{RGB}{242,161,127}
\definecolor{pda10}{RGB}{225,120,96}
\definecolor{pda11}{RGB}{206,79,69}
\definecolor{pda12}{RGB}{233,240,244}
\definecolor{pda13}{RGB}{230,239,244}
\definecolor{pda14}{RGB}{234,241,245}
\definecolor{pda15}{RGB}{237,242,245}
\definecolor{pda16}{RGB}{242,245,246}
\definecolor{pda17}{RGB}{239,243,245}
\definecolor{pda18}{RGB}{245,246,247}
\definecolor{pda19}{RGB}{251,228,214}
\definecolor{pda20}{RGB}{252,213,191}
\definecolor{pda21}{RGB}{251,206,183}
\definecolor{pda22}{RGB}{250,233,223}
\definecolor{pda23}{RGB}{252,223,207}
\definecolor{pda24}{RGB}{249,196,169}
\definecolor{pda25}{RGB}{246,175,142}
\definecolor{pda26}{RGB}{240,156,123}
\definecolor{pda27}{RGB}{252,226,210}
\definecolor{pda28}{RGB}{248,189,161}
\definecolor{pda29}{RGB}{241,158,125}
\definecolor{pda30}{RGB}{234,142,112}
\definecolor{pda31}{RGB}{228,128,102}
\definecolor{pda32}{RGB}{247,185,156}
\definecolor{pda33}{RGB}{233,139,110}
\definecolor{pda34}{RGB}{229,131,104}
\definecolor{pda35}{RGB}{247,181,150}
\definecolor{pda36}{RGB}{245,172,139}
\definecolor{pda37}{RGB}{240,244,246}
\definecolor{pda38}{RGB}{243,245,246}
\definecolor{pda39}{RGB}{247,246,246}
\definecolor{pda40}{RGB}{250,232,222}
\definecolor{pda41}{RGB}{251,227,212}
\definecolor{pda42}{RGB}{253,219,199}
\definecolor{pda43}{RGB}{252,211,188}
\definecolor{pda44}{RGB}{252,224,208}
\definecolor{pda45}{RGB}{239,153,121}
\definecolor{pda46}{RGB}{248,187,158}
\definecolor{pda47}{RGB}{230,134,106}
\definecolor{pda48}{RGB}{251,204,180}
\definecolor{pda49}{RGB}{246,179,148}
\definecolor{pda50}{RGB}{238,150,119}
\definecolor{pda51}{RGB}{227,126,100}
\definecolor{pda52}{RGB}{244,166,131}
\definecolor{pda53}{RGB}{236,147,116}
\definecolor{pda54}{RGB}{255,251,195}
\definecolor{pda55}{RGB}{255,244,176}
\definecolor{pda56}{RGB}{255,236,158}
\definecolor{pda57}{RGB}{255,228,141}
\definecolor{pda58}{RGB}{254,220,123}
\definecolor{pda59}{RGB}{254,206,106}
\definecolor{pda60}{RGB}{254,189,88}
\definecolor{pda61}{RGB}{254,173,74}
\definecolor{pda62}{RGB}{253,158,67}
\definecolor{pda63}{RGB}{253,143,61}
\definecolor{pda64}{RGB}{253,117,53}
\definecolor{pda65}{RGB}{252,91,46}
\definecolor{pda66}{RGB}{254,214,115}
\definecolor{pda67}{RGB}{254,213,113}
\definecolor{pda68}{RGB}{254,209,109}
\definecolor{pda69}{RGB}{254,212,112}
\definecolor{pda70}{RGB}{254,208,108}
\definecolor{pda71}{RGB}{254,208,109}
\definecolor{pda72}{RGB}{254,207,107}
\definecolor{pda73}{RGB}{254,207,108}
\definecolor{pda74}{RGB}{254,224,133}
\definecolor{pda75}{RGB}{254,226,136}
\definecolor{pda76}{RGB}{255,230,145}
\definecolor{pda77}{RGB}{255,235,155}
\definecolor{pda78}{RGB}{255,237,161}
\definecolor{pda79}{RGB}{255,238,162}
\definecolor{pda80}{RGB}{255,239,166}
\definecolor{pda81}{RGB}{255,243,174}
\definecolor{pda82}{RGB}{255,244,177}
\definecolor{pda83}{RGB}{255,240,167}
\definecolor{pda84}{RGB}{255,243,176}
\definecolor{pda85}{RGB}{255,244,178}
\definecolor{pda86}{RGB}{255,245,179}
\definecolor{pda87}{RGB}{255,237,160}
\definecolor{pda88}{RGB}{255,242,172}
\definecolor{pda89}{RGB}{255,242,173}
\definecolor{pda90}{RGB}{255,237,159}
\definecolor{pda91}{RGB}{255,245,178}
\definecolor{pda92}{RGB}{255,239,164}
\definecolor{pda93}{RGB}{255,238,163}
\definecolor{pda94}{RGB}{255,230,146}
\definecolor{pda95}{RGB}{255,228,142}
\definecolor{pda96}{RGB}{254,226,137}
\definecolor{pda97}{RGB}{254,225,136}
\definecolor{pda98}{RGB}{254,225,135}
\definecolor{pda99}{RGB}{254,222,129}
\definecolor{pda100}{RGB}{254,222,128}
\definecolor{pda101}{RGB}{254,223,130}
\definecolor{pda102}{RGB}{254,221,125}
\definecolor{pda103}{RGB}{254,219,122}
\definecolor{pda104}{RGB}{255,228,140}
\definecolor{pda105}{RGB}{254,225,134}
\definecolor{pda106}{RGB}{254,221,126}
\definecolor{pda107}{RGB}{254,215,116}
\definecolor{pda108}{RGB}{254,218,121}
\definecolor{pda109}{RGB}{254,216,117}
\definecolor{pda110}{RGB}{253,134,58}
\definecolor{pda111}{RGB}{253,135,58}
\definecolor{pda112}{RGB}{253,136,58}
\definecolor{pda113}{RGB}{253,140,60}
\definecolor{pda114}{RGB}{253,141,60}
\definecolor{pda115}{RGB}{253,142,61}
\definecolor{pda116}{RGB}{253,144,61}
\definecolor{pda117}{RGB}{253,156,66}
\definecolor{pda118}{RGB}{253,155,66}
\definecolor{pda119}{RGB}{254,160,68}
\definecolor{pda120}{RGB}{254,167,71}
\definecolor{pda121}{RGB}{254,171,73}
\definecolor{pda122}{RGB}{253,154,66}
\definecolor{pda123}{RGB}{254,162,69}
\definecolor{pda124}{RGB}{254,177,76}
\definecolor{pda125}{RGB}{254,190,89}
\definecolor{pda126}{RGB}{254,202,101}
\definecolor{pda127}{RGB}{254,183,81}
\definecolor{pda128}{RGB}{254,201,101}
\definecolor{pda129}{RGB}{254,211,111}
\definecolor{pda130}{RGB}{254,219,121}
\definecolor{pda131}{RGB}{254,174,74}
\definecolor{pda132}{RGB}{254,188,86}
\definecolor{pda133}{RGB}{254,204,104}
\definecolor{pda134}{RGB}{254,217,119}
\definecolor{pda135}{RGB}{254,186,84}
\definecolor{pda136}{RGB}{254,195,94}
\definecolor{pda137}{RGB}{254,205,105}
\definecolor{pda138}{RGB}{254,218,120}
\caption{Normalized GeneralSheaf on Roman Empire, $d=5$, five layers, against the topological reference.}
\label{tab:perds_rm_norm}
\centering
\footnotesize
\renewcommand{\arraystretch}{1.15}%
\setlength{\perdsbar}{7\baselineskip}\setlength{\perdsbar}{1.15\perdsbar}%
\setlength{\perdsrowtop}{0.7\baselineskip}\setlength{\perdsrowtop}{1.15\perdsrowtop}%
\begin{tabular}{l>{\centering\arraybackslash}p{1.55cm}>{\centering\arraybackslash}p{1.55cm}>{\centering\arraybackslash}p{1.55cm}>{\centering\arraybackslash}p{1.55cm}>{\centering\arraybackslash}p{1.55cm}l}
\toprule
\multicolumn{7}{@{}l}{\textbf{GeneralSheaf, normalized: Roman Empire}} \\
 & \multicolumn{5}{c}{Layer} & \\
\cmidrule(l){2-6}
Epoch & 1 & 2 & 3 & 4 & 5 & \\
\midrule
\multicolumn{7}{@{}l}{\emph{Spearman $\rho$}} \\
0 & \cellcolor{pda12}$-0.11_{\pm0.01}$ & \cellcolor{pda13}$-0.13_{\pm0.02}$ & \cellcolor{pda14}$-0.09_{\pm0.03}$ & \cellcolor{pda12}$-0.11_{\pm0.02}$ & \cellcolor{pda15}$-0.07_{\pm0.02}$ & \smash{\begin{tikzpicture}[y=\perdsbar,baseline={([yshift=-\perdsrowtop]current bounding box.north)}]
\fill[pda0] (0,0.0000) rectangle (0.18,0.0833);
\fill[pda1] (0,0.0833) rectangle (0.18,0.1667);
\fill[pda2] (0,0.1667) rectangle (0.18,0.2500);
\fill[pda3] (0,0.2500) rectangle (0.18,0.3333);
\fill[pda4] (0,0.3333) rectangle (0.18,0.4167);
\fill[pda5] (0,0.4167) rectangle (0.18,0.5000);
\fill[pda6] (0,0.5000) rectangle (0.18,0.5833);
\fill[pda7] (0,0.5833) rectangle (0.18,0.6667);
\fill[pda8] (0,0.6667) rectangle (0.18,0.7500);
\fill[pda9] (0,0.7500) rectangle (0.18,0.8333);
\fill[pda10] (0,0.8333) rectangle (0.18,0.9167);
\fill[pda11] (0,0.9167) rectangle (0.18,1.0000);
\draw[line width=0.3pt] (0,0) rectangle (0.18,1);
\node[anchor=north west,font=\tiny,inner sep=1pt] at (0.18,1) {1.00};
\node[anchor=west,font=\tiny,inner sep=1pt] at (0.18,0.5) {0.00};
\node[anchor=south west,font=\tiny,inner sep=1pt] at (0.18,0) {-1.00};
\end{tikzpicture}} \\
1 & \cellcolor{pda16}$-0.04_{\pm0.03}$ & \cellcolor{pda17}$-0.06_{\pm0.02}$ & \cellcolor{pda16}$-0.04_{\pm0.03}$ & \cellcolor{pda18}$-0.02_{\pm0.02}$ & \cellcolor{pda16}$-0.04_{\pm0.02}$ &  \\
5 & \cellcolor{pda19}$0.20_{\pm0.05}$ & \cellcolor{pda19}$0.19_{\pm0.03}$ & \cellcolor{pda7}$0.25_{\pm0.02}$ & \cellcolor{pda20}$0.31_{\pm0.04}$ & \cellcolor{pda21}$0.35_{\pm0.02}$ &  \\
15 & \cellcolor{pda22}$0.13_{\pm0.02}$ & \cellcolor{pda23}$0.24_{\pm0.02}$ & \cellcolor{pda24}$0.41_{\pm0.02}$ & \cellcolor{pda25}$0.52_{\pm0.02}$ & \cellcolor{pda26}$0.61_{\pm0.02}$ &  \\
50 & \cellcolor{pda27}$0.21_{\pm0.03}$ & \cellcolor{pda28}$0.44_{\pm0.03}$ & \cellcolor{pda29}$0.59_{\pm0.02}$ & \cellcolor{pda30}$0.66_{\pm0.02}$ & \cellcolor{pda31}$0.72_{\pm0.02}$ &  \\
200 & \cellcolor{pda21}$0.35_{\pm0.02}$ & \cellcolor{pda32}$0.46_{\pm0.04}$ & \cellcolor{pda29}$0.59_{\pm0.02}$ & \cellcolor{pda33}$0.67_{\pm0.02}$ & \cellcolor{pda34}$0.70_{\pm0.02}$ &  \\
final & \cellcolor{pda35}$0.49_{\pm0.03}$ & \cellcolor{pda36}$0.53_{\pm0.03}$ & \cellcolor{pda9}$0.59_{\pm0.02}$ & \cellcolor{pda33}$0.68_{\pm0.02}$ & \cellcolor{pda34}$0.70_{\pm0.02}$ &  \\
\midrule
\multicolumn{7}{@{}l}{\emph{Pearson $r$}} \\
0 & \cellcolor{pda12}$-0.10_{\pm0.01}$ & \cellcolor{pda12}$-0.10_{\pm0.01}$ & \cellcolor{pda14}$-0.09_{\pm0.02}$ & \cellcolor{pda5}$-0.09_{\pm0.02}$ & \cellcolor{pda37}$-0.05_{\pm0.02}$ & \smash{\begin{tikzpicture}[y=\perdsbar,baseline={([yshift=-\perdsrowtop]current bounding box.north)}]
\fill[pda0] (0,0.0000) rectangle (0.18,0.0833);
\fill[pda1] (0,0.0833) rectangle (0.18,0.1667);
\fill[pda2] (0,0.1667) rectangle (0.18,0.2500);
\fill[pda3] (0,0.2500) rectangle (0.18,0.3333);
\fill[pda4] (0,0.3333) rectangle (0.18,0.4167);
\fill[pda5] (0,0.4167) rectangle (0.18,0.5000);
\fill[pda6] (0,0.5000) rectangle (0.18,0.5833);
\fill[pda7] (0,0.5833) rectangle (0.18,0.6667);
\fill[pda8] (0,0.6667) rectangle (0.18,0.7500);
\fill[pda9] (0,0.7500) rectangle (0.18,0.8333);
\fill[pda10] (0,0.8333) rectangle (0.18,0.9167);
\fill[pda11] (0,0.9167) rectangle (0.18,1.0000);
\draw[line width=0.3pt] (0,0) rectangle (0.18,1);
\node[anchor=north west,font=\tiny,inner sep=1pt] at (0.18,1) {1.00};
\node[anchor=west,font=\tiny,inner sep=1pt] at (0.18,0.5) {0.00};
\node[anchor=south west,font=\tiny,inner sep=1pt] at (0.18,0) {-1.00};
\end{tikzpicture}} \\
1 & \cellcolor{pda18}$-0.02_{\pm0.03}$ & \cellcolor{pda16}$-0.04_{\pm0.03}$ & \cellcolor{pda38}$-0.03_{\pm0.03}$ & \cellcolor{pda39}$0.00_{\pm0.02}$ & \cellcolor{pda38}$-0.03_{\pm0.02}$ &  \\
5 & \cellcolor{pda40}$0.15_{\pm0.04}$ & \cellcolor{pda40}$0.15_{\pm0.03}$ & \cellcolor{pda41}$0.21_{\pm0.02}$ & \cellcolor{pda42}$0.28_{\pm0.04}$ & \cellcolor{pda43}$0.33_{\pm0.02}$ &  \\
15 & \cellcolor{pda22}$0.14_{\pm0.02}$ & \cellcolor{pda44}$0.23_{\pm0.02}$ & \cellcolor{pda24}$0.40_{\pm0.02}$ & \cellcolor{pda36}$0.53_{\pm0.03}$ & \cellcolor{pda45}$0.62_{\pm0.02}$ &  \\
50 & \cellcolor{pda23}$0.23_{\pm0.03}$ & \cellcolor{pda46}$0.46_{\pm0.03}$ & \cellcolor{pda45}$0.62_{\pm0.02}$ & \cellcolor{pda47}$0.69_{\pm0.02}$ & \cellcolor{pda10}$0.75_{\pm0.02}$ &  \\
200 & \cellcolor{pda48}$0.37_{\pm0.02}$ & \cellcolor{pda49}$0.50_{\pm0.04}$ & \cellcolor{pda50}$0.64_{\pm0.02}$ & \cellcolor{pda34}$0.71_{\pm0.02}$ & \cellcolor{pda51}$0.73_{\pm0.02}$ &  \\
final & \cellcolor{pda35}$0.49_{\pm0.03}$ & \cellcolor{pda52}$0.56_{\pm0.03}$ & \cellcolor{pda53}$0.64_{\pm0.02}$ & \cellcolor{pda31}$0.72_{\pm0.02}$ & \cellcolor{pda51}$0.74_{\pm0.02}$ &  \\
\midrule
\multicolumn{7}{@{}l}{\emph{$W_1$}} \\
0 & \cellcolor{pda66}$0.37_{\pm0.01}$ & \cellcolor{pda66}$0.37_{\pm0.01}$ & \cellcolor{pda67}$0.38_{\pm0.00}$ & \cellcolor{pda68}$0.40_{\pm0.01}$ & \cellcolor{pda69}$0.38_{\pm0.01}$ & \smash{\begin{tikzpicture}[y=\perdsbar,baseline={([yshift=-\perdsrowtop]current bounding box.north)}]
\fill[pda54] (0,0.0000) rectangle (0.18,0.0833);
\fill[pda55] (0,0.0833) rectangle (0.18,0.1667);
\fill[pda56] (0,0.1667) rectangle (0.18,0.2500);
\fill[pda57] (0,0.2500) rectangle (0.18,0.3333);
\fill[pda58] (0,0.3333) rectangle (0.18,0.4167);
\fill[pda59] (0,0.4167) rectangle (0.18,0.5000);
\fill[pda60] (0,0.5000) rectangle (0.18,0.5833);
\fill[pda61] (0,0.5833) rectangle (0.18,0.6667);
\fill[pda62] (0,0.6667) rectangle (0.18,0.7500);
\fill[pda63] (0,0.7500) rectangle (0.18,0.8333);
\fill[pda64] (0,0.8333) rectangle (0.18,0.9167);
\fill[pda65] (0,0.9167) rectangle (0.18,1.0000);
\draw[line width=0.3pt] (0,0) rectangle (0.18,1);
\node[anchor=north west,font=\tiny,inner sep=1pt] at (0.18,1) {0.90};
\node[anchor=west,font=\tiny,inner sep=1pt] at (0.18,0.5) {0.45};
\node[anchor=south west,font=\tiny,inner sep=1pt] at (0.18,0) {0.00};
\end{tikzpicture}} \\
1 & \cellcolor{pda68}$0.40_{\pm0.01}$ & \cellcolor{pda70}$0.40_{\pm0.01}$ & \cellcolor{pda71}$0.40_{\pm0.02}$ & \cellcolor{pda72}$0.41_{\pm0.01}$ & \cellcolor{pda73}$0.40_{\pm0.01}$ &  \\
5 & \cellcolor{pda74}$0.29_{\pm0.01}$ & \cellcolor{pda75}$0.28_{\pm0.02}$ & \cellcolor{pda76}$0.24_{\pm0.01}$ & \cellcolor{pda77}$0.20_{\pm0.01}$ & \cellcolor{pda78}$0.18_{\pm0.01}$ &  \\
15 & \cellcolor{pda79}$0.17_{\pm0.01}$ & \cellcolor{pda80}$0.16_{\pm0.01}$ & \cellcolor{pda81}$0.12_{\pm0.01}$ & \cellcolor{pda82}$0.11_{\pm0.01}$ & \cellcolor{pda82}$0.11_{\pm0.00}$ &  \\
50 & \cellcolor{pda83}$0.15_{\pm0.00}$ & \cellcolor{pda81}$0.12_{\pm0.01}$ & \cellcolor{pda84}$0.12_{\pm0.00}$ & \cellcolor{pda85}$0.11_{\pm0.00}$ & \cellcolor{pda86}$0.10_{\pm0.01}$ &  \\
200 & \cellcolor{pda87}$0.18_{\pm0.01}$ & \cellcolor{pda88}$0.13_{\pm0.01}$ & \cellcolor{pda89}$0.13_{\pm0.01}$ & \cellcolor{pda86}$0.10_{\pm0.01}$ & \cellcolor{pda85}$0.11_{\pm0.01}$ &  \\
final & \cellcolor{pda90}$0.18_{\pm0.01}$ & \cellcolor{pda88}$0.13_{\pm0.01}$ & \cellcolor{pda89}$0.13_{\pm0.01}$ & \cellcolor{pda91}$0.10_{\pm0.01}$ & \cellcolor{pda85}$0.11_{\pm0.01}$ &  \\
\midrule
\multicolumn{7}{@{}l}{\emph{RMSE}} \\
0 & \cellcolor{pda92}$0.11_{\pm0.00}$ & \cellcolor{pda92}$0.11_{\pm0.00}$ & \cellcolor{pda92}$0.11_{\pm0.00}$ & \cellcolor{pda92}$0.11_{\pm0.00}$ & \cellcolor{pda92}$0.11_{\pm0.00}$ & \smash{\begin{tikzpicture}[y=\perdsbar,baseline={([yshift=-\perdsrowtop]current bounding box.north)}]
\fill[pda54] (0,0.0000) rectangle (0.18,0.0833);
\fill[pda55] (0,0.0833) rectangle (0.18,0.1667);
\fill[pda56] (0,0.1667) rectangle (0.18,0.2500);
\fill[pda57] (0,0.2500) rectangle (0.18,0.3333);
\fill[pda58] (0,0.3333) rectangle (0.18,0.4167);
\fill[pda59] (0,0.4167) rectangle (0.18,0.5000);
\fill[pda60] (0,0.5000) rectangle (0.18,0.5833);
\fill[pda61] (0,0.5833) rectangle (0.18,0.6667);
\fill[pda62] (0,0.6667) rectangle (0.18,0.7500);
\fill[pda63] (0,0.7500) rectangle (0.18,0.8333);
\fill[pda64] (0,0.8333) rectangle (0.18,0.9167);
\fill[pda65] (0,0.9167) rectangle (0.18,1.0000);
\draw[line width=0.3pt] (0,0) rectangle (0.18,1);
\node[anchor=north west,font=\tiny,inner sep=1pt] at (0.18,1) {0.60};
\node[anchor=west,font=\tiny,inner sep=1pt] at (0.18,0.5) {0.30};
\node[anchor=south west,font=\tiny,inner sep=1pt] at (0.18,0) {0.00};
\end{tikzpicture}} \\
1 & \cellcolor{pda93}$0.11_{\pm0.00}$ & \cellcolor{pda93}$0.11_{\pm0.00}$ & \cellcolor{pda93}$0.11_{\pm0.00}$ & \cellcolor{pda93}$0.11_{\pm0.00}$ & \cellcolor{pda93}$0.11_{\pm0.00}$ &  \\
5 & \cellcolor{pda94}$0.16_{\pm0.00}$ & \cellcolor{pda76}$0.16_{\pm0.01}$ & \cellcolor{pda95}$0.17_{\pm0.01}$ & \cellcolor{pda96}$0.19_{\pm0.00}$ & \cellcolor{pda97}$0.19_{\pm0.00}$ &  \\
15 & \cellcolor{pda98}$0.19_{\pm0.00}$ & \cellcolor{pda74}$0.20_{\pm0.00}$ & \cellcolor{pda74}$0.20_{\pm0.00}$ & \cellcolor{pda99}$0.21_{\pm0.00}$ & \cellcolor{pda100}$0.21_{\pm0.00}$ &  \\
50 & \cellcolor{pda75}$0.19_{\pm0.00}$ & \cellcolor{pda98}$0.19_{\pm0.00}$ & \cellcolor{pda101}$0.20_{\pm0.00}$ & \cellcolor{pda102}$0.22_{\pm0.00}$ & \cellcolor{pda103}$0.23_{\pm0.01}$ &  \\
200 & \cellcolor{pda104}$0.18_{\pm0.00}$ & \cellcolor{pda105}$0.19_{\pm0.00}$ & \cellcolor{pda106}$0.22_{\pm0.00}$ & \cellcolor{pda103}$0.23_{\pm0.00}$ & \cellcolor{pda107}$0.24_{\pm0.01}$ &  \\
final & \cellcolor{pda95}$0.17_{\pm0.00}$ & \cellcolor{pda75}$0.19_{\pm0.00}$ & \cellcolor{pda106}$0.22_{\pm0.00}$ & \cellcolor{pda108}$0.23_{\pm0.00}$ & \cellcolor{pda109}$0.24_{\pm0.01}$ &  \\
\midrule
\multicolumn{7}{@{}l}{\emph{NRMSE}} \\
0 & \cellcolor{pda110}$1.48_{\pm0.01}$ & \cellcolor{pda110}$1.48_{\pm0.01}$ & \cellcolor{pda111}$1.47_{\pm0.02}$ & \cellcolor{pda112}$1.47_{\pm0.01}$ & \cellcolor{pda113}$1.44_{\pm0.01}$ & \smash{\begin{tikzpicture}[y=\perdsbar,baseline={([yshift=-\perdsrowtop]current bounding box.north)}]
\fill[pda54] (0,0.0000) rectangle (0.18,0.0833);
\fill[pda55] (0,0.0833) rectangle (0.18,0.1667);
\fill[pda56] (0,0.1667) rectangle (0.18,0.2500);
\fill[pda57] (0,0.2500) rectangle (0.18,0.3333);
\fill[pda58] (0,0.3333) rectangle (0.18,0.4167);
\fill[pda59] (0,0.4167) rectangle (0.18,0.5000);
\fill[pda60] (0,0.5000) rectangle (0.18,0.5833);
\fill[pda61] (0,0.5833) rectangle (0.18,0.6667);
\fill[pda62] (0,0.6667) rectangle (0.18,0.7500);
\fill[pda63] (0,0.7500) rectangle (0.18,0.8333);
\fill[pda64] (0,0.8333) rectangle (0.18,0.9167);
\fill[pda65] (0,0.9167) rectangle (0.18,1.0000);
\draw[line width=0.3pt] (0,0) rectangle (0.18,1);
\node[anchor=north west,font=\tiny,inner sep=1pt] at (0.18,1) {1.80};
\node[anchor=west,font=\tiny,inner sep=1pt] at (0.18,0.5) {0.90};
\node[anchor=south west,font=\tiny,inner sep=1pt] at (0.18,0) {0.00};
\end{tikzpicture}} \\
1 & \cellcolor{pda63}$1.42_{\pm0.02}$ & \cellcolor{pda114}$1.44_{\pm0.02}$ & \cellcolor{pda115}$1.43_{\pm0.02}$ & \cellcolor{pda116}$1.41_{\pm0.01}$ & \cellcolor{pda115}$1.43_{\pm0.02}$ &  \\
5 & \cellcolor{pda117}$1.30_{\pm0.03}$ & \cellcolor{pda118}$1.30_{\pm0.02}$ & \cellcolor{pda119}$1.26_{\pm0.02}$ & \cellcolor{pda120}$1.19_{\pm0.03}$ & \cellcolor{pda121}$1.15_{\pm0.02}$ &  \\
15 & \cellcolor{pda122}$1.31_{\pm0.02}$ & \cellcolor{pda123}$1.24_{\pm0.02}$ & \cellcolor{pda124}$1.09_{\pm0.02}$ & \cellcolor{pda125}$0.96_{\pm0.03}$ & \cellcolor{pda126}$0.86_{\pm0.02}$ &  \\
50 & \cellcolor{pda123}$1.23_{\pm0.02}$ & \cellcolor{pda127}$1.04_{\pm0.03}$ & \cellcolor{pda128}$0.87_{\pm0.02}$ & \cellcolor{pda129}$0.78_{\pm0.02}$ & \cellcolor{pda130}$0.69_{\pm0.03}$ &  \\
200 & \cellcolor{pda131}$1.12_{\pm0.02}$ & \cellcolor{pda132}$0.99_{\pm0.04}$ & \cellcolor{pda133}$0.84_{\pm0.02}$ & \cellcolor{pda66}$0.75_{\pm0.03}$ & \cellcolor{pda134}$0.71_{\pm0.03}$ &  \\
final & \cellcolor{pda135}$1.01_{\pm0.03}$ & \cellcolor{pda136}$0.93_{\pm0.04}$ & \cellcolor{pda137}$0.83_{\pm0.02}$ & \cellcolor{pda107}$0.74_{\pm0.02}$ & \cellcolor{pda138}$0.71_{\pm0.03}$ &  \\
\bottomrule
\end{tabular}
\end{table}

\begin{table}[p]
\definecolor{pdb0}{RGB}{255,251,195}
\definecolor{pdb1}{RGB}{255,244,176}
\definecolor{pdb2}{RGB}{255,236,158}
\definecolor{pdb3}{RGB}{255,228,141}
\definecolor{pdb4}{RGB}{254,220,123}
\definecolor{pdb5}{RGB}{254,206,106}
\definecolor{pdb6}{RGB}{254,189,88}
\definecolor{pdb7}{RGB}{254,173,74}
\definecolor{pdb8}{RGB}{253,158,67}
\definecolor{pdb9}{RGB}{253,143,61}
\definecolor{pdb10}{RGB}{253,117,53}
\definecolor{pdb11}{RGB}{252,91,46}
\definecolor{pdb12}{RGB}{254,219,121}
\definecolor{pdb13}{RGB}{254,218,120}
\definecolor{pdb14}{RGB}{254,217,119}
\definecolor{pdb15}{RGB}{254,218,121}
\definecolor{pdb16}{RGB}{254,217,117}
\definecolor{pdb17}{RGB}{254,216,117}
\definecolor{pdb18}{RGB}{254,212,112}
\definecolor{pdb19}{RGB}{254,214,115}
\definecolor{pdb20}{RGB}{254,203,103}
\definecolor{pdb21}{RGB}{254,185,83}
\definecolor{pdb22}{RGB}{254,180,78}
\definecolor{pdb23}{RGB}{254,161,69}
\definecolor{pdb24}{RGB}{254,187,85}
\definecolor{pdb25}{RGB}{254,163,69}
\definecolor{pdb26}{RGB}{254,162,69}
\definecolor{pdb27}{RGB}{254,166,71}
\definecolor{pdb28}{RGB}{254,190,89}
\definecolor{pdb29}{RGB}{254,208,108}
\definecolor{pdb30}{RGB}{254,191,90}
\definecolor{pdb31}{RGB}{254,202,101}
\definecolor{pdb32}{RGB}{254,208,109}
\definecolor{pdb33}{RGB}{255,255,204}
\definecolor{pdb34}{RGB}{255,254,200}
\definecolor{pdb35}{RGB}{255,255,203}
\definecolor{pdb36}{RGB}{255,253,198}
\definecolor{pdb37}{RGB}{255,245,179}
\definecolor{pdb38}{RGB}{255,231,147}
\definecolor{pdb39}{RGB}{255,249,189}
\definecolor{pdb40}{RGB}{253,115,53}
\definecolor{pdb41}{RGB}{252,78,42}
\definecolor{pdb42}{RGB}{255,248,188}
\definecolor{pdb43}{RGB}{252,108,51}
\definecolor{pdb44}{RGB}{252,93,46}
\definecolor{pdb45}{RGB}{59,136,190}
\definecolor{pdb46}{RGB}{95,165,205}
\definecolor{pdb47}{RGB}{141,194,220}
\definecolor{pdb48}{RGB}{179,214,232}
\definecolor{pdb49}{RGB}{213,231,241}
\definecolor{pdb50}{RGB}{236,242,245}
\definecolor{pdb51}{RGB}{249,239,233}
\definecolor{pdb52}{RGB}{252,222,205}
\definecolor{pdb53}{RGB}{249,194,167}
\definecolor{pdb54}{RGB}{242,161,127}
\definecolor{pdb55}{RGB}{225,120,96}
\definecolor{pdb56}{RGB}{206,79,69}
\definecolor{pdb57}{RGB}{239,243,245}
\definecolor{pdb58}{RGB}{242,245,246}
\definecolor{pdb59}{RGB}{243,245,246}
\definecolor{pdb60}{RGB}{240,244,246}
\definecolor{pdb61}{RGB}{245,246,247}
\definecolor{pdb62}{RGB}{247,246,246}
\definecolor{pdb63}{RGB}{237,242,245}
\definecolor{pdb64}{RGB}{233,240,244}
\definecolor{pdb65}{RGB}{231,240,244}
\definecolor{pdb66}{RGB}{253,145,62}
\definecolor{pdb67}{RGB}{253,146,62}
\definecolor{pdb68}{RGB}{253,141,60}
\definecolor{pdb69}{RGB}{253,144,61}
\definecolor{pdb70}{RGB}{253,147,62}
\definecolor{pdb71}{RGB}{253,142,61}
\definecolor{pdb72}{RGB}{253,139,59}
\definecolor{pdb73}{RGB}{255,254,201}
\definecolor{pdb74}{RGB}{255,234,153}
\definecolor{pdb75}{RGB}{255,235,156}
\definecolor{pdb76}{RGB}{253,157,67}
\definecolor{pdb77}{RGB}{255,240,168}
\definecolor{pdb78}{RGB}{252,96,47}
\definecolor{pdb79}{RGB}{252,95,47}
\definecolor{pdb80}{RGB}{252,80,43}
\definecolor{pdb81}{RGB}{255,240,166}
\definecolor{pdb82}{RGB}{252,83,43}
\definecolor{pdb83}{RGB}{252,101,49}
\caption{Unnormalized GeneralSheaf on Citeseer, $d=5$, four layers, against the topological and the zero reference.}
\label{tab:perds_citeseer_unnorm}
\centering
\footnotesize
\renewcommand{\arraystretch}{1.15}%
\setlength{\perdsbar}{7\baselineskip}\setlength{\perdsbar}{1.15\perdsbar}%
\setlength{\perdsrowtop}{0.7\baselineskip}\setlength{\perdsrowtop}{1.15\perdsrowtop}%
\begin{tabular}{l>{\centering\arraybackslash}p{1.85cm}>{\centering\arraybackslash}p{1.85cm}>{\centering\arraybackslash}p{1.85cm}>{\centering\arraybackslash}p{1.85cm}l}
\toprule
\multicolumn{6}{@{}l}{\textbf{GeneralSheaf, unnormalized: Citeseer}} \\
 & \multicolumn{4}{c}{Layer} & \\
\cmidrule(l){2-5}
Epoch & 1 & 2 & 3 & 4 & \\
\midrule
\multicolumn{6}{@{}l}{\emph{$W_1$, topological}} \\
0 & \cellcolor{pdb12}$0.35_{\pm0.00}$ & \cellcolor{pdb13}$0.35_{\pm0.00}$ & \cellcolor{pdb13}$0.35_{\pm0.00}$ & \cellcolor{pdb14}$0.36_{\pm0.00}$ & \smash{\begin{tikzpicture}[y=\perdsbar,baseline={([yshift=-\perdsrowtop]current bounding box.north)}]
\fill[pdb0] (0,0.0000) rectangle (0.18,0.0833);
\fill[pdb1] (0,0.0833) rectangle (0.18,0.1667);
\fill[pdb2] (0,0.1667) rectangle (0.18,0.2500);
\fill[pdb3] (0,0.2500) rectangle (0.18,0.3333);
\fill[pdb4] (0,0.3333) rectangle (0.18,0.4167);
\fill[pdb5] (0,0.4167) rectangle (0.18,0.5000);
\fill[pdb6] (0,0.5000) rectangle (0.18,0.5833);
\fill[pdb7] (0,0.5833) rectangle (0.18,0.6667);
\fill[pdb8] (0,0.6667) rectangle (0.18,0.7500);
\fill[pdb9] (0,0.7500) rectangle (0.18,0.8333);
\fill[pdb10] (0,0.8333) rectangle (0.18,0.9167);
\fill[pdb11] (0,0.9167) rectangle (0.18,1.0000);
\draw[line width=0.3pt] (0,0) rectangle (0.18,1);
\node[anchor=north west,font=\tiny,inner sep=1pt] at (0.18,1) {0.90};
\node[anchor=west,font=\tiny,inner sep=1pt] at (0.18,0.5) {0.45};
\node[anchor=south west,font=\tiny,inner sep=1pt] at (0.18,0) {0.00};
\end{tikzpicture}} \\
1 & \cellcolor{pdb14}$0.36_{\pm0.01}$ & \cellcolor{pdb13}$0.35_{\pm0.00}$ & \cellcolor{pdb14}$0.36_{\pm0.00}$ & \cellcolor{pdb15}$0.35_{\pm0.00}$ &  \\
5 & \cellcolor{pdb16}$0.36_{\pm0.00}$ & \cellcolor{pdb17}$0.36_{\pm0.01}$ & \cellcolor{pdb18}$0.38_{\pm0.01}$ & \cellcolor{pdb19}$0.37_{\pm0.01}$ &  \\
15 & \cellcolor{pdb20}$0.42_{\pm0.01}$ & \cellcolor{pdb21}$0.51_{\pm0.02}$ & \cellcolor{pdb22}$0.53_{\pm0.02}$ & \cellcolor{pdb23}$0.62_{\pm0.02}$ &  \\
50 & \cellcolor{pdb24}$0.50_{\pm0.01}$ & \cellcolor{pdb25}$0.62_{\pm0.01}$ & \cellcolor{pdb26}$0.62_{\pm0.01}$ & \cellcolor{pdb27}$0.60_{\pm0.01}$ &  \\
200 & \cellcolor{pdb28}$0.48_{\pm0.01}$ & \cellcolor{pdb28}$0.49_{\pm0.01}$ & \cellcolor{pdb20}$0.42_{\pm0.01}$ & \cellcolor{pdb29}$0.40_{\pm0.01}$ &  \\
final & \cellcolor{pdb30}$0.48_{\pm0.01}$ & \cellcolor{pdb30}$0.48_{\pm0.01}$ & \cellcolor{pdb31}$0.43_{\pm0.01}$ & \cellcolor{pdb32}$0.40_{\pm0.01}$ &  \\
\midrule
\multicolumn{6}{@{}l}{\emph{$W_1$, zero}} \\
0 & \cellcolor{pdb33}$0.00_{\pm0.00}$ & \cellcolor{pdb33}$0.00_{\pm0.00}$ & \cellcolor{pdb33}$0.00_{\pm0.00}$ & \cellcolor{pdb33}$0.00_{\pm0.00}$ & \smash{\begin{tikzpicture}[y=\perdsbar,baseline={([yshift=-\perdsrowtop]current bounding box.north)}]
\fill[pdb0] (0,0.0000) rectangle (0.18,0.0833);
\fill[pdb1] (0,0.0833) rectangle (0.18,0.1667);
\fill[pdb2] (0,0.1667) rectangle (0.18,0.2500);
\fill[pdb3] (0,0.2500) rectangle (0.18,0.3333);
\fill[pdb4] (0,0.3333) rectangle (0.18,0.4167);
\fill[pdb5] (0,0.4167) rectangle (0.18,0.5000);
\fill[pdb6] (0,0.5000) rectangle (0.18,0.5833);
\fill[pdb7] (0,0.5833) rectangle (0.18,0.6667);
\fill[pdb8] (0,0.6667) rectangle (0.18,0.7500);
\fill[pdb9] (0,0.7500) rectangle (0.18,0.8333);
\fill[pdb10] (0,0.8333) rectangle (0.18,0.9167);
\fill[pdb11] (0,0.9167) rectangle (0.18,1.0000);
\draw[line width=0.3pt] (0,0) rectangle (0.18,1);
\node[anchor=north west,font=\tiny,inner sep=1pt] at (0.18,1) {3.59};
\node[anchor=west,font=\tiny,inner sep=1pt] at (0.18,0.5) {1.80};
\node[anchor=south west,font=\tiny,inner sep=1pt] at (0.18,0) {0.00};
\end{tikzpicture}} \\
1 & \cellcolor{pdb33}$0.00_{\pm0.00}$ & \cellcolor{pdb33}$0.00_{\pm0.00}$ & \cellcolor{pdb33}$0.00_{\pm0.00}$ & \cellcolor{pdb33}$0.00_{\pm0.00}$ &  \\
5 & \cellcolor{pdb33}$0.00_{\pm0.00}$ & \cellcolor{pdb33}$0.00_{\pm0.00}$ & \cellcolor{pdb33}$0.00_{\pm0.00}$ & \cellcolor{pdb33}$0.00_{\pm0.00}$ &  \\
15 & \cellcolor{pdb33}$0.00_{\pm0.00}$ & \cellcolor{pdb33}$0.00_{\pm0.00}$ & \cellcolor{pdb33}$0.01_{\pm0.00}$ & \cellcolor{pdb34}$0.06_{\pm0.02}$ &  \\
50 & \cellcolor{pdb35}$0.02_{\pm0.00}$ & \cellcolor{pdb36}$0.09_{\pm0.04}$ & \cellcolor{pdb37}$0.40_{\pm0.08}$ & \cellcolor{pdb38}$0.94_{\pm0.16}$ &  \\
200 & \cellcolor{pdb39}$0.25_{\pm0.02}$ & \cellcolor{pdb7}$2.26_{\pm0.27}$ & \cellcolor{pdb40}$3.17_{\pm0.22}$ & \cellcolor{pdb41}$3.59_{\pm0.26}$ &  \\
final & \cellcolor{pdb42}$0.26_{\pm0.02}$ & \cellcolor{pdb27}$2.38_{\pm0.30}$ & \cellcolor{pdb43}$3.25_{\pm0.20}$ & \cellcolor{pdb44}$3.42_{\pm0.27}$ &  \\
\midrule
\multicolumn{6}{@{}l}{\emph{Pearson $r$, topological}} \\
0 & \cellcolor{pdb57}$-0.06_{\pm0.01}$ & \cellcolor{pdb58}$-0.03_{\pm0.01}$ & \cellcolor{pdb59}$-0.03_{\pm0.02}$ & \cellcolor{pdb50}$-0.09_{\pm0.02}$ & \smash{\begin{tikzpicture}[y=\perdsbar,baseline={([yshift=-\perdsrowtop]current bounding box.north)}]
\fill[pdb45] (0,0.0000) rectangle (0.18,0.0833);
\fill[pdb46] (0,0.0833) rectangle (0.18,0.1667);
\fill[pdb47] (0,0.1667) rectangle (0.18,0.2500);
\fill[pdb48] (0,0.2500) rectangle (0.18,0.3333);
\fill[pdb49] (0,0.3333) rectangle (0.18,0.4167);
\fill[pdb50] (0,0.4167) rectangle (0.18,0.5000);
\fill[pdb51] (0,0.5000) rectangle (0.18,0.5833);
\fill[pdb52] (0,0.5833) rectangle (0.18,0.6667);
\fill[pdb53] (0,0.6667) rectangle (0.18,0.7500);
\fill[pdb54] (0,0.7500) rectangle (0.18,0.8333);
\fill[pdb55] (0,0.8333) rectangle (0.18,0.9167);
\fill[pdb56] (0,0.9167) rectangle (0.18,1.0000);
\draw[line width=0.3pt] (0,0) rectangle (0.18,1);
\node[anchor=north west,font=\tiny,inner sep=1pt] at (0.18,1) {1.00};
\node[anchor=west,font=\tiny,inner sep=1pt] at (0.18,0.5) {0.00};
\node[anchor=south west,font=\tiny,inner sep=1pt] at (0.18,0) {-1.00};
\end{tikzpicture}} \\
1 & \cellcolor{pdb60}$-0.06_{\pm0.01}$ & \cellcolor{pdb58}$-0.04_{\pm0.01}$ & \cellcolor{pdb59}$-0.03_{\pm0.01}$ & \cellcolor{pdb59}$-0.03_{\pm0.01}$ &  \\
5 & \cellcolor{pdb58}$-0.04_{\pm0.01}$ & \cellcolor{pdb61}$-0.02_{\pm0.01}$ & \cellcolor{pdb59}$-0.03_{\pm0.01}$ & \cellcolor{pdb61}$-0.02_{\pm0.01}$ &  \\
15 & \cellcolor{pdb62}$0.01_{\pm0.01}$ & \cellcolor{pdb61}$-0.02_{\pm0.02}$ & \cellcolor{pdb59}$-0.03_{\pm0.02}$ & \cellcolor{pdb58}$-0.03_{\pm0.02}$ &  \\
50 & \cellcolor{pdb58}$-0.04_{\pm0.03}$ & \cellcolor{pdb58}$-0.04_{\pm0.03}$ & \cellcolor{pdb61}$-0.02_{\pm0.01}$ & \cellcolor{pdb59}$-0.03_{\pm0.02}$ &  \\
200 & \cellcolor{pdb61}$-0.02_{\pm0.02}$ & \cellcolor{pdb63}$-0.07_{\pm0.01}$ & \cellcolor{pdb64}$-0.10_{\pm0.02}$ & \cellcolor{pdb64}$-0.10_{\pm0.01}$ &  \\
final & \cellcolor{pdb59}$-0.03_{\pm0.02}$ & \cellcolor{pdb63}$-0.07_{\pm0.01}$ & \cellcolor{pdb50}$-0.08_{\pm0.02}$ & \cellcolor{pdb65}$-0.12_{\pm0.01}$ &  \\
\midrule
\multicolumn{6}{@{}l}{\emph{NRMSE, topological}} \\
0 & \cellcolor{pdb9}$1.42_{\pm0.01}$ & \cellcolor{pdb66}$1.40_{\pm0.01}$ & \cellcolor{pdb67}$1.39_{\pm0.01}$ & \cellcolor{pdb68}$1.44_{\pm0.01}$ & \smash{\begin{tikzpicture}[y=\perdsbar,baseline={([yshift=-\perdsrowtop]current bounding box.north)}]
\fill[pdb0] (0,0.0000) rectangle (0.18,0.0833);
\fill[pdb1] (0,0.0833) rectangle (0.18,0.1667);
\fill[pdb2] (0,0.1667) rectangle (0.18,0.2500);
\fill[pdb3] (0,0.2500) rectangle (0.18,0.3333);
\fill[pdb4] (0,0.3333) rectangle (0.18,0.4167);
\fill[pdb5] (0,0.4167) rectangle (0.18,0.5000);
\fill[pdb6] (0,0.5000) rectangle (0.18,0.5833);
\fill[pdb7] (0,0.5833) rectangle (0.18,0.6667);
\fill[pdb8] (0,0.6667) rectangle (0.18,0.7500);
\fill[pdb9] (0,0.7500) rectangle (0.18,0.8333);
\fill[pdb10] (0,0.8333) rectangle (0.18,0.9167);
\fill[pdb11] (0,0.9167) rectangle (0.18,1.0000);
\draw[line width=0.3pt] (0,0) rectangle (0.18,1);
\node[anchor=north west,font=\tiny,inner sep=1pt] at (0.18,1) {1.80};
\node[anchor=west,font=\tiny,inner sep=1pt] at (0.18,0.5) {0.90};
\node[anchor=south west,font=\tiny,inner sep=1pt] at (0.18,0) {0.00};
\end{tikzpicture}} \\
1 & \cellcolor{pdb9}$1.42_{\pm0.01}$ & \cellcolor{pdb66}$1.40_{\pm0.01}$ & \cellcolor{pdb66}$1.40_{\pm0.01}$ & \cellcolor{pdb66}$1.40_{\pm0.01}$ &  \\
5 & \cellcolor{pdb69}$1.41_{\pm0.01}$ & \cellcolor{pdb67}$1.39_{\pm0.01}$ & \cellcolor{pdb66}$1.40_{\pm0.01}$ & \cellcolor{pdb66}$1.40_{\pm0.01}$ &  \\
15 & \cellcolor{pdb70}$1.38_{\pm0.01}$ & \cellcolor{pdb69}$1.41_{\pm0.02}$ & \cellcolor{pdb9}$1.42_{\pm0.01}$ & \cellcolor{pdb9}$1.42_{\pm0.01}$ &  \\
50 & \cellcolor{pdb71}$1.43_{\pm0.02}$ & \cellcolor{pdb9}$1.43_{\pm0.02}$ & \cellcolor{pdb69}$1.41_{\pm0.01}$ & \cellcolor{pdb69}$1.41_{\pm0.01}$ &  \\
200 & \cellcolor{pdb66}$1.40_{\pm0.01}$ & \cellcolor{pdb71}$1.43_{\pm0.00}$ & \cellcolor{pdb68}$1.44_{\pm0.01}$ & \cellcolor{pdb68}$1.44_{\pm0.01}$ &  \\
final & \cellcolor{pdb66}$1.41_{\pm0.01}$ & \cellcolor{pdb71}$1.43_{\pm0.01}$ & \cellcolor{pdb71}$1.43_{\pm0.01}$ & \cellcolor{pdb72}$1.45_{\pm0.01}$ &  \\
\midrule
\multicolumn{6}{@{}l}{\emph{RMSE, zero}} \\
0 & \cellcolor{pdb33}$0.00_{\pm0.00}$ & \cellcolor{pdb33}$0.00_{\pm0.00}$ & \cellcolor{pdb33}$0.00_{\pm0.00}$ & \cellcolor{pdb33}$0.00_{\pm0.00}$ & \smash{\begin{tikzpicture}[y=\perdsbar,baseline={([yshift=-\perdsrowtop]current bounding box.north)}]
\fill[pdb0] (0,0.0000) rectangle (0.18,0.0833);
\fill[pdb1] (0,0.0833) rectangle (0.18,0.1667);
\fill[pdb2] (0,0.1667) rectangle (0.18,0.2500);
\fill[pdb3] (0,0.2500) rectangle (0.18,0.3333);
\fill[pdb4] (0,0.3333) rectangle (0.18,0.4167);
\fill[pdb5] (0,0.4167) rectangle (0.18,0.5000);
\fill[pdb6] (0,0.5000) rectangle (0.18,0.5833);
\fill[pdb7] (0,0.5833) rectangle (0.18,0.6667);
\fill[pdb8] (0,0.6667) rectangle (0.18,0.7500);
\fill[pdb9] (0,0.7500) rectangle (0.18,0.8333);
\fill[pdb10] (0,0.8333) rectangle (0.18,0.9167);
\fill[pdb11] (0,0.9167) rectangle (0.18,1.0000);
\draw[line width=0.3pt] (0,0) rectangle (0.18,1);
\node[anchor=north west,font=\tiny,inner sep=1pt] at (0.18,1) {7.55};
\node[anchor=west,font=\tiny,inner sep=1pt] at (0.18,0.5) {3.77};
\node[anchor=south west,font=\tiny,inner sep=1pt] at (0.18,0) {0.00};
\end{tikzpicture}} \\
1 & \cellcolor{pdb33}$0.00_{\pm0.00}$ & \cellcolor{pdb33}$0.00_{\pm0.00}$ & \cellcolor{pdb33}$0.00_{\pm0.00}$ & \cellcolor{pdb33}$0.00_{\pm0.00}$ &  \\
5 & \cellcolor{pdb33}$0.00_{\pm0.00}$ & \cellcolor{pdb33}$0.00_{\pm0.00}$ & \cellcolor{pdb33}$0.00_{\pm0.00}$ & \cellcolor{pdb33}$0.00_{\pm0.00}$ &  \\
15 & \cellcolor{pdb33}$0.00_{\pm0.00}$ & \cellcolor{pdb35}$0.02_{\pm0.01}$ & \cellcolor{pdb73}$0.11_{\pm0.05}$ & \cellcolor{pdb74}$1.74_{\pm0.57}$ &  \\
50 & \cellcolor{pdb73}$0.11_{\pm0.02}$ & \cellcolor{pdb75}$1.67_{\pm0.39}$ & \cellcolor{pdb28}$4.05_{\pm0.53}$ & \cellcolor{pdb76}$5.40_{\pm0.52}$ &  \\
200 & \cellcolor{pdb77}$1.22_{\pm0.13}$ & \cellcolor{pdb78}$7.11_{\pm0.34}$ & \cellcolor{pdb79}$7.15_{\pm0.32}$ & \cellcolor{pdb80}$7.50_{\pm0.54}$ &  \\
final & \cellcolor{pdb81}$1.29_{\pm0.15}$ & \cellcolor{pdb82}$7.43_{\pm0.42}$ & \cellcolor{pdb41}$7.55_{\pm0.30}$ & \cellcolor{pdb83}$7.00_{\pm0.53}$ &  \\
\bottomrule
\end{tabular}
\end{table}

\begin{table}[p]
\definecolor{pdc0}{RGB}{255,251,195}
\definecolor{pdc1}{RGB}{255,244,176}
\definecolor{pdc2}{RGB}{255,236,158}
\definecolor{pdc3}{RGB}{255,228,141}
\definecolor{pdc4}{RGB}{254,220,123}
\definecolor{pdc5}{RGB}{254,206,106}
\definecolor{pdc6}{RGB}{254,189,88}
\definecolor{pdc7}{RGB}{254,173,74}
\definecolor{pdc8}{RGB}{253,158,67}
\definecolor{pdc9}{RGB}{253,143,61}
\definecolor{pdc10}{RGB}{253,117,53}
\definecolor{pdc11}{RGB}{252,91,46}
\definecolor{pdc12}{RGB}{254,199,99}
\definecolor{pdc13}{RGB}{254,212,113}
\definecolor{pdc14}{RGB}{254,192,91}
\definecolor{pdc15}{RGB}{254,214,115}
\definecolor{pdc16}{RGB}{254,195,95}
\definecolor{pdc17}{RGB}{254,195,94}
\definecolor{pdc18}{RGB}{254,215,116}
\definecolor{pdc19}{RGB}{254,202,101}
\definecolor{pdc20}{RGB}{254,201,101}
\definecolor{pdc21}{RGB}{254,205,105}
\definecolor{pdc22}{RGB}{254,209,109}
\definecolor{pdc23}{RGB}{254,210,111}
\definecolor{pdc24}{RGB}{254,218,120}
\definecolor{pdc25}{RGB}{254,218,119}
\definecolor{pdc26}{RGB}{254,217,119}
\definecolor{pdc27}{RGB}{255,255,204}
\definecolor{pdc28}{RGB}{255,255,203}
\definecolor{pdc29}{RGB}{252,78,42}
\definecolor{pdc30}{RGB}{253,142,61}
\definecolor{pdc31}{RGB}{255,254,202}
\definecolor{pdc32}{RGB}{252,79,42}
\definecolor{pdc33}{RGB}{253,133,58}
\definecolor{pdc34}{RGB}{255,252,197}
\definecolor{pdc35}{RGB}{253,120,54}
\definecolor{pdc36}{RGB}{253,141,60}
\definecolor{pdc37}{RGB}{255,250,191}
\definecolor{pdc38}{RGB}{253,157,67}
\definecolor{pdc39}{RGB}{253,147,62}
\definecolor{pdc40}{RGB}{253,138,59}
\definecolor{pdc41}{RGB}{255,247,184}
\definecolor{pdc42}{RGB}{255,240,167}
\definecolor{pdc43}{RGB}{255,247,185}
\definecolor{pdc44}{RGB}{59,136,190}
\definecolor{pdc45}{RGB}{95,165,205}
\definecolor{pdc46}{RGB}{141,194,220}
\definecolor{pdc47}{RGB}{179,214,232}
\definecolor{pdc48}{RGB}{213,231,241}
\definecolor{pdc49}{RGB}{236,242,245}
\definecolor{pdc50}{RGB}{249,239,233}
\definecolor{pdc51}{RGB}{252,222,205}
\definecolor{pdc52}{RGB}{249,194,167}
\definecolor{pdc53}{RGB}{242,161,127}
\definecolor{pdc54}{RGB}{225,120,96}
\definecolor{pdc55}{RGB}{206,79,69}
\definecolor{pdc56}{RGB}{218,233,242}
\definecolor{pdc57}{RGB}{150,199,223}
\definecolor{pdc58}{RGB}{135,190,218}
\definecolor{pdc59}{RGB}{227,237,243}
\definecolor{pdc60}{RGB}{169,209,229}
\definecolor{pdc61}{RGB}{212,230,241}
\definecolor{pdc62}{RGB}{174,211,230}
\definecolor{pdc63}{RGB}{234,241,245}
\definecolor{pdc64}{RGB}{224,236,243}
\definecolor{pdc65}{RGB}{230,239,244}
\definecolor{pdc66}{RGB}{209,229,240}
\definecolor{pdc67}{RGB}{216,233,241}
\definecolor{pdc68}{RGB}{237,242,245}
\definecolor{pdc69}{RGB}{215,232,241}
\definecolor{pdc70}{RGB}{233,240,244}
\definecolor{pdc71}{RGB}{247,246,246}
\definecolor{pdc72}{RGB}{254,187,85}
\definecolor{pdc73}{RGB}{253,126,56}
\definecolor{pdc74}{RGB}{252,86,44}
\definecolor{pdc75}{RGB}{252,80,43}
\definecolor{pdc76}{RGB}{252,98,48}
\definecolor{pdc77}{RGB}{253,139,60}
\definecolor{pdc78}{RGB}{253,123,55}
\definecolor{pdc79}{RGB}{252,99,48}
\definecolor{pdc80}{RGB}{253,140,60}
\definecolor{pdc81}{RGB}{253,129,56}
\definecolor{pdc82}{RGB}{253,139,59}
\definecolor{pdc83}{RGB}{253,136,59}
\definecolor{pdc84}{RGB}{253,122,55}
\definecolor{pdc85}{RGB}{253,121,54}
\definecolor{pdc86}{RGB}{254,182,80}
\definecolor{pdc87}{RGB}{253,145,62}
\definecolor{pdc88}{RGB}{254,171,73}
\definecolor{pdc89}{RGB}{255,252,196}
\definecolor{pdc90}{RGB}{254,162,69}
\definecolor{pdc91}{RGB}{253,155,66}
\definecolor{pdc92}{RGB}{255,249,189}
\definecolor{pdc93}{RGB}{254,179,77}
\definecolor{pdc94}{RGB}{253,142,60}
\definecolor{pdc95}{RGB}{255,251,193}
\definecolor{pdc96}{RGB}{254,213,114}
\definecolor{pdc97}{RGB}{255,234,153}
\definecolor{pdc98}{RGB}{255,235,157}
\definecolor{pdc99}{RGB}{255,242,171}
\caption{Unnormalized JointSheafParams on Questions, $d=5$, four layers, started at the identity sheaf, against the topological and the zero reference. In the first layer the profile is constant, so its correlation is undefined.}
\label{tab:perds_questions_joint_unnorm}
\centering
\footnotesize
\renewcommand{\arraystretch}{1.15}%
\setlength{\perdsbar}{7\baselineskip}\setlength{\perdsbar}{1.15\perdsbar}%
\setlength{\perdsrowtop}{0.7\baselineskip}\setlength{\perdsrowtop}{1.15\perdsrowtop}%
\begin{tabular}{l>{\centering\arraybackslash}p{1.85cm}>{\centering\arraybackslash}p{1.85cm}>{\centering\arraybackslash}p{1.85cm}>{\centering\arraybackslash}p{1.85cm}l}
\toprule
\multicolumn{6}{@{}l}{\textbf{JointSheafParams, unnormalized: Questions}} \\
 & \multicolumn{4}{c}{Layer} & \\
\cmidrule(l){2-5}
Epoch & 1 & 2 & 3 & 4 & \\
\midrule
\multicolumn{6}{@{}l}{\emph{$W_1$, topological}} \\
0 & \cellcolor{pdc12}$0.44_{\pm0.00}$ & \cellcolor{pdc13}$0.38_{\pm0.00}$ & \cellcolor{pdc14}$0.48_{\pm0.00}$ & \cellcolor{pdc14}$0.48_{\pm0.00}$ & \smash{\begin{tikzpicture}[y=\perdsbar,baseline={([yshift=-\perdsrowtop]current bounding box.north)}]
\fill[pdc0] (0,0.0000) rectangle (0.18,0.0833);
\fill[pdc1] (0,0.0833) rectangle (0.18,0.1667);
\fill[pdc2] (0,0.1667) rectangle (0.18,0.2500);
\fill[pdc3] (0,0.2500) rectangle (0.18,0.3333);
\fill[pdc4] (0,0.3333) rectangle (0.18,0.4167);
\fill[pdc5] (0,0.4167) rectangle (0.18,0.5000);
\fill[pdc6] (0,0.5000) rectangle (0.18,0.5833);
\fill[pdc7] (0,0.5833) rectangle (0.18,0.6667);
\fill[pdc8] (0,0.6667) rectangle (0.18,0.7500);
\fill[pdc9] (0,0.7500) rectangle (0.18,0.8333);
\fill[pdc10] (0,0.8333) rectangle (0.18,0.9167);
\fill[pdc11] (0,0.9167) rectangle (0.18,1.0000);
\draw[line width=0.3pt] (0,0) rectangle (0.18,1);
\node[anchor=north west,font=\tiny,inner sep=1pt] at (0.18,1) {0.90};
\node[anchor=west,font=\tiny,inner sep=1pt] at (0.18,0.5) {0.45};
\node[anchor=south west,font=\tiny,inner sep=1pt] at (0.18,0) {0.00};
\end{tikzpicture}} \\
1 & \cellcolor{pdc12}$0.44_{\pm0.00}$ & \cellcolor{pdc15}$0.37_{\pm0.00}$ & \cellcolor{pdc16}$0.46_{\pm0.00}$ & \cellcolor{pdc17}$0.46_{\pm0.00}$ &  \\
5 & \cellcolor{pdc12}$0.44_{\pm0.00}$ & \cellcolor{pdc18}$0.37_{\pm0.00}$ & \cellcolor{pdc19}$0.43_{\pm0.00}$ & \cellcolor{pdc20}$0.43_{\pm0.00}$ &  \\
15 & \cellcolor{pdc12}$0.44_{\pm0.00}$ & \cellcolor{pdc18}$0.37_{\pm0.00}$ & \cellcolor{pdc21}$0.42_{\pm0.00}$ & \cellcolor{pdc5}$0.41_{\pm0.00}$ &  \\
50 & \cellcolor{pdc12}$0.44_{\pm0.00}$ & \cellcolor{pdc18}$0.37_{\pm0.00}$ & \cellcolor{pdc5}$0.41_{\pm0.00}$ & \cellcolor{pdc22}$0.40_{\pm0.00}$ &  \\
200 & \cellcolor{pdc12}$0.44_{\pm0.00}$ & \cellcolor{pdc23}$0.39_{\pm0.00}$ & \cellcolor{pdc24}$0.35_{\pm0.01}$ & \cellcolor{pdc25}$0.36_{\pm0.01}$ &  \\
final & \cellcolor{pdc12}$0.44_{\pm0.00}$ & \cellcolor{pdc22}$0.40_{\pm0.00}$ & \cellcolor{pdc24}$0.35_{\pm0.01}$ & \cellcolor{pdc26}$0.36_{\pm0.02}$ &  \\
\midrule
\multicolumn{6}{@{}l}{\emph{$W_1$, zero}} \\
0 & \cellcolor{pdc27}$0.00_{\pm0.00}$ & \cellcolor{pdc28}$0.01_{\pm0.00}$ & \cellcolor{pdc29}$2.75_{\pm0.01}$ & \cellcolor{pdc30}$2.18_{\pm0.01}$ & \smash{\begin{tikzpicture}[y=\perdsbar,baseline={([yshift=-\perdsrowtop]current bounding box.north)}]
\fill[pdc0] (0,0.0000) rectangle (0.18,0.0833);
\fill[pdc1] (0,0.0833) rectangle (0.18,0.1667);
\fill[pdc2] (0,0.1667) rectangle (0.18,0.2500);
\fill[pdc3] (0,0.2500) rectangle (0.18,0.3333);
\fill[pdc4] (0,0.3333) rectangle (0.18,0.4167);
\fill[pdc5] (0,0.4167) rectangle (0.18,0.5000);
\fill[pdc6] (0,0.5000) rectangle (0.18,0.5833);
\fill[pdc7] (0,0.5833) rectangle (0.18,0.6667);
\fill[pdc8] (0,0.6667) rectangle (0.18,0.7500);
\fill[pdc9] (0,0.7500) rectangle (0.18,0.8333);
\fill[pdc10] (0,0.8333) rectangle (0.18,0.9167);
\fill[pdc11] (0,0.9167) rectangle (0.18,1.0000);
\draw[line width=0.3pt] (0,0) rectangle (0.18,1);
\node[anchor=north west,font=\tiny,inner sep=1pt] at (0.18,1) {2.75};
\node[anchor=west,font=\tiny,inner sep=1pt] at (0.18,0.5) {1.37};
\node[anchor=south west,font=\tiny,inner sep=1pt] at (0.18,0) {0.00};
\end{tikzpicture}} \\
1 & \cellcolor{pdc27}$0.00_{\pm0.00}$ & \cellcolor{pdc31}$0.02_{\pm0.00}$ & \cellcolor{pdc32}$2.74_{\pm0.03}$ & \cellcolor{pdc33}$2.27_{\pm0.03}$ &  \\
5 & \cellcolor{pdc27}$0.00_{\pm0.00}$ & \cellcolor{pdc34}$0.08_{\pm0.00}$ & \cellcolor{pdc35}$2.38_{\pm0.06}$ & \cellcolor{pdc36}$2.20_{\pm0.05}$ &  \\
15 & \cellcolor{pdc27}$0.00_{\pm0.00}$ & \cellcolor{pdc37}$0.17_{\pm0.01}$ & \cellcolor{pdc38}$1.96_{\pm0.07}$ & \cellcolor{pdc39}$2.11_{\pm0.10}$ &  \\
50 & \cellcolor{pdc27}$0.00_{\pm0.00}$ & \cellcolor{pdc0}$0.11_{\pm0.01}$ & \cellcolor{pdc5}$1.26_{\pm0.07}$ & \cellcolor{pdc40}$2.23_{\pm0.20}$ &  \\
200 & \cellcolor{pdc27}$0.00_{\pm0.00}$ & \cellcolor{pdc27}$0.00_{\pm0.00}$ & \cellcolor{pdc41}$0.24_{\pm0.07}$ & \cellcolor{pdc42}$0.47_{\pm0.12}$ &  \\
final & \cellcolor{pdc27}$0.00_{\pm0.00}$ & \cellcolor{pdc27}$0.01_{\pm0.00}$ & \cellcolor{pdc1}$0.35_{\pm0.11}$ & \cellcolor{pdc43}$0.24_{\pm0.08}$ &  \\
\midrule
\multicolumn{6}{@{}l}{\emph{Pearson $r$, topological}} \\
0 &  & \cellcolor{pdc56}$-0.22_{\pm0.01}$ & \cellcolor{pdc57}$-0.56_{\pm0.01}$ & \cellcolor{pdc58}$-0.61_{\pm0.01}$ & \smash{\begin{tikzpicture}[y=\perdsbar,baseline={([yshift=-\perdsrowtop]current bounding box.north)}]
\fill[pdc44] (0,0.0000) rectangle (0.18,0.0833);
\fill[pdc45] (0,0.0833) rectangle (0.18,0.1667);
\fill[pdc46] (0,0.1667) rectangle (0.18,0.2500);
\fill[pdc47] (0,0.2500) rectangle (0.18,0.3333);
\fill[pdc48] (0,0.3333) rectangle (0.18,0.4167);
\fill[pdc49] (0,0.4167) rectangle (0.18,0.5000);
\fill[pdc50] (0,0.5000) rectangle (0.18,0.5833);
\fill[pdc51] (0,0.5833) rectangle (0.18,0.6667);
\fill[pdc52] (0,0.6667) rectangle (0.18,0.7500);
\fill[pdc53] (0,0.7500) rectangle (0.18,0.8333);
\fill[pdc54] (0,0.8333) rectangle (0.18,0.9167);
\fill[pdc55] (0,0.9167) rectangle (0.18,1.0000);
\draw[line width=0.3pt] (0,0) rectangle (0.18,1);
\node[anchor=north west,font=\tiny,inner sep=1pt] at (0.18,1) {1.00};
\node[anchor=west,font=\tiny,inner sep=1pt] at (0.18,0.5) {0.00};
\node[anchor=south west,font=\tiny,inner sep=1pt] at (0.18,0) {-1.00};
\end{tikzpicture}} \\
1 &  & \cellcolor{pdc59}$-0.16_{\pm0.01}$ & \cellcolor{pdc60}$-0.46_{\pm0.01}$ & \cellcolor{pdc57}$-0.56_{\pm0.01}$ &  \\
5 &  & \cellcolor{pdc49}$-0.09_{\pm0.00}$ & \cellcolor{pdc61}$-0.26_{\pm0.01}$ & \cellcolor{pdc62}$-0.44_{\pm0.02}$ &  \\
15 &  & \cellcolor{pdc49}$-0.08_{\pm0.00}$ & \cellcolor{pdc56}$-0.22_{\pm0.01}$ & \cellcolor{pdc61}$-0.26_{\pm0.03}$ &  \\
50 &  & \cellcolor{pdc63}$-0.09_{\pm0.00}$ & \cellcolor{pdc64}$-0.17_{\pm0.02}$ & \cellcolor{pdc65}$-0.13_{\pm0.03}$ &  \\
200 &  & \cellcolor{pdc66}$-0.29_{\pm0.01}$ & \cellcolor{pdc67}$-0.22_{\pm0.04}$ & \cellcolor{pdc68}$-0.07_{\pm0.04}$ &  \\
final &  & \cellcolor{pdc69}$-0.24_{\pm0.01}$ & \cellcolor{pdc70}$-0.11_{\pm0.03}$ & \cellcolor{pdc71}$0.01_{\pm0.06}$ &  \\
\midrule
\multicolumn{6}{@{}l}{\emph{NRMSE, topological}} \\
0 & \cellcolor{pdc72}$1.00_{\pm0.00}$ & \cellcolor{pdc73}$1.53_{\pm0.01}$ & \cellcolor{pdc74}$1.75_{\pm0.01}$ & \cellcolor{pdc75}$1.79_{\pm0.00}$ & \smash{\begin{tikzpicture}[y=\perdsbar,baseline={([yshift=-\perdsrowtop]current bounding box.north)}]
\fill[pdc0] (0,0.0000) rectangle (0.18,0.0833);
\fill[pdc1] (0,0.0833) rectangle (0.18,0.1667);
\fill[pdc2] (0,0.1667) rectangle (0.18,0.2500);
\fill[pdc3] (0,0.2500) rectangle (0.18,0.3333);
\fill[pdc4] (0,0.3333) rectangle (0.18,0.4167);
\fill[pdc5] (0,0.4167) rectangle (0.18,0.5000);
\fill[pdc6] (0,0.5000) rectangle (0.18,0.5833);
\fill[pdc7] (0,0.5833) rectangle (0.18,0.6667);
\fill[pdc8] (0,0.6667) rectangle (0.18,0.7500);
\fill[pdc9] (0,0.7500) rectangle (0.18,0.8333);
\fill[pdc10] (0,0.8333) rectangle (0.18,0.9167);
\fill[pdc11] (0,0.9167) rectangle (0.18,1.0000);
\draw[line width=0.3pt] (0,0) rectangle (0.18,1);
\node[anchor=north west,font=\tiny,inner sep=1pt] at (0.18,1) {1.80};
\node[anchor=west,font=\tiny,inner sep=1pt] at (0.18,0.5) {0.90};
\node[anchor=south west,font=\tiny,inner sep=1pt] at (0.18,0) {0.00};
\end{tikzpicture}} \\
1 & \cellcolor{pdc72}$1.00_{\pm0.00}$ & \cellcolor{pdc33}$1.49_{\pm0.01}$ & \cellcolor{pdc76}$1.68_{\pm0.01}$ & \cellcolor{pdc74}$1.75_{\pm0.01}$ &  \\
5 & \cellcolor{pdc72}$1.00_{\pm0.00}$ & \cellcolor{pdc77}$1.45_{\pm0.00}$ & \cellcolor{pdc78}$1.54_{\pm0.01}$ & \cellcolor{pdc79}$1.68_{\pm0.01}$ &  \\
15 & \cellcolor{pdc72}$1.00_{\pm0.00}$ & \cellcolor{pdc80}$1.45_{\pm0.00}$ & \cellcolor{pdc81}$1.51_{\pm0.01}$ & \cellcolor{pdc35}$1.56_{\pm0.02}$ &  \\
50 & \cellcolor{pdc72}$1.00_{\pm0.00}$ & \cellcolor{pdc82}$1.45_{\pm0.00}$ & \cellcolor{pdc33}$1.48_{\pm0.02}$ & \cellcolor{pdc83}$1.47_{\pm0.02}$ &  \\
200 & \cellcolor{pdc72}$1.00_{\pm0.00}$ & \cellcolor{pdc10}$1.58_{\pm0.01}$ & \cellcolor{pdc84}$1.55_{\pm0.02}$ & \cellcolor{pdc80}$1.44_{\pm0.03}$ &  \\
final & \cellcolor{pdc72}$1.00_{\pm0.00}$ & \cellcolor{pdc85}$1.56_{\pm0.01}$ & \cellcolor{pdc83}$1.47_{\pm0.03}$ & \cellcolor{pdc39}$1.39_{\pm0.04}$ &  \\
\midrule
\multicolumn{6}{@{}l}{\emph{RMSE, zero}} \\
0 & \cellcolor{pdc27}$0.00_{\pm0.00}$ & \cellcolor{pdc28}$0.06_{\pm0.00}$ & \cellcolor{pdc9}$11.35_{\pm0.11}$ & \cellcolor{pdc86}$8.33_{\pm0.09}$ & \smash{\begin{tikzpicture}[y=\perdsbar,baseline={([yshift=-\perdsrowtop]current bounding box.north)}]
\fill[pdc0] (0,0.0000) rectangle (0.18,0.0833);
\fill[pdc1] (0,0.0833) rectangle (0.18,0.1667);
\fill[pdc2] (0,0.1667) rectangle (0.18,0.2500);
\fill[pdc3] (0,0.2500) rectangle (0.18,0.3333);
\fill[pdc4] (0,0.3333) rectangle (0.18,0.4167);
\fill[pdc5] (0,0.4167) rectangle (0.18,0.5000);
\fill[pdc6] (0,0.5000) rectangle (0.18,0.5833);
\fill[pdc7] (0,0.5833) rectangle (0.18,0.6667);
\fill[pdc8] (0,0.6667) rectangle (0.18,0.7500);
\fill[pdc9] (0,0.7500) rectangle (0.18,0.8333);
\fill[pdc10] (0,0.8333) rectangle (0.18,0.9167);
\fill[pdc11] (0,0.9167) rectangle (0.18,1.0000);
\draw[line width=0.3pt] (0,0) rectangle (0.18,1);
\node[anchor=north west,font=\tiny,inner sep=1pt] at (0.18,1) {14.38};
\node[anchor=west,font=\tiny,inner sep=1pt] at (0.18,0.5) {7.19};
\node[anchor=south west,font=\tiny,inner sep=1pt] at (0.18,0) {0.00};
\end{tikzpicture}} \\
1 & \cellcolor{pdc27}$0.00_{\pm0.00}$ & \cellcolor{pdc31}$0.11_{\pm0.00}$ & \cellcolor{pdc87}$11.23_{\pm0.17}$ & \cellcolor{pdc88}$9.16_{\pm0.19}$ &  \\
5 & \cellcolor{pdc27}$0.00_{\pm0.00}$ & \cellcolor{pdc89}$0.50_{\pm0.03}$ & \cellcolor{pdc90}$9.87_{\pm0.19}$ & \cellcolor{pdc91}$10.42_{\pm0.36}$ &  \\
15 & \cellcolor{pdc27}$0.00_{\pm0.00}$ & \cellcolor{pdc92}$1.00_{\pm0.08}$ & \cellcolor{pdc93}$8.59_{\pm0.25}$ & \cellcolor{pdc94}$11.42_{\pm0.65}$ &  \\
50 & \cellcolor{pdc27}$0.00_{\pm0.00}$ & \cellcolor{pdc95}$0.70_{\pm0.06}$ & \cellcolor{pdc96}$6.02_{\pm0.30}$ & \cellcolor{pdc29}$14.38_{\pm1.74}$ &  \\
200 & \cellcolor{pdc27}$0.00_{\pm0.00}$ & \cellcolor{pdc27}$0.03_{\pm0.00}$ & \cellcolor{pdc1}$1.81_{\pm0.52}$ & \cellcolor{pdc97}$3.37_{\pm0.72}$ &  \\
final & \cellcolor{pdc27}$0.00_{\pm0.00}$ & \cellcolor{pdc28}$0.04_{\pm0.01}$ & \cellcolor{pdc98}$3.10_{\pm1.01}$ & \cellcolor{pdc99}$2.15_{\pm0.67}$ &  \\
\bottomrule
\end{tabular}
\end{table}

\FloatBarrier

\end{document}